\documentclass[letterpaper,onecolumn,11pt]{quantumarticle}
\pdfoutput=1
\usepackage[utf8]{inputenc}

\usepackage{xcolor}
\usepackage{amsmath,amsthm,amsfonts,amssymb,braket}
\usepackage[pdftex,colorlinks=true,linkcolor=blue,citecolor=darkred,urlcolor=darkblue]{hyperref}
\usepackage[all]{xy}
\usepackage{enumitem}	
\usepackage{multirow}
\usepackage{fullpage}
\usepackage{mathtools}
\usepackage[numbers]{natbib}

\usepackage{thmtools}

\newtheorem{lemma}{Lemma}[section]
\newtheorem{theorem}[lemma]{Theorem}

\newtheorem{corollary}[lemma]{Corollary}

\newtheorem{proposition}[lemma]{Proposition}
\newtheorem{proposition-definition}[lemma]{Proposition-Definition}
\theoremstyle{definition}
\newtheorem{definition}[lemma]{Definition}

\newcommand{\CC}{{\mathbb{C}}}

\newcommand{\ZZ}{{\mathbb{Z}}}

\newcommand{\FF}{{\mathbb{F}}}

\newcommand{\calA}{{\mathcal{A}}}
\newcommand{\calB}{{\mathcal{B}}}
\newcommand{\calC}{{\mathcal{C}}}

\newcommand{\calE}{{\mathcal{E}}}
\newcommand{\calF}{{\mathcal{F}}}

\newcommand{\calL}{{\mathcal{L}}}

\newcommand{\calP}{{\mathcal{P}}}

\newcommand{\calS}{{\mathcal{S}}}

\newcommand{\calX}{{\mathcal{X}}}
\newcommand{\calY}{{\mathcal{Y}}}

\renewcommand{\Pr}{\operatorname{Prob}}		
\newcommand{\one}{{\mathbf{1}}}

\newcommand{\ISG}{\mathrm{ISG}}
\newcommand{\stepR}{\mathbf{R}}		
\newcommand{\stepG}{\mathbf{G}}		
\newcommand{\stepB}{\mathbf{B}}		

\newcommand{\half}{{\frac{1}{2}}}
\newcommand{\thalf}{{\tfrac{1}{2}}}

\newcommand{\lgc}{\calL^{/\calS}}

\newcommand{\qca}{\alpha}
\newcommand{\qcaR}{\beta}
\newcommand{\barqca}{\bar{\alpha}}

\newcommand{\circuit}[1]{{\mathfrak{#1}}} 

\newcommand{\dd}{{\mathsf{d}}} 

\newcommand{\ii}{{\mathbf{i}}}
\newcommand{\id}{{\mathbf{id}}}
\newcommand{\MCAlg}{{Majorana chain algebra}}
\newcommand{\MCAlgEq}{{Majorana chain algebra~\eqref{eq:anomalous_1D_L}}}

\newcommand{\glued}{{\mathrm{glued}}}
\newcommand{\strip}{{\mathrm{strip}}}

\newcommand{\ot}{\leftarrow} 

\DeclareMathOperator{\dom}{dom}

\DeclareMathOperator*{\im}{{im}}

\DeclareMathOperator*{\Supp}{{Supp}}

\DeclareMathOperator{\coker}{{coker}}

\DeclareMathOperator*{\ind}{ind}
\DeclareMathOperator*{\Ind}{Ind^{M}}

\DeclareMathOperator{\Cent}{{Cent}}

\DeclarePairedDelimiter{\abs}{|}{|}

\definecolor{darkblue}{rgb}{0.,0.,0.4}
\definecolor{darkred}{rgb}{0.5,0.,0.}
\definecolor{darkpurple}{rgb}{0.5,0.,0.5}
\definecolor{ltgreen}{rgb}{0.1,.59,.43}
\definecolor{orange}{rgb}{1.0, 0.5, 0.0}
\newcommand{\roger}[1]{{\ifmmode \text{\color{darkpurple}\footnotesize (RM) #1} \else {\color{darkpurple}\small (RM) #1} \fi}}
\newcommand{\jh}[1]{{\ifmmode \text{\color{blue}\footnotesize (JH) #1} \else {\color{blue}\small (JH) #1} \fi}}
\newcommand{\dave}[1]{{\ifmmode \text{\color{ltgreen}\footnotesize (DA) #1} \else {\color{ltgreen}\small (DA) #1} \fi}}
\newcommand{\zhi}[1]{{\ifmmode \text{\color{orange}\footnotesize (Zhi) #1} \else {\color{orange}\small (Zhi) #1} \fi}}

\allowdisplaybreaks		

\makeatletter
\def\l@subsubsection#1#2{}
\makeatother
\renewcommand{\nocontentsline}[3]{}
\renewcommand{\tocless}[2]{\bgroup\let\addcontentsline=\nocontentsline#1{#2}\egroup}
\makeatletter
\newcommand{\insertlabel}[2]{
  \protected@edef\@currentlabelname{#2}
  \label{#1}
}
\makeatother

\begin{document}
\title{Measurement Quantum Cellular Automata and Anomalies in Floquet Codes}

\author{David Aasen}
\affiliation{Microsoft Quantum, Station Q, Santa Barbara, California, USA}
\author{Jeongwan Haah}
\affiliation{Microsoft Quantum, Station Q, Santa Barbara, California, USA}
\affiliation{Microsoft Quantum, Redmond, Washington, USA}
\author{Zhi Li}
\affiliation{Perimeter Institute for Theoretical Physics, Waterloo, Ontario, Canada}
\author{Roger S. K. Mong}
\affiliation{Department of Physics and Astronomy, University of Pittsburgh, Pittsburgh, Pennsylvania, USA}

\begin{abstract}
We investigate the evolution of quantum information under Pauli measurement circuits.
We focus on the case of one- and two-dimensional systems, which are relevant to the recently introduced Floquet topological codes.
We define local reversibility in context of measurement circuits,
which allows us to treat finite depth measurement circuits on a similar footing to finite depth unitary circuits.
In contrast to the unitary case, a finite depth locally reversible measurement circuit 
can implement a translation in one dimension.
A locally reversible measurement circuit in two dimensions 
may also induce a flow of logical information along the boundary.
We introduce ``measurement quantum cellular automata'' which unifies these ideas
and define an index in one dimension to characterize the flow of logical operators.
We find a $\mathbb{Z}_2$ bulk invariant for two-dimensional Floquet topological codes 
which indicates an obstruction to having a trivial boundary.
We prove that the Hastings--Haah honeycomb code belongs to a class with such obstruction,
which means that any boundary must have either nonlocal dynamics, period doubled, 
or admits anomalous boundary flow of quantum information.
\end{abstract}

\maketitle

\tableofcontents

\section{Introduction}

Quantum measurements are critical to virtually any aspect of quantum information processing.
Besides the evident necessity of measurements to read out information in a quantum state,
they are valuable for applications such as 
entanglement and magic state distillation~\cite{Knill2004a,Bravyi2005},
driving entanglement phase transitions~\cite{Fisher2023},
or state preparation~\cite{Gottesman1999,BriegelRaussendorf,Nielsen2003}.
Measurements are also indispensable in quantum error correction,
allowing for the delocalization and robust storage of logical information.
To this end, the sequence of measurements must be carefully designed 
to isolate errors while preserving the integrity of the quantum information.

Meanwhile, new classes of phases of unitary quantum dynamics have been discovered,
which cannot be realized by static Hamiltonians~\cite{Harper2019, Rudner2020}.
For example, periodically driving a noninteracting system can change the topology of the band structure, leading to new topological phases~\cite{Takashi2009,Kitagawa2010,Lindner2011}.
Explicit examples have been constructed where nontrivial boundary dynamics is manifest~\cite{Rudner2013, PoFidkowski2016, Harper2017, fermionGNVW1}.
As new types of emergent phenomena and topological order are discovered in periodically driven systems,
it is natural to ask of their analogues in periodic measurement dynamics.

In this paper, we investigate periodic sequences of measurements and their effect on the dyanmics of quantum information.
Our investigation is, in part,
motivated by a new class of recently developed quantum error correcting codes~\cite{Hastings2021, Aasen2022, Davydova2022, Kesselring2022, ZhangXCube2022}
with both nontrivial Floquet and topological characteristics.
Some of these Floquet codes only require neighboring pairwise measurements,
	which allows for easier implementation than their static counterparts 
	that involve joint measurement of three or more qubits.
Such measurements are expected to be natively available in Majorana-based quantum hardware~\cite{Hastings2021,Paetznick2023}.

In particular, the Hastings--Haah honeycomb code (HH code)~\cite{Hastings2021}
	implements a nontrivial transformation of the code space every measurement period.
The HH code is defined on a two-dimensional plaquette three-colorable lattice
with a period-three measurement schedule (i.e., three measurement steps per cycle) in the absence of boundaries. 
However, it has been difficult to introduce boundaries and construct planar realizations of the HH code.
Refs.~\cite{Hastings2021, Haah2022, Gidney_2022, Paetznick2023} 
maintain logical qubits by doubling the periodicity of the measurement schedule to six steps.
In constrast, Ref.~\cite{Vuillot2021} showed a period-three planar circuit,
but the resulting dynamics is nonlocal and is not fault-tolerant against certain errors.
A natural question follows: 
is the period doubling a fundamental aspect of the HH code with boundaries?
To this end, we show that either the boundary of the HH code must be period doubled, or admit dynamical boundary degrees of freedom, as long as locality is preserved.
This is reminiscent of the bulk-boundary correspondence for certain classes of gapped symmetry protected/enriched topologically ordered states,
	where the boundary must either be symmetry-breaking or hosts anomalous gapless degrees of freedom.
The main result of this paper is to characterize such boundary anomalies for measurement circuits.

Quantum cellular automata (QCA) are an abstraction of local unitary dynamics.
A QCA is a locality preserving $\ast$-automorphism of a local operator algebra; 
simply put, it takes local operators on a product Hilbert space to nearby local operators.
The unitarity of a QCA means that it is invertible; its effect can be reversed by another QCA.
While finite-depth unitary circuits (FDUC or shallow unitary circuits) 
provide a class of QCA, not all QCA can be implemented as a FDUC~%
	\cite{GNVW, nta3, FreedmanHastings2019QCA, Wilbur2022, Chen2021, Haah2022b}.
In particular, the action of shifting all the qubits on an infinite chain by one site (say, to the left) is a QCA, but cannot be implemented as a shallow unitary circuit.
In one dimension, QCA are classified by the GNVW invariant~\cite{GNVW} up to FDUC, a rational index which captures the flow of quantum information along a chain.
The GNVW invariant is akin to a chiral central charge
in the sense that while it must vanish for any circuit in strictly one dimension,
	it can be nonzero at the boundary of a two-dimensional shallow unitary circuit~\cite{PoFidkowski2016, Harper2017, Harper2019, ZhangLevin2022}.

\begin{figure}[ht]
	\centering
	\includegraphics[width=.9\linewidth]{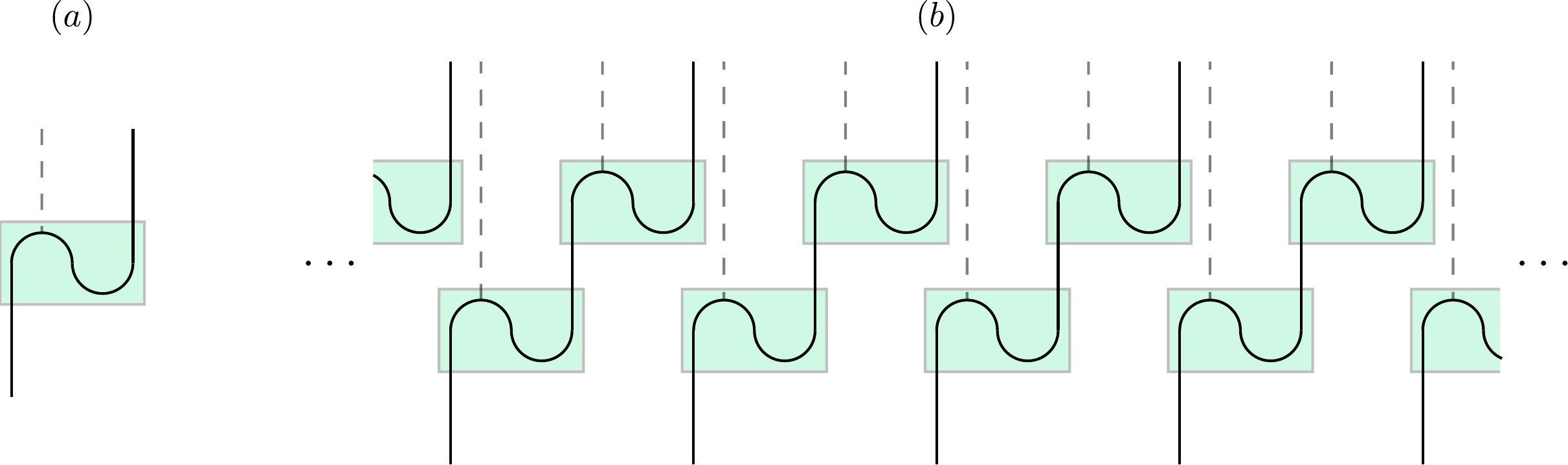} 
   \caption{%
	\textbf{(a)} A schematic for {\bf quantum teleportion} of a single qubit.
	Broken into steps (from bottom to top),
		we take an unknown qubit (left), create a Bell pair (cup),
		measure the unknown qubit with one of the qubit of the Bell pair (cap),
		and the remaining qubit (right) carries the original quantum information,
		up to Pauli corrections.
	\textbf{(b)} A schematic for a measurement circuit which implements a {\bf qubit-translation} along an infinite chain.
	The circuit, acting on a chain of qubits along with ancilla between qubits, comprises of two (composite) steps.
	First the qubit on each site is teleported to its adjacent ancilla, followed by another teleportation of the ancilla to the next qubit.
	This results in a one-dimensional locally reversible measurement cycle which have a nontrivial MQCA index.
	In both subfigures, solid lines indicate flow of quantum information,
		the dashed lines are classical information which results from the measurements.
	} \label{fig:teleportation_schematic}
\end{figure}

We consider analogous ideas for measurements circuits.
In light of the fact that a finite-depth measurement circuit can generate long ranged correlations,
	we introduce a {\bf local reversibility} condition for Pauli measurements circuits,
	which guarantees that quantum information is preserved and remains local over a sequence of measurements.
We define a {\bf locally reversible measurement cycle} (LRMC) as an analogue to shallow unitary circuits.
	It is a cyclic sequence of measurements which after every period leaves the logical space unchanged while transforming the logical operators.
We encapsulate such transformations in a notion of {\bf measurement quantum cellular automata} (MQCA),
	which requires that local logical operators map to nearby local logical operators.
We then define an {\bf MQCA index} which characterizes the flow of information along a one-dimensional strip or chain.

A crucial difference between the measurement versus unitary case is that the degrees of freedom do not necessarily form a product operator algebra.
In the standard QCA formulation, the operator algebra is a tensor product of simple local algebras, 
each living on a lattice site.
Here, the MQCA is defined over the logical operator algebra of a Pauli stabilizer/gauge group,
	which in many interesting cases do not have a tensor product structure.
Some of our examples occur at the boundary of a two-dimensional topological stabilizer group.
The unitary QCA, if Clifford, is a special case of the MQCA.
Hence, the MQCA covers a richer set of dynamics than its unitary counterpart.
As such, the MQCA index is a generalization of the GNVW index for Clifford quantum cellular automata.

While shallow unitary circuits must have trivial GNVW invariant, 
the analogous case is not true for locally reversible measurement cycles.
A finite depth, locally reversible measurement cycle can implement qubit translations along an infinite chain.
Our example, illustrated in Fig.~\ref{fig:teleportation_schematic}, consists of repeated application of the quantum teleportation protocol~\cite{BBCJPW}.
The effect of the circuit is to teleport the qubit at site $j$ to that at $j+1$; this circuit has the MQCA index~$1$.
The first main result of this paper is that the MQCA index of any one-dimensonal system is always an integer.
In terms of logical algebras this means that every relevant degree of freedom can be characterized by a pair of anticommuting operators, say $X_j$ and $Z_j$ localized around a site~$j$.
The translation circuit transforms these operators, up to a sign, to $X_{j+1}$ and $Z_{j+1}$ respectively.
The MQCA index precisely measures the flow of these operators.

\begin{figure}[ht]
	\centering
	\includegraphics[width=0.75\linewidth]{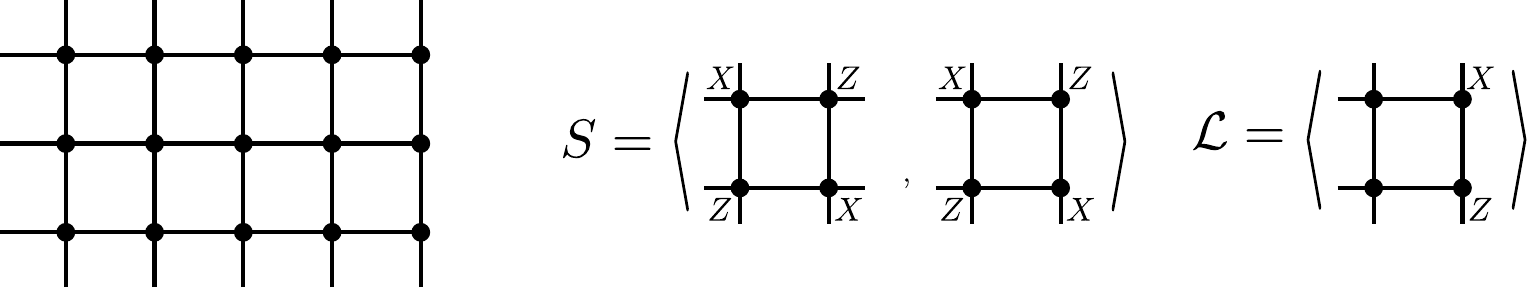} 
   \caption{%
	Wen's plaquette model with a vertical boundary.
	(Left) A square lattice with a vertical boundary to its right.
	(Middle) Stabilizers of the model, which include both bulk terms and terms touching the boundary.
	(Right) Generators of the logical algebra, which are 2-body terms which lives on the boundary.
   } \label{fig:WenPlaq_boundaryR}
\end{figure}

Furthermore, we also consider two-dimensional Floquet measurement circuits and induced MQCA along the boundary.
Our first example is based on Wen's plaquette model~\cite{WenP}, realizing the toric code topological order, stabilized by four-body operators.
Consider a vertical boundary of the model (which results from simply dropping terms which straddles the boundary), shown in Fig.~\ref{fig:WenPlaq_boundaryR}.
The logical algebra is generated by two-body operators along the edge, with one logical generator per unit cell.
When we implement the translation circuit above, one along every column, every site of the model shifts upward by one unit.
While the logical algebra remains the same, i.e., the translation implements an automorphism of the algebra,
	each individual logical operator shifts by one, realizing a nontrivial MQCA along the boundary.
Because there is only one logical generator per unit cell (as opposed to two per unit cell in the 1d chain),
	this MQCA transports the equivalent of a half-qubit of information along the boundary.
The MQCA index is $\half$, indicative of a boundary anomaly that cannot be realized in a purely 1d circuit.

The boundary MQCA index depends on the specific choice of circuit termination along the edge;
	one may alter the index by an integer by attaching a pure 1d circuit to the edge.
Our second main result is that the fractional portion of the MQCA index is independent of the chosen boundary termination,
	and is a property the bulk circuit.
That is, two-dimensional topological locally reversible measurement cycles are characterized by a $\thalf \ZZ / \ZZ \cong \ZZ_2$ index,
	which provides an obstruction to constructing a smooth boundary with no nontrivial logical operators.%
	\footnote{In this paper, we do not prove that this $\ZZ_2$ invariant is the \emph{only} obstruction to having a trivial boundary.}
In particular, the HH code belongs to the nontrivial class.
This allows us to say that any boundary of the HH code must (a) be not locally reversible, or (b) be period doubled (or $6N$ steps) or nonperiodic, or (c) admit a boundary with an anomalous MQCA index.

The organization of the paper is as follows.
\begin{itemize}[itemsep=1pt, parsep=0pt, topsep=2pt]
\item \S\ref{sec:reversibility}
	defines {\bf reversibility} and {\bf local reversibility} in the context of measurement circuits.
	These concepts are analogous to unitarity and local unitarity for quantum circuits (without measurements).
	The section lays the foundations for this work and is filled with examples and nonexamples.
\item \S\ref{sec:LRMC}
	delves into the properties of locally reversible measurement dynamics.
	We explain circuit {\bf blendings}, colloquially boundaries between different circuits.
	We define measurement quantum cellular automata (MQCA).
\item \S\ref{sec:MQCA_1D}
	defines an MQCA index, quantifying the ``flow'' of information in a circuit.
	We prove that every one-dimensional circuit must have an integer index,
	and that two-dimensional topological locally reversible measurement cycles\footnote{More precisely, any topological LR cycle which admits a boundary with the vacuum (trivial circuit).}
	are characterized by a $\ZZ_2$ invariant.
\item \S\ref{sec:WPT}
	introduces the Wen-plaquette translation (WPT) circuit.
	We construct various boundaries for this model.
\item \S\ref{sec:HHcode}
	reviews the Hastings--Haah honeycomb code (HH code).
	Again we construct various boundaries for this model.
\item
	\S\ref{sec:discussion} contains a brief summary of results and concluding remarks.
	We explore connections between the MQCA index and TQFT, and generalizations of the present work.
	We also discuss potential applications to quantum error correction.
\end{itemize}
Section~\ref{sec:MQCA_1D} does the heavy lifting using mathematics of Pauli groups and Fredholm index.
	While we encourage readers of all expertise to study these details,
	we expect readers to be able to fully understand the contents of our examples \S\ref{sec:WPT} and \S\ref{sec:HHcode}
	with only the background contained in \S\ref{sec:reversibility}.
For reference, a list of the common terms and abbreviations used throughout this paper is given in Table~\ref{tab:terms}.

\begin{table}[t]
\centering
\caption{Common terms and abbreviations used in this paper}
\label{tab:terms}
\begin{tabular}{l @{\quad} r}
\hline
	base stabilizer group / base code  & following Eq.~\eqref{eq:LRMcircuit}
\\
	conjugate bases / conjugate pairs  & Def.~\ref{defn:lrt}
\\
	\textbf{ISG}: instantaneous stabilizer group  & Eq.~\eqref{eq:def_ISG}
\\
	locally reversible (LR)  & Def.~\ref{defn:lrt}
\\
	\textbf{LRMC} / \textbf{LR cycle}: locally reversible measurement cycle & \S\ref{sec:LRMC}
\\
	\quad	topological LRMC  & Def.~\ref{defn:TopologicalCode}
\\
	blending  & \S\ref{sec:blending}
\\
	\quad topological blending  & Def.~\ref{defn:TopologicalCode}
\\
	\quad vacuum blending  & \S\ref{sec:boundaryactions}
\\
	\textbf{MQCA}: measurement quantum cellular automata  & Def.~\ref{defn:MQCA}
\\
	\quad canonical MQCA  & preceeding Thm.~\ref{thm:IntegerIndexOf1dLRMC}
\\
	MQCA index / $\Ind$  & \S\ref{sec:MQCA-def-prop}
\\
	\MCAlg  & Eq.~\eqref{eq:anomalous_1D_L}
\\
	\textbf{WPT}: Wen-plaquette translation model  & \S\ref{sec:WPT}
\\
	\textbf{HH code}: Hastings--Haah honeycomb code  & \S\ref{sec:HHcode}
\\\hline
\end{tabular}
\end{table}

\section{Pauli Measurement dynamics}
\label{sec:reversibility}

In analogy with local unitary circuits,
we consider local {\bf measurement circuits}
which consists of finitely many layers (time steps) of 
simultaneous measurements of nonoverlapping or just commuting local operators.
That is, in each layer, we measure a set of local operators, 
each of which acts on a few neighboring qubits, commuting with all other.
One may consider interleaving local measurements and local unitaries.
However, a pure measurement circuit is no less general than a unitary-measurement interleaved circuit,
since a local unitary only changes the basis of the measurement locally:
\begin{equation}\label{eq-measurementseries}
	\Pi_n U_n \cdots \Pi_2 U_2 \Pi_1 U_1=\tilde{\Pi}_n\cdots \tilde{\Pi}_2\tilde{\Pi}_1 U_n \cdots U_2 U_1 \,,
\end{equation}
where $\tilde{\Pi}_k=U_n\cdots U_{k+1}\Pi_k U_{k+1}^\dagger\cdots U_{n}^\dagger$ can be realized as a measurement circuit.
Hence, the most general short-time dynamics including measurements 
can be thought of as two-stage dynamics,
where in the first stage one applies a unitary circuit
and in second stage one applies a measurement circuit.

In this paper, we consider measurements of Pauli operators over qubits.
A Pauli operator is a tensor product of $2$-by-$2$ Pauli matrices,
such as $X_i$ acting on a qubit~$i$ and $Z_i Z_{i+1}$ acting on two qubits~$i$ and~$i+1$.
Since any Pauli operator~$P$ squares to~$P^2  = \one$,
the associated measurement takes two values~$\pm 1$,
and correspondingly the state is collapsed to
\begin{align}
\ket \psi \mapsto \frac{\one \pm P}{2} \ket \psi
\end{align}
depending on the measurement outcome, 
where we neglected overall normalization 
but wrote the projector as a genuine projector.

An interesting aspect of Pauli measurements is that 
they exhibit the Heisenberg uncertainty principle 
to a maximal degree that we are going to examine more carefully below.
We will first ignore locality of the dynamics,
but look at an abstract situation,
and then tailor our discussion with locality.

\subsection{Reversible measurements}

Suppose we measure a system that is in one of (continuously) many possible quantum states
and obtain an outcome that is independent of the underlying state.
There is no information transfer from the system to the observer,
and the underlying quantum information should be undisturbed.
The outcome being independent of the state 
does not necessarily mean a fixed outcome.
It only means that the outcome is drawn from a fixed probability distribution!

Quantum teleportation is a manifestation of such no-information measurement.
Consider two qubits, where the first qubit contains some state we want to transfer (a logical qubit),
and the second qubit is in an eigenstate of Pauli~$X$.
The teleportation protocol is that we measure two-qubit Pauli~$ZZ$
followed by a single-qubit Pauli~$X$ on the first qubit.
The first measurement outcome~$\pm 1$ by~$ZZ$ must be uniformly random ($\Pr[+] = \Pr[-]$)
because  the initial state~$\ket \psi$ is an eigenstate of an anticommuting operator~$IX$:
\begin{align}
\Pr[+] = \bra{\psi} \frac{\one + ZZ}{2} \ket{\psi}
= \bra{\psi} (IX) \frac{\one + ZZ}{2} (IX) \ket{\psi}
= \bra{\psi} \frac{\one - ZZ}{2} \ket{\psi} = \Pr[-]. \label{eq:RandomOutcome}
\end{align}
For the same reason, the next measurement by~$XI$ has a uniformly random outcome,
and direct calculation shows that the second qubit holds a quantum state
that is different from the original by a Pauli that depends only on the measurement outcomes.
We measured, learned nothing, and hence managed to preserve a quantum state.

To properly describe the transformation of the quantum information in a sequence of measurements,
	we need to track the stabilizer groups and the dynamics of the logical operators.
At each step, the {\bf instantaneous stabilizer group} (ISG) is the set of Pauli operators with definite eigenvalue.
For example, both $\one = II$ and $IX$ are stabilizers of the initial state, hence the initial ISG is $\ISG_0 = \{\one, IX\}$.
After measurement of $ZZ$, the ISG become $\ISG_1 = \{\one, ZZ\}$ (note that $IX$ is no longer a stabilizer).
The final ISG is $\ISG_2 = \{\one, XI\}$.
For step $t$, a {\bf logical operator} is any Pauli operator that commutes with all the elements of $\ISG_t$,
	the set of which is denoted $\ISG_t^\perp$.
They generate actions that can be applied to the state while staying in the stabilizer space.
For example, $\ISG_0^\perp = \{\one,XI,YI,ZI,IX,XX,YX,ZX\}$.
Two logical operators are equivalent with resepect to an ISG if they are related by a stabilizer element,
	e.g.\ $XI \mathrel{\overset{\ISG_0}{\scalebox{2.5}[1]{$\sim$}}} XX$, $YI \mathrel{\overset{\ISG_0}{\scalebox{2.5}[1]{$\sim$}}} YX$,
	as they lead to equilvalent actions when acting on states stabilized by $\ISG_0$.
(Each stabilizer, by definition, is also a logical operator and is equivalent to the identity.)

The \emph{measurement dynamics} of the circuit is characterized by the evolution of the logical operators through each step.
Logical operators evolve from step~$t$ to step~$t+1$ via the following procedure:
	for a logical operator $L \in \ISG_t$, find an equivalent operator
	$L' \mathrel{\overset{\ISG_t}{\scalebox{2.5}[1]{$\sim$}}} L$ such that $L'$ also belongs to $\ISG_{t+1}$,
	then $L'$ become the logical operator at step~$t+1$;
	this amounts to finding a stabilizer $S \in \ISG_t$ such that $LS \in \ISG_{t+1}^\perp$.
For instance, consider the logical operator $XI$ in the initial state.
$XI$ is not a logical operator of $\ISG_1$ (failing to commute with $ZZ$),
	but as $XI \mathrel{\overset{\ISG_0}{\scalebox{2.5}[1]{$\sim$}}} XX$ and $XX$ does commute with $ZZ$,
	$XI \mapsto XX$ through the process of measuring $ZZ$.
Measuring $XI$ next, the logical $XX$ (at $\ISG_1$) \emph{does} commute with the second measurement,
	and hence the logical operator remains as $XX$.
Through our sequence of stabilizers $\ISG_0 \to \ISG_1 \to \ISG_2$, the initial logical operators transform as
\begin{align}
	(XI,YI,ZI) \xmapsto{\text{measure $ZZ$}} (XX,YX,ZI) \xmapsto{\text{measure $XI$}} (XX,XY,IZ) \mathrel{\overset{\ISG_2}{\scalebox{2.8}[1]{$\sim$}}} (IX,IY,IZ) .
\end{align}
Indeed, operators acting on the first qubit is ``teleported'' to the operators on the second qubit.

After the teleporation, 
the meaningful part of the quantum state (the logical qubit)
is in the second qubit and the first qubit holds an eigenstate of Pauli~$X$.
So, the situation is the same as before the teleportation and 
we may teleport the quantum state back to the first qubit.
The protocol is to measure~$ZZ$ and then~$IX$.
We observe a forward-backward symmetry 
if we consider the Pauli operators that stabilizes the instantaneous state:
\begin{align}
\xymatrix{
	\braket{ \pm IX } \ar@/^/[rr]^{1_\text{forward}} & & \langle \pm ZZ \rangle \ar@/^/[rr]^{2_\text{forward}} \ar@/^/[ll]^{2_\text{backward}} & & \langle \pm XI \rangle \ar@/^/[ll]^{1_\text{backward}}
}. \label{eq:ModernTeleport}
\end{align}
Here, the operator in each chevron $\braket{\bullet}$ takes eigenvalue~$+1$ on the underlying state.
Since measurement outcomes are completely random, the signs must be regarded as irrelevant,
though one should note that to achieve proper teleportation 
those signs must be used to apply Pauli ``corrections.''

Actually, the very fact that a Pauli measurement is
completely random means that the measurement can be undone 
modulo Pauli corrections depending on measurement outcomes:
\begin{align}
\xymatrix{
	\langle \pm IX \rangle \ar@/^/[rr]^{\text{measure }ZZ} & & \langle \pm ZZ \rangle \ar@/^/[ll]^{\text{measure }IX}
} .
\end{align}
The intermediate $ZZ$ measurement in~\eqref{eq:ModernTeleport} is crucial in the teleportation protocol.
If $XI$ was to be measured directly as a first step (while the system is stabilized by $IX$), the quantum information stored in the first qubit would be destroyed.
This is a nonreversible action which we avoid when discussing dynamics of quantum information.
The underlying reason is that $IX$ and $XI$ commute and are independent, and so a direct transition from the first to last step is impossible.

We will find it useful to abstract the situation.
For a group, a {\bf basis} is a nonredundant generating set of the group.

\begin{proposition}\label{thm:Reversibility}
Let $\calA$ and $\calB$ be abelian groups of Pauli operators under multiplication 
on a finite set of qubits  (Pauli stabilizer groups).
Let $\calS = \calA \cap \calB$.
The following are equivalent.
\begin{enumerate}
	\item[(a)] (anticommuting conjugate) 
	There exist bases $\{A_i \calS\}$ of~$\calA/\calS$ and $\{B_j \calS\}$ of~$\calB/\calS$
	such that for each $A_{i'}$ there is $B_{j'}$ that commutes with all~$A_i$ but $A_{i'}$,
	and for each $B_{j'}$ there is $A_{i'}$ that commutes with all~$B_j$ but $B_{j'}$.

	\item[(b)] (no enlargement of stabilizer groups) If $A \in \calA$ commutes with all elements of~$\calB$, then $A \in \calB$;
	and if $B \in \calB$ commutes with all elements of~$\calA$, then $B \in \calA$.
	\item[(c)] (nonsingular commutation relations) 
	For any bases $\{A_a \calS \,|\, a = 1,2,\ldots, m \}$ of~$\calA/\calS$ and $\{B_b \calS \,|\, b = 1,2,\ldots, n \}$ 
	of~$\calB/\calS$,
	the binary matrix~$M$ with entries $M_{ab} = 0$ if $A_a$ and $B_b$ commute and $M_{ab} = 1$ otherwise 
	(for $ 1 \le a \le m$ and $1 \le b \le n$),
	is invertible over~$\FF_2$.

	\item[(d)] (shared logical operators) There exists a group~$\calL$ of Pauli operators
	that elementwise commutes with~$\calA$ and with~$\calB$ 
	such that $\calL\calA$ is the group of all logical operators for~$\calA$
	and $\calL\calB$ is that for~$\calB$.
	
	\item[(e)] (resolution of anticommutation) For any Pauli operator~$O$ that commutes with $\calS$ elementwise,
	there exist $A \in \calA$ and $B \in \calB$ such that $OA$ commutes with $\calB$ elementwise 
	and $OB$ commutes with $\calA$ elementwise.
\end{enumerate}
\end{proposition}

Note that the conditions~(b) and~(d) do not require any basis
and in particular do not need any information about~$\calS$.

\begin{proof}
See~\ref{thm:Reversible-LinearVersion}.
\end{proof}

\begin{definition}
A {\bf reversible} pair of Pauli stabilizer groups is one 
that satisfies any one (and hence all) of the conditions of~\ref{thm:Reversibility}.
\end{definition}

We interpret these characterizations as follows.
Let two Pauli stabilizer groups~$\calA$ and $\calB$ be a reversible pair.
Suppose we have some logical state encoded in a Pauli stabilizer code~$\calA$
and then we measure stabilizers of~$\calB$.
For each stabilizer measurement $B \in \calB$, either it commutes with all of $\calA$ or it anticommute with at least an element of $\calA$.
If $B$ anticommutes with an element of $\calA$%
---guaranteed for $B \notin \calA$ by condition~(a)---%
then every outcome of the measurements will be completely random in the sense that any allowed results appear with equal probabality (see~\eqref{eq:RandomOutcome}),
hence revealing no information about the logical state.
If $B$ were to commute with every element of $\calA$, then 
it still does not reveal any logical information 
because the measurement is merely a consistency check: $B \in \calA$. This is~(b).
The condition~(c) would be useful because it is an efficiently verifiable condition.
The existence of the shared logical operators, the condition~(d), tells us how to
access the logical qubit after the measurement:
use the operators of~$\calL$ which is invariant under the measurement transition.
The condition~(d) also means that for any logical operator~$P$ of~$\calA$, 
one can dress it by some $A\in\calA$ such that~$P A$ also commutes with~$\calB$. 
Given~$P$, such a product $P A$ is unique modulo~$\calA \cap \calB$.
The condition~(e) gives a minimal condition for an operator to turn to a logical operator;
namely, the operator has to commute with the common stabilizers.

In the example at the beginning of this section, the groups $\braket{IX}$ and $\braket{ZZ}$ form a reversible pair,
	as do $\big({\braket{ZZ}},{\braket{XI}}\big)$, illustrate via~\eqref{eq:ModernTeleport}.
Notably, $\braket{IX}$ and $\braket{XI}$ do not form a reverisble pair.
(Throughout, we use $\braket{P_1, P_2, \dots}$ to denote the group generated by Pauli operator(s) $\{P_i\}$.)

\subsection{Example: iterated but instantaneous quantum teleportation}
\label{sec:teleportation}

The above teleportation example in~\eqref{eq:ModernTeleport} 
is an analog (that is perhaps more modern) of the original quantum teleportation protocol~\cite{BBCJPW}.
Let us review and chain up the orignal protocol.
Alice holds a data (logical) qubit and another qubit that is in the Bell state 
with a qubit of Bob's.
Alice measures her two qubits in the Bell basis.
Instantly the Bob's qubit is set to exactly the same state as Alice's logical qubit,
up to a Pauli correction that depends on the Alice's measurement outcome.
Bob could have held a half of another Bell pair shared with Charlie,
and measured his two qubits in the Bell basis.
Then, up to a Pauli correction, Charlie's qubit would be set to Alice's logical qubit.
It is interesting that Alice's measurement and Bob's commute,
so the time ordering of the measurements is irrelevant.%
\footnote{%
though an interpretation would depend on it;
in one ordering the logical qubit is teleported in sequence,
while in the other ordering the logical qubit is teleported through a derived Bell pair between Alice and Charlie.
}
It is straightforward how to iterate this teleportation protocol
for $n + 1 \ge 2$ parties that stand on a line.
The first party holds a logical qubit, indexed~$1$,
and each party~$j$ holds two qubits indexed~$2j-1$ and~$2j$ 
except for the last one which only holds qubit~$2n+1$.
Parties~$j$ and $j+1$ share a Bell pair between qubits~$2j$ and~$2j+1$
and each party~$j$ makes a Bell measurement on the two qubits~$2j-1$ and~$2j$  where $1 \le j \le n$.
The logical qubit at the first party is teleported to the last party instantaneously,
up to a Pauli correction that depends on all measurement outcomes.
\begin{figure}[ht]
	\centering
	\includegraphics[width=.4\linewidth]{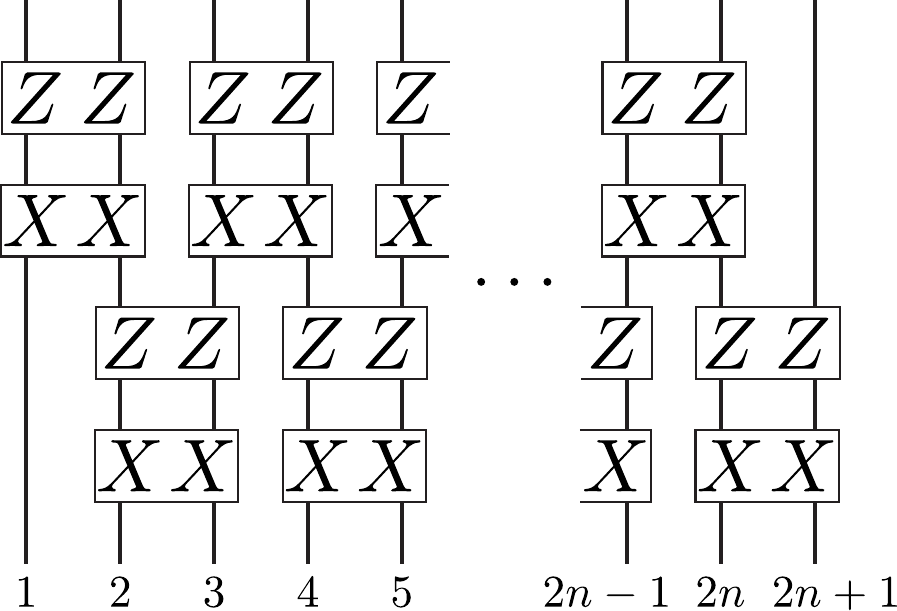}
	\caption{%
	Circuit picture for the teleportation protocol discussed in the main text. 
	The first two measurements (read from the bottom) prepare bells states between parties $2n$ and $2n+1$. 
	The second two measurements measures parties $2j-1$ and $2j$ in the Bell basis,
		which results in teleportation of the state on qubit $1$ to qubit $2n+1$.
	This circuit is reversible, but not locally reversible (as $n$ become large while fixing $\ell$).
	}
	\label{fig:Teleportation}
\end{figure}

The state of the initial shared Bell pairs is
a common eigenstate of Pauli operators of
\begin{subequations}
\begin{align}
	\calA &= \braket{ X_{2j} X_{2j+1} \,, Z_{2j} Z_{2j+1} ~|~ j = 1, 2,\ldots, n }.
\end{align}
The Bell measurements are equivalent to measuring every element of an abelian Pauli group
\begin{align}
	\calB &= \braket{ X_{2j-1} X_{2j} \,, Z_{2j-1} Z_{2j}~|~ j = 1,2,\ldots, n }.
\end{align}
\end{subequations}
We check that these two groups form a reversible pair.
Take a group
\begin{align}
\calL = \braket{ X_1 X_2 X_3 \cdots X_{2n+1} \,, Z_1 Z_2 Z_3 \cdots Z_{2n+1} }.
\end{align}
It is clear that $\calL \calA$ is the group of all logical operators of~$\calA$,
and $\calL \calB$ is the group of all logical operators of~$\calB$.
Therefore, by~\ref{thm:Reversibility}(d), $\calA$ and $\calB$ are a reversible pair.

Equivalently, we can check the condition~(b).
Since $\calA$ is a direct product of nonoverlapping Pauli subgroups,
each containing exactly two generators for a Bell pair,
any hypothetical nonidentity element $A \in \calA$ that commutes with~$\calB$
must have a left-most nonidentity tensor factor on an even-indexed qubit, say~$2j$.
Then, the commutation relation with $X_{2j-1}X_{2j}$ and~$Z_{2j-1}Z_{2j}$ of~$\calB$
forces this tensor factor to be the identity, a contradiction.
Hence, the condition~(b) is satisfied.

Note that as the name ``reversible pair'' suggests,
the iterated teleportation in the reverse direction from qubit~$2n+1$ to~$1$ 
is implemented by measuring~$\calA$ on a state stabilized by~$\calB$.

The other equivalent conditions~\ref{thm:Reversibility}(a,c) are instantiated as follows.
Some generating sets that satisfy~(a) will be found by
diagonalizing the matrix~$M$ in~(c).
Let us order the generators of~$\calA$ and $\calB$ as
\begin{align} \begin{split}
	\calA &= \braket{ X_2 X_3 \,, Z_2 Z_3 \,, X_4 X_5 \,, Z_4 Z_5 \,, \ldots },
\\	\calB &= \braket{ X_1 X_2 \,, Z_1 Z_2 \,, X_3 X_4 \,, Z_3 Z_4 \,, \ldots }.
\end{split} \end{align}
Then, the $2n \times 2n$ binary commutation matrix is
\begin{align}
\renewcommand{\arraystretch}{0.9}
M = \begin{pmatrix}
N & N &       \\
  & N & N     \\
  &   & \ddots & \ddots \\
  &    &   & N & N \\
    &    &   &  & N \\
\end{pmatrix} \text{ where } 
N = \begin{pmatrix}
0 & 1 \\ 1 & 0
\end{pmatrix}.
\end{align}
Choosing a different generating set for~$\calA$ amounts to a row operation,
while that for~$\calB$ does to a column operation.
Hence, we are free to manipulate~$M$ by any invertible matrices on the left and right separately.
An easy choice is by the inverse of~$M$ on the right:
\begin{align}
\renewcommand{\arraystretch}{0.9}
\begin{pmatrix}
N & N &       \\
  & N & N     \\
  &   & \ddots & \ddots \\
  &    &   & N & N \\
    &    &   &  & N \\
\end{pmatrix}
\begin{pmatrix}
N & N &  \cdots & N & N \\
 & N &  \cdots & N & N \\
 & &  \ddots & \vdots & \vdots \\
 & & & N & N \\
 & & & & N \\
\end{pmatrix}
=
\begin{pmatrix}
\,I\, &  &  &   &  \\
 & \,I\, &  &   & \\
 &  &  \ddots &  &  \\
 &  & &   \,I\, &  \\
  & &  &   & \,I\, \\
\end{pmatrix} .
\end{align}
The upper triangular matrix combines many generators of~$\calB$ into a new generator.
This means that a local stabilizer of~$\calA$ is conjugate to a nonlocal generator of~$\calB$ in the sense of~(a).
We cannot yet conclude that this nonlocality is unavoidable;
however, we will confirm that the nonlocality follows
because the quantum teleportation is essentially ``instantaneous.''

\subsection{Locally reversible transitions}
\label{sec:lrtrans}

By a lattice of qubits we mean a set of qubits where the distance between two qubits is defined.
For any $\ell \geq 0$, an $\ell$-ball is a metric ball of diameter~$\ell$.
Here, the diameter of a region is the largest distance between two points in the region.
An $\ell$-neighborhood of a region $R$ is the set of all points within distance $\ell$ from the region: $R^{+\ell} = \{ x ~|~ \exists a \in R, \mathrm{distance}(a,x) \leq \ell \}$.
For an operator $O$, its support, denoted by~$\Supp(O)$, 
is the smallest region for which $O$ acts by identity on the complement.
Let $\calP$ be the group of Pauli operators with finite support.

We say that a subgroup~$\calA$ of~$\calP$ is {\bf $\ell$-locally generated} 
above a subgroup~$\calS$,
if there is a set~$\{A_i \in \calA\}$ of $\ell$-local operators
such that $\calA$ is generated by $\{A_i\}$ and~$\calS$.

\begin{definition}\label{defn:lrt}
A pair $(\calA,\calB)$ of abelian subgroups of the group of all finitely supported Pauli operators 
is {\bf locally reversible} if
there exist $\ell$-local generating sets~$\{A_i \in \calA\}$ and~$\{B_j \in \calB\}$ above~$\calS=\calA\cap\calB$ 
such that for each $A_{i'}$ there is $B_{j'}$ that commutes with all~$A_i$ but $A_{i'}$,
and for each $B_{j'}$ there is $A_{i'}$ that commutes with all~$B_j$ but $B_{j'}$.
We say the anticommuting elements $A_{i'}$ and $B_{j'}$ are {\bf conjugate};
we call the pair $\{A_i \in \calA\}$ and $\{B_i \in \calB\}$ an {\bf $\ell$-local conjugate bases} for $(\calA,\calB)$.
\end{definition}

A {\bf locally finite product}~$\hat A$ from~$\calA$ is a subset of~$\calA$
where there are only finitely many elements whose supports include any given point.
The product of all elements from such a set may not be well defined,
especially if $\calP$ is infinite, and hence may not be an element of~$\calP$,
but the infinite product is well defined locally as a quantum circuit of depth~1 
and in particular does have a well-defined action on any finitely supported Pauli operator by conjugation.
The action by conjugation introduces a sign to any finitely supported Pauli operator.
We identify any two locally finite products from~$\calA$ if they give the same action by conjugation.
The set of all locally finite products of~$\calA$ is denoted by~$\hat \calA$,
which is a group.

\begin{proposition}\label{thm:LocallyReversible}
	Suppose $(\calA,\calB)$ is a locally reversible pair.
	For any~$\hat P \in \hat \calP$ supported on~$R$ that commutes with $\calS = \calA\cap\calB$ elementwise,
	there exist locally finite products $\hat A \in \hat \calA$ and $\hat B \in \hat \calB$, both supported on~$R^{+2\ell}$,
	such that $\hat A \hat P$ commutes with~$\calB$ elementwise and $\hat B \hat P$ commutes with~$\calA$ elementwise.
\end{proposition}

\begin{proof}
Take a conjugate $\ell$-local bases for~$(\calA, \calB)$.
Look at all the basis elements of~$\calA$ that anticommute with~$\hat P$,
	they must each be supported on $R^{+\ell}$.
Multiply $\hat P$ by the conjugate basis elements of~$\calB$.
These conjugate elements of~$\calB$ are supported on $R^{+2\ell}$
and the set is locally finite because in any ball there are only finitely many independent elements.
\end{proof}

\begin{corollary}\label{thm:LP}
For any~$\hat  P \in \hat \calP$ supported on~$R$ that commutes with~$\calA$ elementwise,
there exists a locally finite product~$\hat A \in \hat \calA$ supported on~$R^{+2\ell}$
such that $\hat A \hat P$ commutes with~$\calB$ elementwise.
Such a product~$\hat A \hat P$ is unique up to~$\widehat{\calA \cap \calB}$.
\end{corollary}

\begin{proof}
We only have to show the uniqueness.
The difference of two such products is an element~$\hat A'$ of~$\hat \calA$
that has to commute with~$\calB$ elementwise.
This means that
the expansion of~$\hat A'$ in the basis of~$\calA$ that is conjugate to~$\calB$,
should not contain any nontrivial basis element above~$\calS = \calA \cap \calB$.
This is what we wanted to show.
\end{proof}

This corollary means that the measurement dynamics by a locally reversible {\bf transition}~$\calA \to \calB$
is locality preserving.
A locally finite product~$\hat P$ that commutes with $\calA$ elementwise
can be thought of as a logical operator%
\footnote{%
One may wonder why we do not say it \emph{is} a logical operator.
If $\hat P$ is finitely supported, then it \emph{is};
otherwise, two infinitely supported ``operators''
do not always have a well-defined commutation relation,
and we may not always speak of algebra of logical operators.
Of course, this issue never arises in any finite system.
}
of~$\calA$. 
A string operator in a toric code, extended to infinity in both directions,
is an example of an infinitely supported but locally finite product of Pauli operators
that commutes with the code's stabilizer group.
Then, under the locally reversible transition
the logical operators are updated locally.
We can now conclude that the iterated, instantaneous teleportation circuit in~\S\ref{sec:teleportation}
is \emph{not} locally reversible at some step.
If it were, then every logical operator, in particular the single-qubit~$X$ and~$Z$
on the data qubit that is going to be teleported,
must have evolved to a nearby operator.
This is a contradiction since the circuit teleports the logical qubit to a distant position,
where the post-teleportation logical qubit does not have any logical operator represented 
near the pre-teleportation qubit.

\subsection{Example revisited: translation by measurements}
\label{sec:trans1d_LRMC}

Here we give an example of a one-dimensional locally reversible measurement circuit,
	implementing a uniform translation of the logical qubits (similar in spirit to Fig.~\ref{fig:teleportation_schematic} from the introduction).

Consider a one-dimensional chain of sites, 
each having two qubits indexed by even and odd integers respectively.
Put all even-indexed qubits in an $X$-eigenstate,
and regard all odd-indexed qubits as logical qubits.
By~\eqref{eq:ModernTeleport} we can teleport the logical qubits at $2j-1$ to the even-indexed qubits at $2j$,
putting all odd-indexed qubits in an $X$-eigenstate.
This teleportation is performed by a measurement circuit of depth~2.
The resulting situation is almost the same as the beginning,
except that the role of even and odd qubits are interchanged.
Now, implement another teleportation circuit for each pair of qubits~$2j$ and~$2j+1$ for all~$j$.
Then, we are back the the original stabilizer group where all even-indexed qubits
are individually in an $X$-eigenstate.
A circuit diagram for this teleportation is provided in Fig.~\ref{fig:Teleportation_chain}.
\begin{figure}[ht]
	\centering
	\includegraphics[width=.4\linewidth]{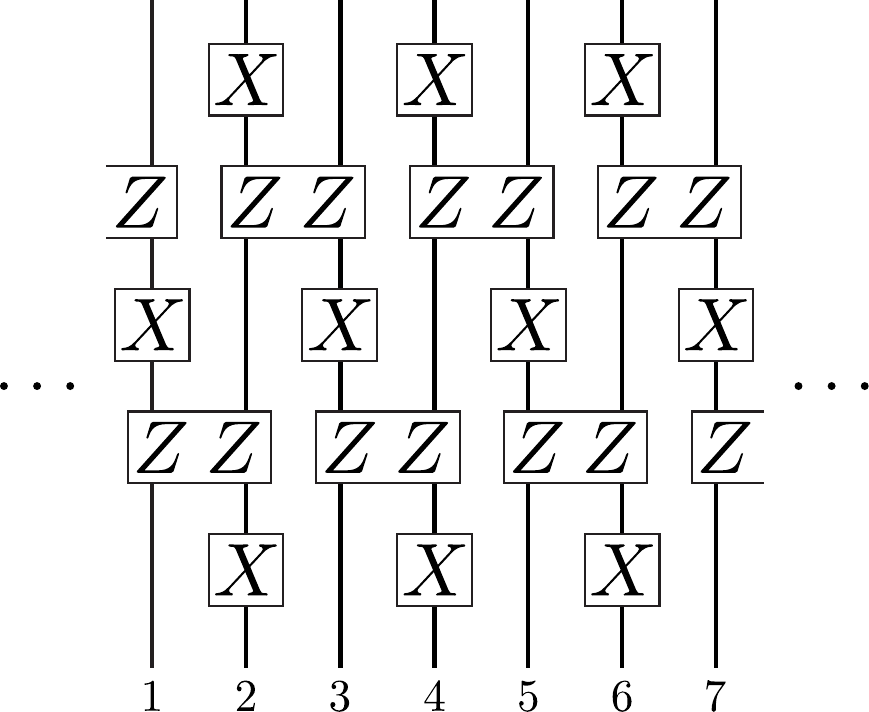}
	\caption{%
	A locally reversible circuit which implements a translation on a chain of $n$ qubits with $n$ ancilla's. 
	After one cycle, the logical operators at odd-numbered sites are transports two sites over to the right.
	} \label{fig:Teleportation_chain}
\end{figure}
Overall, the measurement circuit has depth~4.
Ignoring Pauli corrections that depend on the measurement outcomes,
we may think of this depth-4 circuit as a measurement dynamics of period~4.
\begin{align}
\xymatrix{
\braket{X_{2j}} \ar[r]^{1}& \braket{Z_{2j-1} Z_{2j}} \ar[d]^{2} \\
\braket{Z_{2j} Z_{2j+1}}\ar[u]^{4} & \braket{X_{2j-1}} \ar[l]^{3}
}
\end{align}
A logical qubit, initially at site~$j$ (on qubit~$2j-1$),
is teleported after the second measurement layer to qubit~$2j$,
and then to qubit~$2j+1$ on site~$j+1$.
That is, up to Pauli corrections, every logical qubit is teleported by one site to the right.
Modulo classical information of measurement outcomes that must be transferred towards the right,
we have achieved translation by measurement circuit.
This is contrasted to unitary dynamics where a net translation
is never achieved by a shallow unitary circuit~\cite{GNVW}.
Note that no matter what state the logical qubits are in,
the classical information that we are ignoring 
has exactly the same distribution --- the uniform distribution.

We can check the reversible pair condition to this translation circuit.
As the depth is~4, we consider four pairs,
but each pair is the same as any other modulo regrouping of qubits.
Any pair of measurement layers in the circuit consists of nonoverlapping two-qubit circuits.
It is thus obvious to see the condition satisfied.
A major difference from the iterated teleportation example above
is that now the stabilizer groups' bases that show 
conjugate pairs according to~\ref{thm:Reversibility}(a)
are all local, acting on at most two neighboring qubits.

\section{Locally reversible measurement cycles}
\label{sec:LRMC}

A {\bf measurement circuit} is a sequence of measurement steps $M_t$ ($t = 0,1,2,\ldots$),
where each measurement step is a set of mutually commuting $\ell$-local Pauli operators:
$M_t \subset \calP$ and $M_t^\perp \supseteq M_t$.
($\bullet^\perp$ denotes the set of all Paulis which commute with every element of $\bullet$.)
At each time step, the {\bf instantaneous stabilizer group} (ISG) 
is a set of finitely supported Pauli operators with definite eigenvalue.
From a measurement circuit, its instantaneous stabilizer groups~\cite{Hastings2021} are given by
\begin{align}
	\ISG_t = \Braket{ \ISG_{t-1} \cap M_t^\perp , M_t} .
	\label{eq:def_ISG}
\end{align}
At each time step, the $\ISG_t$ contains the group generated by $M_t$,
but also retains elements from the previous step which commute with the current measurement set.

Given a time-periodic measurement circuit, 
i.e., $M_t = M_{t+T}$ for all $t \ge 0$ and some $T > 0$,
if we start with $\ISG_{-1} = \{ \one \}$,
the evolution according to~\eqref{eq:def_ISG}
may lead to a steady, time-periodic sequence of ISGs.
That is, $\ISG_{t + T'} \equiv \ISG_{t}$ for all $t \ge t_0$ and some $T' > 0$,
where $\equiv$ means that the two stabilizer groups are the same up to signs.
This is the situation of our interest in this paper.
We do not necessarily claim that the time series of ISGs of any time-periodic measurement circuit
must enter into a steady state.
It would be interesting on its own to study the transient part of the dynamics,
but we do not address this question in this paper.

From now on we focus on {\bf locally reversible measurement cycles}.
This will be abbreviated as LR cycles or LRMC.
Formally an LRMC is a tuple $(\calA_0,\circuit C)$,
where $\calA_0$ is a stabilizer code and $\circuit C$ is a finite depth measurement circuit specified by $M_1,M_2,\ldots,M_{T}$,
subject to the condition that 
the ISGs determined by~\eqref{eq:def_ISG} form a cyclic chain of locally reversible pairs
\begin{align} \label{eq:LRMcircuit}
	\calA_0 \to \calA_1 \to \calA_2 \to \cdots \to \calA_{T-1} \to \calA_T \equiv \calA_0
\end{align}
where the last ISG is the same as the first up to signs.
The initial stabilizer group~$\calA_0$ is called the {\bf base} or {\bf background}.
Note that there are exactly $T$ locally reversible pairs in this sequence.
We will simply say that the sequence of stabilizer groups is locally reversible.

Remark that a tuple $(\calA_0,\circuit C)$ defines a sequence of ISGs by~\eqref{eq:def_ISG},
but there can be many different circuits $\circuit C$ that give the same sequence of ISGs up to signs.
Indeed, given a sequence of stabilizer groups, each of which admits an $\ell$-local generating set,
the transition can be obtained by measuring all elements of the $\ell$-local generating set,
the order of which is not important because they all commute with each other.
There can be many choices of an $\ell$-local generating set for a given stabilizer group,
and any choice will result in the same sequence of instantaneous stabilizer groups.
We wish to equate these different measurement circuits that give the same sequence of ISGs up to signs.
We make this precise with a simple equivalence relation.

\subsection{Simple equivalence and group structure}\label{sec:SimpleEquivalence}

Since we consider finite depth circuits,
we can compose them for any finite number of times.
The composed dynamics is still locally reversible by definition.
It may be useful to think of an abstract simple graph,
where the vertices are each a stabilizer group and an edge exists if and only if 
the two stabilizer groups at the ends of the edge are a locally reversible pair.
A locally reversible measurement cycle corresponds to a loop in this abstract graph
and the composition of dynamics corresponds to the composition of the loops.
Hence, we have a monoid of all locally reversible measurement cycles
based at a fixed background stabilizer group.
(A monoid is an algebraic structure similar to a group but not necessarily with the inverse operation.)
We can turn this monoid into a group by introducing a {\bf simple equivalence} relation
in a manner analogous to the group of local unitary circuits.

It is useful to construct a unitary circuit corresponding to a reversible measurement circuit.
Consider the measurement of Pauli operator~$B$ whose outcome is uniformly random
for the reason that there was an anticommuting stabilizer~$A$.
For any pair of anticommuting Pauli operators~$A,B$, we have a unitary
\begin{align}
	U_{A,B} = \frac{A + B}{\sqrt 2} .
\end{align}
If a state~$\ket \psi = A \ket \psi$ was stabilized by~$A$,
then $U_{A,B} \ket \psi = B U_{A,B} \ket \psi$ is stabilized by~$B$:
\begin{align}
	B U \ket \psi = \frac{1}{\sqrt 2} B (\one + B) \ket \psi = \frac{1}{\sqrt 2}(B + \one) \ket \psi = \frac{1}{\sqrt 2}(B + A) \ket \psi = U \ket \psi .
\end{align}
In fact, $U_{A,B} \ket{\psi}$ is the state after obtaining measurement outcome $B= +1$.
(A generalization for higher dimensional qudits can be found in~\ref{lem:Fourier}.)
More generally, for any locally reversible transition $\calA\to\calB$,
we may replace \emph{all} measurements by a product of such unitaries,
\begin{align}
	U_{\{A_i\},\{B_i\}} = \prod_i \frac{A_i+B_i}{\sqrt 2}
\end{align}
with $\{A_i\}$, $\{B_i\}$ being local conjugate bases for $(\calA,\calB)$.
Because $A_i$ commutes with $B_j$ for $i \neq j$, 
the unitary $U_{\{A_i\},\{B_i\}}$ can be implemented by a finite-depth unitary circuit.
If the signs for $\{A_i\}$ and $\{B_i\}$ were chosen such that $\ket{\psi}$ was stabilized
	by $\{A_i\}$ (with $+1$ eigenvalue) and mesaurement outcome for $\{B_i\}$ are all $+1$,
then $U_{A_i,B_i} \ket{\psi}$ is the state after the transition $\calA\to\calB$.
Since any states with different measurement outcomes 	differ only by a (locally finite) Pauli operator,
we may say that the evolution of the ISG by measuring generators of~$\calB$ 
is the same as that by applying~$U$.
Note that the unitary is \emph{not canonical};
different choices of conjugate bases for $(\calA,\calB)$ will give different unitaries.

Returning to the subject of simple equivalence for circuits,
first, we allow insertion and deletion of redundant measurements.
Measuring a Pauli on a state that is already stabilized by the Pauli does not change the state.
In particular, if a measurement gate in a layer is a product of some other operator that is being measured,
we may omit this redundant measurement gate.
We say that two circuits are simply equivalent if 
one is obtained from the other 
by adding or removing such redundant measurements 
as long as the operators of the redundant measurements are $\ell$-local.
It is possible that an entire layer is eliminated after removing redundant measurements.
This may shorten the sequence of instantaneous stabilizer groups.
For example, a sequence $(\calA \to \calA)$ can always be shortened to~$(\calA)$.
In this case, the corresponding unitary circuit for $\calA \to \calA$ can be chosen to be the identity.

Second, we allow insertion and deletion of sandwiched anticommuting measurements.
Concretely, suppose that a locally reversible transition~$\calA \to \calB \to \calC$ of length~2
is achieved by measuring all elements of some local conjugate bases of the ISGs.
Consider two consecutive measurements by~$B \in \calB$ and $C \in \calC$ which are a conjugate pair between~$\calB$ and~$\calC$,
and suppose that $C \equiv A$ was also an element of~$\calA$ 
participating in the local conjugate basis between~$\calA$ and~$\calB$.
(For example, this may happen if $\calA = \calC$.)
Then, we allow deletion of the measurement of~$B$.

The deletion of measuring~$B$ \emph{does change} the underlying sequence of ISGs;
however, local reversibility is retained and the resulting state remains the same (up to Pauli corrections).
To see this, consider the converted unitary circuit.
Since $A \equiv C$ and $B$ are a conjugate pair,
we may replace the measurements of~$B$ followed by~$C$ with~$U_{B,C} U_{A,B}$.
Observe that $U_{B,C}$ is equal to $U_{A,B}^{-1}$ up to a Pauli ($U_{A,B} U_{A,-B} = BA$);
hence the deletion of~$B$ amounts to removing these unitary gates.
The only difference between the measurement circuit and the converted unitary circuit 
is a Pauli operator that only depends on random measurement outcomes.
Application of this rule allows the locally reversible sequence $\calA \to \calB \to \calA$ 
to be shortened to $\calA$.

As the name implies, a locally reversible measurement cycle can be undone,
reversing the sequence of ISGs.
\begin{align} \label{eq:LRMcircuitRev}
	\calA_0 \equiv \calA_T \to \calA_{T-1} \to \cdots \to \calA_{1} \to \calA_0 \,.
\end{align}
Note that this reverse dynamics is in general 
not achieved by simply reversing the time order of measurement gates.
If an LR cycle consists of a base code $\calA_0$ and a circuit
$\circuit C = (M_1, M_2, \ldots, M_T)$ of measurement layers,
the order-reversed measurement circuit $(M_T, M_{T-1},\ldots, M_1)$ 
with background~$\calA_0$ does not make much sense
since it is redundant to measure $M_T$ on the base code $\calA_T \equiv \calA_0$
and it is not clear why the last ISG after measuring $M_1$ must be the base code.
A better attempt is to consider a measurement circuit
$(M_{T-1}, M_{T-2}, \ldots, M_1, M_T)$
as the measured operators in each layer belong to the next ISG in~\eqref{eq:LRMcircuitRev}
so that the circuit may induce transitions
$\calA_0 \equiv \calA_T
\underset{\substack{\\[-1ex]?}}{\xrightarrow{M_{T-1}}} \calA_{T-1}
\underset{\substack{\\[-1ex]?}}{\xrightarrow{M_{T-2}}} \cdots
\underset{\substack{\\[-1ex]?}}{\xrightarrow{M_1}} \calA_1
\underset{\substack{\\[-1ex]?}}{\xrightarrow{M_T}} \calA_0$.
This circuit often works, but not always.\footnote{%
For example, a cycle $\big(\calA_0 = \braket{ZZ,XY}, \circuit C=(M_1=\{IZ\}, M_2=\{XI\}, M_3=\{IY\}, M_4=\{ZZ\}) \big)$
gives an ISG sequence 
$\braket{Z_1Z_2,X_1Y_2} \xrightarrow{M_1} \braket{Z_1,Z_2} \xrightarrow{M_2} \braket{X_1,Z_2} \xrightarrow{M_3} \braket{X_1,Y_2} \xrightarrow{M_4} \braket{Z_1Z_2,X_1Y_2}$,
but $\calA_0=\braket{Z_1Z_2,X_1Y_2} \xrightarrow{M_3} \braket{X_1,Y_2} \xrightarrow{M_2} \braket{X_1,Y_2} \xrightarrow{M_1} \braket{X_1,Z_2} \xrightarrow{M_0=M_4} \braket{Z_1,Z_2}$
which does not even reach the initial stabilizer group up to sign.
That is, $\big(\calA_0, \{M_3,M_2,M_1,M_4\}\big)$ 
does not constitute a valid LR cycle.
}
However, we know from previous discussion that
we can always find a circuit~$\circuit C' = (M_1', M_2', \dots, M_T')$ that is simply equivalent to~$\circuit C$
such that its order-reversed version $(M_{T-1}',\dots,M_1',M_T')$ produces the reversed sequence of ISGs.

Hence, based at a fixed stabilizer group~$\calA_0$,
the collection of all simple equivalence classes of LR cycles, is a group
because for any LR cycle $(\calA_0, \circuit C)$ with an ISG sequence~\eqref{eq:LRMcircuit} 
we can construct an inverse cycle~$(\calA_0, \circuit D)$ with ISG sequence~\eqref{eq:LRMcircuitRev}
such that the composition~$(\calA_0, \circuit C \circ \circuit D )$
is simply equivalent to the trivial LR cycle~$(\calA_0, \emptyset)$.
Indeed, the first primitive allows us to replace a given measurement round of~$\circuit C$ 
with one that measures elements of conjugate bases,
and then the second primitive removes all the middle measurements in $\calA_{t-1} \to \calA_t \to \calA_{t-1}$,
shortening the sequence inductively.

Thus, the simple equivalence among locally reversible measurement cycles
is defined by the two primitives.
This simple equivalence relation allows us to think of any LR cycle
as a sequence of stabilizer groups that is periodic up to signs,
without explicitly referring to the measurement gates.
However, two simply equivalent dynamics may not behave the same way,
especially if one is interested in its behavior against perturbations
as in fault tolerant gadgets in quantum computing literature.
There, it is often essential to repeat measurements.
This aspect is not considered in this paper.

\subsection{Blendings and topological codes}
\label{sec:blending}

For LR cycles in a lattice of the same dimension~$\dd$,
we consider their blending or interpolation 
in analogy with that of quantum cellular automata~\cite{GNVW,FreedmanHastings2019QCA,FHH2019,Haah2022b}.
There can be various adaptations of the unitary notion to our measurement setting,
but we will consider two versions below.
Each version can be thought of as a certain class of spatial boundary conditions.

\begin{figure}[ht]
	\centering
	\begin{minipage}{162mm} \raggedright \includegraphics[width=162mm]{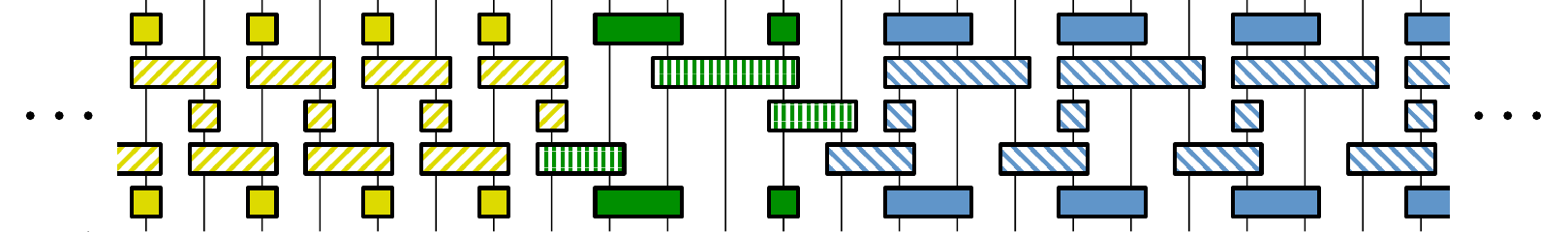}
	\\[-2ex] ${}\mkern2mu\mkern6mu	
		\mathord{\underbrace{\hspace{72mm}\mkern-0.7mu\mkern-12mu}_{\displaystyle\strut A}} \mkern12mu
		\mathord{\underbrace{\hspace{90mm}\mkern-0.7mu\mkern-12mu}_{\displaystyle\strut B}}$
	\\[-3ex] ${}\mkern1mu\mkern6mu \hspace{60mm}
		\mathord{\underbrace{\hspace{24mm}\mkern-0.7mu\mkern-12mu}_{\displaystyle\strut I}}$
	\end{minipage}
	\caption{%
	Illustrative example of a blending in one dimension, between $(\calA_0, \circuit A)$ (yellow) with $(\calB_0, \circuit B)$ (blue).
	The lattice is divided into $A$ (left half) and $B$ (right half); the interface region $I$ has 4-site width.
	Stabilizers and circuit operators supported on $A \setminus I$ and $B \setminus I$ matches those of $(\calA_0, \circuit A)$ and $(\calB_0, \circuit B)$ respectively.
	} \label{fig:circuit_blending}
\end{figure}

The first that we call a {\bf blending} is defined as follows.
Take two LR cycles $(\calA_0,\circuit A)$ and $(\calB_0,\circuit B)$ 
on a $\dd$-dimensional lattice.
Without loss of generality, we can always assume two circuits $\circuit A$ and $\circuit B$ 
act on the same lattice; if not, we insert inert ancillary degrees of freedom.
Then, given a spatial boundary, a blending between the two dynamics 
is a choice of an LR cycle~$(\calC_0, \circuit C)$ on the full lattice
subject to conditions as follows.
The spatial boundary divides the full lattice into two regions~$A$ and~$B$,
and we consider an interface region~$I$ that is a small-distance neighborhood 
of the geometric boundary between~$A$ and~$B$.
The interface~$I$ is extended along the boundary, but has a microscopic width.
The blending base group~$\calC_0$ is required 
to include any element of~$\calA_0$ if it is supported on~$A \setminus I$
and any element of~$\calB_0$ if it is supported on~$B \setminus I$.
Similarly, the blending circuit~$\circuit C$ is required to include measurement gates 
of~$\circuit A$ over~$A \setminus I$ and those of~$\circuit B$ over~$B \setminus I$ 
up to some insertion of idling steps in between.
Due to the possible idling steps, 
two circuits~$\circuit A$ and $\circuit B$ that is being blended may have different periods.
We do not impose any requirement on the gates over the interface region~$I$
or the base group~$\calC_0$ besides local-reversibility.
See Fig.~\ref{fig:circuit_blending} for an illustrative example of blending.
However, naive deletion of measurement gates may not give a blending
because either periodicity or the local reversibility can be broken.
Hence, the existence of a blending between LR cycles is not always obvious.

The second is a subclass of blendings
and applies only for a smaller family of LR cycles with topological background codes that we now define.
Given an $\ell$-locally generated Pauli stabilizer group~$\calA$ on a $\dd$-dimensional lattice of qubits,
an {\bf excitation} by a Pauli operator~$P$ is define to be 
an $\ell$-local stabilizer generator that anticommutes with~$P$.
For any region~$R$ of the lattice we denote by~$R^{+\ell}$ the $\ell$-neighborhood of~$R$.

\begin{definition}\label{defn:TopologicalCode}
$\calA$ is {\bf topological} if all of the following conditions are met with some $\ell > 0$: 
\begin{enumerate}[itemsep=1pt, parsep=0pt, topsep=2pt]
	\item[(i)] for any Pauli operator~$P$ supported on a finite convex region~$R$,
if $P$ commutes with all stabilizers, 
then $P$ is a product of $\ell$-local stabilizer generators,
each of which is supported on~$R^{+\ell}$,

	\item[(ii)] if a subset of excitations by a Pauli operator~$P$ 
is contained in a finite convex region~$R$,
but no excitation's support overlaps $R^{+\ell}\setminus R$,
then for any ball~$B$ of diameter~$\ell$ in~$R$
there exists another Pauli operator~$Q$ supported on~$R^{+\ell}$
such that $PQ$ causes no excitation in $R^{+\ell} \setminus B$, and

\item[(iii)] if the set of all excitations by a finitely supported Pauli operator~$P$
is contained in a convex region~$R$,
then there exists another Pauli operator~$Q$ supported on~$R^{+\ell}$
such that $PQ$ causes no excitation.
\end{enumerate}
A {\bf topological} LR cycle is one that has a topological code as a base stabilizer group.
A {\bf topological blending} of topological LR cycles is a blending
such that the blending base stabilizer group is topological on the whole lattice.
\end{definition}

Here are a few remarks about the present definition of topological codes.
We will liberally use conventional languages that are not defined in the present paper.

The condition~(i) for topological codes is a tailored version of 
what is called ``local topological order condition'' or TQO-1,2 
that are used in an energy gap stability proof~\cite{BravyiHastingsMichalakis2010stability}.
This condition implies that the local reduced density matrix of any code state 
is determined by the stabilizer generators around them.
In particular, there cannot be any nontrivial logical operator supported on any finite region.
All ``exact'' codes in the sense of~\cite{Haah2013} satisfy this first condition
at least for diamond-shaped regions (balls in the $\ell_1$-metric).
Many examples of fracton phases~\cite{Nandkishore2018,Pretko2020} satisfy the condition~(i).

The condition~(ii) says that we should be able to localize any excitation at any point in space
as long as it is isolated by a small distance~$\ell$, 
having no other excitation in~$R^{+\ell} \setminus R$.
Indeed, if we have a bunch of point-like excitations as in the 2d toric code 
or gauge charges in higher dimensions, 
we can always bring them to any point by some string operators.
On the other hand, a flux loop in a gauge theory in~$\dd \ge 3$ must remain as a flux loop,
and if we choose a region~$R$ to enclose a proper portion of a flux loop,
we cannot localize the excitation within~$R$ to a point.
This is not a violation of the condition~(ii) since a loop is not isolated.
Note that the two conditions~(i) and~(ii) allows for gapped boundaries~\cite{BravyiKitaevSurfaceCode}.
Indeed, the surface code on a left half plane has a gapped vertical boundary
and there is no nontrivial finitely supported logical operator, satisfying~(i).
Any excitation can be transported to any given point by some string operator, implying~(ii).
On the contrary, fracton codes~\cite{Nandkishore2018,Pretko2020} in~$\dd \ge 3$ violates the condition~(ii).
By definition a fracton is an isolated, nonlocally created excitation
whose antiparticles are never point-like.
If we choose a region~$R$ that contains an isolated, single fracton,
then the condition requires that we find an operator~$Q$ that can move this fracton 
to any given point within~$R$, which is forbidden.

The condition~(iii) rules out any gapped boundaries and enforces more strict spatial homogeneity.
A gapped boundary allows for a topologically charged excitation to enter the bulk from the boundary
by a finitely supported operator.
If we take~$R$ to be a small ball enclosing the excitation in the bulk,
then the condition~(iii) demands that we should be able to annihilate it by acting around the small ball,
which contradicts the fact that excitation is topologically nontrivial.

Finally, we note that every translation invariant code with macroscopic code distance in~$\dd=2$
is a topological code in the present sense 
because it is finite depth quantum circuit equivalent to finitely many copies of the toric code~\cite{Haah2018}.
Virtually any Pauli stabilizer code that has a topological quantum field theory description
fits into the present definition of topological codes.

\subsection{Boundary actions of topological LR cycles: MQCA}
\label{sec:boundaryactions}

For convenience, we say a {\bf blending into vacuum} or simply {\bf vacuum blending}
to mean a blending between some LR cycle and the ``trivial'' LR cycle based
at the stabilizer group generated by the single-qubit~$Z$ on every qubit.
A vacuum blending is not necessarily topological.
The two major examples in \S\ref{sec:WPT} and \S\ref{sec:HHcode} below 
will have nontopological vacuum blendings with various interfaces.

A vacuum blending for a topological LR cycle starts with a base stabilizer code~$\calA$
that has a spatial boundary and that is topological on one side,
by which we mean that the three conditions of topological codes
hold for regions~$R$ contained on the half space.
Let us look at this base code~$\calA$ more closely.
Since the vacuum side has single-qubit stabilizers, we can ignore that region.
By definition, the base code~$\calA$ 
contains all generators of a topological code~$\calB$
that are supported $5\ell$-away from the spatial boundary of the blending.
(Of course, the constant~$5$ is just conveniently chosen to clarify what we do.)
There can be many finitely supported nontrivial logical operators due to the boundary,
although there is none in the bulk.

Assuming that the spatial boundary is a flat hyperplane,
we show that there is a sharply defined region near the boundary, say within distance~$20\ell$,
on which all finite logical operators are represented.
We may say that this is the interface region~$I$ of the vacuum blending.
To this end, let $P \in \calA^\perp$ be any logical operator that is finitely supported.
The lattice for~$\calA$ is the same as that for~$\calB$ at least in the bulk,
that is away from the boundary by distance, say $5\ell$.
By the topological code condition~(iii),
$P$ can be capped off by some Pauli operator~$P'$
such that $P P'$ is a stabilizer of~$\calB$
and $P'$ is supported on the $\ell$-neighborhood of the minimal convex region 
that supports all the noncommuting stabilizer generators of~$\calB$ against~$P$.
The convex region has to be spatially near the boundary
because that is the only region where $P$ may not commute with~$\calB$.
So, $P'$ is supported on within distance $10\ell$ from the boundary.
Then, the topological code condition~(i)
implies that $PP'$ is a product of stabilizer generators of~$\calB$,
each of which is $\ell$-close to the support of~$P P'$.
In this product, the generators supported in the bulk is also a member of~$\calA$
and they cancel the tensor components of~$P$ over the bulk.
Hence, we have a stabilizer of~$\calA$ 
that cancels the tensor components of~$P$ over the bulk, 
leaving some operator near the boundary, within distance~$15\ell$.
This final operator is within distance~$5\ell$ from~$\Supp(P)$.
We conclude that all finitely supported logical operators for the base stabilizer group~$\calA$ 
have a representative localized at the boundary, strictly within the distance~$15\ell$ from the boundary.
Note that our claimed interface region has width~$20\ell$ from the boundary.\footnote{%
	We have not used the condition~(ii) of~\ref{defn:TopologicalCode} for this localization.
	We will use (ii) in the proof of~\ref{lem:IndexSeparation}.
}

This localization implies that, to consider the full group of logical operators,
it suffices to restrict stabilizers in $\calA$ to the interface~$I$
and take the commutant~$\calL$ of the restriction
within the full Pauli group on~$I$.
Here, the restriction means that we drop any tensor component that is not acting on~$I$.
For example, if $B \otimes F$ is a Pauli operator acting on the ``interior of the bulk'' by~$B$
and on~$I$ by~$F$, then the restriction of~$B \otimes F$ to~$I$ is $\one \otimes F$.
(There is a phase ambiguity in the restriction, which will be inconsequential.)
The restriction of~$\calA$ to the interface region~$I$ 
generates a group~$\Pi_I \calA$ that is not necessarily abelian even though $\calA$ is abelian.
If two operators~$P,Q \in \calL$, which are supported on the interface~$I$ by construction,
are equivalent up to a stabilizer~$S \in \calA$,
{\it i.e.}, $P = QS$,
then, trivially, $S = Q^{-1}P $ is also supported on~$I$.
So, the group of all equivalence classes of finitely supported logical operators of~$\calA$
is precisely~$\calL / \calA_I$ where
$\calA_I = \{ S \in \calA ~|~\Supp(S) \subseteq I \} \subseteq \Cent(\calL)$.
We will refer to $\calL / \calA_I$ as the {\bf boundary logical quotient group}.

Now, the circuit dynamics that is a composition of locally reversible transitions
gives an invertible evolution rule of~\ref{thm:LP}
under which every logical operator representative~$P$, 
which can be found within distance~$15\ell$ from the boundary,
is mapped to an equivalence class~$[Q]$ of logical operators up to stabilizers
with a representative $Q$ supported near $P$
within distance that depends only on $\ell$ times the circuit depth.
Localizing $Q$ towards the boundary, we obtain~$Q'$
that is supported within distance~$5\ell$ along the direction that the interface is extended
and within distance~$20\ell$ from the boundary.
Hence, $P \mapsto [Q']$ is well defined map on the interface region~$I$.
It follows that the circuit dynamics defines an automorphism on~$\calL / \calA_I$
with the extra locality-preserving property.
This is an example of the following.

\begin{definition}\label{defn:MQCA}
Let $\calL$ be a group of finitely supported Pauli operators 
on a $\dd$-dimensional lattice with finitely many qubits per site.
For a normal subgroup $\calS \subseteq \calL$,
an automorphism~$\alpha$ of the quotient group~$\calL/\calS$ is an {\bf MQCA} 
if there exists a constant length~$r$ (called a range or spread) such that 
for any~$P \in \calL$ the image~$\alpha([P])$
can always be represented by~$Q \in \calL$ that is supported on the $r$-neighborhood of the support of~$P$.
\end{definition}

A Clifford QCA in the existing sense~\cite{clifQCA,clifQCAclassification} 
is an MQCA with $\calL$ being the full Pauli group and $\calS = \{ \one \}$.
Thus, our MQCA is a generalization of Clifford QCA.
As we have shown above, every vacuum-blending of a topological LR cycle
in $\dd$ dimensions gives an MQCA in $\dd-1$ dimensions.
Though our definition~\ref{defn:MQCA} of MQCA does not mention measurements explicitly,
our examples come from measurement circuits, hence the name ``M''QCA.
In all cases we consider, $\calL$ is a group of logical operators.
The translation in~\S\ref{sec:trans1d_LRMC} induced by an LR cycle in one dimension
is an example of MQCA in one dimension, which happens to be a unitary Clifford QCA.
Of course, to have a definite automorphism we need to fix the measurement outcomes (i.e., postselection),
but we have discussed that different measurement outcomes lead to the same automorphism 
up to a Pauli circuit of depth~1.

\section{Index of one-dimensional MQCA}
\label{sec:MQCA_1D}

For one-dimensional unitary QCA and quantum walk~\cite{Kitaev_2005,GNVW},
it has been fruitful to use Fredholm operators~\cite{avron1994index,Lax},
which we will review briefly.
Since a group of Pauli operators over qubits is an $\FF_2$-linear space if we ignore all the phase factors,
our MQCA can be thought of as an invertible linear map acting on an infinite dimensional $\FF_2$-vector space.
Following~Refs.~\cite{avron1994index,Kitaev_2005,GNVW},
we adapt the construction of an index for a linear map to our MQCA.

\subsection{Fredholm maps}

Let $\FF$ be a field ({\it e.g.}, the field of complex numbers, real numbers, rational numbers, or a finite field~$\FF_2$).
A {\bf Fredholm} map~$\phi : V \to W$ is a linear map between $\FF$-vector spaces,
such that the kernel $\ker \phi = \{ v \in V ~|~ \phi(v) = 0 \}$ and the cokernel
$\coker \phi = W / \im \phi = W / \{ \phi(v) ~|~ v \in V \}$ are both finite dimensional.
If $V$ and $W$ are finite dimensional, all linear maps are Fredholm.
Hence, this notion is interesting only in an infinite dimensional setting.
It is well known~\cite[Chap.~2]{Lax} that an $\FF$-linear map~$\phi : V \to W$ is Fredholm 
if and only if
there exists another $\FF$-linear map $\eta : W \to V$ such that 
$(\phi\eta - \one) : W \to W$ and $(\eta \phi - \one) : V \to V$
are both finite rank ({\it i.e.}, having a finite dimensional image).
Such $\eta$ is called a {\bf pseudoinverse} of~$\phi$.
Thus, colloquially speaking, 
a Fredholm operator is an almost invertible linear map except for a tiny subspace.

The {\bf index} of a Fredholm map~$\phi$ is an integer defined as
\begin{align}
	\ind \phi = (\dim_\FF \ker \phi) - (\dim_\FF \coker \phi).
\end{align}
The two dimensions are assumed to be finite, so is the index.
The Fredholm index is stable against any finite dimensional perturbations:
if $\delta : V \to W$ is any linear map of finite rank
and $\phi : V \to W$ is Fredholm,
then
\begin{align}
	\ind(\phi + \delta) = \ind(\phi). \label{eq:FredholmStability}
\end{align}
In addition, the Fredholm index satisfies
\begin{align} \begin{split}
\label{eq:FredholmPlus}
	\ind(\phi_1 \oplus \phi_2) &= \ind(\phi_1) + \ind(\phi_2),
\\	\ind(\phi_1 \phi_2) &= \ind(\phi_1) + \ind(\phi_2)
\end{split} \end{align}
whenever they are defined.
The proof in~\cite[Chap.~2]{Lax} for these properties~\cite{Sarason1987} is surprisingly elementary.
Note that the cofficient field~$\FF$ is arbitrary.

\subsection{The MQCA index}\label{sec:MQCA-def-prop}

We have to identify an appropriate Fredholm map 
starting with a one-dimensional MQCA
\begin{align}
	\alpha : \calL / \calS \to \calL/\calS.
\end{align}
Although $\alpha$ can be regarded as a linear map on an $\FF_2$-vector space
by forgetting phase factors,
this is invertible, so the Fredholm index of~$\alpha$ is automatically zero.
A guiding principle to avoid this triviality
is to consider a semi-infinite interval
and examine the restricted action by~$\alpha$.

Motivated by the infinite Ising chain, we consider a quotient as follows.
For any region~$R \subseteq \ZZ$,
let $\calL_R$ denote the subgroup of~$\calL$ consisting of all elements supported on~$R$.
If $R$ is an interval, we will abbreviate $R$ by an obvious inequality such as $\calL_{<a} = \calL_{(-\infty,a)}$.
Similarly, let $\calL^{/\calS}_R$ denote the subgroup of~$\calL/\calS$
consisting of all elements that have a representative supported on~$R$.
With these, we define the following.%
\footnote{%
For the infinite Ising chain, we can consider $\calL$, the set of all finite tensor products of Pauli~$Z$,
and $\calS$, the subset of all finite tensor products of even number of Pauli~$Z$.
In this example, $\calL^{/\calS}_{-\infty} = \calL^{/\calS}_{+\infty} = \calL / \calS$,
because $Z_0 \calS$ has a representative on any site, so $Z_0 \calS \in \calL^{/\calS}_R$ for any nonempty~$R$.
If we have a collection of independent Ising chains,
and if $\calL$ is the set of all finitely supported logical operators
and $\calS$ be the stabilizers,
then $\big( \calL^{/\calS}_{-\infty} \calL^{/\calS}_{+\infty} \big) := \bigl\langle \calL^{/\calS}_{-\infty}, \calL^{/\calS}_{+\infty} \bigr\rangle \subseteq \calL/\calS$ 
captures the logical operators for all infinite and semi-infinite chains (e.g.\ Ising model with support on $(-\infty,a)$ or $(a,+\infty)$ sites).
For the Ising chain collection,
$\calF$ is generated by the logical operators 
of all the finite Ising chains;
$\calF \cong \calL/\Cent(\calL)$.
So, it might seem extravagant to form~$\calF$;
however, our construction allows for more general~$\calS \subseteq \calL$.
}
\begin{align} \label{eq:ConstructF} \begin{aligned}
	\calL^{/\calS}_R &= \big\{ x\calS \in \calL/\calS ~\big|~ x \in \calL ,\, \Supp(x) \subseteq R \big\} ,\\
	\calL^{/\calS}_{+\infty} &= \bigcap_{k \in \ZZ} \calL^{/\calS}_{ (k,\infty)  }  \,,\\
	\calL^{/\calS}_{-\infty} &= \bigcap_{k \in \ZZ} \calL^{/\calS}_{ (-\infty,k)  } \,,\\
	\calF &= (\calL/\calS) \Big/ \big( \calL^{/\calS}_{-\infty} \calL^{/\calS}_{+\infty} \big) ,\\
	\calF_{<a} &= \big( \calL^{/\calS}_{<a} \calL^{/\calS}_{+\infty} \big) \Big/ \big( \calL^{/\calS}_{-\infty} \calL^{/\calS}_{+\infty} \big) .
\end{aligned} \end{align}
The automorphism~$\alpha$ is defined on~$\calL/\calS$,
but we observe that $\alpha$ preserves~$\calL^{/\calS}_{-\infty}\calL^{/\calS}_{+\infty}$,
so we have an induced automorphism
\begin{align}
\bar \alpha :  \calF \to \calF.
\end{align}
To see this, we have to show
for any $x \calS \in \calL^{/\calS}_{+\infty}$
that $\alpha(x \calS) \in \bigcap_{k} \calL^{/\calS}_{(k,\infty)}$,
but this is straightforward:
for any~$k \in \ZZ$, we know that $x \calS \in \calL^{/\calS}_{(k+r,\infty)}$ 
and the locality-preserving property implies that $\alpha(x\calS) \in \calL^{/\calS}_{(k,\infty)}$.
A similar argument shows that 
$\alpha(\calL^{/\calS}_{-\infty}) \subseteq \calL^{/\calS}_{-\infty}$.

Let $a \le b$ be sites of the lattice~$\ZZ$
and choose a subgroup~$\calF_{\circ}$ 
such that $\calF_{<a} \subseteq \calF_\circ \subseteq \calF_{<b}$.
Let $\nu$ be a projection onto~$\calF_\circ$:
\begin{align}
	\calF \xrightarrow{\quad \nu \quad} \calF_{\circ} \,,
\end{align}
which is a left inverse of the inclusion~$\iota :\calF_\circ \to \calF$
so that $\nu \iota = \one$ on~$\calF_\circ$ and $(\iota \nu)^2 = \iota \nu$.
Consider a composition
\begin{align}\label{eq-nualphaiota}
	\phi &= \nu \bar \alpha |_{\calF_{\circ}} : 
	\calF_{\circ} \lhook\joinrel\xrightarrow{\quad\iota\quad} 
	\calF \xrightarrow{\quad\bar \alpha \quad} 
	\calF \xrightarrow{\quad\nu\quad} \calF_{\circ} \,.
\end{align}
\begin{definition}
	The MQCA index of~$\alpha$ is defined as:
	\begin{align}
		\Ind(\alpha) &= \half \ind(\phi) \in \half \ZZ .\label{eq:IndexMQCA}
	\end{align}
\end{definition}
The factor of half is a convention, 
which will make our~$\Ind$ equal to the usual index of one-dimensional unitary QCA~\cite{GNVW}
if $\alpha$ acts on a full local operator algebra.
This is a legitimate definition because

\begin{lemma}\label{lem:Fredholm2}
$\phi$ is Fredholm, and $\Ind(\alpha)$ does not depend on the choices of~$a,b,\nu$.
\end{lemma}
\begin{proof}[Proof sketch] 
	The map~$\phi$ is a Fredholm operator 
	because the operator $\nu\bar{\alpha}^{-1}$ gives a pseudoinverse.
	Under different choices of $a,b,\nu$, 
	the Fredholm map~$\phi$ differs by a finite rank map, 
	which does not change the index. 
	See \S\ref{sec:indexproof}.
\end{proof}

\begin{proposition}\label{MQCA-properties}
The MQCA index satisfies the following properties.
\begin{enumerate}
	\item 
	$\Ind(\qca\otimes \qcaR)=\Ind(\qca)+\Ind(\qcaR)$. 		 	
	\item
	$\Ind(\qca\qcaR)=\Ind(\qca)+\Ind(\qcaR)$.
	\\In particular, $\Ind(\id)=0$ and $\Ind(\qca^{-1})=-\Ind(\qca)$.
	\item
	If $U: \calL/\calS\to\calL'/\calS'$ is a locality-preserving isomorphism between groups of equivalent logical operators, and if $\qca: \calL/\calS\to\calL/\calS$ is an MQCA on $\calL/\calS$, then $U\qca U^{-1}: \calL'/\calS'\to\calL'/\calS'$ is an MQCA on $\calL'/\calS'$, and $\Ind(U\qca U^{-1})=\Ind(\qca)$.
	\item
	If $\calL/\calS$ is the full Pauli group, then $\Ind$ equals 
	($\log_2$ of) the GNVW index of 1d unitary QCA~\cite{GNVW} over qubits.
\end{enumerate}
\end{proposition}

\noindent
In item~1, $\qca \otimes \qcaR$ denotes the MQCA on a double layer system defined by applying $\qca$ and $\qcaR$ separately on each layer. 
The underlying logical groups for $\qca$ and $\qcaR$ can be different.
In item~2, $\qca\qcaR$ is a composition: first applying $\qcaR$ and then applying $\qca$. 
The underlying logical groups for $\qca$ and $\qcaR$ must be the same for $\qca\qcaR$ to make sense. 
It is clear that $\qca\qcaR$ is still an MQCA. 
In item~3, we conjugate a MQCA by a locality-preserving isomorphism of the logical group 
and obtain a new MQCA.
Here, $U$ being locality-preserving means that, similarly to the definition~\ref{defn:MQCA} of MQCA,
for any $P \in \calL$ the image~$U([P])$ can be represented by~$Q \in \calL'$ that is supported near~$P$.
Item~3 says that the index is invariant under the conjugation by~$U$.
Note that this does not follow from item~3, since $U$ is merely an isomorphism, not necessarily an automorphism.
Item~4 shows that our index is a proper generalization 
of the usual index of 1D Clifford QCA defined on a full Pauli group.
The proof of~\ref{MQCA-properties} is deferred to \S\nameref{app:MQCAprop}.

The following result, proved in \S\nameref{app:MQCAflow}, substantiates the intuition that 1d MQCA index is a measure of flow.

\begin{proposition}\label{prop:IndAsFlow}
	Suppose that $\calL$ is the commutant of a locally generated group of Pauli operators 
	within a full one-dimensional Pauli group
	and that $\calS / (\calS_{<a} \calS_{>b})$ is finite dimensional for some $a,b \in \ZZ$.
	Then,
	\begin{enumerate}
		\item $\Ind$ is a blending invariant, i.e., 
		if for two MQCA $\alpha,\beta$ on~$\calL / \calS$ there exists a third MQCA~$\gamma$ and $a, b \in \ZZ$ 
		such that $\gamma|_{\calF_{<a}} = \alpha|_{\calF_{<a}}$ and $\gamma|_{\calF_{>b}} = \beta|_{\calF_{>b}}$,
		then $\Ind(\alpha) = \Ind(\beta)$.
		\item If $\alpha^\mathrm{refl}$ on~$\calL^\mathrm{refl}/ \calS^\mathrm{refl}$ 
		is the MQCA obtained by the spatial reflection about any point in the 1d line of an MQCA~$\alpha$ on~$\calL/\calS$, 
		then $\Ind(\alpha^\mathrm{refl}) = - \Ind(\alpha)$.
		\item For any $a \in \ZZ$ there exists $b \in \ZZ$ such that $b < a$ and
		\begin{equation}
			\Ind(\alpha) = \half \Bigl( 
			\dim(\calF_{<a} \cap \bar \alpha^{-1} \calF_{>b}) 
			- 
			\dim(\calF_{<a} \cap \calF_{>b}) 
			\Bigr) . \label{eq:indexFormula}
		\end{equation} 
	\end{enumerate}
	The assumption that $\dim(\calS / (\calS_{<a}\calS_{>b})) < \infty$
	is fulfilled if either
	$\calS$ is locally generated
	or
	$\calS$ consists of all elements of a two-dimensional topological Pauli stabilizer group
	that are supported on a finite width strip.
\end{proposition}

In the formula~\eqref{eq:indexFormula} 
the first term $\dim(\calF_{<a}\cap\bar{\qca}^{-1}\calF_{>b})$ 
counts the number of logical operators 
that are originally represented in $(-\infty,a)$ 
and are mapped to a logical operator represented in $(b,\infty)$ under the MQCA. 
The second term calibrates the index so that it vanishes if $\alpha$ is the identity.

\subsection{Shift on Majorana chain algebra}

As an example of a noninteger index,
we consider $\calL$ generated by $L_i = X_i Z_{i+1}$ on a lattice with one qubit per site.
The generators satisfy the following commutation relations:
\begin{align}
	L_j L_k &= s L_k L_j \quad \text{ where } s = \begin{cases} -1 & \text{if } |j-k| = 1, \\ +1 & \text{otherwise}.  \end{cases}
	\label{eq:anomalous_1D_L}
\end{align}
This algebra is ubiquitous among the examples discussed in the following sections;
for any set of operators $\{L_k | k\in\ZZ\}$ obeying~\eqref{eq:anomalous_1D_L},
regardless of the context that they are constructed,
we call the \textbf{Majorana chain algebra}.%
\footnote{%
	Consider a 1D chain of Majorana zero modes $\{ \gamma_k \,|\, k\in\ZZ \}$, 
	such that the fermion operators obey $\gamma_k^\dag = \gamma_k$ 
	and $\gamma_k \gamma_l + \gamma_l \gamma_k = 2 \delta_{k,l}$.
	Local observables consist of an even number of fermion operators, 
	and are generated by bilinears $\ii \gamma_k \gamma_{k+1}$. 
	These bilinears obey the \MCAlg, hence the name.
}
In fact, 
any translation invariant algebra in one dimension generated by Pauli operators
has commutation relations consisting of three pieces:
those of a commutative subalgebra, 
those of a full Pauli algebra,
and those of a Majorana chain algebra~\cite[IV.22]{nta3}.

Set $\calS = \{ \one \} \subset \calL$.
Suppose an MQCA $\alpha : \calL/\calS \to \calL / \calS$ is given as follows:
\begin{align}
	\qca(L_j) = L_{j+1} \,.
	\label{eq:MC_shift_qca}
\end{align}
To calculate the index, we follow the construction of~$\calF$ in~\eqref{eq:ConstructF}.
The subgroups~$\calL^{/\calS}_{\pm \infty}$ at infinities are both~$\{ \pm \one \}$.
So, $\calF = \calL/\{ \pm \one \}$.
We choose a projection $\nu : \calF \to \calF_{\le 0}$ defined by
\begin{align}
\nu :  \pm L_k \mapsto 
\begin{cases}
\pm L_k  & \text{if } k \le 0, \\
\pm \one & \text{otherwise}.
\end{cases}
\end{align}
The language of vector spaces would make this more clear.
$\calF$ is precisely a vector space of all finite bit strings, one bit for each site~$k$, 
with finitely many nonzero components.
The projection~$\nu$ zeros out all components on sites~$k > 0$.
So, the Fredholm map in terms of the basis vectors~$e_k$,
each of which is the unit bit string vector of a sole nonzero component~$1$ at site~$k \le 0$, is
\begin{align}
	\phi : e_k \mapsto 
	\begin{cases} 
	e_{k+1} & \text{if } k \le -1, \\ 
	0 & \text{if } k = 0.  
	\end{cases}
\end{align}
This is surjective, but has a nonzero kernel spanned by~$e_0$.
Hence, $\ind(\phi) = \dim \ker \phi - \dim \coker \phi = 1 - 0 = 1$,
and
\begin{align}
	\Ind(\alpha) = \half \ind(\phi) = \half.
\end{align}

One can also calculate the index using Prop.~\ref{prop:IndAsFlow}. For any $b\leq a-2$, we have $\dim(\calF_{<a}\cap\calF_{>b})=a-b-2$ (a basis is $\braket{e_{b+1},\cdots,e_{a-2}}$) and $\dim(\calF_{<a}\cap\barqca^{-1}\calF_{>b})=a-b-1$ (a basis is $\braket{e_{b},\cdots,e_{a-2}}$), hence $\Ind(\alpha)=\frac{1}{2}$. Notice that in this example, $(-\infty,a)\cap(b,\infty)$ must be nonempty since $\calL$ is not a direct sum of on-site local groups.

\subsection{MQCA of one-dimensional LR cycles}

Here we show that the MQCA index of any LR cycle in one dimension is an integer.
This result is contrasted with the MQCA example above on the shift on \MCAlg,
which has MQCA index~$\thalf$.

We first analyze the base stabilizer group.
The structure of all stabilizer groups is very complicated in higher dimensions
with numerous examples in fracton phases with or without space translation invariance~\cite{Nandkishore2018,Pretko2020}
and more recently discovered infinite families~\cite{Michnicki2014,Portnoy2023}.
However, known stabilizer groups in one dimension are essentially those of Ising chains,
generated by $Z_i Z_{i+1}^\dag$ on some (possibly infinite) interval on the one-dimensional lattice.
(The hermitian conjugate is redundant for qubits.)
We show that in dimension one this is the only possibility.
An Ising chain over $n$ consecutive sites 
is an $n$ qubit (or qudit) system with a stabilizer group generated by~$Z_i Z_{i+1}^\dag$.
By convention, an Ising chain over one site ($n=1$)
is a qubit (or a qudit) with a stabilizer~$Z$ acting on it.

\begin{figure}[ht]
	\begin{align*} \xymatrix @R=12mm @M=2mm @L=2mm {
		\vcenter{\hbox{\includegraphics[width=100mm]{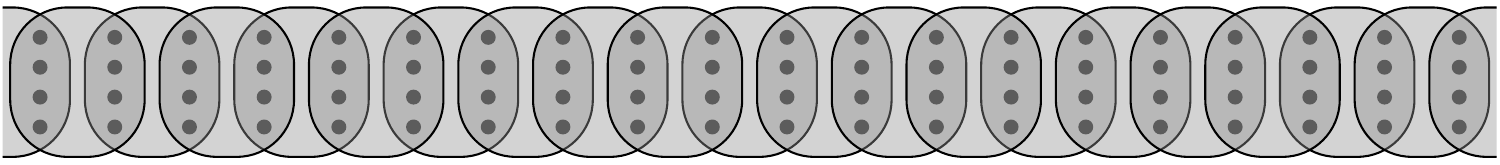}}} \ar[d]^(0.44){\displaystyle \prod_\text{sites $s$} U_s}
	\\	\vcenter{\hbox{\includegraphics[width=100mm]{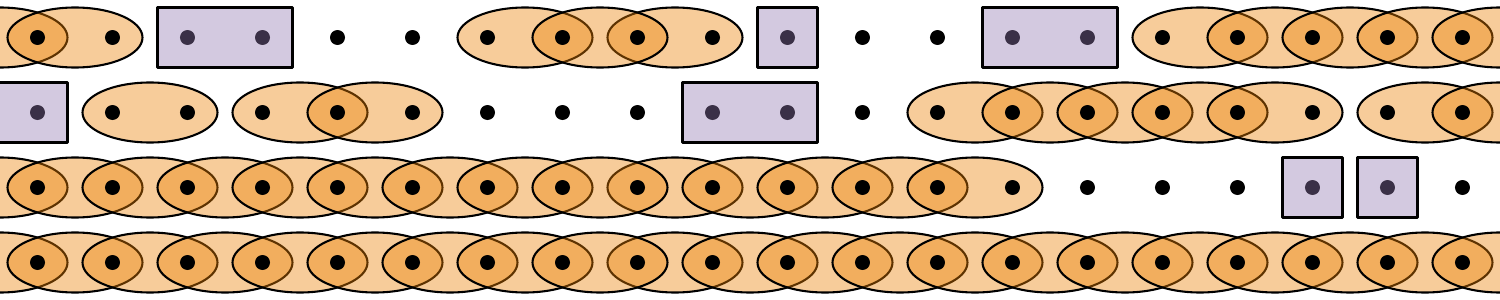}}}
	} \end{align*}
	\caption{%
	Illustration of Theorem~\ref{thm:IsingIn1D}.
	Each site (shown as columns) consists of a finite number of qudits (black dots).
	(Top) Each bubble represents a set of stabilizers supported on two sites.
	(Bottom) Each purple rectangle represent a 1-site stabilizer or 2-site Bell pair.
	The orange ovals represent Ising couplings.
	The theorem says that the stabilizers in the top figure can always be mapped, via conjugation by onsite unitaries,
	to stabilizers belonging to a collection of Ising chains, 1-site stabilizers, and Bell pairs (such as that shown in the bottom figure).
	} \label{fig:ThmIsingIn1D}
\end{figure}

\begin{restatable}{theorem}{IsingOneD}\label{thm:IsingIn1D}
Let there be finite $q_k$ qudits of prime dimension~$p$ 
at each site~$k$ of the one-dimensional lattice~$\ZZ$
or a finite periodic lattice.
Let $\calA$ be a Pauli stabilizer group with generators
acting on at most two neighboring sites.
Then, there exists a Clifford circuit of depth~1, 
consisting of one-site Clifford unitary gates,
by which $\calA$ is mapped to the stabilizer group of a collection of independent Ising chains,
and some completely disentangled qubits and nearest-neighbor Bell pairs.
Each Ising chain may be infinite or finite, 
but each site~$k$ participates in at most $q_k$ Ising chains.
On a finite periodic lattice of length~$L$,
each Ising chain occupies at most $L$ qudits.
\end{restatable}

Though this theorem (see Fig.~\ref{fig:ThmIsingIn1D}) is perhaps not surprising~\cite{BravyiTerhal2009no-go},
it appears that this has only been explicitly proven assuming translation invariance~\cite{Haah2013}.
The assumption that the generators of~$\calA$ acts on at most two neighboring sites
is always satisfied by blocking a few neighboring sites.
Our proof does not assume translation invariance and is more elementary.
Actually, the present theorem applied to translation invariant cases
is stronger than that in~\cite{Haah2013} since our circuit has depth~1.
See~\S\ref{app:IsingIn1D} for two full proofs.

The MQCA of a one-dimensional LR cycle is defined by an obvious choice of~$\calL$ and~$\calS$.
Namely, we set $\calS$ to be the base stabilizer group,
and $\calL$ to be the set of all finitely supported logical operators of $\calS$.
Then, we have an MQCA, which we call the {\bf canonical MQCA} of the one-dimensional LR cycle.

\begin{theorem}\label{thm:IntegerIndexOf1dLRMC}
The canonical MQCA of any one-dimensional LR cycle has an integer MQCA index.
\end{theorem}
\begin{proof}[Proof sketch]
	The base code is equivalent to a collection of independent Ising chains (theorem \ref{thm:IsingIn1D}). Each finite Ising chain is effectively one qubit, while infinite (including semi-infinite) Ising chains are quotient out by definition. The canonical MQCA is essentially a unitary QCA on these Ising qubits, which we know has integer index. See~\S\ref{sec:indexproof}.
\end{proof}

\subsection{Boundary MQCA of two-dimensional topological LR cycles}

In the previous section, we have found $(\dd-1)$-dimensional MQCA
by the action of a topological LR cycle in $\dd$ dimensions
if the measurement circuit admits a vacuum blending.
Here, we specifically consider a~$\dd=2$ topological LR cycle that admits vacuum blendings 
on a right boundary and on a left boundary.
We will show that the boundary MQCA is well defined 
up to the canonical MQCA of a standalone one-dimensional LR cycle,
regardless of how we choose the vacuum blending.
This will imply that the boundary MQCA 
is a topological blending invariant up to the canonical MQCA of standalone one-dimensional LR cycles.
It is plausible that every topological LR cycle admits vacuum blendings on any boundary.

Let~$\calB$ be the base topological code of a 2d topological LR cycle~$(\calB, \circuit B)$.
Suppose that our bulk circuit~$\circuit B$ admits a vacuum blending 
with two boundary components on the left and right of the system with a base code~$\calA$,
so the gates of this vacuum blending are supported on a \emph{strip} with sufficiently large but finite width,
say $1234\ell$.
We assume that the base stabilizer group~$\calA_\strip$ is $\ell$-locally generated 
and that the three properties for topological codes 
hold for all operators supported sufficiently far away, say by distance~$10\ell$, 
from the boundary. See Fig.~\ref{fig:gluedStrip} for the geometry of the regions.

Then, we have three MQCAs: $\alpha_1$, $\alpha_2$, and $\alpha_\strip$.
As in~\S\ref{sec:boundaryactions}, 
the LR cycle will implement a boundary MQCA~$\alpha_1$ on the left boundary, acting on~$\calL_1 / \calA_1$,
and another boundary MQCA~$\alpha_2$ on the right boundary, acting on~$\calL_2 / \calA_2$.
Here, $\calL_i$ consists of all Pauli operators that commutes with~$\calA$ 
and that are finitely supported on the vertically extended interface region~$I_i$.
The subgroup~$\calA_i$ is~$(\calA_\strip)_{I_i}$, the set of all elements of~$\calA_\strip$ that are supported on~$I_i$.
We require no relation between these two MQCA other than that $I_1$ and $I_2$ must be parallel to each other.
In addition, since the overall circuit is effectively on a one-dimensional system, the strip,
with the number of qubits per effective site depending on the distance between~$I_1$ and~$I_2$, 
we have the canonical 1d MQCA $\alpha_\strip$ acting on $\calL_\strip / \calA_\strip$,
where $\calL_\strip$ is the set of all finitely supported Pauli operators on the strip 
that commute with every element of~$\calA_\strip$.

\begin{figure}
	\centering
	\includegraphics[width=60mm]{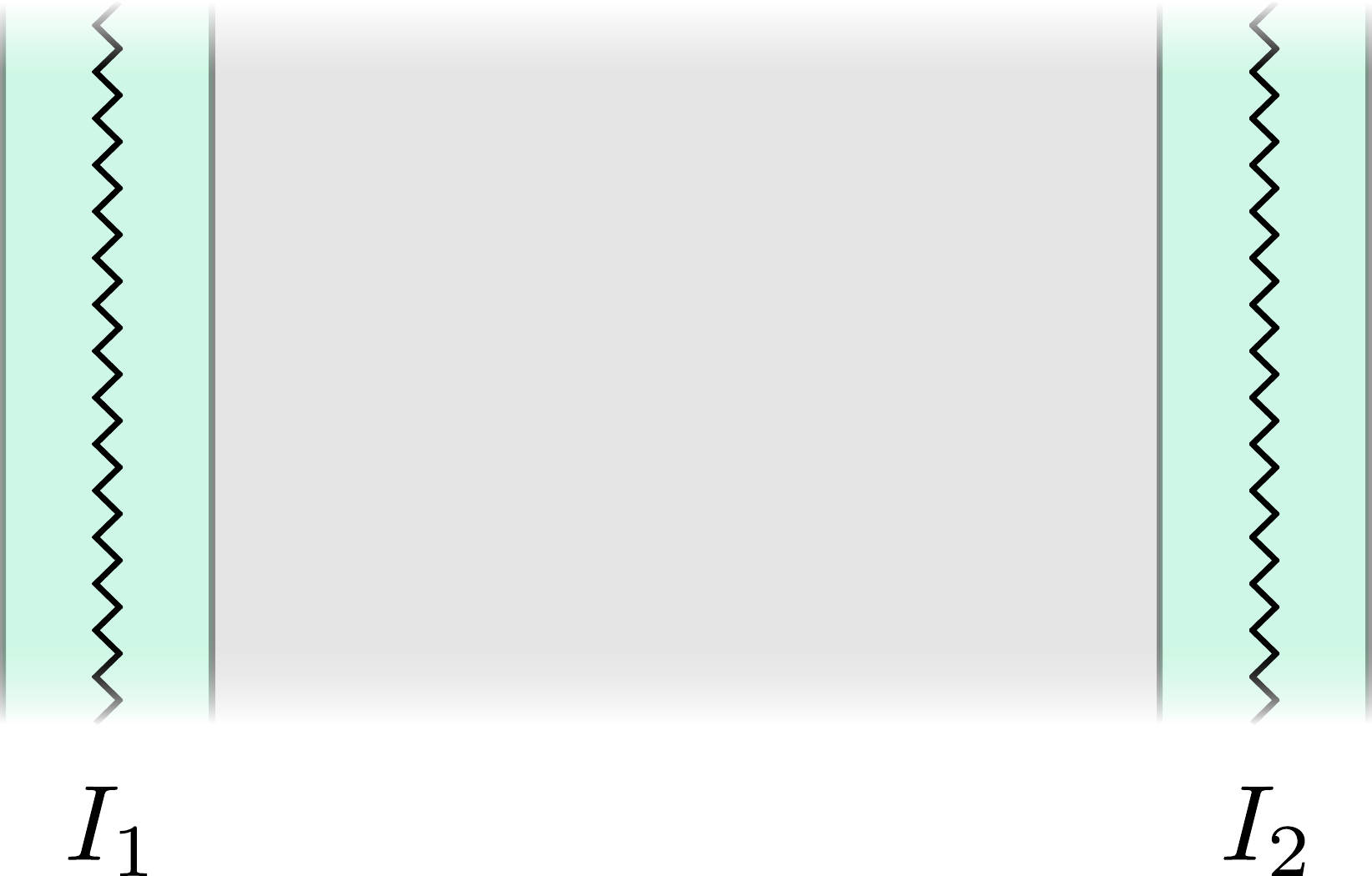}
	\caption{The geometry of the strip and the two interface regions. 
	Every region is infinite vertically, but uniformly bounded widths.}
	\label{fig:gluedStrip}
\end{figure}

\begin{lemma}\label{lem:IndexSeparation}
	$\Ind(\alpha_\strip) = \Ind(\alpha_1) + \Ind(\alpha_2)$.
\end{lemma}
\begin{proof}[Proof sketch]
	The proof can be found in \S\nameref{app:MQCAstrip}.
	Roughly speaking, logical operators at the boundaries~$I_1$ and $I_2$ 
	are still logical operators for the strip, and undergo the same dynamics. 
	The strip may have extra logical operators not supported on either individual boundary, 
	namely, string operators connecting $I_1$ and $I_2$. 
	However, there are only finitely many such string operators
	because the condition~(ii) for topological codes implies that excitations are mobile.
	These finitely many string operators do not contribute to the index.
\end{proof}

Combining~\ref{thm:IntegerIndexOf1dLRMC} and~\ref{lem:IndexSeparation}, we get the following:
\begin{theorem}\label{thm:twodimZ2}
	Suppose that a two-dimensional topological locally reversible measurement cycle~$(\calB, \circuit C)$ 
	admits vacuum blendings on the left and right boundaries.
	Then, the index~$\Ind(\alpha)$ of the MQCA~$\alpha$ on the right boundary (and hence on the left, too)
	is independent of specific vacuum blendings up to~$\ZZ$.
	Therefore, $\Ind(\alpha) + \ZZ \in \thalf \ZZ / \ZZ \cong \ZZ_2$ is an invariant of~$(\calB,\circuit C)$
	under topological blendings.
\end{theorem}

\begin{proof}
	Applying~\ref{lem:IndexSeparation} on an LR cycle~$(\calA,\circuit C)$ with left and right boundaries, 
	we get $\Ind(\alpha_\strip) = \Ind(\alpha_1) + \Ind(\alpha_2)$ 
	where $\alpha_1$ and $\alpha_2$ are the MQCA on the left and right boundary, respectively.
	Suppose we have another LR cycle~$(\calA',\circuit C')$,
	which is identical to $(\calA,\circuit C)$ except at the right boundary. 
	We know $\Ind(\alpha_\strip') = \Ind(\alpha_1) + \Ind(\alpha_2')$.
	Theorem~\ref{thm:IsingIn1D} tells us that $\Ind(\alpha_\strip'), \Ind(\alpha_\strip) \in \ZZ$.
	Hence, $\Ind(\alpha_2') - \Ind(\alpha_2) \in \ZZ$, 
	which means that the index of the boundary MQCA is independent of specific vacuum blendings up to~$\ZZ$.
	
	It remains to explain the invariance of~$\Ind(\alpha_2) + \ZZ \in \thalf \ZZ / \ZZ \cong \ZZ_2$ 
	under topological blendings.
	Suppose that there is a topological blending with a vertical interface region 
	between $(\calB,\circuit C)$ on the left and $(\calB',\circuit C')$ on the right.
	One can also take a right boundary on the right of the interface,
	in which case the the right MQCA index is fully determined by~$(\calB',\circuit C')$.
	The interface between $(\calB,\circuit C)$ and $(\calB', \circuit C')$
	and the elements of~$\circuit C'$ with its right boundary,
	can be regarded as as a blending of~$(\calB,\circuit C)$ into the vacuum.
	The independence of the index on vacuum blendings implies the conclusion.
\end{proof}

In \S\ref{sec:WPT} and \S\ref{sec:HHcode} below,
we will examine two examples, the Wen plaquette translation circuit
and the honeycomb Floquet code.

\begin{corollary}
The Wen plaquette translation circuit and the honeycomb Floquet code
have both nonzero $\ZZ_2$ index.
\end{corollary}

\begin{proof}
The base codes of those LR cycles are finite depth unitary circuit equivalent to the toric code,
and are therefore topological in the sense of Defn.~\ref{defn:TopologicalCode}.
(Any translation invariant base code in 2d with no local logical operator is topological
by the classifcation theorem~\cite{Haah2018}.)
They fulfill the assumptions of~\ref{thm:twodimZ2} by direct calculations below.
\end{proof}

\section{Example of 2d locally reversible cycle: Wen plaquette-translation}
\label{sec:WPT}

Here we define the Wen plaquette-translation (WPT) model.
The Wen plaquette code~\cite{WenP} is a Pauli stabilizer code equivalent to Kitaev toric code.
The WPT model is a locally reversible measurement cycle with period four with the Wen plaquette code as the base code.
The circuit lives on the lattice $\Lambda = \ZZ \times \thalf\ZZ$,
whose sites will be specified by coordinates such as
$(x,y)$ and $(x,y + \thalf )$ with $x,y\in \ZZ$.
The lattice is a square lattice with an extra site at every vertical link.
The base stabilizer group~$\ISG_0$ is given by
\begin{align}
	\ISG_0 &\equiv \Braket{ Z_{x,y} X_{x+1,y} X_{x,y+1} Z_{x+1,y+1} \;,\, Z_{x,\,y+1/2} ~|~ x,y \in \ZZ}.
\end{align}
The sites on $(\ZZ,\ZZ)$ implement Wen's plaquette model realizing the $\ZZ_2$ toric code topological order, 
while the sites on $(\ZZ,\ZZ+\half)$ act as ancillas to facilitate a translation circuit.
The four steps of the circuit implement the (reverse of the) qubit-translation protocol 
of \S\ref{sec:trans1d_LRMC} along each vertical column of the lattice,
parallel to the $y$-coordinate axis.
As a result, the plaquette stabilizers translate by half a unit ($y \mapsto y-\thalf$) after every two steps of the circuit.
They may be represented as follows.
\begin{align} \begin{aligned}
	P_0(x,y) &= Z_{x,y} X_{x+1,y} X_{x,y+1} Z_{x+1,y+1} \,,
\\	P_1(x,y) &= Z_{x,y-1/2} Z_{x,y} X_{x+1,y} X_{x,y+1} Z_{x+1,y+1/2} Z_{x+1,y+1} \,,
\\	P_{t+2}(x,y) &= P_t\big(x,y-\thalf\big) .
\end{aligned} \end{align}
These operators are shown in Fig.~\ref{fig:WPT}.

\begin{figure}[hbt]
	\centering
	\includegraphics[width=.95\linewidth]{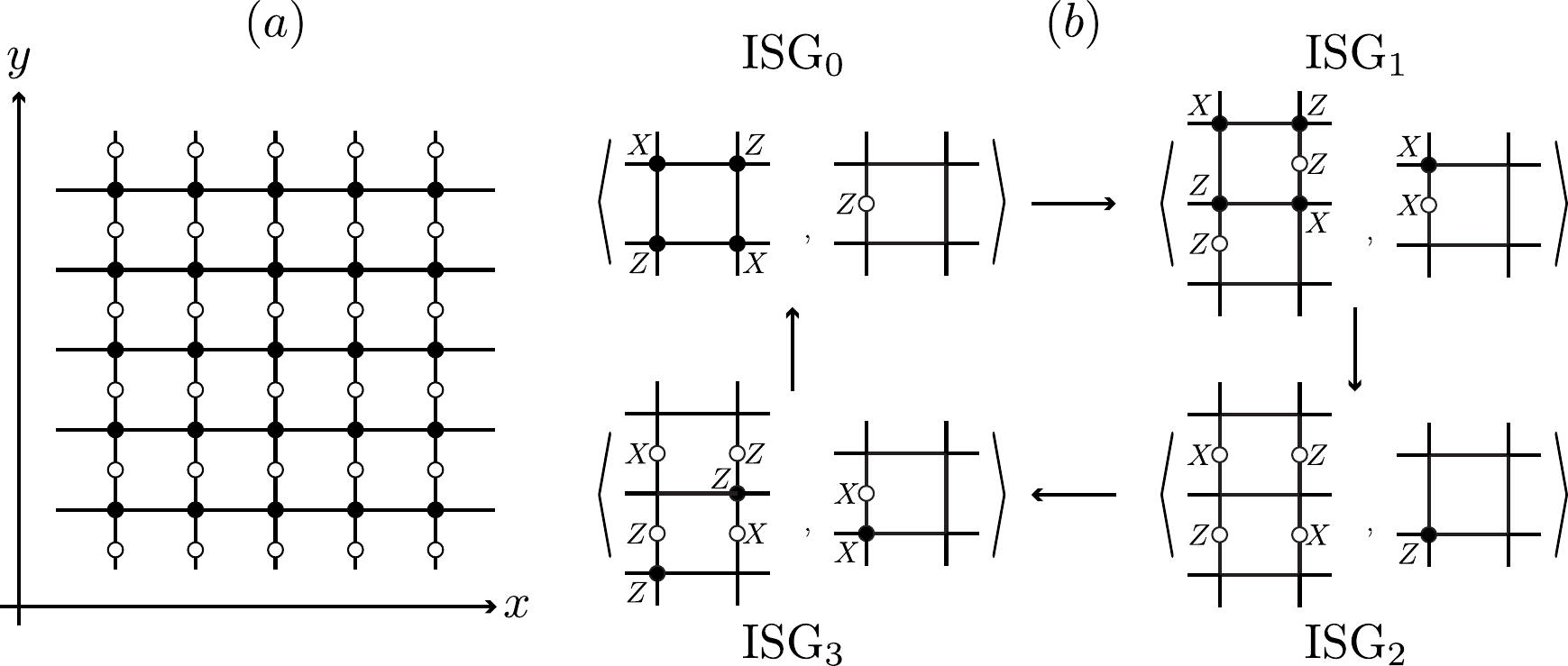}
	\caption{
	In (a) we display the lattice used for the WPT model. The solid dots are data qubits and the open dots are ancilla qubits used in the teleportation protocol. 
	In (b) we provide the ISGs for each step of the measurement sequence.
	}
	\label{fig:WPT}
  \end{figure}

Explicitly, the ISGs of the WPT model are
\begin{subequations} \begin{align}
	\ISG_0 &= \Braket{ P_0(x,y) \;,\, Z_{x,y+1/2} ~|~ x,y \in \ZZ } ,
\\\begin{split}
	\ISG_1 &= \Braket{ P_1(x,y) \;,\, X_{x,y-1/2} X_{x,y} ~|~ x,y \in \ZZ } ,
\end{split}\\\begin{split}
	\ISG_2 &= \Braket{ \vphantom{\Big|} P_2(x,y) \;,\, Z_{x,y} ~|~ x,y \in \ZZ } ,
\end{split}\\\begin{split}
	\ISG_3 &= \Braket{ P_3(x,y) \;,\, X_{x,y} X_{x,y+1/2} ~|~ x,y \in \ZZ } .
\end{split}
\end{align} \end{subequations}
This circuit translates logicals by one unit cell every cycle.
For example, the locally finite product $\prod_n X_{n,n}$ is mapped to $\prod_n X_{n,n-1}$ 
after one cycle of the circuit.

In the following subsections, we construct various boundaries of the WPT circuit.
First, we construct a boundary on the east of the bulk by dropping the measurement gates.
The background stabilizer code is generated by plaquette operators in the west half plane.
The construction is arguably the simplest that gives nontrivial MQCA flow,
naturally induced by the vertical translation of the bulk.
Second, we introduce a boundary at the north
with extra measurements at the boundary.
This shows that the nominal vertical motion of the bulk
does not dictate the flow at the boundary.
These east and north boundaries are examples of vacuum blendings.
We also discuss gluing two copies of the same system along a boundary,
demonstrating how two LR cycles of the same boundary MQCA index
may admit a topological blending.
The discussion there is more generally applicable beyond this WPT example.
In \S\ref{app:WPT}, we present further calculations showing additional vacuum blendings 
along a south boundary and along the east boundary with reversed flow.
Since translation is possible in a standalone 1d LR cycle,
it may seem trivial to change the net flow (the MQCA index),
but our construction does not explicitly bring an extra 1d system.

\subsection{Vertical (right) boundary}
\label{sec:WPT_vright}
We construct a blend of the bulk plaquette circuit ($x \leq 0$) with a trivial circuit ($x > 0$).
The stabilizer groups are
\begin{subequations} \label{eq:WPT_right_blendd} \begin{align}
	\ISG^{(R)}_{0} &= \Big\langle \big\{ P_0(x,y) \,, Z_{x+1,y-1/2} ~\big|~ x\in\ZZ_{-},\, y\in\ZZ \big\} \cup \big\{ Z_{x,y} ~\big|~ x\in\ZZ_+,\, y\in\thalf\ZZ \big\} \Big\rangle \,,
\\	\ISG^{(R)}_{1} &= \Big\langle \big\{ P_1(x,y) \,, X_{x+1,y-1/2} X_{x+1,y} ~\big|~ x\in\ZZ_{-},\, y\in\ZZ \big\} \cup \big\{ Z_{x,y} ~\big|~ x\in\ZZ_+,\, y\in\thalf\ZZ \big\} \Big\rangle \,,
\\	\ISG^{(R)}_{2} &= \Big\langle \big\{ P_2(x,y) \,, Z_{x+1,y} ~\big|~ x\in\ZZ_{-},\, y\in\ZZ \big\} \cup \big\{ Z_{x,y} ~\big|~ x\in\ZZ_+,\, y\in\thalf\ZZ \big\} \Big\rangle \,,
\\	\ISG^{(R)}_{3} &= \Big\langle \big\{ P_3(x,y) \,, X_{x+1,y} X_{x+1,y+1/2} ~\big|~ x\in\ZZ_{-},\, y\in\ZZ \big\} \cup \big\{ Z_{x,y} ~\big|~ x\in\ZZ_+,\, y\in\thalf\ZZ \big\} \Big\rangle \,,
\end{align} \end{subequations}
with $\ISG^{(R)}_{n+4} = \ISG^{(R)}_{n}$.

The local logical algebra of $\ISG^{(R)}_{0}$ are generated by $Z_{0,y} X_{0,y+1}$ along the boundary.
That is,
\begin{align}
	\big(\ISG^{(R)}_{0}\big)^\perp &= \Braket{ L_j \,|\, j\in\ZZ } \ISG^{(R)}_{0} ,
&	L_j &= Z_{0,j} X_{0,j+1} \,.
\end{align}
These logical operators obey the \MCAlgEq, with commutation relation
\begin{align}
	L_j L_k &= \begin{cases} -L_k L_j & |j-k|=1, \\ \phantom{-} L_k L_j & \text{$j=k$ or $|j-k|\geq2$}. \end{cases}
\end{align}
Under a cycle of the WPT circuit, the logical operators undergoes the following dynamics:
\begin{align}
	\begin{aligned}
		Z_{0,j}X_{0,j+1}
		&\xmapsto[(Z_{0,j-1/2})]{\mathbf{0}\to\mathbf{1}}
		Z_{0,j-1/2}Z_{0,j}X_{0,j+1}
		\xmapsto[(X_{0,j+1/2}X_{0,{j+1}})]{\mathbf{1}\to\mathbf{2}}
		Z_{0,j-1/2}Z_{0,j}X_{0,j+1/2}\\
		&\xmapsto[(Z_{0,j-1}Z_{0,j})]{\mathbf{2}\to\mathbf{3}}
		Z_{0,j-1}Z_{0,j-1/2}X_{0,j+1/2}
		\xmapsto[(X_{0,j}X_{0,j+1/2})]{\mathbf{3}\to\mathbf{4}}
		Z_{0,j-1}Z_{0,j-\frac{1}{2}}X_{0,j}\equiv Z_{0,j-1} X_{0,j}
		\,.
	\end{aligned}
\end{align}
Here, $L\xmapsto[(A)]{\mathbf{i}\to\mathbf{i+1}}L'$ means that when we measure $\ISG^{(R)}_{i+1}$,
the logical operator $L$ for $\ISG^{(R)}_{i}$ should be dressed by $A\in\ISG_{i}$ in order to commute with $\ISG^{(R)}_{i+1}$ and survive the measurement,
resulting in the logical operator $L'\equiv LA$ for $\ISG^{(R)}_{i+1}$ (see discussion below Prop.~\ref{thm:Reversibility} for details).

The net effect after one cycle is a translation of logical operators, with MQCA index $-\half$.

\subsection{Horizontal (top) boundary}
\label{sec:WPT_htop}
Again, we construct a blending of the bulk plaquette circuit ($y \leq y_b$) with a trivial circuit ($y > y_b$), with the boundary located at $y_b$.
Define (for $y_b \in \half\ZZ$) the stabilizer group
\begin{align} \begin{aligned}
	\ISG_{T(y_b)} &= \Big\langle \big\{ P_0(x,y) \,, Z_{x,y+1/2} ~\big|~ x\in\ZZ,\, y_b-y\in\ZZ_+ \big\} \cup \big\{ Z_{x,y} ~\big|~ (x,y)\in\Lambda,\, y>y_b \big\} \Big\rangle \,.
\end{aligned} \end{align}
(Recall $\Lambda = \ZZ \times \thalf\ZZ$ is the lattice of qubits.)
Beginning with the boundary at $y_b=0$, a unit translation would take $\ISG_{T(0)} \mapsto \ISG_{T(-1)}$.
The key point in our blending construction is that the direct transition $\ISG_{T(y_b)} \to \ISG_{T(y_b+1)}$ is locally reversible.

Indeed, the generators of $\ISG_{T(y_b)}$ and $\ISG_{T(y_b+1)}$ are identical almost everywhere, except for a horizontal strip near $y_b$.
Their intersection $\calS=\ISG_{T(y_b)}\cap\ISG_{T(y_b+1)}$ is generated by
	$\{ P_0(x,y) \,, Z_{x,y+3/2} \,|\, x\in\ZZ,\, y_b-y\in\ZZ_+ \}$ below
	and $\{ Z_{x,y} \,|\, x\in\ZZ,\, y>y_b+1 \}$ above \mbox{$y_b+1$}.
Above $\calS$, $\ISG_{T(y_b)}$ has extra stabilzers $\braket{ Z_{x,y_b+1} \,|\, x\in\ZZ }$;
$\ISG_{T(y_{b+1})}$ has extra stabilzers $\braket{ P_0(x,y_b) \,|\, x\in\ZZ }$.
$Z_{x,y_b+1}$ and $P_0(x,y_b)$ form conjugate pairs (cf.~Def.~\ref{defn:lrt}):
\begin{align}
	Z_{x,y_b+1}P_0(x',y_b) = (-1)^{\delta_{x,x'}} P_0(x',y_b)Z_{x,y_b+1} \,.
\end{align}

This blending consists of 5 steps, the first four steps implements the bulk translation moving all the plaquette stabilizers one unit away from the boundary;
	the 5\textsuperscript{th} step restores the ISG back to its starting point by measuring $\{P_0(x,-1) \,|\, x\in\ZZ\}$ supported on rows $y=-1,0$.
Explicitly, the blend of the WPT circuit is:
\begin{subequations} \label{eq:WPT_top_blendl} \begin{align}
	\ISG^{(T)}_{0} &= \ISG_{T(0)} \,,
\\	\ISG^{(T)}_{1} &= \Big\langle \big\{ P_1(x,y) \,, X_{x,y+1/2} X_{x,y+1} ~\big|~ x\in\ZZ,\, y\in\ZZ_- \big\} \cup \big\{ Z_{x,y} ~\big|~ (x,y)\in\Lambda,\, y>0 \big\} \Big\rangle \,,
\\	\begin{split}
	\ISG^{(T)}_{2} &= \Big\langle \big\{ P_2(x,y) \,, Z_{x,y+1} ~\big|~ x\in\ZZ,\, y\in\ZZ_- \big\} \cup \big\{ Z_{x,y} ~\big|~ (x,y)\in\Lambda,\, y>0 \big\} \Big\rangle
	\\ &= \ISG_{T(-1/2)} \,,
\end{split} \\
	\ISG^{(T)}_{3} &= \Big\langle \big\{ P_3(x,y) \,, X_{x,y} X_{x,y+1/2} ~\big|~ x\in\ZZ,\, y\in\ZZ_- \big\} \cup \big\{ Z_{x,y} ~\big|~ (x,y)\in\Lambda,\, y\geq0 \big\} \Big\rangle \,,
\\	\ISG^{(T)}_{4} &= \ISG_{T(-1)} \,,
\\	\ISG^{(T)}_{n+5} &= \ISG^{(T)}_{n} .
\end{align} \end{subequations}
For $y \leq -1$, the circuits matches that of the bulk, except with an extra step $\ISG_{5n+4} \to \ISG_{5n+5}$ which does nothing as the two ISGs have identical sets of stabilizers away from the boundary.
Similarly, for $y > 0$, the circuit is stabilized by onsite operators and has trivial dynamics.

The local logical operators of $\ISG^{(T)}_{0}$ are
\begin{align}
	\big(\ISG^{(T)}_{0}\big)^\perp &= \Braket{ L_j \,|\, j\in\ZZ } \ISG^{(T)}_{0} ,
&	L_j &= Z_{j,0} X_{j+1,0} \,.
\end{align}
Again, $\{L_j\}$ obey the \MCAlgEq.
After four steps of the circuit, the logical operators transforms to $L_j \mapsto Z_{j,-1} X_{j+1,-1}$,
	which is equivalent to $Z_{j-1,0} Z_{j,-1} X_{j+1,-1} Z_{j+1,0}$ (under $\ISG^{(T)}_{4}$ stabilizer group).
Upon the the final step, the logical operator become (in the equivalence class of) $Z_{j-1,0} X_{j,0}$.
Hence a cycle of this blending circuit takes every logical operator a unit translation to the left $L_j \mapsto L_{j-1}$; moving quantum information left by a half-qubit.

We can construct a circuit that moves quantum information rightward with a slight modification.
The stabilizer group $\ISG^{(T)}_4$ consists of a row of $Z$-stabilizers along $y=0$;
	in the new circuit, we replace this with a row of $X$-stabilizers, such that it remains locally reversible with $\ISG^{(T)}_3$ and $\ISG^{(T)}_5 = \ISG_{T(0)}$.
\begin{align} \label{eq:WPT_top_blendr} \begin{aligned}
	\ISG^{(T')}_t &= \ISG^{(T)}_t \qquad t \notin 5\ZZ+4 ,
\\	\ISG^{(T')}_{5n+4} &= \Big\langle \big\{ P_0(x,y-1) \,, Z_{x,y+1/2} ~\big|~ x\in\ZZ,\, y\in\ZZ_- \big\}
	\\&\qquad \cup \big\{ X_{(x,0)} ~\big|~ x\in\ZZ \big\} \cup \big\{ Z_{x,y} ~\big|~ (x,y)\in\Lambda,\, y>0 \big\} \Big\rangle \,.
\end{aligned} \end{align}
Under this circuit, the logical operator $L_j = Z_{j,0} X_{j+1,0}$ still maps to $Z_{j,-1} X_{j+1,-1}$ after four steps, and transforms to $L_{j+1} = Z_{j+1,0} X_{j+2,0}$ after a full cycle.

\subsection{Double WPT: period doubling and gluing}
\label{sec:WPT_double}

We now demonstrate that doubling the WPT model results in a boundary that can always be trivialized.
There are two ways to double the LR cycle:
(1) by taking a double-period, letting $T \to 2T$,
or (2) we stack two copies by taking a tensor product of two copies of the WPT circuit.

In the first scenario, doubling the period means that the boundary \MCAlg 
generated by $L_k$ transform as $L_k \mapsto L_{k\pm2}$.
By appending the circuit~\eqref{eq:anomalous_1D_L_shift2} (or its inverse),
we can produce a stationary boundary, i.e., a blend such that the boundary MQCA is the identity map.

In the second scenario, we can ``gap out'' the edge as follows.
Take two identical copies of WPT with a boundary (e.g., \S\ref{sec:WPT_vright}).
Each copy admits a boundary \MCAlg{} generated by $L^{1,2}_k$ acted on by a nontrivial MQCA which shifts $L^{1,2}_k \mapsto L^{1,2}_{k\pm1}$.
We add to the stabilizer groups pairs of logical operators $L^1_k L^2_k$ (which mutually commute!).

This gapping out construction is an instance of the following more general situation.
Consider two copies of a topological LR cycle in $\dd$ dimensions.
For illustrative purposes, we imagine that each system is a two-dimensional sheet.
Suppose that the sheet admits a vacuum blending on the left,
so all the bulk gates stay on the right half,
and the interface region~$I$ is vertically extended but have a uniformly finite thickness.
The blending starts with a background code~$\calA$.
We find the set~$\calL$ of all logical operators of~$\calA$ 
that are supported finitely on~$I$.
Having two identical sheets, we have two identical base groups~$\calA \otimes \one, \one \otimes \calA$,
identical logical groups~$\calL \otimes \one ,\one \otimes \calL$ and identical boundary MQCA~$\alpha \otimes \one,\one \otimes \alpha$.

\begin{figure}[htbp]
   \centering
\includegraphics[width=.9\textwidth]{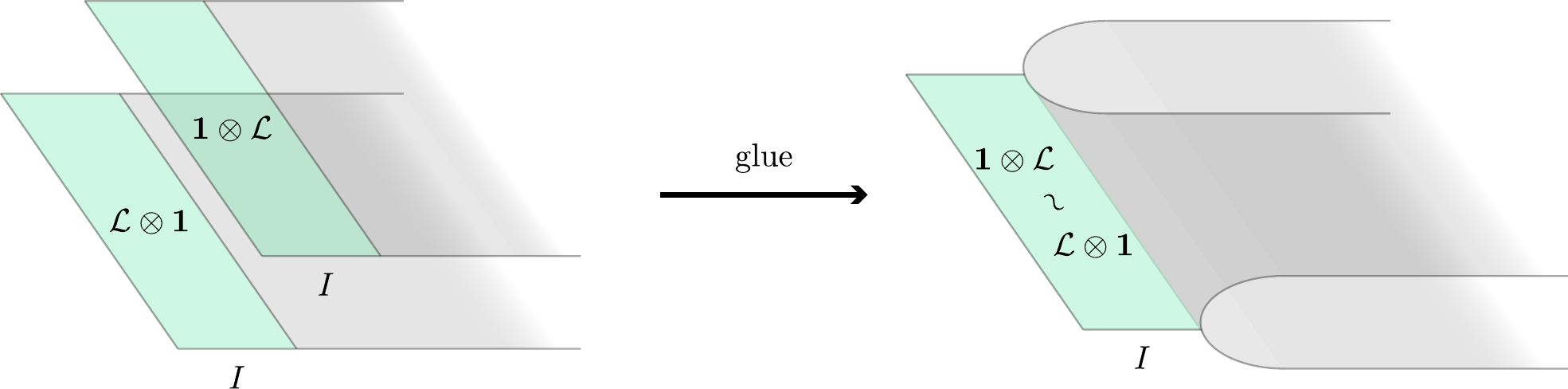}
   \caption{Gluing two identical vacuum blendings along the interface regions.
   Since the two MQCA along the boundaries are identical, 
   the abelian subgroup of the boundary logical group 
   generated by $P \otimes P$ with $P \in \calL$
   remains stationary under the dynamics.}
   \label{fig:glue}
\end{figure}

Leaving intact all the gates of the two-sheet system with the boundaries,
we modify the base code to define
\begin{align}
	\calA_\text{glued} = \Braket{\strut \calA \otimes \one, \quad \one \otimes \calA, \quad \{P \otimes P \,\vert\, P \in \calL \} } .
\end{align}
Then, it is easy to see that $\calA_\glued$ is abelian,%
\footnote{%
	With odd prime dimensional qudits,
	we take the complex conjugate for the second sheet
	to make~$\calA_\text{glued}$ abelian,
	where the complex conjugation is in the basis 
	under which ``Pauli''~$X$ is real and ``Pauli''~$Z$ is diagonal.
	$\braket{\calA \otimes \one, \quad \one \otimes \calA^\ast, \quad \{P \otimes P^\ast \,|\, P \in \calL \} }$.
}
and that $\alpha \otimes \alpha$ maps~$\calA_\glued / (\calA \times \calA)$ into itself.

As the third group of generators for~$\calA_\glued$ 
couples the two-sheet system with virtually all possible operators,
one may expect that $\calA_\text{glued}$ is also a topological code.
However, this is not always true (with the current definition of vacuum blending).
A counterexample is given by a ``floating'' infinite Ising chain,
which sits within the interface region, but which is not acted on by any gates in the circuit.
The only finitely supported logical operator of the Ising chain is a single-site~$Z$,
so the third group of generators for~$\calA_\text{glued}$ will consist of~$Z \otimes Z$,
making an Ising system on a ladder geometry, 
which fails to be a topological code.

If a vacuum blending for a topological LR cycle has a \emph{centerless} 
boundary logical quotient group~$\calL / \calA_I$,
then we can show that there is no nontrivial finitely supported logical operator for~$\calA_\glued$.
By the localization of logical operators for~$\calA$ at the interface region discussed in~\S\ref{sec:boundaryactions},
it suffices to look at a potential logical operator~$P \otimes Q \in \calL \times \calL$.
Since $\calA_\glued$ contains~$Q \otimes Q$,
the potential logical operator is equivalent to~$PQ \otimes \one \in \calL \otimes \one$,
which should commute with all $P' \otimes P' \in \calA_\glued$,
implying that $PQ \in \calL$ commutes with all of~$\calL$.
So, $PQ \in \Cent(\calL) \subseteq \calA$ since $\calL / \calA_I$ is centerless,
and $PQ \otimes \one \in \calA \otimes \one \subseteq \calA_\glued$,
i.e., the potential logical operator is trivial.

We summarize the conclusion here:

\begin{proposition}\label{prop:GluingPristineBoundaries}
For two copies of any vacuum blending of a topological LR cycle
giving the boundary MQCA~$\alpha \otimes \alpha$ on a centerless group at the interface,
there is a new base code~$\calA_\glued$ that includes the base code~$\calA \times \calA$ of the two blendings
such that $\calA_\glued$ has no finitely supported nontrivial logical operator
and $\alpha \otimes \alpha $ takes $\calA_\glued / (\calA \times \calA)$ into itself.
\end{proposition}

\noindent
We obtain a new LR cycle, 
starting with $\calA_\glued$ and evolving 
by the gates of the pre-gluing two identical blendings.
This is mostly easily seen by considering a unitary circuit constructed from the locally reversible circuit.
This circuit starts with the background code $\calA \times \calA$
and certainly can start with a larger background code $\calA_\glued \supseteq \calA \times \calA$.
We do not claim that this new LR cycle is always topological.
It is almost topological for the absence of nontrivial logical operators, 
but we do not know if the third group of generators of the form~$P \otimes P$
is uniformly locally generated,
which would depend on~$\calA$ that is largely unconstrained.

If we unfold the glued region to keep one sheet on the right and put the other sheet on the left,
then $\calA_\glued$ is topological for the WPT model by inspection.
This blending between the LR cycle of WPT model and the space-reflected LR cycle is topological.

\section{The HH honeycomb Floquet code and its boundary dynamics}
\label{sec:HHcode}

\newcommand{\HHzzSiteNumbering}{ \begin{minipage}{70mm} \centering \footnotesize \ensuremath{\xymatrix @!0 @C=4.12mm{ \cdots&\bar3&\bar2&\bar1&0&1&2&3&4&5&6&7&8&9&\cdots }} \end{minipage} }

In the previous section we have constructed an LR cycle based on the Wen plaquette model.
Here, we discuss another example, the Hastings--Haah Honeycomb code (HH code)~\cite{Hastings2021},
realizing an LR cycle,
and calculate the boundary MQCA along flat boundaries.
A boundary condition below will induce a boundary MQCA that has nontrivial spread,
by which the size of operators actually grows linearly in time.

In addition, we show how to ``gap out'' a boundary by doubling the period.
If we double the period of an LR cycle, the boundary MQCA index $\bmod\,\ZZ$ always vanishes.
Then, it is conceivable that we may find a vacuum blending
with no local logical operators at the boundary.
We show that this is indeed the case for the HH code.
This construction gives a measurement circuit of period~6
and is different from the planar honeycomb code of~\cite{Haah2022,Gidney_2022,Paetznick2023}
in that our bulk dynamics is simply repeated in the second half
whereas in~\cite{Haah2022,Gidney_2022,Paetznick2023} it is reversed in the second half.
This is a necessary modification to implement an error correcting code on a finite patch.

In \S\ref{app:HH},
we show a topological blending between the LR cycle of the HH code and that of the WPT model,
a chiral edge dynamics on the boundary of a finite disk,
and a variety of other boundary conditions.

\subsection{The LR cycle of the HH code}

The HH code is a two-dimensional Floquet code defined on any trivalent, plaquette-three-colorable lattice.
For illustrative purposes, the figures will all be illustrated with the honeycomb lattice, although the results apply more generally.
See Fig.~\ref{fig:honeycomb3color} for the coloring scheme.
We will take the plaquette colors to be $R$, $G$, $B$, for red, green, blue, respectively. 
The plaquette coloring induces an edge coloring; since every edge must connect two plaquettes with the same coloring, we denote red edges connect red plaquettes, blue edges blue plaquettes, and green edges green plaquettes.
The edge shared by neighboring red and green plaquette is always blue, the edge shared by neighboring green and blue plaquette is always red, etc.

The model has one qubit per vertex of the lattice.
We denote a set of edge operators by a color and Pauli label: $E_\text{color}^p$ for $\text{color} \in \{R,G,B\}$ and $p \in \{X,Y,Z\}$.
For example $E_R^X$ denotes the set of two-qubit operators given by $XX$ acting on the pair of vertices along red edges.
Similarly, we denote a set of plaquette operators by a color and a Pauli label: $P_\text{color}^p$ for $\text{color} \in \{R,G,B\}$ and $p \in \{X,Y,Z\}$.
For example, $P_G^X$ denotes the set of plaquette operators given by a product of Pauli $X$ operators which reside on the vertices along the boundary of the green plaquettes. 
Denote $\Braket{E_\text{color}^p}$ and $\Braket{P_\text{color}^p}$ as the stabilizer group generated by $E_\text{color}^p$ and $P_\text{color}^p$ respectively.
Notice that $\Braket{P_{c'}^p} \subset \Braket{E_c^p}$ for different colors $c' \neq c$ (see Fig.~\ref{fig:honeycomb3color}).

\begin{figure}[ht]
	\centering
	\begin{minipage}{85mm} \raisebox{33mm}{\small\bf (a)} \; \includegraphics[width=75mm, trim={0 5mm 0 30mm}, clip]{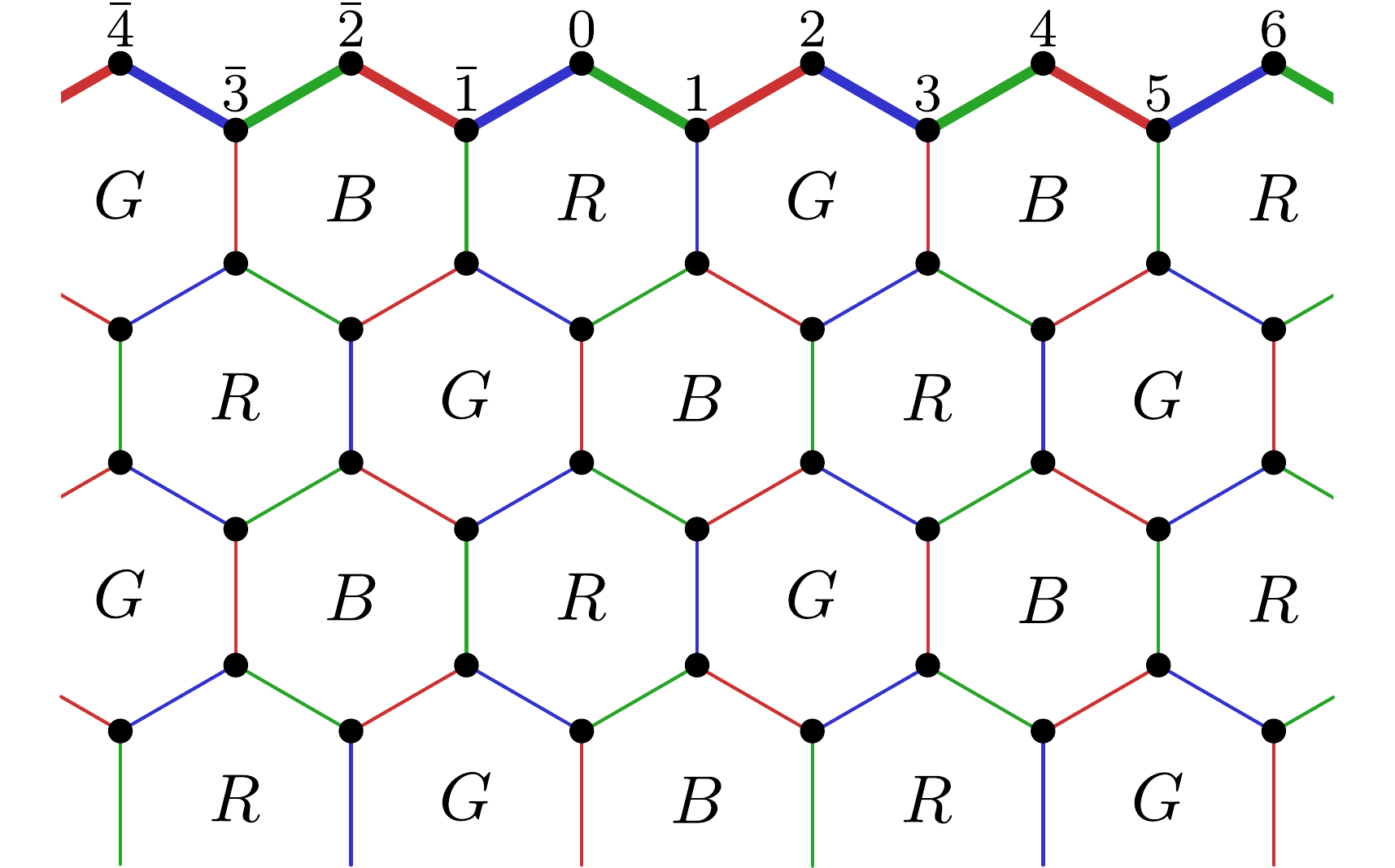} \end{minipage}
	\hspace{10mm}
	\begin{minipage}{44mm} \raisebox{33mm}{\small\bf (b)} \; \includegraphics[width=34mm]{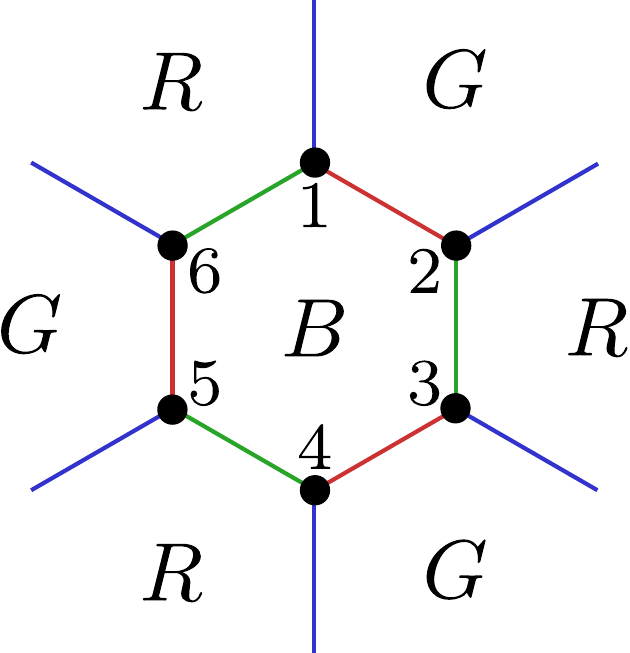} \end{minipage}
	\caption{%
	\textbf{(a)} Plaquette coloring on the honeycomb lattice.
	Each plaquette and edge is colored red ($R$), green ($G$), or blue ($B$).
	\textbf{(b)} A blue plaquette with six vertices. 
	Around the blue plaquette consists of alternating red and green edges.
	} \label{fig:honeycomb3color}
\end{figure}

The HH code consists of a period-three measurement schedule $\dots, \stepR, \stepG, \stepB, \stepR, \dots$ given by 
\begin{align}
	\text{Step }\stepR &: E_R^X \,,
&	\text{Step }\stepG &: E_G^Y \,,
&	\text{Step }\stepB &: E_B^Z \,.
	\label{eq:HHbulk_MeasureSeq}
\end{align}
The three steps ($\stepR$, $\stepG$, $\stepB$) repeats as a ``Floquet code.''
The measurement sequence implements a circuit between ISGs
\begin{align}
	\cdots \xrightarrow{E_R^X} \ISG_\stepR \xrightarrow{E_G^Y} \ISG_\stepG \xrightarrow{E_B^Z} \ISG_\stepB \xrightarrow{E_R^X} \ISG_\stepR \xrightarrow{E_G^Y} \cdots,
\end{align}
where
\begin{subequations} \begin{align}
	\ISG_\stepR &\equiv \Braket{ E_R^X, P_R^X, P_G^Y, P_B^Z },
\\	\ISG_\stepG &\equiv \Braket{ E_G^Y, P_R^X, P_G^Y, P_B^Z },
\\	\ISG_\stepB &\equiv \Braket{ E_B^Z, P_R^X, P_G^Y, P_B^Z }.
\end{align} \end{subequations}
The ISGs are strictly larger than their respective measurement set (e.g.\ $\bigl\langle E_G^Y \bigr\rangle \subset \ISG_\stepG$)
	because there are stabilizers inferred after pairs of measurements which commute with subsequent measurements.
For example measuring $E_R^X$ and then $E_G^Y$ infers the $P_B^Z$ plaquette stabilizers;
	around a blue plaquette $E_R^X$ generates $\prod_{v \in \text{plaq}} X_v$ and $E_G^Y$ generates $\prod_{v \in \text{plaq}} Y_v$, which together puts $\prod_{v \in \text{plaq}} Z _v\in P_B^Z$ in $\ISG_\stepG$.%
	\footnote{If one were to implement the HH code from a ``cold start'' with the measurement sequence~\eqref{eq:HHbulk_MeasureSeq}, it would take a full period to ``warm up'' to reversible transitions between the ISGs.}

We now explicitly show that transitions between the ISGs are locally reversible.
Consider the pair $(\ISG_\stepR, \ISG_\stepG)$.
Their intersection is given by $\ISG_\stepR \cap \ISG_\stepG = \Braket{P_R^X, P_G^Y, P_B^X, P_B^Y, P_B^Z}$.
Observe that every operator in $E_R^X$ and $E_G^Y$ lives completely along the boundary of some blue plaquette, a property of the trivalent, three-colorable graph.
For each blue plaquette with $2n$ vertices, order the vertices in clockwise order $1, 2, \dots, 2n$,
	such that the edge $(1,2)$ (connecting vertices $1$ and $2$) is labeled red,
	illustrated in Fig.~\ref{fig:honeycomb3color}(b).
As such, edges $(2j-1,2j)$ are all red, edges $(2j,2j+1)$ along with $(2n,1)$ are all green.
Within a single blue plaquette,
	the quotient $\ISG_\stepR / (\ISG_\stepR \cap \ISG_\stepG)$ is generated by a basis $\{A_i\} = \{ X_3 X_4, X_5 X_6, \dots, X_{2n-1}X_{2n} \}$; 
	the quotient $\ISG_\stepG / (\ISG_\stepR \cap \ISG_\stepG)$ is generated by basis $\{ Y_2 Y_3, Y_4 Y_5, \dots, Y_{2n-2}Y_{2n-1} \}$.
Both quotients consist of $(n-1)$ generators (not $n$) because the products of all $X$'s (and $Y$'s) are in $\ISG_\stepR \cap \ISG_\stepG$.
For each~$k$ with $2 \leq k \leq n$, the operator $\prod_{j=2}^{2k-1} Y_j \in \ISG_\stepG$ anticommutes with $X_{2k-1}X_{2k}$ but commutes with all other generators of $\{A_i\}$.
Hence the pair $\{A_i\}, \bigl\{ \prod_{j=2}^{2k-1} Y_j \bigr\}$
	form a conjugate $\ell$-local bases%
	\footnote{These bases are $\ell$-local with $\ell=\sup \{n_\text{plaq}\}$, under the graph metric.
		It is possible to lower $\ell$ with alternative choices of conjugate bases.
		For example, the bulk HH code on the honeycomb lattice admits conjugate bases with $\ell=1$.}
	for $(\ISG_\stepR, \ISG_\stepG)$ which satisfies the defintion~\ref{defn:lrt} of locally reversible transitions.

By similar arguments, the circuit $\ISG_\stepG \to \ISG_\stepB$ and $\ISG_\stepB \to \ISG_\stepR$ are also locally reversible.
Hence the HH code is a locally reversible measurement cycle.

\subsection{Truncated zigzag boundary}
\label{sec:HHboundary1}
We now consider various boundaries of the HH code and explicitly compute the evolution of the boundary logical algebra.
For the remainder of this section, we assume the bulk HH code sits on a honeycomb lattice.

\begin{figure}[ht]
	\centering
	\includegraphics[width=90mm, trim={0 90mm 0 0}, clip]{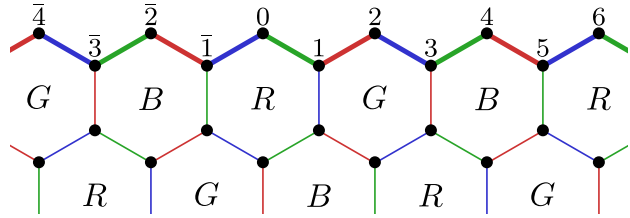}
	\caption{A portion of the HH honeycomb code along with a zigzag boundary exposed. 
	Sites along the boundary are numbered ($\bar{n}$ denotes $-n$).
	} \label{fig:HHboundary_zigzag}
\end{figure}
The zigzag boundary and labeling of vertices as shown in Fig.~\ref{fig:HHboundary_zigzag}.
Consider the measurement circuit consists of both bulk edges ($E_R^X$, $E_G^Y$, or $E_B^Z$), as well as one- and two-body measurements along the boundary.
\begin{align} \label{eq:HHboundary1_MeasSeq} \begin{aligned}
	M_\stepR' &= E_R^X \cup \big\{ X_{6j-2}X_{6j-1} \;,\, X_{6j+1}X_{6j+2} \;,\, X_{6j+0} ~\big|~ j\in\ZZ \big\} ,
\\	M_\stepG' &= E_G^Y \cup \big\{ Y_{6j+0}Y_{6j+1} \;,\, Y_{6j+3}Y_{6j+4} \;,\, Y_{6j+2} ~\big|~ j\in\ZZ \big\} ,
\\	M_\stepB' &= E_B^Z \cup \big\{ Z_{6j-4}Z_{6j-3} \;,\, Z_{6j-1}Z_{6j+0} \;,\, Z_{6j-2} ~\big|~ j\in\ZZ \big\} .
\end{aligned} \end{align}
These edge measurements result from taking a cut through the lattice which has a zigzag boundary as shown in Fig.~\ref{fig:HHboundary_zigzag},
	and then truncating measurement stabilizers along cut edges into a single site stabilizer.
Again, each two subsequent measurement steps determine a bulk plaquette,
	and so the ISG corresponding to each measurement round includes the bulk plaquettes,%
	\footnote{We consider all plaquettes in Fig.~\ref{fig:HHboundary_zigzag} with a color label part of the bulk.}
	bulk edge stabilizers, in additional to the boundary stabilizers.
\begin{align} \label{eq:HHboundary1_ISG} \begin{aligned}
	\ISG_\stepR' &= \Braket{ P_R^X, P_G^Y, P_B^Z, M_\stepR' },  &&&& \vcenter{\hbox{\includegraphics[width=70mm]{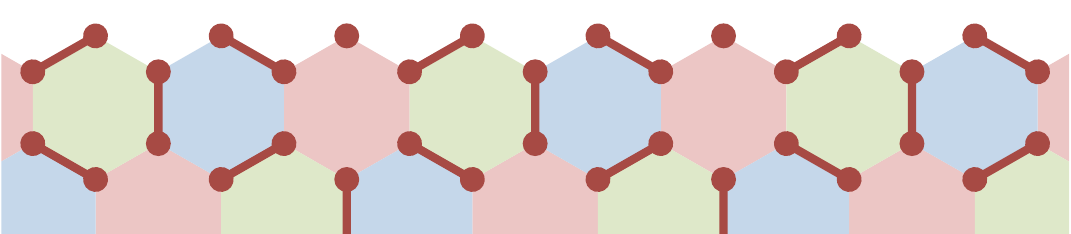}  }}
\\ \ISG_\stepG' &= \Braket{ P_R^X, P_G^Y, P_B^Z, M_\stepG' },  &&&& \vcenter{\hbox{\includegraphics[width=70mm]{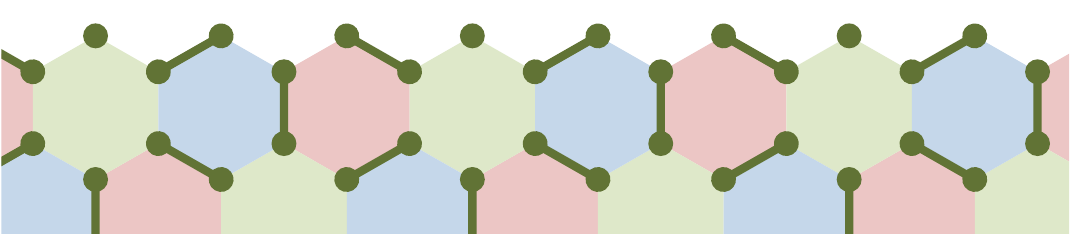}}}
\\	\ISG_\stepB' &= \Braket{ P_R^X, P_G^Y, P_B^Z, M_\stepB' }.  &&&& \vcenter{\hbox{\includegraphics[width=70mm]{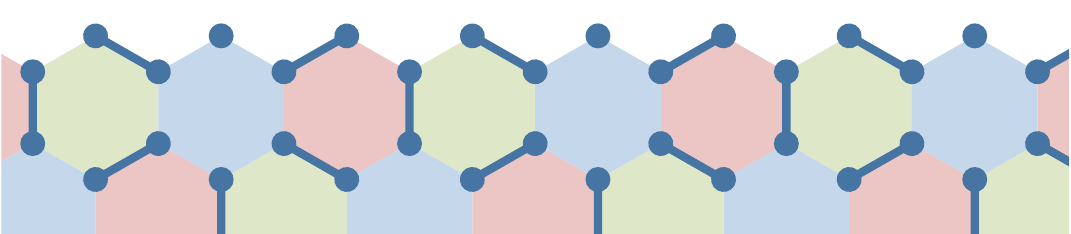}}}
\end{aligned} \end{align}
Here we provide a visualization of the measurements sets and ISGs at each step.
Measurement operators are indicated by the dots;
	a single dot denotes a single-site Pauli operator, connected dots indicate multi-site operators;
	the dark colors red, green, blue corresponds to Pauli $X$, $Y$, $Z$ respectively.
The colored hexagon are used to indicate the presence of certain plaquette operators in the ISG,
	the light colors red, green, blue corresponding to products of Pauli $X$, $Y$, $Z$ respectively.
The ISGs at each step are generated by the plaquette and measurements operators shown.

One can compute the boundary logical algebra for each ISG.
We denote the boundary logical generators at step $\stepR$ via $L^{\stepR\prime}_k$ for $k \in \ZZ$, 
	logicals at step $\stepG$ via $L^{\stepG\prime}_k$ for $k \in \ZZ-\tfrac{1}{3}$,
	logicals at step $\stepB$ via $L^{\stepB\prime}_k$ for $k \in \ZZ-\tfrac{2}{3}$.
Explicitly, they are given by (for $j\in\ZZ$):
\begin{subequations} \label{eq:HHboundary1_logicals} \begin{align}
	L^{\stepR\prime}_{2j    } &= Z_{6j-2}Z_{6j-1}Z_{6j+1}Z_{6j+2} \,,
&	L^{\stepR\prime}_{2j+1  } &= X_{6j+2}X_{6j+3}X_{6j+4} \,,
\\	L^{\stepG\prime}_{2j+2/3} &= X_{6j+0}X_{6j+1}X_{6j+3}X_{6j+4} \,,
&	L^{\stepG\prime}_{2j+5/3} &= Y_{6j+4}Y_{6j+5}Y_{6j+6} \,,
\\	L^{\stepB\prime}_{2j-2/3} &= Y_{6j-4}Y_{6j-3}Y_{6j-1}Y_{6j+0} \,,
&	L^{\stepB\prime}_{2j+1/3} &= Z_{6j+0}Z_{6j+1}Z_{6j+2} \,.
\end{align} \end{subequations}
At each step, the boundary algebra is a \MCAlgEq:
	$L^C_i L^C_j = (-1)^{\delta_{|i-j|,1}} L^C_j L^C_i$ for each step $C$.
The transitions between ISGs at each step induces an isomorphism between these boundary logical algebra.
\begin{align}
	\cdots \mapsto L^{\stepR\prime}_j
	\xmapsto{\stepR\to\stepG} L^{\stepG\prime}_{j-1/3}
	\xmapsto{\stepG\to\stepB} L^{\stepB\prime}_{j-2/3}
	\xmapsto{\stepB\to\stepR} L^{\stepR\prime}_{j-1  }
	\mapsto \cdots
\end{align}
In particular after one measurement period we have $L^C_{k} \to L^C_{k-1}$,
	which results in a translation of half a qubit per measurement period.
Just as the MQCA~\eqref{eq:MC_shift_qca}, this is a nontrivial boundary action,
and has MQCA index $-\half$.

\subsection{MQCA with nontrivial spread}

Example MQCA so far have been translations.
Technically these require nonzero spread parameter,
but the image of an operator does not change its support size.
Here we give an example of boundary MQCA that exhibits 
nontrivial growth of an operator.

\label{sec:HHboundary2}
Consider the following edge measurement sequence on the zigzag edge (Fig.~\ref{fig:HHboundary_zigzag}).
\begin{subequations} \label{eq:HHboundary2_MeasSeq} \begin{align}
	M_\stepR &= \big\{ X_{6j+1} X_{6j+2} \;,\, X_{6j+4} X_{6j+5} ~\big|~ j\in\ZZ \big\} ,
	&&	\vcenter{\hbox{\includegraphics[width=70mm]{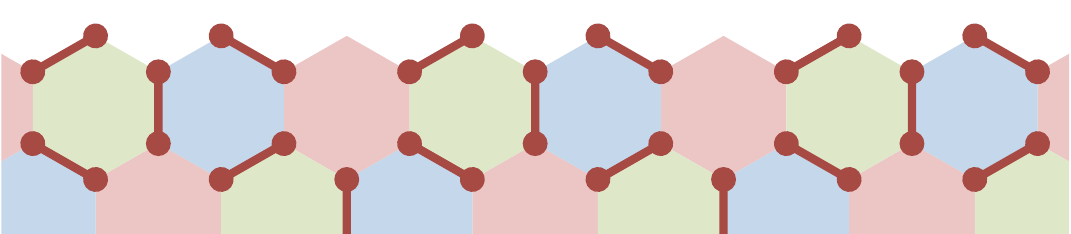}}}
	\label{eq:HHboundary_R2}
\\	M_\stepG &= \big\{ Y_{6j+3} Y_{6j+4} \;,\, Y_{6j+0} Y_{6j+1} ~\big|~ j\in\ZZ \big\} ,
	&&	\vcenter{\hbox{\includegraphics[width=70mm]{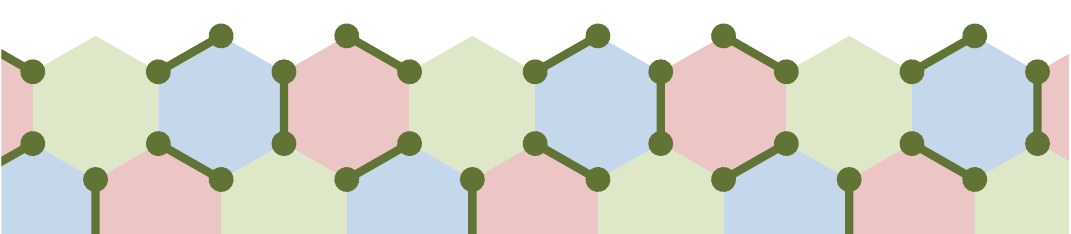}}}
	\label{eq:HHboundary_G2}
\\	M_\stepB &= \big\{ Z_{6j-1} Z_{6j+0} \;,\, Z_{6j+2} Z_{6j+3} ~\big|~ j\in\ZZ \big\} .
	&&	\vcenter{\hbox{\includegraphics[width=70mm]{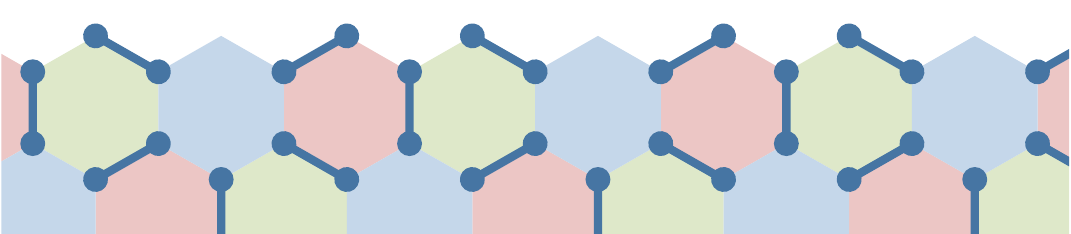}}}
	\label{eq:HHboundary_B2}
\end{align} \end{subequations}
We use the same coloring scheme as Eq.~\eqref{eq:HHboundary1_ISG} to denote the Pauli measurements.
These are the edge measurements that would result from truncating the lattice at the zigzag edge,
	and keeping only measurements where the entire term lives on the lattice that remains.
Notably, this differs from the circuit of \S\ref{sec:HHboundary1} by dropping single-site measurements.

The corresponding ISGs are generated by these measurements operators, along with the bulk plaquette operators.
The three ISGs share a common set of logical operators
\begin{align}
    \big\{ X_{6j+2} X_{6j+3} X_{6j+4} \;,\, Y_{6j-2} Y_{6j-1} Y_{6j} \;,\, Z_{6j} Z_{6j+1} Z_{6j+2} \;\big|\; j\in\ZZ \big\} .
\end{align}
Because these operators commute with all three ISGs, they are fixed points of the boundary MQCA;
	i.e., they map to themselves after a period of the circuit.
In addition,
	$\ISG_\stepR$ has additional logical generators $X_{6j}$,
	$\ISG_\stepG$ has $Y_{6j+2}$, and $\ISG_\stepB$ has $Y_{6j+4}$, for $j\in\ZZ$.
These operators have nontrivial dynamics:
\begin{align}
	X_0
	\xmapsto{\stepR\to\stepG}
	(Z_0 Z_1 Z_2) Y_2
	\xmapsto{\stepG\to\stepB}
	(Z_0 Z_1 Z_2) (X_2 X_3 X_4) Z_4
	\xmapsto{\stepB\to\stepR}
	(Z_0 Z_1 Z_2) (X_2 X_3 X_4) (Y_4 Y_5 Y_6) X_6
	\,.
	\label{eq:HHboundary2_opstring}
\end{align}
In contrast to the previous cases, the support of a logical under the LR cycle increases in size over time.
There is no contradiction here, as the operator growth is at most linear in time.
The logical algebra for $\ISG_\stepR$ is indeed a \MCAlg,
	with assignment
\begin{align}
	L_{4j} &= X_{6j} \,,
&	L_{4j+1} &= Y_{6j} Z_{6j+1} Z_{6j+2} \,,
&	L_{4j+2} &= X_{6j+2} X_{6j+3} X_{6j+4} \,,
&	L_{4j+3} &= Y_{6j+4} Y_{6j+5} Y_{6j+6} \,.
\end{align}
However, a cycle of this circuit implements the automorphism
	$L_{4j} \mapsto L_{4j} L_{4j+1} L_{4j+2} L_{4j+3} L_{4j+4}$,
	$L_{4j+1} \mapsto L_{4j+2} L_{4j+3} L_{4j+4}$,
	$L_{4j+2} \mapsto L_{4j+2}$,
	$L_{4j+3} \mapsto L_{4j+3}$.
A calculation shows that this circuit has boundary MQCA index of $+\half$.

It is enlightening to use the Majorana basis to describe these operators:
Consider a chain of Majorana zero modes $\{\gamma_k \,|\, k\in\ZZ \}$, where $\gamma_k$ are Hermitian fermionic operators obeying anticommutation relations $\gamma_k \gamma_l + \gamma_l \gamma_k = 2 \delta_{k,l}$.
Local observables are comprised of an even number of Majorana operators, which are generated by neighboring bilinears $\ii \gamma_{k} \gamma_{k+1}$.
The neighboring bilinears $A_k = \ii \gamma_{k-1} \gamma_k$ form a \MCAlgEq,  that is, $A_j A_k = (-1)^{\delta_{|j-k|,1}} A_k A_j$.
For this reason, we make the assignment $L_k \cong \ii \gamma_{k-1} \gamma_k$.%
	\footnote{%
		The operators $L_k$ form a representation of~$\ii \gamma_{k-1}\gamma_k$.
		An inverse of this representation 
		(as an algebra homomorphism from the algebra of Majorana operators to a matrix algebra)
		is called the Jordan--Wigner transformation:
		$\gamma_k \sim \prod_{i \leq k} L_i$.}
Then this LR cycle implements an automorphism equilvalent (up to signs) to the map
\begin{align}
	\gamma_k \mapsto \begin{cases} \gamma_{k+4} & k \equiv 0 \pmod{4} \\ \gamma_k & k \not\equiv 0 \pmod{4} \end{cases} .
\end{align}
For example, the circuit transforms the operator $X_0 = L_0 \cong \ii \gamma_{-1} \gamma_0 \mapsto \ii \gamma_{-1} \gamma_4 \cong L_0 L_1 L_2 L_3 L_4$, explaining why the length of the string grows.
We observe that along any cut $\ZZ = \big( (-\infty,\mathfrak{c}) \cap \ZZ \big) \cup \big( [\mathfrak{c},\infty) \cap \ZZ \big)$,
	a single Majorana crosses the cut per cycle (moving from position $< \mathfrak{c}$ to $\geq \mathfrak{c}$).%
	\footnote{We do not make any rigorous claims regarding fermionic (M)QCA here.  However such connection is likely not a coincidence.}

\subsection{Gapped dynamics with period doubled}
\label{sec:HHzz_gapped}

For applications to quantum error correction, it is desirable to find a planar implementation
whose boundary consists of multiple gapped segments.
For a vacuum blending of a LR cycle, we say a particular segment of the boundary is \textbf{gapped} 
if the base code does not admit any nontrivial logical operator 
supported entirely in a small neighborhood of the boundary segment.
Nontrivial logical operators must involve objects that are extended beyond the gapped region.
For example, a string logical operator that ends on a gapped region but extends to the bulk is allowed.

If the boundary MQCA is the identity,
then one can gap out any planar realization by measuring boundary logical operators.
For instance, consider the triangle planar configuration shown in Fig.~\ref{fig:HHtriangle},
where the logical algebra is generated by $\{ L_k \,|\, k \in \ZZ_{2N} \}$ obeying \MCAlg.
One can construct gapped regions along the boundary by adding
	$\{L_{2k} \,|\, k\in[a,b]\}$ or  $\{L_{2k+1} \,|\, k\in[a,b]\}$ to the base stabilizer group.
For the boundary circuits discussed so far, these involves measuring 3-body or higher-weight operators.

For some quantum computing architectures,
it may be advantageous to resolve the boundary logical operators using only weight-2 operators.
The following circuit, which uses additional ancilla qubits (placed next to the ``even'' site along the edge)
	implements a gapped boundary for the zigzag edge.
\begin{align} \begin{split}
	\HHzzSiteNumbering &
\\[0ex]	\vcenter{\hbox{\includegraphics[width=70mm]{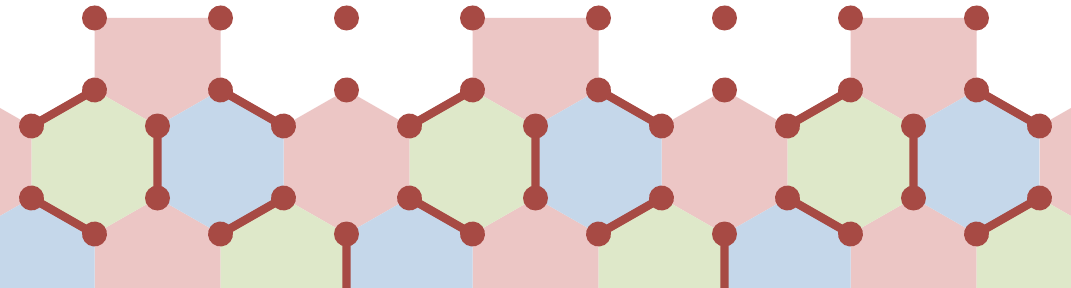}}} & \quad\begin{aligned} M_0 = E_R^X &\cup \big\{ X_{6j-2}X_{6j-1} \,, X_{6j+1}X_{6j+2} \,, \\&\qquad X^{6j-2} \,, X^{6j+0} \,, X^{6j+2} ~|~ j\in\ZZ \big\} , \end{aligned}
\\[1ex]	\vcenter{\hbox{\includegraphics[width=70mm]{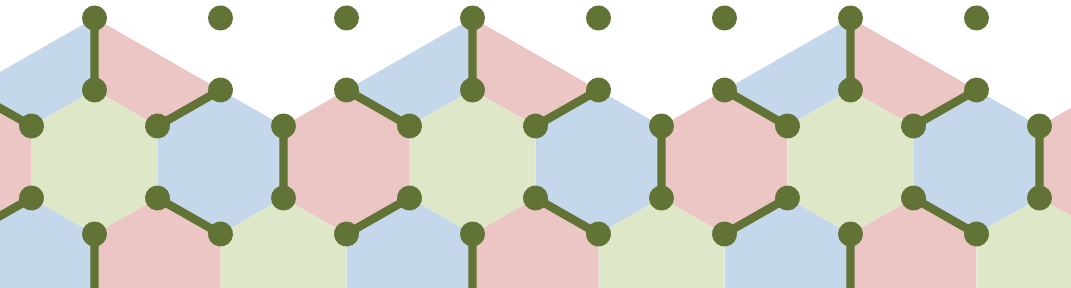}}} & \quad\begin{aligned} M_1 = E_G^Y &\cup \big\{ Y_{6j+0}Y_{6j+1} \,, Y_{6j+3}Y_{6j+4} \,, \\&\qquad Y^{6j+0} \,, Y^{6j+2}Y_{6j+2} \,, Y^{6j+4} ~|~ j\in\ZZ \big\} , \end{aligned}
\\[1ex]	\vcenter{\hbox{\includegraphics[width=70mm]{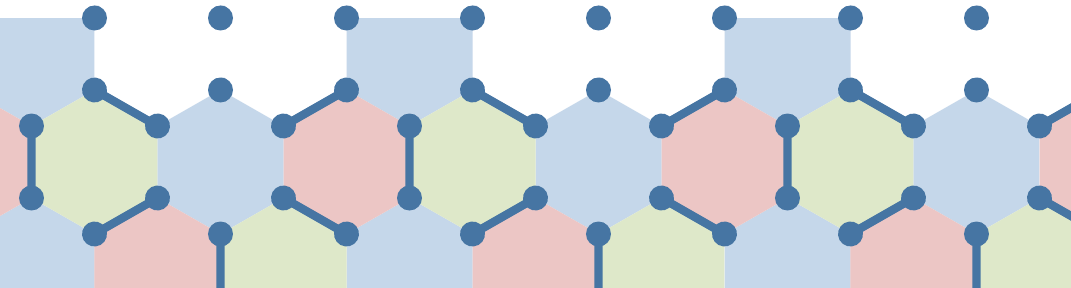}}} & \quad\begin{aligned} M_2 = E_B^Z &\cup \big\{ Z_{6j-4}Z_{6j-3} \,, Z_{6j-1}Z_{6j+0} \,, \\&\qquad Z^{6j-4} \,, Z^{6j-2} \,, Z^{6j+0} ~|~ j\in\ZZ \big\} , \end{aligned}
\\[1ex]	\vcenter{\hbox{\includegraphics[width=70mm]{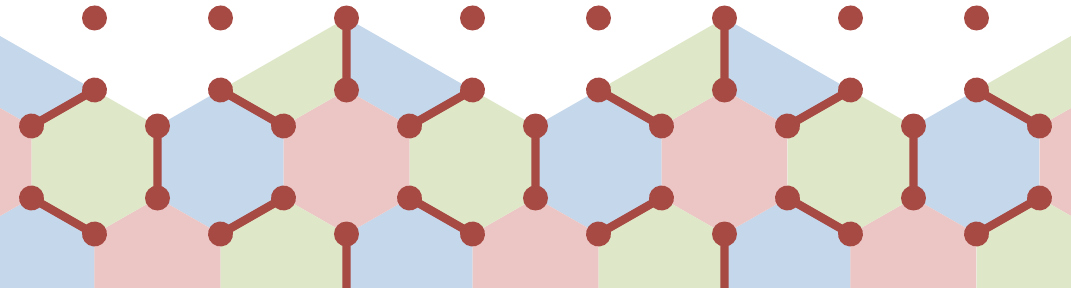}}} & \quad\begin{aligned} M_3 = E_R^X &\cup \big\{ X_{6j-2}X_{6j-1} \,, X_{6j+1}X_{6j+2} \,, \\&\qquad X^{6j-2} \,, X^{6j+0}X_{6j+0} \,, X^{6j+2} ~|~ j\in\ZZ \big\} , \end{aligned} 
\\[1ex]	\vcenter{\hbox{\includegraphics[width=70mm]{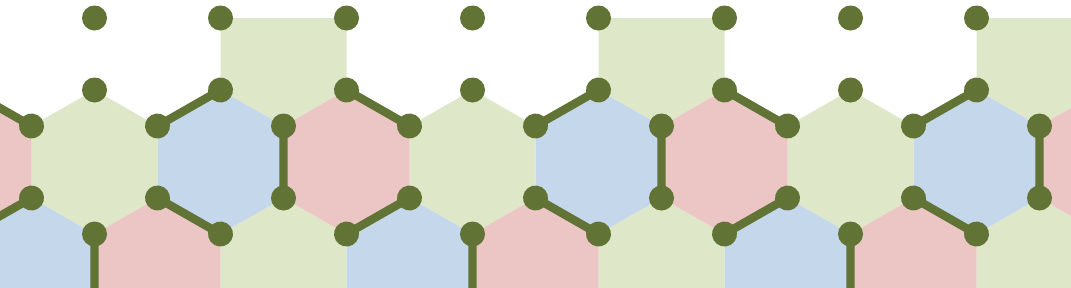}}} & \quad\begin{aligned} M_4 = E_G^Y &\cup \big\{ Y_{6j+0}Y_{6j+1} \,, Y_{6j+3}Y_{6j+4} \,, \\&\qquad Y^{6j+0} \,, Y^{6j+2} \,, Y^{6j+4} ~|~ j\in\ZZ \big\} , \end{aligned}
\\[1ex]	\vcenter{\hbox{\includegraphics[width=70mm]{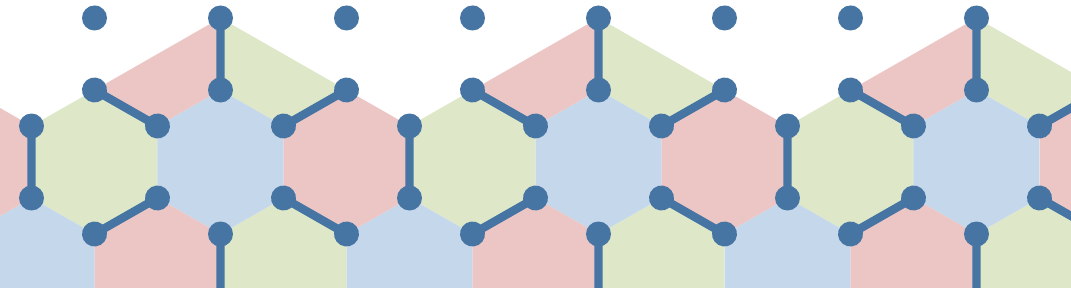}}} & \quad\begin{aligned} M_5 = E_B^Z &\cup \big\{ Z_{6j-4}Z_{6j-3} \,, Z_{6j-1}Z_{6j+0} \,, \\&\qquad Z^{6j-4} \,, Z^{6j-2}Z_{6j-2} \,, Z^{6j+0} ~|~ j\in\ZZ \big\} . \end{aligned}
\end{split} \notag \\ \HHzzSiteNumbering
\end{align}

Consider steps $4 \to 5 \to 0$.
On step~4, $Y_2$, $Y_3Y_4$, and $Y^4$ (acting on ancilla next to site 4) are stablized.
After step~5, where $Z_2Z_3$ and $Z_4Z^4$ are both measured, the product $Y_2Y_3Y_4Y^4$ from step~4 remains in the ISG.
Their combination also place the operator $X_2X_3X_4X^4$ in $\ISG_5$.
On step~0, $X^4$ is measured, which infers the operator $X_2X_3X_4$---a logical operator~\eqref{eq:HHboundary1_logicals} of the circuit~\eqref{eq:HHboundary1_MeasSeq}.
Similar arguments will show that the $Z_0Z_1Z_2 \in \ISG_2$ and $Y_4Y_5Y_6 \in \ISG_4$.
Further detail is given in \S\nameref{app:HHzz:3anc}.
Upon every cycle,
	the bulk plaquette operators are measured \emph{twice},
	the boundary operators are checked \emph{three times}.

A planar error correction code can be realized by executing this boundary sequence along segments of the planar edge,
	with alternating segments offset by 3 measurement steps of the period.%
	\footnote{In the TQFT picture (expounded in the discussion \S\ref{sec:discussion}), the segments alternates with $e$- and $m$-anyon condensation.  Between each adjacent pair of segment traps a Majorana zero mode.}
This scheme utilizes three ancilla/unit cell (for visual aesthetic).
Here each ancilla is only used for a fraction of the cycle;
a similar version with one ancilla per unit cell is given in \S\nameref{app:HHzz:1anc}.

\subsection{Remarks}

Suchara et al.~\cite{Suchara2011constructions} pointed out that the Kitaev honeycomb model
viewed as a static subsystem code does not have any logical qubit.
It seems natural to tripartition the set of edges of the honeycomb lattice based on orientation,\footnote{%
	The tripartition splits the terms of the Kitaev honeycomb model into $X$-, $Y$-, and $Z$-bonds, as described in~\cite{Kitaev_2005}.}
	rather than the Kekule pattern in the HH code.
One may then consider a measurement circuit with period 3, measuring all the edges of one orientation at a time.
We remark that this measurement circuit is \emph{not} locally reversible, as neighboring steps consist of large conjugate elements akin to the iterated teleportation example in~\S\ref{sec:teleportation}.

The bulk $\ZZ_2$-valued index is nonchiral in the sense that it is unchanged under space-reflection or circuit-reversal.
Notably, the HH code is invariant under a reflection that preserves the plaquette coloring.
However, for any termination of the HH code which is period-preserving and locally reversible
	(or any such boundary of any bulk LR cycle with a nontrivial $\ZZ_2$ index)
	the boundary MQCA must take on a nonzero index which indeed breaks reflection symmetry 
	(Prop.~\ref{prop:IndAsFlow}.2).
Along the zigzag boundary employed in the previous subsections,
	reflection symmetry is broken by the plaquette-coloring
	and the boundary circuit has chiral dynamics.
However, along the ``armchair'' boundary shown in Fig.~\ref{fig:HHboundary_armchair}
the edge measurement sequence must either
explicitly break reflection, alter the circuit periodicity, or violate local reversibility.

\begin{figure}[ht]
	\centering
	\includegraphics[width=0.4\linewidth]{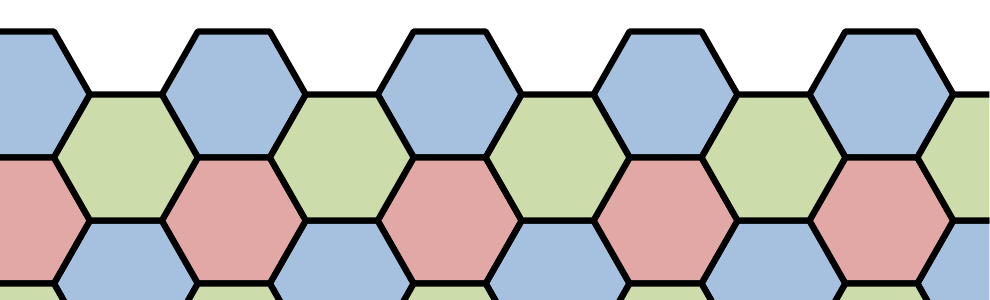} \qquad
	\includegraphics[width=0.4\linewidth]{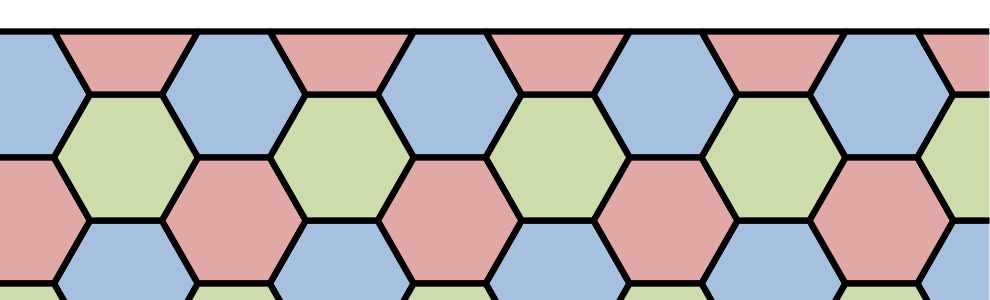}
	\caption{Possible boundary terminations of the HH honeycomb code along the ``armchair'' edge.
	Both geometries are (left-right) reflection symmetric.
	Hence any circuit along these boundary must either
	explicitly break reflection (via the set of measurement operations),
	alter the circuit periodicity, or violate local reversibility.
	}
	\label{fig:HHboundary_armchair}
\end{figure}

For example, Vuillot's planar realization~\cite{Vuillot2021} of the HH code has
	boundary geometry illustrated in Fig.~\ref{fig:HHboundary_armchair}(right).
Labeling the boundary sites by integers,
	the $\stepR$-step consists of measurements $\{X_{2a} X_{2a+1} ~|~ a\in\ZZ\}$,
	while the $\stepB$-step consists of measurements $\{Z_{2a-1} Z_{2a} ~|~ a\in\ZZ\}$.
These two steps, similar to example \S\ref{sec:teleportation}, is not locally reversible.
Appendix\,\nameref{app:HHac:gapped} provides a period-doubled (6-step) LR cycle with gapped boundary, i.e., no nontrivial boundary logical operators.
Appendix\,\nameref{app:HHac:p3r} provides a 3-step LR cycle breaking reflection symmetry.
On the other hand, Ref.~\cite{Haah2022} circumvents the issue by altering the bulk sequence into the period-6 cycle:
	$\stepR \to \stepG \to \stepB \to \stepR \to \stepB \to \stepG \to \stepR$.
	This sequence of length~6 is actually simply equivalent in the sense of~\S\ref{sec:SimpleEquivalence} 
	to the sequence of length zero because the second half is precisely the reverse of the first half.
	Therefore, its $\ZZ_2$ index is zero.

\section{Discussion}
\label{sec:discussion}

By considering locality and reversibility to measurement circuits,
we have introduced measurement quantum cellular automata, 
which enabled the characterization of boundary anomalies in Floquet measurement circuits.
Much of the concepts from the unitary case carry over to measurement circuits:
	locally reversible measurement cycles (LRMC) are analogous to finite depth unitary circuits (FDUC);
	measurement quantum cellular automata (MQCA) generalizes quantum cellular automata (QCA);
	and the MQCA index is a generalization of the GNVW index for Clifford QCA.
If we were to make a comparison between unitary and measurement circuits and their respective QCAs, 
our investigation shows that in one dimension LRMCs can implement any 1d QCA.
Moreover, at the boundary of two-dimensional LRMCs, 
we have observed anomalous effects which cannot occur in purely one-dimensional LRMCs.
In an abuse of notation, we have the hierarchy of classes
\begin{align}
	\mathrm{FDUC} \subset \mathrm{QCA} \cong \mathrm{LRMC} \subset \mathrm{MQCA}
\end{align}
in one dimension.

For topological LR cycles in two dimensions we have defined a $\ZZ_2$-valued index 
under a mild (perhaps redundant) assumption
by considering induced boundary dynamics.
By analyzing standalone one-dimensional LR cycles 
via the classification of one-dimensional Pauli stabilizer codes,
we have shown that the index is well defined for the two-dimensional bulk
regardless of the details of boundaries.
Our $\ZZ_2$ index is an invariant under topological blending equivalence,
but is not complete
for a rather uninteresting reason that our definition of topological blending 
does not allow nontrivial gapped boundaries between base codes.
We expect that a complete set of invariants for topological LR cycles over qubits in two dimensions
is our $\ZZ_2$ index 
combined with a nonnegative integer counting the copies of toric code
that the base code is FDUC equivalent to.
If the latter is zero, then the former should also be zero.
Based on examples, we further expect that our $\ZZ_2$ index completely determines
whether there exists a vacuum blending with no local logical operators at the boundary
for two-dimensional topological LR cycles over qubits.
Note that we have not considered any symmetries (onsite, antiunitary, {\it etc}.),
with which the classification seems much richer~\cite{Jones:DHR:2023}.

We have examined two examples of topological LR cycles,
the WPT model in \S\ref{sec:WPT} and the HH code in \S\ref{sec:HHcode},
each showcasing anomalous boundary dynamics.
The local logical algebras that appear at their boundaries are both isomorphic to the \MCAlg.
Both LR cycles realize a nontrivial bulk automorphism interchanging electric and magnetic logical operators.
It is well known that a translation by one site in the Wen plaquette model 
realizes the nontrivial automorphism on the emergent anyon theory,
which is exactly what the translation circuit implements.
In the HH code, each period pumps an invertible domain wall across the system.
Both of these observations are consistent with the TQFT picture described in Fig.~\ref{fig:meme}.

\begin{figure}[htbp] 
	\centering
	\includegraphics[width=.5\textwidth]{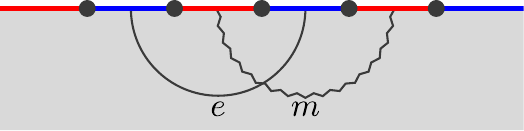} 
	\caption{%
	A boundary of the 2D toric code consisting of alternating $e$ (blue) and $m$ (red) condensates.
	The local logical operators are generated by $e$- and $m$-strings which connect adjacent boundaries.}
	\label{fig:meme}
\end{figure}

The topological quantum field theory picture of this index is rather simple. 
In Fig.~\ref{fig:meme} we show a boundary of the 2D toric code described by alternating $e$ and $m$ condensates.
Between every neighbouring pairs of $e$ and $m$ edges (illustruated by the black dots in Fig.~\ref{fig:meme}) is a Majorana zero mode trapped at the domain wall.
Local logical operators are generated by $e$- and $m$-strings which connect second-neighboring boundaries of the same type, as shown in Fig.~\ref{fig:meme}.
These local logical operators obey \MCAlgEq{} and are analogues of the operators $\{L_{2j}\}$ and $\{L_{2j+1}\}$ 
discussed in \S\ref{sec:WPT} and \S\ref{sec:HHcode}.
Now consider the case where we have an LR cycle whose base ISG is described by Fig.~\ref{fig:meme} 
(or anything equivalent up to a finite depth unitary circuit). 
If the circuit implements a bulk automorphism $\varphi$ that is nontrivial (i.e., switches $e$ and $m$), 
then every $e$-string must map to an $m$-string and vice versa.
However, because the base ISG is fixed, the boundary pattern of $e$- and $m$-condensation does not change.
This means that the boundary logical algebra (generated by boundary string operators) must transform nontrivially under the automorphism~$\varphi$.
One way this can happen while maintaining locality is a translation by a Majorana;
	every $e$- and $m$-string maps to the string immediately to its left (or right),
	corresponding to the boundary MQCA index $\Ind \alpha_\varphi = \pm\half$.
(Here $\alpha_\varphi$ is some boundary MQCA induced by the bulk topological LR cycle 
which implements the topological automorphism~$\alpha$.)
This anomalous boundary action is a consequence of the nontrivial bulk automorphism
and a bulk-boundary correspondence between bulk and edge dynamics.
In light of the fact that Majoranas have quantum dimension~$\sqrt{2}$,
we can exponentiate the MQCA index with base~2: $I(\alpha) = 2^{\Ind\alpha}$.
We may interpret the exponential MQCA index as the quantum dimension of the domain walls, modulo an integer power of~2.

Our results generalize to Pauli stabilizers on $p$-dimensional qudits for any prime $p$.
In dimension one, every LR cycle has an integer MQCA index;
Theorem~\ref{thm:IntegerIndexOf1dLRMC} does not use anything special about the qubits having dimension~2,
but it is important that Pauli groups modulo phase factors are vector spaces, 
and hence we need prime dimensional qudits.
Technically, we have used the fact that any subspace of a vector space is a direct summand,
which is not always true over a coefficient ring that is not a field.
The WPT model of \S\ref{sec:WPT} extends to $\ZZ_p$ in a straightforward manner, with boundary algebra obeying
\begin{align} L_j L_k = s_{j,k} L_k L_j \quad
	\text{where } s_{j,k} = \begin{cases} \exp\mathopen\big[ \frac{2\pi i}{p} (j-k) \big] & |j-k|=1 , \\ 1 & \text{otherwise} . \end{cases} \end{align}
A cycle of the model again transforms $L_k \mapsto L_{k+1}$ with MQCA index of~$\half$.
Hence, we have a $\ZZ_2$ index for a class of 2d topological LR cycle for each prime~$p$,
where the class is specified by the conditions of Theorem~\ref{thm:twodimZ2}.
More generally, working within Pauli stabilizers for $N$-dimensional qudits for any positive integer $N$,
it is conceivable that the bulk topological LR cycle in two dimensions have an invariant valued in
$
	\bigoplus_{\text{primes $p|N$}} \ZZ_2
$,
a $\ZZ_2$-index for every prime factor of $N$.
If this is going to be true, different prime dimensional qudits should not mix up under any dynamics.

We can also extend these ideas to fermionic systems.
In the unitary case, the main difference between the fermionic and bosonic case is that the (Clifford) f-QCA GNVW index is quantized in units of $\half$~\cite{fermionGNVW2}.
We can define f-LRMCs, where the stabilizers are bosonic (i.e., consists of an even number of fermion operators),
	and f-MQCA as a locality-preserving automorphism of a fermionic algebra (or its quotient).
Suppose that we can also generalize the MQCA index to its fermionic counterpart,
	then as $\mathrm{LRMC} \subseteq \mathrm{f\mbox{-}LRMC}$, we can construct one-dimensional circuits with arbitrary integer f-MQCA index.
It is conceivable that the f-MQCA index must be an integer for a pure 1d circuit by the following argument.
If we write a 1d fermionic chain in the Majorana basis, then the bosonic sector of the algebra is generated by fermion bilinears, which is exactly the \MCAlgEq.
Any f-LRMC must be constructed from operators in the \MCAlg.
If it were possible to construct an f-LRMC $\circuit C_f$ with f-MQCA index of $\half$,
we would be able to construct a ``representation'' circuit $\circuit C_b$ on the boundary of Wen's plaquette model 
from its boundary logical algebra that is a representation of the algebra of fermion bilinears.
This would change the boundary MQCA index by $\half$, in direct contradiction with Theorem~\ref{thm:twodimZ2}, 
one of our main results.
We can give another argument. We can think of the translation by measurements
as a unitary circuit on a background code that has a translation invariant state in a subsystem.
For the MQCA index to be a half, we would be shifting the operator algebra by just one Majorana mode,
but there is no 1d state specified by a complete set of commuting operators
that is invariant under one-Majorana translation.
Thus, we expect $\mathrm{LRMC} \cong \mathrm{f\mbox{-}LRMC}$ in one dimensions.%
	\footnote{The equivalence probably holds in all dimensions, since it is possible to construct emergent fermions with bosonic systems in $\dd \geq 2$.}

Throughout this work we have studied the dynamics of measurement circuits without noise or error considerations.
We observe a pattern that nonlocally reversible codes such as
	the Bacon--Shor code~\cite{Bacon2006} 
	viewed as a period 2 sequence of measurements of the ``gauge'' operators
	or Vuillot's boundary of the HH code~\cite{Vuillot2021}
	are susceptible to errors~\cite{NappBaconShor2012, Vuillot2021} 
	and do not have an error threshold in the thermodynamic limit.
Akin to the iterated teleportation example in \S\ref{sec:teleportation}, the nonlocality of the conjugate bases means that
	tracking the dynamics of logical operators require classical corrections
	that depends on a long chain of results from the measurements.
As system size grows, cummulative measurement errors would become uncorrectable.
It is worth exploring if for a topological code the local reversibility 
provides an easily-checkable condition for the existence of an error threshold.

Finally, we envisage the qualifications of a \emph{Floquet topological code}.
Suppose we interpret a ``Floquet code'' as one which involves some measurement schedule between noncommuting stabilizers.
We see that many subsystem codes, such as the Bacon--Shor code~\cite{Bacon2006},
	Bombin's ruby lattice code~\cite{Bombin2010TopoSubsysCode},
	and the 3-qubit subsystem surface code~\cite{Bravyi2012} can be categorized as Floquet.
We can let Floquet topological code to mean any LR cycle with a topological base stabilizer code,
	which would certainly include many of the recent works~\cite{Aasen2022, Davydova2022, Kesselring2022}.
Of course, one can always ``Floquetize'' a static topological error correction code
	by breaking up high-weight operators into sequences of 2- and 3-body operators while retaining the original code's topological character~\cite{Bravyi2012, Chao2020optimizationof, Gidney2022pentagon}.
One could further demand that the measurement circuit implements a nontrivial automorphism among the logical operators,
	which would rule out such Floquetized codes; but would also rule out the CSS honeycomb Floquet code~\cite{Davydova2022, Kesselring2022}!

\textit{Note}:
During the preparation of this manuscript, we became aware of~\cite{Sullivan2023, Roberts2023}
	that also studies the boundary dynamics of Floquet codes, using complementary methods.
Ref.~\cite{Sullivan2023} studies a \emph{unitarized} version of the HH code written in terms of fermions 
and characterizes the boundary dynamics of the unitary circuit via the ``chiral unitary index''~\cite{fermionGNVW1, fermionGNVW2}.

\begin{acknowledgments}
D.A. is grateful for illuminating discussions with Margarita Davydova, Andrew Potter, and Zhenghan Wang.
Z.L. thanks Beni Yoshida for helpful discussions.
R.M. acknowledges Margarita Davydova and Nathanan Tantivasadakarn for illuminating conversations.

Z.L. is supported by Perimeter Institute;
research at Perimeter Institute is supported in part by the Government of Canada through the Department of Innovation, Science and Economic Development and by the Province of Ontario through the Ministry of Colleges and Universities.
R.M. is supported by US National Science Foundation grant No.~DMR-1848336.
\end{acknowledgments}

\phantomsection
\addcontentsline{toc}{section}{References}
\nocite{apsrev42Control}
\bibliographystyle{apsrev4-2}
\bibliography{mqcarefs}

\begin{thebibliography}{61}%
\makeatletter
\providecommand \@ifxundefined [1]{%
 \@ifx{#1\undefined}
}%
\providecommand \@ifnum [1]{%
 \ifnum #1\expandafter \@firstoftwo
 \else \expandafter \@secondoftwo
 \fi
}%
\providecommand \@ifx [1]{%
 \ifx #1\expandafter \@firstoftwo
 \else \expandafter \@secondoftwo
 \fi
}%
\providecommand \natexlab [1]{#1}%
\providecommand \enquote  [1]{``#1''}%
\providecommand \bibnamefont  [1]{#1}%
\providecommand \bibfnamefont [1]{#1}%
\providecommand \citenamefont [1]{#1}%
\providecommand \href@noop [0]{\@secondoftwo}%
\providecommand \href [0]{\begingroup \@sanitize@url \@href}%
\providecommand \@href[1]{\@@startlink{#1}\@@href}%
\providecommand \@@href[1]{\endgroup#1\@@endlink}%
\providecommand \@sanitize@url [0]{\catcode `\\12\catcode `\$12\catcode
  `\&12\catcode `\#12\catcode `\^12\catcode `\_12\catcode `\%12\relax}%
\providecommand \@@startlink[1]{}%
\providecommand \@@endlink[0]{}%
\providecommand \url  [0]{\begingroup\@sanitize@url \@url }%
\providecommand \@url [1]{\endgroup\@href {#1}{\urlprefix }}%
\providecommand \urlprefix  [0]{URL }%
\providecommand \Eprint [0]{\href }%
\providecommand \doibase [0]{https://doi.org/}%
\providecommand \selectlanguage [0]{\@gobble}%
\providecommand \bibinfo  [0]{\@secondoftwo}%
\providecommand \bibfield  [0]{\@secondoftwo}%
\providecommand \translation [1]{[#1]}%
\providecommand \BibitemOpen [0]{}%
\providecommand \bibitemStop [0]{}%
\providecommand \bibitemNoStop [0]{.\EOS\space}%
\providecommand \EOS [0]{\spacefactor3000\relax}%
\providecommand \BibitemShut  [1]{\csname bibitem#1\endcsname}%
\let\auto@bib@innerbib\@empty
\bibitem [{\citenamefont {Knill}(2004)}]{Knill2004a}%
  \BibitemOpen
  \bibfield  {author} {\bibinfo {author} {\bibfnamefont {E.}~\bibnamefont
  {Knill}},\ }\href@noop {} {\bibinfo {title} {Fault-tolerant postselected
  quantum computation: Schemes}} (\bibinfo {year} {2004}),\ \Eprint
  {https://arxiv.org/abs/quant-ph/0402171} {arXiv:quant-ph/0402171}%
  \BibitemShut {NoStop}%
\bibitem [{\citenamefont {Bravyi}\ and\ \citenamefont
  {Kitaev}(2005)}]{Bravyi2005}%
  \BibitemOpen
  \bibfield  {author} {\bibinfo {author} {\bibfnamefont {Sergey}\ \bibnamefont
  {Bravyi}}\ and\ \bibinfo {author} {\bibfnamefont {Alexei}\ \bibnamefont
  {Kitaev}},\ }\bibfield  {title} {\bibinfo {title} {Universal quantum
  computation with ideal {C}lifford gates and noisy ancillas},\ }\href
  {https://doi.org/10.1103/PhysRevA.71.022316} {\bibfield  {journal} {\bibinfo
  {journal} {Phys. Rev. A}\ }\textbf {\bibinfo {volume} {71}},\ \bibinfo
  {pages} {022316} (\bibinfo {year} {2005})},\ \Eprint
  {https://arxiv.org/abs/quant-ph/0403025} {arXiv:quant-ph/0403025}%
  \BibitemShut {NoStop}%
\bibitem [{\citenamefont {Fisher}\ \emph {et~al.}(2023)\citenamefont {Fisher},
  \citenamefont {Khemani}, \citenamefont {Nahum},\ and\ \citenamefont
  {Vijay}}]{Fisher2023}%
  \BibitemOpen
  \bibfield  {author} {\bibinfo {author} {\bibfnamefont {Matthew~P.A.}\
  \bibnamefont {Fisher}}, \bibinfo {author} {\bibfnamefont {Vedika}\
  \bibnamefont {Khemani}}, \bibinfo {author} {\bibfnamefont {Adam}\
  \bibnamefont {Nahum}},\ and\ \bibinfo {author} {\bibfnamefont {Sagar}\
  \bibnamefont {Vijay}},\ }\bibfield  {title} {\bibinfo {title} {Random quantum
  circuits},\ }\href {https://doi.org/10.1146/annurev-conmatphys-031720-030658}
  {\bibfield  {journal} {\bibinfo  {journal} {Annual Review of Condensed Matter
  Physics}\ }\textbf {\bibinfo {volume} {14}},\ \bibinfo {pages} {335--379}
  (\bibinfo {year} {2023})}\BibitemShut {NoStop}%
\bibitem [{\citenamefont {Gottesman}\ and\ \citenamefont
  {Chuang}(1999)}]{Gottesman1999}%
  \BibitemOpen
  \bibfield  {author} {\bibinfo {author} {\bibfnamefont {Daniel}\ \bibnamefont
  {Gottesman}}\ and\ \bibinfo {author} {\bibfnamefont {Isaac~L.}\ \bibnamefont
  {Chuang}},\ }\bibfield  {title} {\bibinfo {title} {Demonstrating the
  viability of universal quantum computation using teleportation and
  single-qubit operations},\ }\href {https://doi.org/10.1038/46503} {\bibfield
  {journal} {\bibinfo  {journal} {Nature}\ }\textbf {\bibinfo {volume} {402}},\
  \bibinfo {pages} {390--393} (\bibinfo {year} {1999})},\ \Eprint
  {https://arxiv.org/abs/quant-ph/9908010} {arXiv:quant-ph/9908010}%
  \BibitemShut {NoStop}%
\bibitem [{\citenamefont {Briegel}\ and\ \citenamefont
  {Raussendorf}(2001)}]{BriegelRaussendorf}%
  \BibitemOpen
  \bibfield  {author} {\bibinfo {author} {\bibfnamefont {Hans~J.}\ \bibnamefont
  {Briegel}}\ and\ \bibinfo {author} {\bibfnamefont {Robert}\ \bibnamefont
  {Raussendorf}},\ }\bibfield  {title} {\bibinfo {title} {Persistent
  entanglement in arrays of interacting particles},\ }\href
  {https://doi.org/10.1103/PhysRevLett.86.910} {\bibfield  {journal} {\bibinfo
  {journal} {Phys. Rev. Lett.}\ }\textbf {\bibinfo {volume} {86}},\ \bibinfo
  {pages} {910} (\bibinfo {year} {2001})},\ \Eprint
  {https://arxiv.org/abs/quant-ph/0004051} {arXiv:quant-ph/0004051}%
  \BibitemShut {NoStop}%
\bibitem [{\citenamefont {Nielsen}(2003)}]{Nielsen2003}%
  \BibitemOpen
  \bibfield  {author} {\bibinfo {author} {\bibfnamefont {Michael~A.}\
  \bibnamefont {Nielsen}},\ }\bibfield  {title} {\bibinfo {title} {Quantum
  computation by measurement and quantum memory},\ }\href
  {https://doi.org/https://doi.org/10.1016/S0375-9601(02)01803-0} {\bibfield
  {journal} {\bibinfo  {journal} {Physics Letters A}\ }\textbf {\bibinfo
  {volume} {308}},\ \bibinfo {pages} {96--100} (\bibinfo {year} {2003})},\
  \Eprint {https://arxiv.org/abs/quant-ph/0108020} {arXiv:quant-ph/0108020}%
  \BibitemShut {NoStop}%
\bibitem [{\citenamefont {Harper}\ \emph {et~al.}(2020)\citenamefont {Harper},
  \citenamefont {Roy}, \citenamefont {Rudner},\ and\ \citenamefont
  {Sondhi}}]{Harper2019}%
  \BibitemOpen
  \bibfield  {author} {\bibinfo {author} {\bibfnamefont {Fenner}\ \bibnamefont
  {Harper}}, \bibinfo {author} {\bibfnamefont {Rahul}\ \bibnamefont {Roy}},
  \bibinfo {author} {\bibfnamefont {Mark~S.}\ \bibnamefont {Rudner}},\ and\
  \bibinfo {author} {\bibfnamefont {S.~L.}\ \bibnamefont {Sondhi}},\ }\bibfield
   {title} {\bibinfo {title} {Topology and broken symmetry in {Floquet}
  systems},\ }\href {https://doi.org/10.1146/annurev-conmatphys-031218-013721}
  {\bibfield  {journal} {\bibinfo  {journal} {Annual Review of Condensed Matter
  Physics}\ }\textbf {\bibinfo {volume} {11}},\ \bibinfo {pages} {345--368}
  (\bibinfo {year} {2020})},\ \Eprint {https://arxiv.org/abs/1905.01317}
  {arXiv:1905.01317}\BibitemShut {NoStop}%
\bibitem [{\citenamefont {Rudner}\ and\ \citenamefont
  {Lindner}(2020)}]{Rudner2020}%
  \BibitemOpen
  \bibfield  {author} {\bibinfo {author} {\bibfnamefont {Mark~S}\ \bibnamefont
  {Rudner}}\ and\ \bibinfo {author} {\bibfnamefont {Netanel~H}\ \bibnamefont
  {Lindner}},\ }\bibfield  {title} {\bibinfo {title} {Band structure
  engineering and non-equilibrium dynamics in {F}loquet topological
  insulators},\ }\href {https://doi.org/10.1038/s42254-020-0170-z} {\bibfield
  {journal} {\bibinfo  {journal} {Nature reviews physics}\ }\textbf {\bibinfo
  {volume} {2}},\ \bibinfo {pages} {229--244} (\bibinfo {year}
  {2020})}\BibitemShut {NoStop}%
\bibitem [{\citenamefont {Oka}\ and\ \citenamefont {Aoki}(2009)}]{Takashi2009}%
  \BibitemOpen
  \bibfield  {author} {\bibinfo {author} {\bibfnamefont {Takashi}\ \bibnamefont
  {Oka}}\ and\ \bibinfo {author} {\bibfnamefont {Hideo}\ \bibnamefont {Aoki}},\
  }\bibfield  {title} {\bibinfo {title} {Photovoltaic {H}all effect in
  graphene},\ }\href {https://doi.org/10.1103/PhysRevB.79.081406} {\bibfield
  {journal} {\bibinfo  {journal} {Phys. Rev. B}\ }\textbf {\bibinfo {volume}
  {79}},\ \bibinfo {pages} {081406} (\bibinfo {year} {2009})},\ \Eprint
  {https://arxiv.org/abs/0807.4767} {arXiv:0807.4767}\BibitemShut {NoStop}%
\bibitem [{\citenamefont {Kitagawa}\ \emph {et~al.}(2010)\citenamefont
  {Kitagawa}, \citenamefont {Berg}, \citenamefont {Rudner},\ and\ \citenamefont
  {Demler}}]{Kitagawa2010}%
  \BibitemOpen
  \bibfield  {author} {\bibinfo {author} {\bibfnamefont {Takuya}\ \bibnamefont
  {Kitagawa}}, \bibinfo {author} {\bibfnamefont {Erez}\ \bibnamefont {Berg}},
  \bibinfo {author} {\bibfnamefont {Mark}\ \bibnamefont {Rudner}},\ and\
  \bibinfo {author} {\bibfnamefont {Eugene}\ \bibnamefont {Demler}},\
  }\bibfield  {title} {\bibinfo {title} {Topological characterization of
  periodically driven quantum systems},\ }\href
  {https://doi.org/10.1103/PhysRevB.82.235114} {\bibfield  {journal} {\bibinfo
  {journal} {Phys. Rev. B}\ }\textbf {\bibinfo {volume} {82}},\ \bibinfo
  {pages} {235114} (\bibinfo {year} {2010})},\ \Eprint
  {https://arxiv.org/abs/1010.6126} {arXiv:1010.6126}\BibitemShut {NoStop}%
\bibitem [{\citenamefont {Lindner}\ \emph {et~al.}(2011)\citenamefont
  {Lindner}, \citenamefont {Refael},\ and\ \citenamefont
  {Galitski}}]{Lindner2011}%
  \BibitemOpen
  \bibfield  {author} {\bibinfo {author} {\bibfnamefont {Netanel~H.}\
  \bibnamefont {Lindner}}, \bibinfo {author} {\bibfnamefont {Gil}\ \bibnamefont
  {Refael}},\ and\ \bibinfo {author} {\bibfnamefont {Victor}\ \bibnamefont
  {Galitski}},\ }\bibfield  {title} {\bibinfo {title} {Floquet topological
  insulator in semiconductor quantum wells},\ }\href
  {https://doi.org/10.1038/nphys1926} {\bibfield  {journal} {\bibinfo
  {journal} {Nature Physics}\ }\textbf {\bibinfo {volume} {7}},\ \bibinfo
  {pages} {490--495} (\bibinfo {year} {2011})},\ \Eprint
  {https://arxiv.org/abs/1008.1792} {arXiv:1008.1792}\BibitemShut {NoStop}%
\bibitem [{\citenamefont {Rudner}\ \emph {et~al.}(2013)\citenamefont {Rudner},
  \citenamefont {Lindner}, \citenamefont {Berg},\ and\ \citenamefont
  {Levin}}]{Rudner2013}%
  \BibitemOpen
  \bibfield  {author} {\bibinfo {author} {\bibfnamefont {Mark~S.}\ \bibnamefont
  {Rudner}}, \bibinfo {author} {\bibfnamefont {Netanel~H.}\ \bibnamefont
  {Lindner}}, \bibinfo {author} {\bibfnamefont {Erez}\ \bibnamefont {Berg}},\
  and\ \bibinfo {author} {\bibfnamefont {Michael}\ \bibnamefont {Levin}},\
  }\bibfield  {title} {\bibinfo {title} {Anomalous edge states and the
  bulk-edge correspondence for periodically-driven two dimensional systems},\
  }\href@noop {} {\bibfield  {journal} {\bibinfo  {journal} {Phys. Rev. X}\
  }\textbf {\bibinfo {volume} {3}},\ \bibinfo {pages} {031005} (\bibinfo {year}
  {2013})},\ \Eprint {https://arxiv.org/abs/1212.3324} {arXiv:1212.3324}%
  \BibitemShut {NoStop}%
\bibitem [{\citenamefont {Po}\ \emph {et~al.}(2016)\citenamefont {Po},
  \citenamefont {Fidkowski}, \citenamefont {Morimoto}, \citenamefont {Potter},\
  and\ \citenamefont {Vishwanath}}]{PoFidkowski2016}%
  \BibitemOpen
  \bibfield  {author} {\bibinfo {author} {\bibfnamefont {Hoi~Chun}\
  \bibnamefont {Po}}, \bibinfo {author} {\bibfnamefont {Lukasz}\ \bibnamefont
  {Fidkowski}}, \bibinfo {author} {\bibfnamefont {Takahiro}\ \bibnamefont
  {Morimoto}}, \bibinfo {author} {\bibfnamefont {Andrew~C.}\ \bibnamefont
  {Potter}},\ and\ \bibinfo {author} {\bibfnamefont {Ashvin}\ \bibnamefont
  {Vishwanath}},\ }\bibfield  {title} {\bibinfo {title} {Chiral {F}loquet
  phases of many-body localized bosons},\ }\href
  {https://doi.org/10.1103/PhysRevX.6.041070} {\bibfield  {journal} {\bibinfo
  {journal} {Phys. Rev. X}\ }\textbf {\bibinfo {volume} {6}},\ \bibinfo {pages}
  {041070} (\bibinfo {year} {2016})}\BibitemShut {NoStop}%
\bibitem [{\citenamefont {Harper}\ and\ \citenamefont
  {Roy}(2017)}]{Harper2017}%
  \BibitemOpen
  \bibfield  {author} {\bibinfo {author} {\bibfnamefont {Fenner}\ \bibnamefont
  {Harper}}\ and\ \bibinfo {author} {\bibfnamefont {Rahul}\ \bibnamefont
  {Roy}},\ }\bibfield  {title} {\bibinfo {title} {Floquet topological order in
  interacting systems of bosons and fermions},\ }\href
  {https://doi.org/10.1103/PhysRevLett.118.115301} {\bibfield  {journal}
  {\bibinfo  {journal} {Phys. Rev. Lett.}\ }\textbf {\bibinfo {volume} {118}},\
  \bibinfo {pages} {115301} (\bibinfo {year} {2017})},\ \Eprint
  {https://arxiv.org/abs/1609.06303} {arXiv:1609.06303}\BibitemShut {NoStop}%
\bibitem [{\citenamefont {Po}\ \emph {et~al.}(2017)\citenamefont {Po},
  \citenamefont {Fidkowski}, \citenamefont {Vishwanath},\ and\ \citenamefont
  {Potter}}]{fermionGNVW1}%
  \BibitemOpen
  \bibfield  {author} {\bibinfo {author} {\bibfnamefont {Hoi~Chun}\
  \bibnamefont {Po}}, \bibinfo {author} {\bibfnamefont {Lukasz}\ \bibnamefont
  {Fidkowski}}, \bibinfo {author} {\bibfnamefont {Ashvin}\ \bibnamefont
  {Vishwanath}},\ and\ \bibinfo {author} {\bibfnamefont {Andrew~C.}\
  \bibnamefont {Potter}},\ }\bibfield  {title} {\bibinfo {title} {Radical
  chiral {F}loquet phases in a periodically driven {Kitaev} model and beyond},\
  }\href {https://doi.org/10.1103/PhysRevB.96.245116} {\bibfield  {journal}
  {\bibinfo  {journal} {Phys. Rev. B}\ }\textbf {\bibinfo {volume} {96}},\
  \bibinfo {pages} {245116} (\bibinfo {year} {2017})},\ \Eprint
  {https://arxiv.org/abs/1701.01440} {arXiv:1701.01440}\BibitemShut {NoStop}%
\bibitem [{\citenamefont {Hastings}\ and\ \citenamefont
  {Haah}(2021)}]{Hastings2021}%
  \BibitemOpen
  \bibfield  {author} {\bibinfo {author} {\bibfnamefont {Matthew~B.}\
  \bibnamefont {Hastings}}\ and\ \bibinfo {author} {\bibfnamefont {Jeongwan}\
  \bibnamefont {Haah}},\ }\bibfield  {title} {\bibinfo {title} {Dynamically
  {G}enerated {L}ogical {Q}ubits},\ }\href
  {https://doi.org/10.22331/q-2021-10-19-564} {\bibfield  {journal} {\bibinfo
  {journal} {{Quantum}}\ }\textbf {\bibinfo {volume} {5}},\ \bibinfo {pages}
  {564} (\bibinfo {year} {2021})},\ \Eprint {https://arxiv.org/abs/2107.02194}
  {arXiv:2107.02194}\BibitemShut {NoStop}%
\bibitem [{\citenamefont {Aasen}\ \emph {et~al.}(2022)\citenamefont {Aasen},
  \citenamefont {Wang},\ and\ \citenamefont {Hastings}}]{Aasen2022}%
  \BibitemOpen
  \bibfield  {author} {\bibinfo {author} {\bibfnamefont {David}\ \bibnamefont
  {Aasen}}, \bibinfo {author} {\bibfnamefont {Zhenghan}\ \bibnamefont {Wang}},\
  and\ \bibinfo {author} {\bibfnamefont {Matthew~B.}\ \bibnamefont
  {Hastings}},\ }\bibfield  {title} {\bibinfo {title} {Adiabatic paths of
  {H}amiltonians, symmetries of topological order, and automorphism codes},\
  }\href {https://doi.org/10.1103/PhysRevB.106.085122} {\bibfield  {journal}
  {\bibinfo  {journal} {Phys. Rev. B}\ }\textbf {\bibinfo {volume} {106}},\
  \bibinfo {pages} {085122} (\bibinfo {year} {2022})},\ \Eprint
  {https://arxiv.org/abs/2203.11137} {arXiv:2203.11137}\BibitemShut {NoStop}%
\bibitem [{\citenamefont {Davydova}\ \emph {et~al.}(2022)\citenamefont
  {Davydova}, \citenamefont {Tantivasadakarn},\ and\ \citenamefont
  {Balasubramanian}}]{Davydova2022}%
  \BibitemOpen
  \bibfield  {author} {\bibinfo {author} {\bibfnamefont {Margarita}\
  \bibnamefont {Davydova}}, \bibinfo {author} {\bibfnamefont {Nathanan}\
  \bibnamefont {Tantivasadakarn}},\ and\ \bibinfo {author} {\bibfnamefont
  {Shankar}\ \bibnamefont {Balasubramanian}},\ }\href@noop {} {\bibinfo {title}
  {Floquet codes without parent subsystem codes}} (\bibinfo {year} {2022}),\
  \Eprint {https://arxiv.org/abs/2210.02468} {arXiv:2210.02468}\BibitemShut
  {NoStop}%
\bibitem [{\citenamefont {Kesselring}\ \emph {et~al.}(2022)\citenamefont
  {Kesselring}, \citenamefont {de~la Fuente}, \citenamefont {Thomsen},
  \citenamefont {Eisert}, \citenamefont {Bartlett},\ and\ \citenamefont
  {Brown}}]{Kesselring2022}%
  \BibitemOpen
  \bibfield  {author} {\bibinfo {author} {\bibfnamefont {Markus~S.}\
  \bibnamefont {Kesselring}}, \bibinfo {author} {\bibfnamefont {Julio
  C.~Magdalena}\ \bibnamefont {de~la Fuente}}, \bibinfo {author} {\bibfnamefont
  {Felix}\ \bibnamefont {Thomsen}}, \bibinfo {author} {\bibfnamefont {Jens}\
  \bibnamefont {Eisert}}, \bibinfo {author} {\bibfnamefont {Stephen~D.}\
  \bibnamefont {Bartlett}},\ and\ \bibinfo {author} {\bibfnamefont
  {Benjamin~J.}\ \bibnamefont {Brown}},\ }\href@noop {} {\bibinfo {title}
  {Anyon condensation and the color code}} (\bibinfo {year} {2022}),\ \Eprint
  {https://arxiv.org/abs/2212.00042} {arXiv:2212.00042}\BibitemShut {NoStop}%
\bibitem [{\citenamefont {Zhang}\ \emph {et~al.}(2022)\citenamefont {Zhang},
  \citenamefont {Aasen},\ and\ \citenamefont {Vijay}}]{ZhangXCube2022}%
  \BibitemOpen
  \bibfield  {author} {\bibinfo {author} {\bibfnamefont {Zhehao}\ \bibnamefont
  {Zhang}}, \bibinfo {author} {\bibfnamefont {David}\ \bibnamefont {Aasen}},\
  and\ \bibinfo {author} {\bibfnamefont {Sagar}\ \bibnamefont {Vijay}},\
  }\href@noop {} {\bibinfo {title} {The {X}-cube {F}loquet code}} (\bibinfo
  {year} {2022}),\ \Eprint {https://arxiv.org/abs/2211.05784}
  {arXiv:2211.05784}\BibitemShut {NoStop}%
\bibitem [{\citenamefont {Paetznick}\ \emph {et~al.}(2023)\citenamefont
  {Paetznick}, \citenamefont {Knapp}, \citenamefont {Delfosse}, \citenamefont
  {Bauer}, \citenamefont {Haah}, \citenamefont {Hastings},\ and\ \citenamefont
  {da~Silva}}]{Paetznick2023}%
  \BibitemOpen
  \bibfield  {author} {\bibinfo {author} {\bibfnamefont {Adam}\ \bibnamefont
  {Paetznick}}, \bibinfo {author} {\bibfnamefont {Christina}\ \bibnamefont
  {Knapp}}, \bibinfo {author} {\bibfnamefont {Nicolas}\ \bibnamefont
  {Delfosse}}, \bibinfo {author} {\bibfnamefont {Bela}\ \bibnamefont {Bauer}},
  \bibinfo {author} {\bibfnamefont {Jeongwan}\ \bibnamefont {Haah}}, \bibinfo
  {author} {\bibfnamefont {Matthew~B.}\ \bibnamefont {Hastings}},\ and\
  \bibinfo {author} {\bibfnamefont {Marcus~P.}\ \bibnamefont {da~Silva}},\
  }\bibfield  {title} {\bibinfo {title} {Performance of planar {F}loquet codes
  with {M}ajorana-based qubits},\ }\href
  {https://doi.org/10.1103/PRXQuantum.4.010310} {\bibfield  {journal} {\bibinfo
   {journal} {PRX Quantum}\ }\textbf {\bibinfo {volume} {4}},\ \bibinfo {pages}
  {010310} (\bibinfo {year} {2023})}\BibitemShut {NoStop}%
\bibitem [{\citenamefont {Haah}\ and\ \citenamefont
  {Hastings}(2022)}]{Haah2022}%
  \BibitemOpen
  \bibfield  {author} {\bibinfo {author} {\bibfnamefont {Jeongwan}\
  \bibnamefont {Haah}}\ and\ \bibinfo {author} {\bibfnamefont {Matthew~B.}\
  \bibnamefont {Hastings}},\ }\bibfield  {title} {\bibinfo {title} {Boundaries
  for the {H}oneycomb {C}ode},\ }\href
  {https://doi.org/10.22331/q-2022-04-21-693} {\bibfield  {journal} {\bibinfo
  {journal} {{Quantum}}\ }\textbf {\bibinfo {volume} {6}},\ \bibinfo {pages}
  {693} (\bibinfo {year} {2022})},\ \Eprint {https://arxiv.org/abs/2110.09545}
  {arXiv:2110.09545}\BibitemShut {NoStop}%
\bibitem [{\citenamefont {Gidney}\ \emph {et~al.}(2022)\citenamefont {Gidney},
  \citenamefont {Newman},\ and\ \citenamefont {McEwen}}]{Gidney_2022}%
  \BibitemOpen
  \bibfield  {author} {\bibinfo {author} {\bibfnamefont {Craig}\ \bibnamefont
  {Gidney}}, \bibinfo {author} {\bibfnamefont {Michael}\ \bibnamefont
  {Newman}},\ and\ \bibinfo {author} {\bibfnamefont {Matt}\ \bibnamefont
  {McEwen}},\ }\bibfield  {title} {\bibinfo {title} {Benchmarking the planar
  honeycomb code},\ }\href {https://doi.org/10.22331/q-2022-09-21-813}
  {\bibfield  {journal} {\bibinfo  {journal} {Quantum}\ }\textbf {\bibinfo
  {volume} {6}},\ \bibinfo {pages} {813} (\bibinfo {year} {2022})},\ \Eprint
  {https://arxiv.org/abs/2202.11845} {arXiv:2202.11845}\BibitemShut {NoStop}%
\bibitem [{\citenamefont {Vuillot}(2021)}]{Vuillot2021}%
  \BibitemOpen
  \bibfield  {author} {\bibinfo {author} {\bibfnamefont {Christophe}\
  \bibnamefont {Vuillot}},\ }\href@noop {} {\bibinfo {title} {Planar {F}loquet
  codes}} (\bibinfo {year} {2021}),\ \Eprint {https://arxiv.org/abs/2110.05348}
  {arXiv:2110.05348}\BibitemShut {NoStop}%
\bibitem [{\citenamefont {Gross}\ \emph {et~al.}(2012)\citenamefont {Gross},
  \citenamefont {Nesme}, \citenamefont {Vogts},\ and\ \citenamefont
  {Werner}}]{GNVW}%
  \BibitemOpen
  \bibfield  {author} {\bibinfo {author} {\bibfnamefont {D.}~\bibnamefont
  {Gross}}, \bibinfo {author} {\bibfnamefont {V.}~\bibnamefont {Nesme}},
  \bibinfo {author} {\bibfnamefont {H.}~\bibnamefont {Vogts}},\ and\ \bibinfo
  {author} {\bibfnamefont {R.~F.}\ \bibnamefont {Werner}},\ }\bibfield  {title}
  {\bibinfo {title} {Index theory of one dimensional quantum walks and cellular
  automata},\ }\href {https://doi.org/10.1007/s00220-012-1423-1} {\bibfield
  {journal} {\bibinfo  {journal} {Commun. Math. Phys.}\ }\textbf {\bibinfo
  {volume} {310}},\ \bibinfo {pages} {419--454} (\bibinfo {year} {2012})},\
  \Eprint {https://arxiv.org/abs/0910.3675} {arXiv:0910.3675}\BibitemShut
  {NoStop}%
\bibitem [{\citenamefont {Haah}\ \emph {et~al.}(2022)\citenamefont {Haah},
  \citenamefont {Fidkowski},\ and\ \citenamefont {Hastings}}]{nta3}%
  \BibitemOpen
  \bibfield  {author} {\bibinfo {author} {\bibfnamefont {Jeongwan}\
  \bibnamefont {Haah}}, \bibinfo {author} {\bibfnamefont {Lukasz}\ \bibnamefont
  {Fidkowski}},\ and\ \bibinfo {author} {\bibfnamefont {Matthew~B.}\
  \bibnamefont {Hastings}},\ }\bibfield  {title} {\bibinfo {title} {Nontrivial
  quantum cellular automata in higher dimensions},\ }\href
  {https://doi.org/10.1007/s00220-022-04528-1} {\bibfield  {journal} {\bibinfo
  {journal} {Commun. Math. Phys.}\ }\textbf {\bibinfo {volume} {398}},\
  \bibinfo {pages} {469--540} (\bibinfo {year} {2022})},\ \Eprint
  {https://arxiv.org/abs/1812.01625} {arXiv:1812.01625}\BibitemShut {NoStop}%
\bibitem [{\citenamefont {Freedman}\ and\ \citenamefont
  {Hastings}(2020)}]{FreedmanHastings2019QCA}%
  \BibitemOpen
  \bibfield  {author} {\bibinfo {author} {\bibfnamefont {M.}~\bibnamefont
  {Freedman}}\ and\ \bibinfo {author} {\bibfnamefont {M.~B.}\ \bibnamefont
  {Hastings}},\ }\bibfield  {title} {\bibinfo {title} {Classification of
  quantum cellular automata},\ }\href
  {https://doi.org/10.1007/s00220-020-03735-y} {\bibfield  {journal} {\bibinfo
  {journal} {Commun. Math. Phys.}\ }\textbf {\bibinfo {volume} {376}},\
  \bibinfo {pages} {1171--1222} (\bibinfo {year} {2020})},\ \Eprint
  {https://arxiv.org/abs/1902.10285} {arXiv:1902.10285}\BibitemShut {NoStop}%
\bibitem [{\citenamefont {Shirley}\ \emph {et~al.}(2022)\citenamefont
  {Shirley}, \citenamefont {Chen}, \citenamefont {Dua}, \citenamefont
  {Ellison}, \citenamefont {Tantivasadakarn},\ and\ \citenamefont
  {Williamson}}]{Wilbur2022}%
  \BibitemOpen
  \bibfield  {author} {\bibinfo {author} {\bibfnamefont {Wilbur}\ \bibnamefont
  {Shirley}}, \bibinfo {author} {\bibfnamefont {Yu-An}\ \bibnamefont {Chen}},
  \bibinfo {author} {\bibfnamefont {Arpit}\ \bibnamefont {Dua}}, \bibinfo
  {author} {\bibfnamefont {Tyler~D.}\ \bibnamefont {Ellison}}, \bibinfo
  {author} {\bibfnamefont {Nathanan}\ \bibnamefont {Tantivasadakarn}},\ and\
  \bibinfo {author} {\bibfnamefont {Dominic~J.}\ \bibnamefont {Williamson}},\
  }\bibfield  {title} {\bibinfo {title} {Three-dimensional quantum cellular
  automata from chiral semion surface topological order and beyond},\ }\href
  {https://doi.org/10.1103/PRXQuantum.3.030326} {\bibfield  {journal} {\bibinfo
   {journal} {PRX Quantum}\ }\textbf {\bibinfo {volume} {3}},\ \bibinfo {pages}
  {030326} (\bibinfo {year} {2022})}\BibitemShut {NoStop}%
\bibitem [{\citenamefont {Chen}\ and\ \citenamefont {Hsin}(2022)}]{Chen2021}%
  \BibitemOpen
  \bibfield  {author} {\bibinfo {author} {\bibfnamefont {Yu-An}\ \bibnamefont
  {Chen}}\ and\ \bibinfo {author} {\bibfnamefont {Po-Shen}\ \bibnamefont
  {Hsin}},\ }\href@noop {} {\bibinfo {title} {Exactly solvable lattice
  hamiltonians and gravitational anomalies}} (\bibinfo {year} {2022}),\ \Eprint
  {https://arxiv.org/abs/2110.14644} {arXiv:2110.14644}\BibitemShut {NoStop}%
\bibitem [{\citenamefont {Haah}(2022{\natexlab{a}})}]{Haah2022b}%
  \BibitemOpen
  \bibfield  {author} {\bibinfo {author} {\bibfnamefont {Jeongwan}\
  \bibnamefont {Haah}},\ }\href@noop {} {\bibinfo {title} {Invertible
  subalgebras}} (\bibinfo {year} {2022}{\natexlab{a}}),\ \Eprint
  {https://arxiv.org/abs/2211.02086} {arXiv:2211.02086}\BibitemShut {NoStop}%
\bibitem [{\citenamefont {Zhang}\ and\ \citenamefont
  {Levin}(2022)}]{ZhangLevin2022}%
  \BibitemOpen
  \bibfield  {author} {\bibinfo {author} {\bibfnamefont {Carolyn}\ \bibnamefont
  {Zhang}}\ and\ \bibinfo {author} {\bibfnamefont {Michael}\ \bibnamefont
  {Levin}},\ }\href@noop {} {\bibinfo {title} {Bulk-boundary correspondence for
  interacting {F}loquet systems in two dimensions}} (\bibinfo {year} {2022}),\
  \Eprint {https://arxiv.org/abs/2209.03975} {arXiv:2209.03975}\BibitemShut
  {NoStop}%
\bibitem [{\citenamefont {Bennett}\ \emph {et~al.}(1993)\citenamefont
  {Bennett}, \citenamefont {Brassard}, \citenamefont {Cr\'{e}peau},
  \citenamefont {Jozsa}, \citenamefont {Peres},\ and\ \citenamefont
  {Wootters}}]{BBCJPW}%
  \BibitemOpen
  \bibfield  {author} {\bibinfo {author} {\bibfnamefont {Charles~H.}\
  \bibnamefont {Bennett}}, \bibinfo {author} {\bibfnamefont {Gilles}\
  \bibnamefont {Brassard}}, \bibinfo {author} {\bibfnamefont {Claude}\
  \bibnamefont {Cr\'{e}peau}}, \bibinfo {author} {\bibfnamefont {Richard}\
  \bibnamefont {Jozsa}}, \bibinfo {author} {\bibfnamefont {Asher}\ \bibnamefont
  {Peres}},\ and\ \bibinfo {author} {\bibfnamefont {William~K.}\ \bibnamefont
  {Wootters}},\ }\bibfield  {title} {\bibinfo {title} {Teleporting an unknown
  quantum state via dual classical and {Einstein--Podolsky--Rosen} channels},\
  }\href {https://doi.org/10.1103/PhysRevLett.70.1895} {\bibfield  {journal}
  {\bibinfo  {journal} {Phys. Rev. Lett.}\ }\textbf {\bibinfo {volume} {70}},\
  \bibinfo {pages} {1895} (\bibinfo {year} {1993})}\BibitemShut {NoStop}%
\bibitem [{\citenamefont {Wen}(2003)}]{WenP}%
  \BibitemOpen
  \bibfield  {author} {\bibinfo {author} {\bibfnamefont {Xiao-Gang}\
  \bibnamefont {Wen}},\ }\bibfield  {title} {\bibinfo {title} {Quantum orders
  in an exact soluble model},\ }\href
  {https://doi.org/10.1103/PhysRevLett.90.016803} {\bibfield  {journal}
  {\bibinfo  {journal} {Phys. Rev. Lett.}\ }\textbf {\bibinfo {volume} {90}},\
  \bibinfo {pages} {016803} (\bibinfo {year} {2003})},\ \Eprint
  {https://arxiv.org/abs/quant-ph/0205004} {arXiv:quant-ph/0205004}%
  \BibitemShut {NoStop}%
\bibitem [{\citenamefont {Freedman}\ \emph {et~al.}(2022)\citenamefont
  {Freedman}, \citenamefont {Haah},\ and\ \citenamefont {Hastings}}]{FHH2019}%
  \BibitemOpen
  \bibfield  {author} {\bibinfo {author} {\bibfnamefont {Michael}\ \bibnamefont
  {Freedman}}, \bibinfo {author} {\bibfnamefont {Jeongwan}\ \bibnamefont
  {Haah}},\ and\ \bibinfo {author} {\bibfnamefont {Matthew~B.}\ \bibnamefont
  {Hastings}},\ }\bibfield  {title} {\bibinfo {title} {The group structure of
  quantum cellular automata},\ }\href
  {https://doi.org/10.1007/s00220-022-04316-x} {\bibfield  {journal} {\bibinfo
  {journal} {Commun. Math. Phys.}\ }\textbf {\bibinfo {volume} {389}},\
  \bibinfo {pages} {1277--1302} (\bibinfo {year} {2022})},\ \Eprint
  {https://arxiv.org/abs/1910.07998} {arXiv:1910.07998}\BibitemShut {NoStop}%
\bibitem [{\citenamefont {Bravyi}\ \emph {et~al.}(2010)\citenamefont {Bravyi},
  \citenamefont {Hastings},\ and\ \citenamefont
  {Michalakis}}]{BravyiHastingsMichalakis2010stability}%
  \BibitemOpen
  \bibfield  {author} {\bibinfo {author} {\bibfnamefont {Sergey}\ \bibnamefont
  {Bravyi}}, \bibinfo {author} {\bibfnamefont {Matthew}\ \bibnamefont
  {Hastings}},\ and\ \bibinfo {author} {\bibfnamefont {Spyridon}\ \bibnamefont
  {Michalakis}},\ }\bibfield  {title} {\bibinfo {title} {Topological quantum
  order: Stability under local perturbations},\ }\href
  {https://doi.org/10.1063/1.3490195} {\bibfield  {journal} {\bibinfo
  {journal} {J. Math. Phys.}\ }\textbf {\bibinfo {volume} {51}},\ \bibinfo
  {pages} {093512} (\bibinfo {year} {2010})},\ \Eprint
  {https://arxiv.org/abs/1001.0344} {arXiv:1001.0344}\BibitemShut {NoStop}%
\bibitem [{\citenamefont {Haah}(2013)}]{Haah2013}%
  \BibitemOpen
  \bibfield  {author} {\bibinfo {author} {\bibfnamefont {Jeongwan}\
  \bibnamefont {Haah}},\ }\bibfield  {title} {\bibinfo {title} {Commuting
  {Pauli} {H}amiltonians as maps between free modules},\ }\href
  {https://doi.org/10.1007/s00220-013-1810-2} {\bibfield  {journal} {\bibinfo
  {journal} {Commun. Math. Phys.}\ }\textbf {\bibinfo {volume} {324}},\
  \bibinfo {pages} {351--399} (\bibinfo {year} {2013})},\ \Eprint
  {https://arxiv.org/abs/1204.1063} {arXiv:1204.1063}\BibitemShut {NoStop}%
\bibitem [{\citenamefont {Nandkishore}\ and\ \citenamefont
  {Hermele}(2019)}]{Nandkishore2018}%
  \BibitemOpen
  \bibfield  {author} {\bibinfo {author} {\bibfnamefont {Rahul~M.}\
  \bibnamefont {Nandkishore}}\ and\ \bibinfo {author} {\bibfnamefont {Michael}\
  \bibnamefont {Hermele}},\ }\bibfield  {title} {\bibinfo {title} {Fractons},\
  }\href {https://doi.org/10.1146/annurev-conmatphys-031218-013604} {\bibfield
  {journal} {\bibinfo  {journal} {Annual Review of Condensed Matter Physics,}\
  }\textbf {\bibinfo {volume} {10}},\ \bibinfo {pages} {295--313} (\bibinfo
  {year} {2019})},\ \Eprint {https://arxiv.org/abs/1803.11196}
  {arXiv:1803.11196}\BibitemShut {NoStop}%
\bibitem [{\citenamefont {Pretko}\ \emph {et~al.}(2020)\citenamefont {Pretko},
  \citenamefont {Chen},\ and\ \citenamefont {You}}]{Pretko2020}%
  \BibitemOpen
  \bibfield  {author} {\bibinfo {author} {\bibfnamefont {Michael}\ \bibnamefont
  {Pretko}}, \bibinfo {author} {\bibfnamefont {Xie}\ \bibnamefont {Chen}},\
  and\ \bibinfo {author} {\bibfnamefont {Yizhi}\ \bibnamefont {You}},\
  }\bibfield  {title} {\bibinfo {title} {Fracton phases of matter},\ }\href
  {https://doi.org/10.1142/S0217751X20300033} {\bibfield  {journal} {\bibinfo
  {journal} {International Journal of Modern Physics A}\ }\textbf {\bibinfo
  {volume} {35}},\ \bibinfo {pages} {2030003} (\bibinfo {year} {2020})},\
  \Eprint {https://arxiv.org/abs/2001.01722} {arXiv:2001.01722}\BibitemShut
  {NoStop}%
\bibitem [{\citenamefont {Bravyi}\ and\ \citenamefont
  {Kitaev}(1998)}]{BravyiKitaevSurfaceCode}%
  \BibitemOpen
  \bibfield  {author} {\bibinfo {author} {\bibfnamefont {S.~B.}\ \bibnamefont
  {Bravyi}}\ and\ \bibinfo {author} {\bibfnamefont {A.~Yu.}\ \bibnamefont
  {Kitaev}},\ }\href@noop {} {\bibinfo {title} {Quantum codes on a lattice with
  boundary}} (\bibinfo {year} {1998}),\ \Eprint
  {https://arxiv.org/abs/quant-ph/9811052} {arXiv:quant-ph/9811052}%
  \BibitemShut {NoStop}%
\bibitem [{\citenamefont {Haah}(2021)}]{Haah2018}%
  \BibitemOpen
  \bibfield  {author} {\bibinfo {author} {\bibfnamefont {Jeongwan}\
  \bibnamefont {Haah}},\ }\bibfield  {title} {\bibinfo {title} {Classification
  of translation invariant topological {Pauli} stabilizer codes for prime
  dimensional qudits on two-dimensional lattices},\ }\href
  {https://doi.org/10.1063/5.0021068} {\bibfield  {journal} {\bibinfo
  {journal} {J. Math. Phys.}\ }\textbf {\bibinfo {volume} {62}},\ \bibinfo
  {pages} {012201} (\bibinfo {year} {2021})},\ \Eprint
  {https://arxiv.org/abs/1812.11193} {arXiv:1812.11193}\BibitemShut {NoStop}%
\bibitem [{\citenamefont {Schlingemann}\ \emph {et~al.}(2008)\citenamefont
  {Schlingemann}, \citenamefont {Vogts},\ and\ \citenamefont
  {Werner}}]{clifQCA}%
  \BibitemOpen
  \bibfield  {author} {\bibinfo {author} {\bibfnamefont {Dirk-M.}\ \bibnamefont
  {Schlingemann}}, \bibinfo {author} {\bibfnamefont {Holger}\ \bibnamefont
  {Vogts}},\ and\ \bibinfo {author} {\bibfnamefont {Reinhard~F.}\ \bibnamefont
  {Werner}},\ }\bibfield  {title} {\bibinfo {title} {On the structure of
  {Clifford} quantum cellular automata},\ }\href
  {https://doi.org/10.1063/1.3005565} {\bibfield  {journal} {\bibinfo
  {journal} {J. Math. Phys.}\ }\textbf {\bibinfo {volume} {49}},\ \bibinfo
  {pages} {112104} (\bibinfo {year} {2008})},\ \Eprint
  {https://arxiv.org/abs/0804.4447} {arXiv:0804.4447}\BibitemShut {NoStop}%
\bibitem [{\citenamefont {Haah}(2022{\natexlab{b}})}]{clifQCAclassification}%
  \BibitemOpen
  \bibfield  {author} {\bibinfo {author} {\bibfnamefont {Jeongwan}\
  \bibnamefont {Haah}},\ }\href@noop {} {\bibinfo {title} {Topological phases
  of unitary dynamics: Classification in {Clifford} category}} (\bibinfo {year}
  {2022}{\natexlab{b}}),\ \Eprint {https://arxiv.org/abs/2205.09141}
  {arXiv:2205.09141}\BibitemShut {NoStop}%
\bibitem [{\citenamefont {{Kitaev}}(2006)}]{Kitaev_2005}%
  \BibitemOpen
  \bibfield  {author} {\bibinfo {author} {\bibfnamefont {A.}~\bibnamefont
  {{Kitaev}}},\ }\bibfield  {title} {\bibinfo {title} {{Anyons in an exactly
  solved model and beyond}},\ }\href
  {https://doi.org/10.1016/j.aop.2005.10.005} {\bibfield  {journal} {\bibinfo
  {journal} {Ann. Phys.}\ }\textbf {\bibinfo {volume} {321}},\ \bibinfo {pages}
  {2--111} (\bibinfo {year} {2006})},\ \Eprint
  {https://arxiv.org/abs/cond-mat/0506438} {arXiv:cond-mat/0506438}%
  \BibitemShut {NoStop}%
\bibitem [{\citenamefont {Avron}\ \emph {et~al.}(1994)\citenamefont {Avron},
  \citenamefont {Seiler},\ and\ \citenamefont {Simon}}]{avron1994index}%
  \BibitemOpen
  \bibfield  {author} {\bibinfo {author} {\bibfnamefont {J.}~\bibnamefont
  {Avron}}, \bibinfo {author} {\bibfnamefont {R.}~\bibnamefont {Seiler}},\ and\
  \bibinfo {author} {\bibfnamefont {B.}~\bibnamefont {Simon}},\ }\bibfield
  {title} {\bibinfo {title} {The index of a pair of projections},\ }\href
  {https://doi.org/10.1006/jfan.1994.1031} {\bibfield  {journal} {\bibinfo
  {journal} {Journal of Functional Analysis}\ }\textbf {\bibinfo {volume}
  {120}},\ \bibinfo {pages} {220--237} (\bibinfo {year} {1994})}\BibitemShut
  {NoStop}%
\bibitem [{\citenamefont {Lax}(2002)}]{Lax}%
  \BibitemOpen
  \bibfield  {author} {\bibinfo {author} {\bibfnamefont {Peter~D.}\
  \bibnamefont {Lax}},\ }\href@noop {} {\emph {\bibinfo {title} {Functional
  Analysis}}},\ Pure and Applied Mathematics\ (\bibinfo  {publisher} {John
  Wiley \& sons},\ \bibinfo {year} {2002})\BibitemShut {NoStop}%
\bibitem [{\citenamefont {Sarason}(1987)}]{Sarason1987}%
  \BibitemOpen
  \bibfield  {author} {\bibinfo {author} {\bibfnamefont {Donald}\ \bibnamefont
  {Sarason}},\ }\bibfield  {title} {\bibinfo {title} {The multiplication
  theorem for {Fredholm} operators},\ }\href
  {https://doi.org/10.1080/00029890.1987.12000596} {\bibfield  {journal}
  {\bibinfo  {journal} {The American Mathematical Monthly}\ }\textbf {\bibinfo
  {volume} {94}},\ \bibinfo {pages} {68--70} (\bibinfo {year}
  {1987})}\BibitemShut {NoStop}%
\bibitem [{\citenamefont {Michnicki}(2014)}]{Michnicki2014}%
  \BibitemOpen
  \bibfield  {author} {\bibinfo {author} {\bibfnamefont {Kamil~P}\ \bibnamefont
  {Michnicki}},\ }\bibfield  {title} {\bibinfo {title} {3-d topological quantum
  memory with a power-law energy barrier},\ }\href
  {https://doi.org/10.1103/PhysRevLett.113.130501} {\bibfield  {journal}
  {\bibinfo  {journal} {Phys. Rev. Lett.}\ }\textbf {\bibinfo {volume} {113}},\
  \bibinfo {pages} {130501} (\bibinfo {year} {2014})},\ \Eprint
  {https://arxiv.org/abs/1406.4227} {arXiv:1406.4227}\BibitemShut {NoStop}%
\bibitem [{\citenamefont {Portnoy}(2023)}]{Portnoy2023}%
  \BibitemOpen
  \bibfield  {author} {\bibinfo {author} {\bibfnamefont {Elia}\ \bibnamefont
  {Portnoy}},\ }\href@noop {} {\bibinfo {title} {Local quantum codes from
  subdivided manifolds}} (\bibinfo {year} {2023}),\ \Eprint
  {https://arxiv.org/abs/2303.06755} {arXiv:2303.06755}\BibitemShut {NoStop}%
\bibitem [{\citenamefont {Bravyi}\ and\ \citenamefont
  {Terhal}(2009)}]{BravyiTerhal2009no-go}%
  \BibitemOpen
  \bibfield  {author} {\bibinfo {author} {\bibfnamefont {Sergey}\ \bibnamefont
  {Bravyi}}\ and\ \bibinfo {author} {\bibfnamefont {Barbara}\ \bibnamefont
  {Terhal}},\ }\bibfield  {title} {\bibinfo {title} {A no-go theorem for a
  two-dimensional self-correcting quantum memory based on stabilizer codes},\
  }\href {https://doi.org/10.1088/1367-2630/11/4/043029} {\bibfield  {journal}
  {\bibinfo  {journal} {New J. Phys.}\ }\textbf {\bibinfo {volume} {11}},\
  \bibinfo {pages} {043029} (\bibinfo {year} {2009})},\ \Eprint
  {https://arxiv.org/abs/0810.1983} {arXiv:0810.1983}\BibitemShut {NoStop}%
\bibitem [{\citenamefont {Suchara}\ \emph {et~al.}(2011)\citenamefont
  {Suchara}, \citenamefont {Bravyi},\ and\ \citenamefont
  {Terhal}}]{Suchara2011constructions}%
  \BibitemOpen
  \bibfield  {author} {\bibinfo {author} {\bibfnamefont {Martin}\ \bibnamefont
  {Suchara}}, \bibinfo {author} {\bibfnamefont {Sergey}\ \bibnamefont
  {Bravyi}},\ and\ \bibinfo {author} {\bibfnamefont {Barbara}\ \bibnamefont
  {Terhal}},\ }\bibfield  {title} {\bibinfo {title} {Constructions and noise
  threshold of topological subsystem codes},\ }\href
  {https://doi.org/10.1088/1751-8113/44/15/155301} {\bibfield  {journal}
  {\bibinfo  {journal} {J. Phys. A: Math. Theor.}\ }\textbf {\bibinfo {volume}
  {44}},\ \bibinfo {pages} {155301} (\bibinfo {year} {2011})}\BibitemShut
  {NoStop}%
\bibitem [{\citenamefont {Jones}(2023)}]{Jones:DHR:2023}%
  \BibitemOpen
  \bibfield  {author} {\bibinfo {author} {\bibfnamefont {Corey}\ \bibnamefont
  {Jones}},\ }\href@noop {} {\bibinfo {title} {{DHR} bimodules of quasi-local
  algebras and symmetric quantum cellular automata}} (\bibinfo {year} {2023}),\
  \Eprint {https://arxiv.org/abs/2304.00068} {arXiv:2304.00068}\BibitemShut
  {NoStop}%
\bibitem [{\citenamefont {{Fidkowski}}\ \emph {et~al.}(2019)\citenamefont
  {{Fidkowski}}, \citenamefont {{Po}}, \citenamefont {{Potter}},\ and\
  \citenamefont {{Vishwanath}}}]{fermionGNVW2}%
  \BibitemOpen
  \bibfield  {author} {\bibinfo {author} {\bibfnamefont {L.}~\bibnamefont
  {{Fidkowski}}}, \bibinfo {author} {\bibfnamefont {H.~C.}\ \bibnamefont
  {{Po}}}, \bibinfo {author} {\bibfnamefont {A.~C.}\ \bibnamefont {{Potter}}},\
  and\ \bibinfo {author} {\bibfnamefont {A.}~\bibnamefont {{Vishwanath}}},\
  }\bibfield  {title} {\bibinfo {title} {{Interacting invariants for {F}loquet
  phases of fermions in two dimensions}},\ }\href
  {https://doi.org/10.1103/PhysRevB.99.085115} {\bibfield  {journal} {\bibinfo
  {journal} {Phys. Rev. B}\ }\textbf {\bibinfo {volume} {99}},\ \bibinfo
  {pages} {085115} (\bibinfo {year} {2019})},\ \Eprint
  {https://arxiv.org/abs/1703.07360} {arXiv:1703.07360}\BibitemShut {NoStop}%
\bibitem [{\citenamefont {Bacon}(2006)}]{Bacon2006}%
  \BibitemOpen
  \bibfield  {author} {\bibinfo {author} {\bibfnamefont {Dave}\ \bibnamefont
  {Bacon}},\ }\bibfield  {title} {\bibinfo {title} {Operator quantum
  error-correcting subsystems for self-correcting quantum memories},\ }\href
  {https://doi.org/10.1103/PhysRevA.73.012340} {\bibfield  {journal} {\bibinfo
  {journal} {Phys. Rev. A}\ }\textbf {\bibinfo {volume} {73}},\ \bibinfo
  {pages} {012340} (\bibinfo {year} {2006})},\ \Eprint
  {https://arxiv.org/abs/quant-ph/0506023} {arXiv:quant-ph/0506023}%
  \BibitemShut {NoStop}%
\bibitem [{\citenamefont {Napp}\ and\ \citenamefont
  {Preskill}(2012)}]{NappBaconShor2012}%
  \BibitemOpen
  \bibfield  {author} {\bibinfo {author} {\bibfnamefont {John}\ \bibnamefont
  {Napp}}\ and\ \bibinfo {author} {\bibfnamefont {John}\ \bibnamefont
  {Preskill}},\ }\href@noop {} {\bibinfo {title} {Optimal {B}acon-{S}hor
  codes}} (\bibinfo {year} {2012}),\ \Eprint {https://arxiv.org/abs/1209.0794}
  {arXiv:1209.0794}\BibitemShut {NoStop}%
\bibitem [{\citenamefont {Bombin}(2010)}]{Bombin2010TopoSubsysCode}%
  \BibitemOpen
  \bibfield  {author} {\bibinfo {author} {\bibfnamefont {H.}~\bibnamefont
  {Bombin}},\ }\bibfield  {title} {\bibinfo {title} {Topological subsystem
  codes},\ }\href {https://doi.org/10.1103/PhysRevA.81.032301} {\bibfield
  {journal} {\bibinfo  {journal} {Phys. Rev. A}\ }\textbf {\bibinfo {volume}
  {81}},\ \bibinfo {pages} {032301} (\bibinfo {year} {2010})},\ \Eprint
  {https://arxiv.org/abs/0908.4246} {arXiv:0908.4246}\BibitemShut {NoStop}%
\bibitem [{\citenamefont {Bravyi}\ \emph {et~al.}(2012)\citenamefont {Bravyi},
  \citenamefont {Duclos-Cianci}, \citenamefont {Poulin},\ and\ \citenamefont
  {Suchara}}]{Bravyi2012}%
  \BibitemOpen
  \bibfield  {author} {\bibinfo {author} {\bibfnamefont {Sergey}\ \bibnamefont
  {Bravyi}}, \bibinfo {author} {\bibfnamefont {Guillaume}\ \bibnamefont
  {Duclos-Cianci}}, \bibinfo {author} {\bibfnamefont {David}\ \bibnamefont
  {Poulin}},\ and\ \bibinfo {author} {\bibfnamefont {Martin}\ \bibnamefont
  {Suchara}},\ }\bibfield  {title} {\bibinfo {title} {Subsystem surface codes
  with three-qubit check operators},\ }\href@noop {} {\bibfield  {journal}
  {\bibinfo  {journal} {Quantum Info. and Comput.}\ }\textbf {\bibinfo {volume}
  {13}},\ \bibinfo {pages} {0963--0985} (\bibinfo {year} {2012})},\ \Eprint
  {https://arxiv.org/abs/1207.1443} {arXiv:1207.1443}\BibitemShut {NoStop}%
\bibitem [{\citenamefont {Chao}\ \emph {et~al.}(2020)\citenamefont {Chao},
  \citenamefont {Beverland}, \citenamefont {Delfosse},\ and\ \citenamefont
  {Haah}}]{Chao2020optimizationof}%
  \BibitemOpen
  \bibfield  {author} {\bibinfo {author} {\bibfnamefont {Rui}\ \bibnamefont
  {Chao}}, \bibinfo {author} {\bibfnamefont {Michael~E.}\ \bibnamefont
  {Beverland}}, \bibinfo {author} {\bibfnamefont {Nicolas}\ \bibnamefont
  {Delfosse}},\ and\ \bibinfo {author} {\bibfnamefont {Jeongwan}\ \bibnamefont
  {Haah}},\ }\bibfield  {title} {\bibinfo {title} {Optimization of the surface
  code design for {M}ajorana-based qubits},\ }\href
  {https://doi.org/10.22331/q-2020-10-28-352} {\bibfield  {journal} {\bibinfo
  {journal} {{Quantum}}\ }\textbf {\bibinfo {volume} {4}},\ \bibinfo {pages}
  {352} (\bibinfo {year} {2020})},\ \Eprint {https://arxiv.org/abs/2007.00307}
  {arXiv:2007.00307}\BibitemShut {NoStop}%
\bibitem [{\citenamefont {Gidney}(2022)}]{Gidney2022pentagon}%
  \BibitemOpen
  \bibfield  {author} {\bibinfo {author} {\bibfnamefont {Craig}\ \bibnamefont
  {Gidney}},\ }\href@noop {} {\bibinfo {title} {A pair measurement surface code
  on pentagons}} (\bibinfo {year} {2022}),\ \Eprint
  {https://arxiv.org/abs/2206.12780} {arXiv:2206.12780}\BibitemShut {NoStop}%
\bibitem [{\citenamefont {Sullivan}\ \emph {et~al.}(2023)\citenamefont
  {Sullivan}, \citenamefont {Wen},\ and\ \citenamefont
  {Potter}}]{Sullivan2023}%
  \BibitemOpen
  \bibfield  {author} {\bibinfo {author} {\bibfnamefont {Joseph}\ \bibnamefont
  {Sullivan}}, \bibinfo {author} {\bibfnamefont {Rui}\ \bibnamefont {Wen}},\
  and\ \bibinfo {author} {\bibfnamefont {Andrew~C.}\ \bibnamefont {Potter}},\
  }\href@noop {} {\bibinfo {title} {Floquet codes and phases in twist-defect
  networks}} (\bibinfo {year} {2023}),\ \Eprint
  {https://arxiv.org/abs/2303.17664} {arXiv:2303.17664}\BibitemShut {NoStop}%
\bibitem [{\citenamefont {Roberts}\ \emph {et~al.}(2023)\citenamefont
  {Roberts}, \citenamefont {Vijay}, \citenamefont {Vishwanath},\ and\
  \citenamefont {Dua}}]{Roberts2023}%
  \BibitemOpen
  \bibfield  {author} {\bibinfo {author} {\bibfnamefont {Brenden}\ \bibnamefont
  {Roberts}}, \bibinfo {author} {\bibfnamefont {Sagar}\ \bibnamefont {Vijay}},
  \bibinfo {author} {\bibfnamefont {Ashvin}\ \bibnamefont {Vishwanath}},\ and\
  \bibinfo {author} {\bibfnamefont {Arpit}\ \bibnamefont {Dua}},\ }\bibfield
  {title} {\bibinfo {title} {Topological invariants of {F}loquet codes}}
  (\bibinfo {year} {2023}),\ \bibinfo {note} {to appear}\BibitemShut {NoStop}%
\bibitem [{\citenamefont {Calderbank}\ \emph {et~al.}(1997)\citenamefont
  {Calderbank}, \citenamefont {Rains}, \citenamefont {Shor},\ and\
  \citenamefont {Sloane}}]{CalderbankRainsShorEtAl1997Quantum}%
  \BibitemOpen
  \bibfield  {author} {\bibinfo {author} {\bibfnamefont {A.~R.}\ \bibnamefont
  {Calderbank}}, \bibinfo {author} {\bibfnamefont {E.~M}\ \bibnamefont
  {Rains}}, \bibinfo {author} {\bibfnamefont {P.~W.}\ \bibnamefont {Shor}},\
  and\ \bibinfo {author} {\bibfnamefont {N.~J.~A.}\ \bibnamefont {Sloane}},\
  }\bibfield  {title} {\bibinfo {title} {Quantum error correction and
  orthogonal geometry},\ }\href {https://doi.org/10.1103/PhysRevLett.78.405}
  {\bibfield  {journal} {\bibinfo  {journal} {Phys. Rev. Lett.}\ }\textbf
  {\bibinfo {volume} {78}},\ \bibinfo {pages} {405--408} (\bibinfo {year}
  {1997})},\ \Eprint {https://arxiv.org/abs/quant-ph/9605005}
  {arXiv:quant-ph/9605005}\BibitemShut {NoStop}%
\end{thebibliography}%

\clearpage
\appendix

\section{Some algebra}\label{app:algebra}

\begin{proposition}\label{thm:Reversible-LinearVersion}
Let $A$ and $B$ be self-orthogonal subspaces of a finite dimensional $\FF$-vector space~$P$ equipped with a symplectic form~$\lambda$.
Put $S = A \cap B$.
For any subspace $\perp$ means its orthogonal complement within~$P$.
The following are equivalent.
\begin{enumerate}
	\item[(a)] Linear maps $B/S \ni b+S \mapsto \lambda(b, \cdot ) \in (A/S)^*$ 
	and $A/S \ni a+S \mapsto \lambda(a,\cdot) \in (B/S)^*$ 
	are both surjective.
	\item[(b)] The linear maps in~(a) are both injective, i.e., $A^\perp \cap B = S = A \cap B^\perp$.
	\item[(c)] The induced bilinear map~$A/S \times B/S \ni (a + S, b+S) \mapsto \lambda(a,b) \in \FF$ is nonsingular, 
	i.e., the linear maps in~(a) are both linear isomorphisms.
	\item[(d)] There exists a subspace~$L \subseteq P$ such that $L+A = A^\perp$ and $L+B = B^\perp$.
	\item[(d')] There exists a subspace~$H \subseteq P$ on which $\lambda$ restricts to a nonsingular form
	such that $H+A = A^\perp$ and $H+B = B^\perp$.
	\item[(e)] For any $p \in S^\perp$,
	there exist $a \in A, b \in B$ such that $\lambda(p-a, B) = 0$ and $\lambda(p-b, A) = 0$.
 \item[(f)] $A^\perp/A$ and $B^\perp/B$ are canonically isomorphic, where the isomorphism $A^\perp/A \to B^\perp/B$ is defined for every $a'+A$ (here $a'\in A^\perp$) by searching $a\in A$ such that $a'+a\in B^\perp$ and defining $a'+A\mapsto a'+a+B$.
\end{enumerate}
If (d') holds, the sums are orthogonal sums.
\end{proposition}

Theorem~\ref{thm:Reversibility} in the main text is a specialized version of this proposition for $\FF = \FF_2$.
Here (a--e) parallel~\ref{thm:Reversibility}(a--e) item-by-item.
The transcription should be straightforward using~\cite{CalderbankRainsShorEtAl1997Quantum}
to any interested reader. (f) is adopted for the abstract definition of MQCA.

\begin{proof}
(a) $\leftrightarrow$ (b) $\leftrightarrow$ (c):
A linear surjection implies that the dimension of the domain is at least that of the codomain.
So, (a) implies that the dimensions of $A/S$, $B/S$, and their duals are all the same.
A linear injection implies that the dimension of the codomain is at least that of the domain.
So, (b) implies that the dimensions of $A/S$, $B/S$, and their duals are all the same.
A surjection or an injection between finite dimensional spaces of the same dimension is automatically an isomorphism.
An isomoprhism is surjective and injective.

(a,b,c) $\to$ (d'):
For any $p \in S^\perp$ we have $\lambda(p, \cdot) \in (B/S)^*$ well defined.
But (a) says that this linear functional is equal to $\lambda(a(p), \cdot)$ for some~$a(p) \in A$, implying $p- a(p)\in B^\perp$ for all $p\in S^\perp$.
By (b), such $a(p)$ is unique up to~$S$.
In particular, since $A^\perp \subseteq S^\perp$, 
we define a linear subspace~$H$ to be the $\FF$-span of~$\{ a' - a(a') | a' \in A' \}$
where $A'$ is such that $\{a' + A | a' \in A' \}$ is any basis of~$A^\perp / A$.
By construction, $H \subseteq B^\perp$. 
Clearly, $H \subseteq A^\perp+A$ and $A^\perp\subseteq H+A$. 
Noticing that $A\subseteq A^\perp$, we get $H + A = A^\perp$.
Since $\lambda(a' - a(a'), a'' - a(a'')) = \lambda(a',a'')$, 
$\lambda$ restricted to~$H$ is the same as that on~$A^\perp/A$
where it is nonsingular. 
Then, $H \perp A$ implies $H \cap A = 0$; similarly $H \cap B = 0$.
Since $H \subseteq B^\perp$, we have $B^\perp \supseteq B + H$.
By~(c), we know $\dim A^\perp = \dim B^\perp$ and $\dim A = \dim B$, so it must be that $B^\perp = B \oplus H$.

(d) $\leftrightarrow$ (d')
Assume~(d).
Decompose $L$ as an orthogonal sum of a hyperbolic subspace~$H$ and a self-orthogonal subspace~$S$:
$L = H \oplus S$.
By assumption, $S \perp A$, and furthermore we see $S \perp S$ and $S \perp H$ by the decomposition.
Hence, $S \subseteq (H + S + A)^\perp = (A^\perp)^\perp = A$.
Therefore, $A^\perp = A + L = A + S + H = A + H$.
Similarly, $S \subseteq (H + S + B)^\perp = B$ and $B^\perp = B + H$.
This is (d').
It is obvious that (d') implies (d).

(d') $\to$ (c):
The nonsingularity of~$\lambda$ on~$L$ 
implies $P = L \oplus L^\perp$. $L\subseteq A^\perp$ implies $A\subseteq L^\perp$; similarly, $B\subseteq L^\perp$. 
Then, $A^\perp \cap B = (L\oplus A) \cap B = A \cap B = S = A \cap (L\oplus B) = A \cap B^\perp$.

(a) $\to$ (e):
By supposition, $\lambda(p,\cdot) \in (A/S)^*$. 
By~(a), choose $b \in B$ such that $\lambda(p,\cdot) = \lambda(b,\cdot)$.
By symmetry, we find $a \in A$ such that $\lambda(p,\cdot) = \lambda(a,\cdot)$.

(e) $\to$ (a):
If $f \in (A/S)^*$, then we can extend~$f$ to some linear functional on~$S^\perp / S$,
which is represented as $\lambda(p,\cdot)$ for some~$p \in S^\perp$.
By~(e), we find some~$b$ such that $\lambda(p-b,A) = 0$.
It follows that $f = \lambda(b,\cdot)$, so $B/S \to (A/S)^*$ is onto.
By symmetry, (a) follows.

(d) $\leftrightarrow$ (f): Assuming (d), then $A^\perp/A \cong L\cong B^\perp/B$. Moreover, since each  $a'\in A^\perp$ can be decomposed as $a'=a_0+l$ (here $a_0\in A$) and any $a\in A$ such that $a'+a\in B^\perp$ must be of the form $a=s-a_0$ (here $s\in S=A\cap B$), we know $a'+a+B$ is well-defined in $B^\perp/B$. Therefore the isomorphic $A^\perp/A \cong B^\perp/B$ does not depend on the choice of $L$. Assuming (f), we define an $L$ (perhaps nonuniquely)  by choosing a basis of $A^\perp/A$, finding an $a$ (perhaps nonuniquely) for each $a'+A$ in the basis, and spanning $L$ by $a'+a$. Clearly $L\oplus A=A^\perp$. $L\cap B=0$ and $L+B=B^\perp$ follow from the injectivity and surjectivity of $A^\perp/A \to B^\perp/B$.
\end{proof}

\paragraph{Generalized Pauli matrices and Fourier transforms.}
Let $p > 1$ be an integer and $\omega = e^{2\pi \ii k /p}$ be a primitive $p$-th root of unity where $\gcd(p,k)=1$.
Let $\ket{a}$ for $a \in \ZZ/p\ZZ = \ZZ_p$ span $\CC^p$.
The generalized Pauli matrices are $X = \sum_{j \in \ZZ/p\ZZ} \ket{j+1} \bra{j}$ 
and $Z = \sum_{j \in \ZZ/p\ZZ} \omega^j \ket j \bra j$.
Define a discrete Fourier transform~$F = \frac{1}{\sqrt{p}} \sum_{a,b \in \ZZ/p\ZZ} \ket a \omega^{ab} \bra b$.
It is readily checked that
\begin{equation}
	XZ = \omega^{-1} ZX, \qquad FF^\dag = \one, \qquad FXF^\dag = Z.
\end{equation}
Since $\{ X^a Z^b \,|\, a,b \in \ZZ/p\ZZ \}$ is an othonormal operator basis under the Hilbert--Schmidt inner product,
we may expand $F$ in this basis.
These are single-qudit matrix identities, but a generalization is straightforward.

\begin{lemma}\label{lem:Fourier}
If $A$ and $B$ are unitaries on a finite dimensional complex vector space such that 
\begin{equation}
A^p = B^p = \one, \quad AB = \omega^{-1} BA, \quad \abs[\Big]{\{ \omega^k \in \CC \,|\, k \in \ZZ \}} = p,
\end{equation}
then there exist some $\lambda_{a,b} \in \CC$ and a unitary~$U$ such that
\begin{equation}
	U = \sum_{a,b \in \ZZ / p \ZZ} \lambda_{a,b} A^a B^b, \qquad
	U A U^\dag = B \, .
\end{equation}
If $p$ is odd, then $\lambda_{a,b}$ can be chosen to be $p^{-1} \omega^{- ((p+1)/2)^2 (a-b)^2}$.
\end{lemma}
\begin{proof}
	Since $B$ is a unitary, there exists an eigenvector $\ket{\bar k}$ of~$B$
	where the eigenvalue must be a $p$-th root of unity, say $\omega^k$.
	Then the commutation relation shows that $\ket{\bar n} = A^{n-k} \ket{\bar k}$
	is an eigvenvector of~$B$ with eigenvalue $\omega^n$ for any $n \in \ZZ / p\ZZ$.
	Hence, the action of $A$ and $B$ on the $\CC$-linear span of~$\{\ket{\bar n} : n \in \ZZ / p \ZZ \}$
	is a representation of~$X,Z$.
	Since the algebra of $X,Z$ is simple, this representation is unique up to a unitary.
	For $A = X$ and $B = Z$, such~$U$ in the claim exists by the Fourier transform,
	and so does for more general~$A, B$.
	The special choice of $\lambda_{a,b}$ when $p$ is odd follows by direct calculation.
\end{proof}

\section{Details on the MQCA index}
\label{sec:indexproof}

Rehashing \eqref{eq:ConstructF},
\begin{align} \begin{aligned}
	\lgc_R &= \big\{ x\calS \in \calL/\calS ~\big|~ x \in \calL ,\, \Supp(x) \subseteq R \big\} ,\\
	\lgc_{+\infty} &= \bigcap_{k \in \ZZ} \lgc_{ (k,\infty) }  \,,\\
	\lgc_{-\infty} &= \bigcap_{k \in \ZZ} \lgc_{ (-\infty,k) } \,,\\
	\lgc_\infty &= \lgc_{+\infty} \lgc_{-\infty} \,, \\
	\calF &= (\calL/\calS) \Big/ \lgc_\infty \,, \\
	\calF_{<a} &= \big( \lgc_{<a} \lgc_{+\infty} \big) \Big/ \lgc_\infty \,.
\end{aligned} \end{align}
Recall that
$a \le b$ and $\calF_{\circ}$ is a subgroup
such that $\calF_{<a} \subseteq \calF_\circ \subseteq \calF_{<b}$.
Let $\nu: \calF \to \calF_{\circ}$ and $\iota :\calF_\circ \to \calF$
	be projection and inclusion maps respectively, such that $\nu\iota = \id_{\calF_\circ}$.
We wish to prove that 
	$\phi = \nu \bar \alpha |_{\calF_{\circ}} : 
	\calF_{\circ} \lhook\joinrel\xrightarrow{\quad\iota\quad} 
	\calF \xrightarrow{\quad\bar \alpha \quad} 
	\calF \xrightarrow{\quad\nu\quad} \calF_{\circ}$
is Fredholm, and that its Fredholm index is independent of the choice of $a,b,\nu$.

\begin{proof}[Proof of~\ref{lem:Fredholm2} that the MQCA index is well defined] \hfill\par

Define a map
\begin{align}
\eta : \calF_{\circ} \lhook\joinrel\xrightarrow{\quad\iota\quad} \calF \xrightarrow{\quad\bar \alpha^{-1} \quad} \calF \xrightarrow{\quad\nu\quad} \calF_{\circ} \,.
\end{align}
Suppose that $\alpha$ and $\alpha^{-1}$ have range at most~$r$,
so $\bar \alpha(\calF_{<k}) \subseteq \calF_{<k+r}$ and $\bar \alpha^{-1}(\calF_{<k}) \subseteq \calF_{<k+r}$
for any~$k \in\ZZ$.
Since the projection $\nu$ acts as the identity on~$\calF_{<a}$,
it follows that $\eta(\phi (f)) = \eta( \bar \alpha(f) ) = f$ for any $f \in \calF_{<a - 2r}$.
Take any subspace $\calE \subset \calF_{\circ}$ such that $\calF_{\circ} = \calF_{<a - 2r} \oplus \calE$.
We see that $(\eta \phi - \one)(\calF_{<a}) = (\eta \phi - \one)(\calE)$.
Here, $\calE \cong \calF_{\circ} / \calF_{<a - 2r}$ is contained in~$\calF_{<b} / \calF_{<a - 2r}$,
that is the image of 
the canonical map from~$\calL_{< b} / \calL_{< a - 2r}$,
which is finite dimensional 
because there are only finitely many qubits in the interval~$[a-2r,b)$.
Hence, $\calE$ is finite dimensional, and so is the image of~$\eta \phi - \one$.
Interchanging the role of $\alpha$ and $\alpha^{-1}$,
we see that $\phi \eta - \one$ is also finite rank.
Therefore, $\eta$ is a pseudoinverse for~$\phi$, implying that $\phi$ is Fredholm.

If $a',b',\nu'$ are some other choices giving a Fredholm map~$\phi'$,
then we set $a'' = \min(a,a')$
and consider a new projection~$\nu''$ onto~$\calF_{<a''}$,
which gives a new Fredholm map~$\phi''$.
It will suffice to show that $\ind(\phi) = \ind(\phi'')$
because by symmetry we will also have $\ind(\phi') = \ind(\phi'')$.
We chase the following diagram
\begin{align}
\xymatrix{
\calF_{\circ} \ar[r]^\iota& \calF \ar[r]^{\bar\alpha} & \calF \ar[dr]_{\nu''} \ar[r]^\nu & \calF_{\circ} \ar[d]^{\nu''}\\
\calF_{<a''} \ar[u]^{\iota''} & & & \calF_{<a''} 
}
\end{align}
where $\iota,\iota''$ are inclusions.

We claim that the difference 
$\Delta = (\nu'' \nu \bar \alpha \iota - \nu'' \bar \alpha \iota) : \calF_{\circ} \to \calF_{<a''}$
associated with the triangle in the diagram,
has finite rank.
To see this,
we recall that both projections~$\nu$ and~$\nu''$ are the identity on~$\calF_{<a''}$.
So, the difference $\Delta$ is zero on~$\calF_{< a'' - r}$ 
where $r$ is the range (also called spread) of~$\bar \alpha$.
Hence, $\Delta(\calF_{\circ}) = \Delta(\calE)$ 
where $\calE$ is a subspace such that $\calF_{<a''-r} \oplus \calE = \calF_{\circ}$.
This direct complement~$\calE$ is not necessarily unique,
but is always isomorphic to $\calF_{\circ} / \calF_{< a'' - r}$,
which is contained in~$\calF_{<b} / \calF_{< a'' - r}$,
which is the canonical image of~$\calL_{< b} / \calL_{<a'' - r}$,
which is finite dimensional 
because there are only finitely many qubits in the interval~$[a'' - r, a) \subset \ZZ$.
Hence, $\Delta$ has finite rank.
It follows that $\ind(\phi'') = \ind( \nu'' \nu \bar \alpha \iota \iota'')$
by the stability~\eqref{eq:FredholmStability} of the Fredholm index,
where the latter map, the longest chain of maps in the diagram,
is equal to $\nu'' \phi \iota''$.

It remains to show that $\ind(\nu'' \phi \iota'') = \ind(\phi)$.
To this end, we use the composition rule~\eqref{eq:FredholmPlus} of the Fredholm index.
The inclusion~$\iota''$ is injective, having zero kernel,
so $\ind(\iota'') = - \dim \coker \iota'' = - \dim \calF_{\circ} / \calF_{<a''}$,
which we have shown is finite.
On the other hand,
the projection $\nu'' : \calF_{\circ} \to \calF_{<a''}$ is surjective,
so $\ind(\nu''|_{\calF_{\circ}}) = \dim \ker (\nu''|_{\calF_{\circ}}) = \dim \calF_{\circ} / \calF_{<a''}$.
Therefore, $\ind(\iota'') + \ind(\nu''|_{\calF_{<a}}) = 0$
and the lemma is proved.
\end{proof}

\tocless\subsection{Properties}
\insertlabel{app:MQCAprop}{\thesubsection{}}
\begin{proof}[Proof of~\ref{MQCA-properties} on properties of the MQCA index] \hfill\par

	(Item~1) 
	This follows from the direct sum rule~\eqref{eq:FredholmPlus} of the Fredholm index.

	(Item~2) 
	Let $\nu : \calF \to \calF_{<a} \subseteq \calF$ be a projection, 
	which is a left inverse of the inclusion~$\iota : \calF_{<a} \to \calF$,
	i.e., $\nu \iota = \one$ on~$\calF_{<a}$.
	Then,
	\begin{align} \begin{split}
		\nu \bar \alpha \bar \beta \iota 
		&= \nu \bar \alpha ( \iota \nu + \one - \iota \nu) \bar \beta \iota
	\\	&= \nu \bar \alpha \iota \nu \bar \beta \iota  + \nu \bar \alpha (\one - \iota \nu) \bar \beta \iota
	\\	&= (\nu \bar \alpha \iota) (\nu \bar \beta \iota)  + \nu \bar \alpha (\one - \iota \nu) \bar \beta \iota \,.
	\end{split} \end{align}
	Here, the last term has finite rank because $\beta$ is locality preserving:
	for all $f \in \calF_{< a - r}$ we know that $\bar\beta(f) = \iota\nu \bar \beta (f)$,
	where $r$ is the spread of~$\beta$,
	and $\calF_{<a} / \calF_{<a -r}$ is finite dimensional.
	Hence,
	\begin{align} \begin{split}
		\ind(\nu \bar \alpha \bar \beta \iota) 
		&= \ind(  (\nu \bar \alpha \iota) (\nu \bar \beta \iota) )
	\\	&= \ind (\nu \bar \alpha \iota) + \ind (\nu \bar \beta \iota),
	\\	\Ind(\alpha \beta) &= \Ind(\alpha) + \Ind(\beta).
	\end{split} \end{align}
	
	(Item~3)
	Since $U$ is locality preserving, it restricts to an isomoprhism $\lgc_{\pm\infty} \to \calL'^{/\calS'}_{\pm\infty}$,
	and hence induces a locality-preserving isomorphism $\bar U: \calF \to \calF'$.
	If $\nu : \calF \to \calF_{<a}$ is a projection, that is a left inverse of the inclusion~$\iota : \calF_{<a} \to \calF$,
	Let $\calF'_\circ = \bar U(\calF_{<a})$ be a subspace of~$\calF'$.
	If $r$ is the spread of~$U$, then $\calF'_{<a - r} \subseteq \calF'_\circ \subseteq \calF'_{< a + r}$.
	Let $\iota' : \calF'_\circ \to \calF'$ be the inclusion,
	for which we choose a left inverse $\nu' = \bar U \iota \nu \bar U ^{-1} : \calF' \to \calF'_\circ$.
	Tautologically, $\nu' \bar U \iota : \calF_{<a} \to \calF'_\circ$ is an isomophism 
	with the inverse~$\nu \bar U^{-1} \iota'$.
	Both are Fredholm of index zero.
	Since $\bar U^{-1} \iota' = \iota \nu \bar U^{-1} \iota'$, we have
	\begin{align} \label{eq:equalindex} \begin{split}
		\Ind(U \alpha U^{-1}) 
		&= \thalf \ind\mathopen\big( \nu' (\bar{U}\bar{\qca}\bar{U}^{-1}) \iota' \big)
	\\	&=\thalf \ind\mathopen\big(( \nu' (\bar{U} \iota \nu \bar{U}^{-1}) (\bar{U}\bar{\qca}\bar{U}^{-1}) \iota' \big)
	\\	&=\thalf \ind( \nu' \bar U \iota ) + \thalf\ind( \nu \bar \alpha \iota) + \thalf\ind\mathopen\big(\nu \bar U^{-1} \iota'\big)
	\\	&= 0 +  \Ind(\alpha) + 0.
	\end{split} \end{align}
		
	(Item~4) 
	In this case $\calL=\calP$, $\calS=\{1\}$. 
	Everything goes back to the usual recipe for 1d unitary Clifford QCA index.
	It is well known that any unitary QCA in 1d is FDQC equivalent to a translation~\cite{GNVW}.
	The GNVW index is known to be invariant under FDQC,
	and our MQCA index is invariant under depth~1 semi-infinite array of unitary gates,
	and hence is invariant for any FDQC.
	For translations, direct calculation gives the result.
\end{proof}

\tocless\subsection{Measure of flow}
\insertlabel{app:MQCAflow}{\thesubsection{}}
\begin{proof} [Proof of~\ref{prop:IndAsFlow} on MQCA index as a measure of flow]\hfill\par

	Combine~\ref{lem:TopologicalBulkToLocallyFiniteBoundaryStabilizer},
	\ref{lem:LocallyGeneratedToLocallyFinite},
	\ref{lem:FaFbSeparation},
	\ref{lem:SemiLocalFormulaForIndex},
	\ref{lem:IndexOfReflection},
	and~\ref{lem:BlendingInvarianceOfMQCAIndex} below.
\end{proof}

\begin{lemma}\label{lem:TopologicalBulkToLocallyFiniteBoundaryStabilizer}
Let $\calB$ be a topological Pauli stabilizer group in a two-dimensional lattice~$\ZZ^2$
with locality parameter~$\ell$.
For a vertically extended strip~$I \subseteq \ZZ^2$ of finite width, 
let $\calS = \calB_I$ be the subgroup of~$\calB$ consisting of all elements supported on~$I$,
which is a one-dimensional Pauli stabilizer group.
Then, $\dim ( \calS / (\calS_{<a} \calS_{>b} ) ) < \infty$ for any~$a,b \in \ZZ$.
\end{lemma}

\begin{proof}
Since $\calS_{<a}$ gets bigger as $a \to +\infty$
and $\calS_{>b}$ gets bigger as $b \to -\infty$,
it suffices to prove the lemma when $a \ll b$.
Suppose $a + \ell < b$.
If the support of an element~$S \in \calS$ 
does not intersect the (vertical) interval~$[a,b]$ of~$I$,
then the tensor factor~$S_\uparrow$ of~$S$ above~$b$ commutes with every $\ell$-local generator of~$\calB$,
and hence is an element of~$\calB$ because $\calB$ is topological.
The same applies to the tensor factor~$S_\downarrow$ of~$S$ below~$a$,
and therefore $S \in \calS_{<a} \calS_{>b}$.
Hence, any nonzero element of the quotient $\calS / (\calS_{<a} \calS_{>b} )$
must act on the interval~$[a,b]$ nontrivially,
and the number of such operators that are independent
is at most the dimension of the Pauli group on the interval~$[a,b]$,
which is finite.
\end{proof}

\begin{lemma}\label{lem:LocallyGeneratedToLocallyFinite}
If $\calS$ is $\ell$-locally generated for some $\ell$,
then $\dim ( \calS / (\calS_{<a} \calS_{>b} ) ) < \infty$ for all~$a,b \in \ZZ$.
\end{lemma}

\begin{proof}
	Since $\calS_{<a}$ gets bigger as $a \to +\infty$
	and $\calS_{>b}$ gets bigger as $b \to -\infty$,
	it suffices to prove the lemma when $a \ll b$.
	$\calS / \calS_{<a} \calS_{>b}$ is generated by the $\ell$-local generators
	supported on $(a - 2\ell, b + 2\ell)$.
	There can only be finitely many generators on that interval.
\end{proof}

\begin{lemma}\label{lem:FaFbSeparation}
	If $\dim ( \calS / (\calS_{<a} \calS_{>b} ) ) < \infty$ for some~$a,b \in \ZZ$,
	then $\dim ( \calF_{<a} \cap \calF_{>b} ) < \infty $ for all~$a, b \in \ZZ$.
\end{lemma}

\begin{proof}
	Note that
	\begin{equation}
		\calF_{<a}\cap\calF_{>b}
		=
		\Bigl( \lgc_{<a}\lgc_{+\infty}\cap\lgc_{>b}\lgc_{-\infty} \Bigr) \big/ \lgc_{\infty}
		=
		\Bigl( \big(\lgc_{<a}\cap\lgc_{>b}\big)\lgc_{\infty} \Bigr) \big/ \lgc_{\infty} \,,
	\end{equation}
	where the second equality is because $\lgc_{-\infty}\subseteq\lgc_{<a}$ and $\lgc_{+\infty}\subseteq\lgc_{>b}$.
	So we only need to prove that $\lgc_{<a}\cap\lgc_{>b}$ is finite dimensional
	for all $a,b \in \ZZ$.
	
	First, assume $a\leq b$.
	We embed $\lgc_{<a}\cap\lgc_{>b}$ into $\calS/(\calS_{<a}\calS_{>b})$ as follows.
	If $L + \calS \in\lgc_{<a}\cap\lgc_{>b}$,
	then, obviously, $L + \calS = L_1 + \calS$ for some $L_1\in\calL_{<a}$
	and $L + \calS = L_2 + \calS$ for some $L_2 \in\calL_{>b}$.
	It follows that $L_1 L_2 \in \calS$.
	Using this, we define a map $[L] \mapsto [L_1 L_2] \in \calS/\calS_{<a}\calS_{>b}$.
	This map is well-defined:
	if $L'_1$ and $L'_2$ are different choices,
	then $L_1 L_1'\in \calS$, and hence $L_1L_1'\in \calS_{<a}$, 
	and similarly, $L_2 L_2' \in \calS_{>b}$,
	implying that $(L_1L_2)(L_1'L_2')\in \calS_{<a}\calS_{>b}$.
	This map must be injective. 
	Indeed, if $[L]$ maps to zero,
	then $L_1 L_2 \in \calS_{<a}\calS_{>b}$,
	implying $L_1 \in \calS_{<a}$ and $L_2\in \calS_{>b}$ since $a\leq b$.
	This means that $[L] = 0 \in \calL/\calS$.
	
	To handle the case where $a>b$,
	we note that we already know that $\dim\lgc_{<a}\cap\lgc_{>a} < \infty$. 
	Then, $\dim\bigl(\lgc_{<a}\cap\lgc_{>b}\bigr)/\bigl(\lgc_{<a}\cap\lgc_{>a}\bigr) \leq \dim \lgc_{>b}/\lgc_{>a}<\infty$,
	which completes the proof.
\end{proof}

\begin{lemma}\label{lem:SemiLocalFormulaForIndex}
	Suppose $\dim(\calF_{<a}\cap \calF_{>b})<\infty$ for some (and hence all) $a,b\in\ZZ$.
	Then, for any $a\in\ZZ$ there exists $b_0\leq a$ such that for all $b\leq b_0$, it holds that 
		\begin{equation}
			\Ind(\alpha) = \half 
			\Bigl( \dim(\calF_{<a}\cap\bar \alpha^{-1}\calF_{>b})
			-
			\dim(\calF_{<a}\cap\calF_{>b}) 
			\Bigr).
		\end{equation}
\end{lemma}

\begin{proof}
We leave it to the reader to show that 
if the assumption $\dim (\calF_{<a} \cap \calF_{>b}) < \infty$
is true for some~$a$ and~$b$, 
then it is true for all~$a$ and~$b$;
the reader can use the fact that 
both $\calF_{<a'} / \calF_{<a}$ and $\calF_{> b'} / \calF_{>b}$ are 
always finite dimensional for any~$a' \ge a$ and $b' \le b$.

Since $\alpha$ is locality preserving, 
there is $a' \ge a$ such that $\bar{\qca}(\calF_{<a}) \subseteq \calF_{<a'}$.
Clearly, for any $b$ we have
\begin{equation}
	\calF_{<a'} \supseteq (\calF_{>b} \cap\calF_{<a'}) + (\bar \alpha (\calF_{<a})\cap\calF_{<a})
\end{equation}
where the union of the right-hand side over all $b < a$ is obviously the left-hand side.
Since $\calF_{<a'}/(\bar{\qca}(\calF_{<a})\cap\calF_{<a})$ is finite dimensional,
this union is stabilized at a finite $b = b' < a$, with which we have the equality
\begin{equation}
	\calF_{<a'} = (\calF_{>b'} \cap\calF_{<a'}) + (\bar \alpha (\calF_{<a})\cap\calF_{<a}) .\label{eq:a1b}
\end{equation}
Now, since $\calF_{>b'} \cap \calF_{< a'} \supseteq \calF_{>b'} \cap \calF_{<a}$,
we find a vector subspace~$\calF_1 \subseteq \calF_{>b'} \cap \calF_{< a'}$ such that
\begin{equation}
	\calF_{>b'} \cap \calF_{< a'} = (\calF_{>b'} \cap \calF_{<a}) \oplus \calF_1 \,. \label{eq:b1ac1}
\end{equation}
By~\eqref{eq:a1b}, we know that the right-hand side~$\calF_{>b'} \cap \calF_{< a'}$
extends to~$\calF_{<a'}$ by some elements of~$\bar \alpha (\calF_{<a})\cap\calF_{<a}$
where the latter span a vector subspace~$\calF_2$ such that
\begin{align}
	\calF_2 &\subseteq \bar \alpha (\calF_{<a})\cap\calF_{<a} \,, \nonumber\\
	\calF_{<a'} &= (\calF_{>b'} \cap \calF_{<a'}) \oplus \calF_2 \,, \label{eq:apbpap2}
\end{align}
so
\begin{equation}
	\calF_{<a'} = (\calF_{>b'} \cap \calF_{<a}) \oplus \calF_1 \oplus \calF_2 \,.
\end{equation}
We take its intersection with $\calF_{<a}$:
\begin{equation}
	\calF_{<a} = (\calF_{>b'} \cap \calF_{<a}) \oplus (\calF_1 \cap \calF_{<a}) \oplus \calF_2 = (\calF_{>b'} \cap \calF_{<a}) \oplus \calF_2 \label{eq:2333}
\end{equation}
where $\calF_1 \cap \calF_{<a} = 0$ by~\eqref{eq:b1ac1}.

To define the Fredholm index, we choose a projection $\nu: \calF \to \calF_{<a}$,
which is a left inverse of the inclusion~$\iota : \calF_{<a} \to \calF$.
We choose another map $\mu : \calF \to \calF_{<a'} \to \calF_{<a}$ 
that is the identity on~$\calF_2$ but $\mu (\calF_{>b'} \cap \calF_{<a'}) = 0$;
see~\eqref{eq:apbpap2}.
The map~$\mu$ is really a projection onto~$\calF_2$ 
but we set the codomain to be the same as that of~$\nu$,
so we may consider~$\mu-\nu$.
Note that, by construction, $(\ker \mu) \cap \calF_{<a'} = \calF_{>b'} \cap \calF_{<a'}$.

We claim that $(\mu - \nu) \bar \alpha \iota : \calF_{<a} \to \calF_{<a}$ is finite rank.
It is routine to check that $(\mu - \nu) \bar \alpha \iota$ vanishes on 
$\bar \alpha^{-1}(\calF_2) \subseteq \calF_{<a} \cap \bar \alpha^{-1}(\calF_{<a})$;
its image under $\bar \alpha$ is~$\calF_2 = \im \mu \subseteq \im \nu$.
Hence, it suffices to see that $\calF_{<a} / \bar \alpha^{-1}(\calF_2)$ is finite dimensional,
but this quotient is isomorphic to $\bar \alpha(\calF_{<a}) / \calF_2$,
which is in~$\calF_{<a'} / \calF_2 \cong \calF_{>b'} \cap \calF_{<a'}$ by~\eqref{eq:apbpap2}.
The last is finite dimensional by assumption.

By~\eqref{eq:FredholmStability}, we know $\ind(\nu \bar \alpha \iota) = \ind(\mu \bar \alpha \iota)$.
Let us calculate the latter.
The kernel consists of~$x \in \calF_{<a}$ 
such that $\bar \alpha(x) \in (\ker \mu) \cap \calF_{<a'} = \calF_{>b'} \cap \calF_{<a'}$.
So, $x \in \calF_{<a} \cap \bar \alpha^{-1}(\calF_{>b'})$.
It is then checked that
\begin{equation}
	\ker (\mu \bar \alpha \iota) = \calF_{<a} \cap \bar \alpha^{-1}(\calF_{>b'}).
\end{equation}
Next, by definition, $\coker (\mu \bar \alpha \iota) = \calF_{<a} / \mu \bar \alpha(\calF_{<a})$,
but $\calF_2 = \im \mu$ is a subspace of~$\bar \alpha(\calF_{<a})$ by construction,
so $\mu \bar \alpha(\calF_{<a}) = \calF_2$. From~\eqref{eq:2333} we see
\begin{equation}
	\dim \coker(\mu \bar \alpha \iota) = \dim (\calF_{<a} / \calF_2) = \dim (\calF_{>b'} \cap \calF_{<a}) .
\end{equation}
Since~\eqref{eq:a1b} is true for any smaller $b'$,
we complete the proof.
\end{proof}

\begin{lemma}\label{lem:IndexOfReflection}
Suppose that $\calL$ is the commutant of a locally generated Pauli group within a full one-dimensional Pauli group
and further that $\dim (\calF_{<a} \cap \calF_{>b}) < \infty$ for some $a,b \in \ZZ$.
If $\alpha^\mathrm{refl}$ on~$\calL^\mathrm{refl}/ \calS^\mathrm{refl}$ 
is the MQCA obtained by the spatial reflection about any point in the 1d line of an MQCA~$\alpha$ on~$\calL/\calS$, 
then $\Ind(\alpha^\mathrm{refl}) = - \Ind(\alpha)$.
\end{lemma}

\begin{proof}
We claim that $\calF / (\calF_{<a} + \calF_{>b})$ is finite dimensional
whenever $a + \ell < b$ where $\ell$ is the locality parameter of an generating set~$G$ in the assumption,
of which $\calL$ is the commutant.
If the support of a Pauli operator~$P \in \calL$ does not intersect the interval~$[a,b]$
then the $\ell$-local generators in~$G$ can be tested for commutation relation
to each left and right tensor factor of~$P$.
Since $\calL$ is the commutant of~$G$,
the left tensor factor of~$P$ belongs to~$\calL_{<a}$ and the right tensor factor to~$\calL_{>b}$,
and therefore $P$ represents zero class in $\calF / (\calF_{<a} + \calF_{>b})$.
Therefore $\calF / (\calF_{<a} + \calF_{>b})$ must be generated by
elements of~$\calL$ that act nontrivially on the finite interval~$[a,b]$.
There are finitely many such elements that are linear independent.

By construction, $\bigcap_{b \in \ZZ} \calF_{>b} = 0$ since we have modded out the ``infinity groups.''
Therefore, $\lim_{b \to \infty} \dim( \calF_{<a} \cap \calF_{>b} ) = 0$.
Since we assume that $\calF_{<a} \cap \calF_{>b}$ is finite dimensional,
this limit zero must be attained by some finite~$b$.
Let $a,b \in \ZZ$ with $a + \ell < b$ be such that
\begin{equation}
	\calF_{<a} \cap \calF_{>b} = 0.
\end{equation}
It follows that there exists a subspace $\calF' \subseteq \calF$ such that
\begin{equation}
	\calF = \calF_{<a} \oplus \calF' \oplus \calF_{>b} \,.
\end{equation}
Let $\nu : \calF \to \calF_{<a}$ be the projection onto the first direct summand
and $\mu : \calF \to \calF' \oplus \calF_{>b}$ be the projection to the rest.
They are left inverses of the inclusions $\iota_\nu : \calF_{<a} \to \calF$ and $\iota_\mu : \calF' \oplus \calF_{>b} \to \calF$.
We may write $\one = (\nu,0) + (0,\mu)$.

Since $\dim \calF' < \infty$,
we have $\calF_{>b} \oplus \calF' \subseteq \calF_{>b'}$ for some~$b' < b$.
Then, it is clear by~\ref{lem:Fredholm2} that
\begin{align}
	\Ind(\alpha) &= \half \ind( \nu \bar \alpha \iota_\nu ), \\
	\Ind(\alpha^\mathrm{refl}) &= \half \ind ( \mu \bar \alpha \iota_\mu ). \nonumber 
\end{align}
By~\eqref{eq:FredholmPlus}, 
the direct sum $(\nu \bar \alpha \iota_\nu) \oplus (\mu \bar \alpha \iota_\mu)$ is Fredholm 
whose index is the sum of indexes of the two.
The direct sum can also be written as $(\nu,0) \bar \alpha (\nu,0) + (0,\mu) \bar \alpha (0,\mu)$.
We claim that a map
\begin{equation}
	(\nu,0) \bar \alpha (\nu,0) + (0,\mu) \bar \alpha (0,\mu) - \one : \calF \to \calF
\end{equation}
has finite rank.
Indeed, the map is zero on $\calF_{< a-r } \oplus \calF_{>b + r}$ where $r$ is the spread of~$\alpha$.
By the first paragraph in this proof, we know the complement of~$\calF_{< a - r} \oplus \calF_{>b + r}$ is finite dimensional.
Since $\ind(\one) = 0$, we conclude by~\eqref{eq:FredholmStability}
that
\begin{equation}
\ind(\nu \bar \alpha \iota_\nu) + \ind(\mu \bar \alpha \iota_\mu) 
=
\ind( (\nu \bar \alpha \iota_\nu) \oplus (\mu \bar \alpha \iota_\mu) ) = 
\ind( \one ) = 0 .
\end{equation}
This complete the proof.
\end{proof}

\begin{lemma}\label{lem:BlendingInvarianceOfMQCAIndex}
Assume the supposition of~\ref{lem:IndexOfReflection}.
If for two MQCA $\alpha,\beta$ on~$\calL / \calS$ there exists a third MQCA~$\gamma$ and $a, b \in \ZZ$ 
such that $\gamma|_{\calF_{<a}} = \alpha|_{\calF_{<a}}$ and $\gamma|_{\calF_{>b}} = \beta|_{\calF_{>b}}$,
then $\Ind(\alpha) = \Ind(\beta)$.
\end{lemma}

\begin{proof}
The MQCA index is defined by the action of the MQCA on~$\calF_{<a}$,
and hence $\Ind(\alpha) = \Ind(\gamma)$.
On the other hand, the MQCA index of a space-reflected MQCA~$\gamma^\mathrm{refl}$ 
is determined by the action of~$\gamma$ on $\calF_{>b}$,
and hence $\Ind(\gamma^\mathrm{refl}) = \Ind(\beta^\mathrm{refl})$.
The claim follows by~\ref{lem:IndexOfReflection}.
\end{proof}

\begin{proof}[Proof of \ref{thm:IntegerIndexOf1dLRMC} that the MQCA index of 1d LR cycle is an integer] \hfill\par

A finite Ising chain has two logical operators~$\bar Z$ and~$\bar X$.
The operator~$\bar X$ is the all-qubit flip logical operator $\bar X = X^{\otimes n}$
where $n \ge 0$ is the length of the chain.
There is no other equivalent representative that is a tensor product of~$X$ 
for the logical operator~$\bar X$.
The conjugate logical operator~$\bar Z$ is a single-qubit operator~$Z$
which can be put anywhere within the chain.
The support of an Ising chain is that of~$\bar X$.
For (semi-)infinite chains the all-qubit flip logical operator is not finitely supported.

Applying the structure theorem~\ref{thm:IsingIn1D},
our $\calL/\calS$ is generated by~$\bar X$ of \emph{finite} Ising chains,
and $\bar Z$ of \emph{all} chains.
The groups at infinities, $\lgc_{+\infty}$ and $\lgc_{-\infty}$,
are generated by~$\bar Z$ of \emph{infinite} Ising chains,
so the quotient~$\calF$ is generated precisely of all logical operators of \emph{finite} chains.
We choose a projection~$\nu$ onto the group of all logical operators of the finite chains 
that are supported on $(-\infty,0)$.
So, the image of~$\nu$ is isomorphic to the full Pauli group of some collection of qubits.%
\footnote{%
	One might have noticed that in our construction towards~$\Ind$,
	the groups at infinity~$\lgc_{\pm \infty}$ did not play any important role
	other than being an invariant subgroup under~$\alpha$.
	Here, the groups at infinities make it easy to reduce the MQCA 
	to a unitary QCA on a full algebra of local operators.
	Any systematic choice of an invariant subgroup, e.g.,~$\lgc_{\pm \infty}$, defines an index.
	It is a separate question whether such an index is useful.
}

To consider the locality preserving property of the MQCA,
we compress the lattice so that every chain is supported on at most two neighboring sites;
this compression or coarse-graining can be very nonuniform, and so can be the number of qubits.
Nonetheless, since there are only a finite number of Ising chains that can overlap in support,
and we only have finite Ising chains,
the number of qubits on any post-compression site is finite.
Replace each Ising chain with a qubit located at either site in the support of the compressed chain.
Then, the MQCA defines a unitary Clifford QCA on this new lattice of qubits with finitely many qubits per site.
The MQCA index has nothing to do with the compression,
and thus is equal to the index of the unitary Clifford QCA,
which we know is an integer.
\end{proof}

\tocless\subsection{Strip geometry}
\insertlabel{app:MQCAstrip}{\thesubsection{}}
\begin{proof}[Proof of \ref{lem:IndexSeparation} that $\Ind(\alpha_\strip) = \Ind(\alpha_1) + \Ind(\alpha_2)$] \hfill\par

In this proof we omit the subscript ``strip'' for simplicity of notation.
So, $\calL$, $\calA$, and $\calF$ are, respectively, the group of all finitely supported on the strip,
the group of all stabilizers on the strip, and the quotient group according to~\eqref{eq:ConstructF}.
From the discussion in~\S\ref{sec:boundaryactions}, we have two injections 
$\calL^{/\calA_1}_1\to\calL^{/\calA}$ and  $\calL^{/\calA_2}_2\to\calL^{/\calA}$.
Combining these, we have a map
\begin{align}
	\psi: \calL^{/\calA_1}_1\oplus \calL^{/\calA_2}_2\to \calL^{/\calA}.
\end{align}
Now consider the following commutative diagram:
\begin{align}\label{eq-defphi}
	\begin{aligned}
		\xymatrix @C=7mm @R=6mm @M=2mm{
0 \ar[r]  &  \calL^{/\calA_1}_{1,\infty}\oplus \calL^{/\calA_2}_{2,\infty} \ar@{^{(}->}[rr]^{} \ar[d]_{\psi_{\infty}}  &&  \calL^{/\calA_1}_1 \oplus \calL^{/\calA_2}_2 \ar[rr]\ar[d]_{\psi}  &&  \calF_1 \oplus \calF_2 \ar@{-->}^\varphi[d] \ar[r]  &  0 \\
0 \ar[r]  &  \calL^{/\calA}_{\infty} \ar@{^{(}->}[rr]^{}  &&  \calL^{/\calA} \ar[rr] && \calF \ar[r]  &  0
		}
	\end{aligned}
\end{align}
The rows in this diagram are exact.
The groups $\calL^{/\calA}_{n,\infty}$ denote $\calL^{/\calA}_{n,+\infty} \calL^{/\calA}_{n,-\infty}$,
the map~$\psi_\infty$ is a restriction of~$\psi$.
The commutative diagram defines a map
\begin{equation}
	\varphi: \calF_1 \oplus \calF_2\to\calF. 
\end{equation}
For any vertical position~$k$ of our strip,
we restrict $\varphi$ to $(\calF_1)_{<k} \oplus (\calF_2)_{<k}$ and define:
\begin{align}
	\varphi_{<k} : (\calF_1)_{<k} \oplus (\calF_2)_{<k} \to (\calF)_{<k} \,.
\end{align}
We will show:
	\begin{itemize}[itemsep=1pt, parsep=0pt, topsep=2pt]
		\item[(i)] $\psi$ is injective.
		\item[(ii)] $\dim\ker\varphi<\infty$, and hence $\dim\ker\varphi_{<k}<\infty$.
		\item[(iii)] $\dim\coker \varphi_{<k} < \infty$.
	\end{itemize}

	These claims will imply that $\Ind(\alpha) = \Ind(\alpha_1) + \Ind(\alpha_2)$ as follows.
	Recall that the indices of the boundary MQCA $\alpha_{1,2}$ and the strip MQCA $\alpha$ 
	are defined by the Fredholm indices of maps as in~\eqref{eq-nualphaiota}.
	Let us consider the following diagram:
	\begin{align}
		\vcenter{\xymatrix{
			\calF_{<k}\ar[rr]^{\iota} && \calF\ar[rr]^{\bar\alpha} && \calF\ar[rr]^\mu && \calF_{<k} \\
			(\calF_1)_{<k} \oplus (\calF_2)_{<k}\ar[rr]^{\iota_1\oplus\iota_2}\ar[u]^{\varphi_{<k}} && \calF_1 \oplus \calF_2\ar[rr]^{\bar\alpha_1 \oplus \bar\alpha_2}\ar[u]^{\varphi} && \calF_1 \oplus \calF_2\ar[u]^\varphi \ar[rr]^{\nu_1\oplus\nu_2} && (\calF_1)_{<k} \oplus (\calF_2)_{<k}\ar[u]^{\varphi_{<k}}
		}}
	\end{align}
	In this diagram, $\varphi_{<k}$ is Fredholm due to the claims that $\varphi_{<k}$ has finite dimensional kernel and cokernel.
	The composition of the three maps in the top line is the Fredholm map for the strip MQCA index
	(see~\eqref{eq-nualphaiota}), where $\iota$ is the inclusion and $\mu$ is a projection.
	The bottom line is the direct sum of the maps for the boundary MQCA indices,
	where $\iota_1$ and $\iota_2$ are the inclusions, and $\nu_1$ and $\nu_2$ are projections.
	The left square commutes because of the definition of~$\varphi$ and~$\varphi_{<k}$.
	The middle square commutes because~$\bar \alpha$ 
	restricts to~$\bar \alpha_1 \oplus \bar \alpha_2$ on~$\calF_1 \oplus \calF_2$.
	The right square may not be commuting,
	but the difference due to this square starting from the bottom-left is finite rank
	because the subgroup~$(\calF_1)_{<k - r} \oplus (\calF_2)_{<k-r}$ of the domain with $r$ larger 
	than the spreads of the MQCAs is mapped the same.
	Hence, using \eqref{eq:FredholmStability} we have
	\begin{align}
		\ind\mathopen\big( \varphi_{<k}(\nu_1\oplus\nu_2)(\bar\alpha_1 \oplus \bar\alpha_2)(\iota_1\oplus\iota_2) \big) = \ind(\mu\bar\alpha\iota\varphi_{<k}).
	\end{align}
	Using the composition rule for the Fredholm index (\ref{MQCA-properties}), we arrive at:
	\begin{align}
		\Ind(\alpha) = \frac{1}{2}\ind(\mu\bar\alpha\iota) = \frac{1}{2}\ind\mathopen\big( (\nu_1\oplus\nu_2)(\bar\alpha_1 \oplus \bar\alpha_2)(\iota_1\oplus\iota_2) \big) = \Ind(\alpha_1)+\Ind(\alpha_2).
	\end{align}
	It remains to prove those three claims. 

	(i) $\psi$ is injective.
	
	Suppose the product of two boundary logical operators $Q_i\in\calL_i$ is a strip stabilizer,
	$Q_1Q_2\in\calA$. We have to show that $Q_1 \in \calA$ and $Q_2 \in \calA$.	
	By assumption, $Q_1Q_2=\prod A$, where $A\in\calA$ are local strip stabilizer generators.
	These $A$'s do not need to be supported near the support of~$Q_1Q_2$,
	but the number of them is finite.
	We partition the whole strip into three regions, two of which are the interface regions~$I_1$ and~$I_2$,
	each having width~$10\ell$, and the bulk denoted as~$I_b$ the complement of~$I_1 \cup I_2$.
	We have assumed that the conditions for topological codes 
	hold for all operators and regions supported on~$I_b$.
	Accordingly, we group those factors~$A$ into three classes: 
	class 1 are those that intersect $I_1$, 
	class 2 are those that intersect $I_2$, 
	class $b$ are those that are fully supported on~$I_b$.
	We have $Q_1Q_2=A_1 A_2 A_b$, where $A_i$ is the product of all those $A$'s in class~$i$ ($i=1,2,b$).
	Equivalently, $(Q_1 A_1^{-1}) (Q_2 A_2^{-1}) = A_b$.
	
	Now, observe that for $i=1,2$, the product~$Q_i A_i^{-1}$ is supported on~$I'_i$, 
	a strip of width~$20\ell$ near boundary~$i$. 
	For their product~$A_b$ to be supported on $I_b$,
	it must be that $Q_i A_i^{-1}$ is supported on~$I'_i \cap I_b$.
	Then, every $\ell$-local stabilizer $S \in\calA$ 
	must commute with both $Q_1 A_1^{-1}$ and $Q_2 A_2^{-1}$ 
	since $S$ cannot overlap with both of them.
	But both $Q_1A_1^{-1}$ and $Q_2 A_2^{-1}$ are supported on~$I_b$,
	where the first condition of topological codes implies that
	$Q_i A_i^{-1} \in \calA$ for $i=1,2$ individually,
	which in turn implies that $Q_i \in \calA$ for $i=1,2$.

	(ii) $\dim\ker\varphi<\infty$. 
	
	Since $\ker \psi = 0$ by~(i),
	the snake lemma applied to the commutative diagram~\eqref{eq-defphi} gives
	\begin{equation}
		\ker\varphi	\cong \ker( \coker \psi_\infty \to \coker \psi ).
	\end{equation}
	Here, the codomain~$\calL^{/ \calA}_\infty$ of~$\psi_\infty$
	is the logical quotient group at infinity of a one-dimensional system,
	in which the number of qudits per effective site at any vertical position~$k$
	depends on the width of the strip at that position, which is uniformly bounded by assumption.
	Therefore, the structure theorem~\ref{thm:IsingIn1D} 
	implies that there are only finitely many (semi-)infinite Ising chains.
	Hence, $\calL^{/\calA}_\infty$ is finite dimensional, 
	so $\dim \ker \varphi \le \dim \coker \psi_\infty \le \dim \calL^{/ \calA}_\infty < \infty$.

	(iii) $\dim\coker \varphi_{<k} < \infty$.
	It suffices to show that
	\begin{equation}	
	\dim \frac{ \calL_{<k} }{
		(\calL_1)_{<k} + (\calL_2)_{<k} + \calA
	} < \infty .
	\end{equation}

	To illustrate what we are going to do, 
	we consider $P = P_1 P_2 \in \calL$
	where $P_1$ and $P_2$ are supported on vertical regions~$V_1 \supseteq I_1$ and $V_2\supseteq I_2$, respectively,
	where $V_1$ and $V_2$ are separated by distance larger than~$\ell$,
	so that every $\ell$-local generator of~$\calA$ may overlap with at most one of~$V_1$ and~$V_2$.
	Then, $P_1$ and~$P_2$ are both in~$\calL_\strip$.
	The operators~$P_1$ and $P_2$ can then be localized to~$I_1$ and~$I_2$, respectively,
	by some elements of~$\calA$.
	Hence, $P$ is in the image of~$\varphi$.
	Note that the localization of $P_1, P_2$ may enlarge their support vertically, but not more than~$10\ell$.
	So, if $P = P_1 P_2  \in \calL_{<k - 10\ell}$, then $P \in \im \varphi_{<k}$.
	Therefore, any representative~$P \in \calL$ of a nonzero class of~$\coker \varphi_{<k}$
	must either not have any vertical ``gap" in its support 
	or act nontrivially on the window of vertical positions between~$k - 10\ell$ and~$k$.
	Those that act on the window is easily seen to be finite dimensional because
	its $\FF_2$-dimension is bounded by the number of qudits in that window of the strip.
	The remaining questions is
	how many operators of~$\calL$ are there that need a no-gap support.
	
	\begin{figure}
		\centering
		\includegraphics[width=70mm]{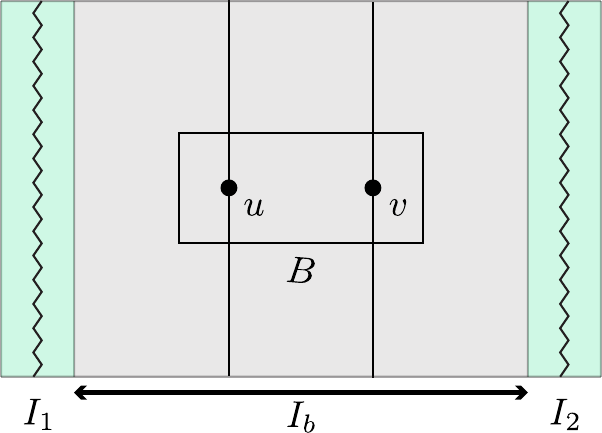}
		\caption{
			The geometry of regions in the proof of~\ref{lem:IndexSeparation}.
			For any logical operator,
			we find an equivalent logical operator that is supported near the boundaries~$I_1, I_2$
			and a bridge~$B$ of small height.
		}
		\label{fig:gluedStripUnfolded}
	\end{figure}

	Take any logical operator~$Q \in \calL_{<k}$ and 
	consider its truncation~$Q_1$ on, say, the $50\ell$-neighbor\-hood of~$I_1$.
	See Fig.~\ref{fig:gluedStripUnfolded}.
	Then, $Q_1$ will create some excitations on~$I_b$ which is the region in between~$I_1$ and~$I_2$, 
	but the condition~(ii) for topological codes in~\ref{defn:TopologicalCode}
	implies that these excitations 
	can be moved to a fixed location, say~$u \in I_b$, by some finitely supported operator 
	supported within distance~$80\ell$ from~$I_1$.
	Similarly, consider another truncation~$Q_2$ of~$Q$ on the $50\ell$-neighborhood of~$I_2$,
	and bring any excitation created by~$Q_2$ to a fixed location~$v \in I_b$ by some finitely supported operator
	supported within distance~$80\ell$ from~$I_2$.
	Since the set of all excitations caused by~$Q_1$ and~$Q_2$ 
	is created by a finitely supported operator,
	the condition~(iii) for topological codes 
	annihilates the moved excitations at~$u$ and~$v$ by some operator
	supported within the $10\ell$-neighborhood~$B$ of the straight line connecting~$u$ and~$v$.
	Thus, we have deformed $Q$ to~$Q'$ where $Q'$ is supported on the ``bridge''~$B$ and two vertical regions,
	that are each within distance~$80\ell$ from~$I_1$ and~$I_2$.
	Since any representative of a nonzero class of~$\coker \varphi_{<k}$ that is supported on $<k - 10\ell$
	must cross this bridge,
	we conclude that $\coker \varphi_{<k}$ has dimension bounded by the number of qudits in~$B$
	plus the number of qudits in the window of vertical coordinate between~$k - 10\ell$ and~$k$.
\end{proof}

\section{Structure of 1d Pauli stabilizer groups}\label{app:IsingIn1D}

Here we prove~\ref{thm:IsingIn1D} which we copy here for readers' convenience.
This section does not depend on any results from other sections of this paper.

\IsingOneD*

Every qudit in this proof will have a fixed prime dimension~$p \in \{ 2,3,5,7,11,\ldots \}$,
so we will write $\FF$ to denote the prime field instead of~$\FF_p$.
As usual, $Z,X$ will denote generalized Pauli or Weyl operators
\begin{align}
	Z &= \sum_{a \in \FF} \exp\left[\frac{2\pi \ii a}{p}\right] \ket{a}\!\bra{a} ,
&	X &= \sum_{a \in \FF} \ket{a+1}\!\bra{a} .
\end{align}
An {\bf Ising coupling} is $Z_j Z_{j+1}^\dag$ across two neighboring sites~$j$ and~$j+1$.
Recall that an Ising chain over one site consists of, by convention, 
one qudit with a stabilizer~$Z$ on it,
which is thus completely disentangled from the rest of the system.

\tocless\subsection{Reduction to infinite lattice}

Suppose we have the theorem for the infinite lattice~$\ZZ$.
Then, we can prove the theorem for any finite periodic lattice~$\ZZ/L\ZZ$ as follows
where $L$ is the number of sites.
If $L=1$ or $2$, the proof is left to the reader; see~\ref{lem:F2} below.
If $L \ge 3$, then since every generator~$P_i P'_{i+1}$ of the stabilizer group~$\calS$ of a finite periodic system
acts on just two neighboring sites $i,i+1 \in \ZZ/L\ZZ$,
we can define a lifted stabilizer group~$\calA$ on the infinite lattice 
by generators~$P_j P'_{j+1}$ for all $j = i \bmod L$;
the lifted stabilizer group is determined by the choice of a $2$-site local generating set of~$\calS$,
not just by the group~$\calS$.

The theorem for the infinite lattice gives us single-site Cliffords~$U_j$ for each site~$j \in \ZZ$
by which this lifted stabilizer group~$\calA$ 
becomes that of a collection of independent Ising chains and Bell pairs.
Note that $U_j$'s do not have to obey the periodicity~$L$,
i.e., it may be that $U_j \neq U_{j+L}$ for some~$j$.
Since each Bell pair is identified on just two sites,
the finite system has a Bell pair whenever the lifted system does.
Extracting all Bell pairs, we are left with Ising chains in~$\calA$.
Consider an Ising chain over $n$ sites in~$\calA$.
If $n = 1$, then one qubit is in a fixed product state, disentangled with the rest of the system.
This means that the finite system also has a disentangled qubit of a fixed state.
Similarly, if $n \le L$, then the Ising chain can also be distinguished in the finite system,
using $n$ single-site Cliffords.
Extracting all such Ising chains, we are left with a finite system
where the Ising chains of the lift, if any, always occupy more than $L$ sites.

Suppose we have chosen a nonredundant $2$-site local generating set for~$\calS$
and defined the lift~$\calA$ that has no Bell pairs or Ising chains over $L$ or fewer sites.
Suppose there is an Ising chain embedded in~$\calA$.
Then, we must have~$Z_j Z_{j+L}^\dag \in (U_j U_{j+L}) \calA (U_j U_{j+L})^\dag$ for some~$j \in \ZZ$.
This stabilizer is, in the finite chain, a product of 2-site local generators of~$\calS$ over $L$ bonds,
which is supported on a single site~$i = j \bmod L$.
By assumption, such a stabilizer must be the identity for the lack of any Ising chain over one site,
but then it violates the assumption that $\calS$ does not have any redundant generator.
Therefore, the lift $\calA$ must become the vacuous group 
after extracting Bell pairs and Ising chains over $L$ or fewer sites.
This completes the reduction from finite periodic cases to the infinite case.

\bigskip
\tocless\subsection{Reduction to classical codes}

\begin{lemma}\label{lem:F1}
For any set of independent commuting Pauli operators $P_1, P_2,\ldots, P_n$ on finitely many qudits,
	there exists a Clifford unitary~$U$ such that $U P_j U = Z_j$ for all~$j$ 
	where $Z_j$ is a single-qudit Pauli acting on qudit~$j$.
\end{lemma}
\begin{proof}
This is a rephrasing of the facts that any self-orthogonal subspace 
of a finite dimensional symplectic space extends to a lagrangian subspace
and that any lagrangian subspaces are isometric to each other.
\end{proof}

\begin{lemma}\label{lem:F2}	
	Given a finite Pauli stabilizer group~$\calS$ on a bipartite system
	with no element supported on either one party,
	there exists an unentangling Clifford~$U \otimes V$ such that $(U \otimes V) \calS (U \otimes V)^\dag$
	is supported on a collection of qudit pairs~$(i,1),(i,2)$ 
	where the second index denotes the party, 
	with a generating set where every generator is either $X_{i,1} \otimes X_{i,2}$ or $Z_{i,1} \otimes Z_{i,2}^\dag$.
\end{lemma}
If a pair~$i$ of qudits supports two generators, it is a Bell pair;
otherwise, the stabilizer group is an Ising coupling. 
\begin{proof}
We use induction in the order of~$\calS$.
Define groups~$\calS_1, \calS_2$ of Pauli operators in each party by
\begin{align} \begin{split}
	\calS_1 &= \{ P  \in \calP_1 \,|\, P \otimes Q \in \calS \text{ for some }Q \},
\\	\calS_2 &= \{ Q  \in \calP_2 \,|\, P \otimes Q \in \calS \text{ for some }P \},
\end{split} \end{align}
where $\calP_1,\calP_2$ are the groups of all Pauli operators on the respective party.
For $P \in \calS_1$, suppose there are two operators~$Q,Q'$ such that both $P \otimes Q, P \otimes Q'$ are in~$\calS$.
Then, their product $\one \otimes Q^{-1} Q'$ is also in $\calS$, supported on the second party,
violating assumption unless $Q = Q'$ because $\calS$ does not contain any nontrivial scalar,
i.e., if $\eta \one \in \calS$ for some~$\eta \in \CC$, then $\eta = 1$.
So, there is a unique~$Q$ on the second party corresponding to~$P$ such that $P \otimes Q \in \calS$.
This gives a group homomorphism
\begin{align}
\phi : \calS_1 \ni P \mapsto Q \in \calS_2 \,,
\end{align}
which is surjective by construction.
Note that $\calS_1$ and $\calS_2$ always include phase factors because $(-\one) \otimes (-\one) \in \calS$.
If $\phi(P) = \one$, then $P \otimes \one \in \calS$ and by assumption $P$ has to be the identity.
So, $\phi$ is injective and hence is a group isomorphism.
If $\calS_1$ is nonabelian, 
then a noncommuting pair $P,P' \in \calS_1$ is associated with a noncommuting pair $Q,Q' \in \calS_2$.
The pair $P \otimes Q, P' \otimes Q'$ defines a Bell pair.
Restricting to the commutant of~$P,P'$ within~$\calS_1$
(the set of all elements of~$\calS_1$ that commute with both~$P,P'$)
and that of~$Q,Q'$ within~$\calS_2$,
we reduce the order of the bipartite stabilizer group.
Hence, we may now assume that $\calS_1$ is abelian, and hence so is~$\calS_2$.
By~\ref{lem:F1} we find a Clifford $U \otimes V$ by which $\calS$ is mapped to
the group generated by~$\{ Z_{1,j} \otimes Z_{2,j}^\dag \,|\, j = 1,2,\ldots \}$,
the Ising couplings.
\end{proof}

Now, let us apply those two facts to a one-dimensional stabilizer group~$\calA$.
Recall the assumption that~$\calA$ is generated by operators acting on neighboring two sites.
Without loss of generality, 
we can assume that $\calA$ does not have any one-site operator;
if it did, we would find a one-site Clifford to map the one-site operator to a single-qudit operator~$Z$
and remove the qudit stabilized by this.
For any site~$i \in \ZZ$,
we denote by $\calA(i,i+1) \subseteq \calA$ the subgroup of all elements of~$\calA$ supported on the two sites.
Applying~\ref{lem:F2}, we drop any Bell pair qudits.
Then, from the proof for~\ref{lem:F2} above,
we know that $\calB_i' = \calA(i,i+1)_1$ (the ``left group'' of of the bond~$(i,i+1)$)
and $\calB_{i+1}'' = \calA(i,i+1)_2$ (the ``right group'' of the bond~$(i,i+1)$) 
are isomorphic and both abelian for all~$i$.
Furthermore, since $\calA$ is abelian, the left group $\calB_i'$ of the bond on the right of~$i$
commutes with the right group $\calB_i''$ of the bond on the left of~$i$.
For each site~$i$, let $\calB_i$ be the abelian group generated by~$\calB_i'$ and~$\calB_i''$.
This bigger group~$\calB_i$ is not necessarily a subgroup of~$\calA$.
Since $\calB_i$ is abelian, there is a Clifford on site~$i$ 
that brings~$\calB_i$ to a group consisting of tensor products of~$Z$ only.
Hence, our group~$\calA$ is generated  entirely by tensor products of single-qudit~$Z$ after some one-site Cliffords,
which we assume till the end of the proof of~\ref{thm:IsingIn1D}.

In the following we will conclude the proof using two complementary methods.
The first method is by analyzing the structure of logical operators, 
and the second is by analyzing the structure of stabilizer operators. 

\tocless\subsection{Looking at logical operators}

Since the stabilizer group~$\calA$ consists of tensor products of Pauli~$Z$ only,
the $X$-type logical operators determine everything.
We allow ``infinite tensor products'' of Pauli~$X$ as logical operators.%
\footnote{%
Strictly speaking, an infinite tensor product is not defined.
We are confident that an interested reader can fill this ``gap.''
}
The set~$\calX$ of all $X$-type logical operators for~$\calA$
can thus be identified with a set of all bitstrings that may be infinite, one bit ($\in \FF$) for each qudit,
such that the dot product of the bitstring with exponents of the stabilizer is zero.
This set~$\calX$ is an $\FF$-vector space.
For any set~$S \subseteq \ZZ$ of sites and for any bit string~$x$,
let $\Pi_S(x)$ be the substring of~$x$ on~$S$.
It is an $\FF$-linear projection.
We simply write $\Pi_{<k} = \Pi_{\{ j \in \ZZ | j < k\}}$,
$\Pi_{>k} = \Pi_{\{ j \in \ZZ | j > k\}}$,
and
$\Pi_k = \Pi_{\{ k\}}$.
The 2-site locality of the stabilizer group~$\calA$ implies the following.

\begin{lemma}\label{lem:F3}
If $x \in \calX$ is zero on a site~$k$, i.e., $\Pi_{\{k\}}(x) = 0$,
then the left substring~$\Pi_{<k}(x)$ and right substring~$\Pi_{>k}(x)$ are both members of~$\calX$.
\end{lemma}

\begin{proof}
Each 2-site local stabilizer generator intersect 
only one of the left and right substrings, not both.
So, if $x$ is logical, 
then every stabilizer generator must commute with both left and right substrings separately.
\end{proof}

Let us write $\calX(S)$ for any~$S \subseteq \ZZ$ to mean
the space of all bitstrings of~$\calX$ that is supported on~$S$.
(In earlier sections we used subscripts for this purpose.)
Equivalently, $\calX(S) = \calX \cap \ker \Pi_{\ZZ \setminus S}$.
Special cases deserve shorter notations:
\begin{align} \begin{aligned}
	\calX(\check k) &:= \calX( \ZZ \setminus \{ k \} ) ,
\\	\calX( < k ) &:= \calX( \{\ldots,k-3,k-2,k-1 \} ) ,
\\	\calX( > k ) &:= \calX( \{ k+1,k+2,k+3,\ldots \} ) ,
\\	\calX(\check k) &= \calX(<k) \oplus \calX(>k) &&\qquad \text{by~\ref{lem:F3}.}
\end{aligned} \end{align}

\begin{lemma}\label{lem:F4}
	If $L \subseteq \ZZ$ is a contiguous interval (possibly infinite),
	then the intersection $\calX(L) \cap \big( \sum_{k \in L} \calX(\check k) \big)$
	is generated by bitstrings, each of which is supported on a proper subset of~$L$.
	That is, $\calX(L) \cap \big( \sum_{k \in L} \calX(\check k) \big) =  \sum_{k \in L} \calX(L) \cap \calX(\check k)$.
\end{lemma}

\begin{proof}
Observe that $\calX(<j) \subseteq \calX(<k)$ whenever $j < k$,
and $\calX( > j ) \subseteq \calX(>k)$ whenever $j > k$.
By definition, any element~$z \in \sum_{k \in L} \calX(\check k)$ 
is a \emph{finite} sum of elements, each of which is an element of~$\calX(\check k')$ for some~$k' \in L$.
Given $z$, let $L'$ be the finite collection of all those~$k'$.
We have
\begin{align}
z \in \sum_{k' \in L'} \calX(\check k') = \sum_{k' \in L'} \calX(< k' ) + \calX( > k') = \calX(< L'_{\max} ) + \calX( > L'_{\min} )
\end{align}
where $L'_{\min}$ is the leftmost site of~$L'$ and $L'_{\max}$ is the rightmost.
Write $z = x + y$ where $x \in \calX( > L'_{\min} )$ and $y \in \calX(< L'_{\max})$.
If $L = \ZZ$, then we are done.
If $L$ is semi-infinite, extended to the left,
then $x$ is solely responsible for the bits of~$z$ on the right of~$L$,
implying that $x \in \calX(L)$ because $L$ is an interval.
Hence, $x$ is supported on a proper, actually finite, subset of~$L$.
Clearly, $y$ is supported on a proper subset of~$L$, and we are done.
If $L$ is semi-infinite, extended to the right,
a symmetric argument applies.
If $L$ is finite, then $x$ accounts for the bits of~$z$ on the right of~$L$,
and $y$ does for those on the left of~$L$.
This means that $x,y \in \calX(L)$ with at least one site missing from their supports, and we are done.
\end{proof}

Using~\ref{lem:F3} and~\ref{lem:F4}, we are going to show

\begin{lemma}\label{lem:extractIsing}
For any interval~$L \subseteq \ZZ$, if $\calX(M) = 0$ for all $M \subsetneq L$, 
then $\calX = \calX(L) \oplus \calY$ 
for some $\calY$ such that $\Pi_k(\calY) \cap \Pi_k(\calX(L)) = 0$ for all~$k \in L$.
Furthermore, $\calX(L)$ is precisely the $X$-logical space of
$n$ independent Ising chains where $n = \dim_{\FF} \calX(L) \le \min_{k \in L} q_k$
and $q_k$ is the number of qudits at site~$k$.
\end{lemma}

After a basis change, two subspaces of~$\FF^{q_k}$ 
(the vector space of all $q_k$-component column vectors)
with zero intersection
can be put to be supported on disjoint sets of components.
Hence, the lemma means that 
\emph{on any interval~$L$ of sites that supports a nonzero element~$x \in \calX$,
if no nonzero element of~$\calX$ is supported on any proper subset of~$L$,
then there exists a basis for~$\calX$ that includes~$x$ 
such that, after some one-site controlled-$X$,
the bitstring $x$ becomes a bitstring on~$L$ that has exactly one nonzero bit at each site of~$L$,
and any other basis element does not overlap with~$x$.}

\begin{proof}
The assumption is that there is no ``shorter'' logical on~$L$.
By~\ref{lem:F4}, this implies that $\calX(L) \cap \sum_{k \in L} \calX(\check k) = 0$.
Therefore, there exists an $\FF$-subspace~$\calY \subset \calX$ such that 
\begin{align}
\calX = \calX(L) \oplus \calY, \qquad \text{and} \qquad \sum_{k \in L} \calX(\check k) \subseteq \calY.
\end{align}
(To find such~$\calY$, 
we can extend a basis of~$\calX(L) \oplus \sum_{k \in L} \calX(\check k)$ to that of~$\calX$.)
Let $k \in L$ be an arbitrary site.
$\Pi_k(x)$ has to be nonzero for all nonzero~$x \in \calX(L)$;
otherwise, \ref{lem:F3} would give us elements of~$\calX(L)$ supported on a proper subset of~$L$,
and those are always zero by assumption.
This means that the projection $\Pi_k : \calX(L) \to \FF^{q_k}$ 
where $q_k$ is the number of qudits at site~$k$,
is injective, so $\Pi_k(\calX(L)) \subseteq \FF^{q_k}$ is isomorphic to~$\calX(L)$
of dimension~$n = \dim_{\FF} \calX(L) \le q_k$.
In particular,
there is a basis change under which 
for every qudit in the support of~$\Pi_k(\calX(L))$
there is a unique basis element of~$\calX(L)$ 
that has the nonzero bit~$1$ on that qudit.
Since~$k$ is arbitrary, 
we see that $\calX(L)$ is precisely the $X$-logical space of independent Ising chains.

It remains to separate the support of~$\calX(L)$ and that of~$\calY$.
Let $g \in \Pi_k(\calY) \cap \Pi_k(\calX(L))$ 
with $g = \Pi_k(y)$ for $y \in \calY$ 
and $g = \Pi_k(x)$ for $x \in \calX(L)$,
so that $\Pi_k(x-y) = 0$.
Then, $x-y \in \calX(\check k) \subseteq \calY$ by construction of~$\calY$,
implying that $x \in \calY$.
But, again by the construction of~$\calY$, we have $x \in \calY \cap \calX(L) = 0$.
Therefore, $\Pi_k(\calY) \cap \Pi_k(\calX(L)) = 0$.
\end{proof}

We conclude the proof of~\ref{thm:IsingIn1D} by~\ref{lem:extractIsing}
as we can always extract a nonzero number of Ising chains at a time,
starting with the shortest ones.
If a reader would like set-theoretic perfection,
then one can consider a collection~$\{ \calX' \}$ 
of all direct summands~$\calX'$ of~$\calX$
such that $\calX = \calX' \oplus \calY'$ for some~$\calY' \subseteq \calX$
where the two summands have disjoint supports after some one-site basis change
and $\calX'$ is the $X$-logical space of some family of independent Ising chains.
This collection is partially ordered by inclusion,
and it is not hard to check that 
any totally ordered chain has an upper bound, the union,
which also belongs to the collection.
Then, Zorn's lemma says that there is a maximal member,
and by~\ref{lem:extractIsing} this maximal member 
cannot leave anything behind from~$\calX$.

\tocless\subsection{Looking at stabilizers}

\newcommand{\cev}[1]{\reflectbox{\ensuremath{\vec{\reflectbox{\ensuremath{#1}}}}}}

It is convenient to regard~$\calB_i$ as an $\FF$-vector space $X_i\cong\FF^{q_i}$ 
(not to be confused with a Pauli~$X$ operator --- we do not need it here),
by forgetting phase factors.
There is no symplectic form here; all relevant groups are abelian.
The subgroups $\calB_i' = \calA(i,i+1)_1 \subseteq \calB_i$ and $\calB_i'' = \calA(i-1,i)_2 \subseteq \calB_i$
give subspaces~$\vec X_i$ and~$\cev X_i$ of~$X_i$, respectively.
The group isomorphism~$\phi : \calS_1 \to \calS_2$ in the proof of~\ref{lem:F2} above,
becomes a $\FF$-linear isomorphism~$\phi_{i \to i+1} : \vec X_i \to \cev X_{i+1}$ for each~$i \in \ZZ$.
This linear isomorphism~$\phi_{i \to i+1}$ is \emph{not} necessarily defined on the whole~$X_i$.
Nor is~$(\phi_{i-1 \to i})^{-1}$ on~$X_i$.

Let us say that a collection of $\FF$-linear spaces $\{ X_i | i \in \ZZ \}$ 
is a {\bf linear space chain} or a chain for short 
if there are subsapces $\cev X_i, \vec X_i \subseteq X_i$ for each~$i$
with linear isomorphisms~$\phi_{i \to i+1} : \vec X_i\to \cev X_{i+1}$.
These subspaces may be zero.
We can take the direct sum of two chains to make another chain:
$\{X_i \} \oplus \{ X'_i\} = \{ X_i \oplus X'_i \}$
with $\phi_{i \to i+1} \oplus \phi'_{ \to i+1} : \vec X_i \oplus \vec X'_i \to \cev X_{i+1} \oplus \cev X'_{i+1}$.
A linear space chain~$\{X_i\}$ is zero if $X_i = 0$ for all~$i \in \ZZ$.

In this proof, we will use intervals notations such as $[a,b)$ 
to denote a set of consecutive sites on the lattice~$\ZZ$.
When we wish to specify boundaries of an interval,
we will write that, for example, 
there is an interval $I = (a-1,b+1)$ with $a \in [-\infty,-1]$ and $b \in [3,\infty]$.
This means that $I$ is unbounded to the left if $a = -\infty$ or
$I$ contains~$a$ as the least element otherwise.
Similarly, $I$ is unbounded to the right if $b = \infty$
or $I$ contains~$b$ as the greatest element otherwise.
Here, by convention, $-\infty - 1 = -\infty$ and $\infty + 1 = \infty$.

An Ising chain on an interval~$I=(a-1,b+1)$ where $a,b \in [-\infty, \infty]$ 
gives a linear space chain: $X_i \cong \FF$ for all~$i \in (a-1,b+1)$ and $X_i = 0$ otherwise; 
$\cev X_i = X_i$ for $i > a$ but $\cev X_a=0$ for finite~$a$;
$\vec X_i = X_i$ for $i < b$ but $\vec X_b=0$ for finite~$b$.
Alternatively, an Ising chain is determined by
a sequence of vectors~$x_i \in X_i$
such that $x_{i+1} = \phi_{i \to i+1}(x_i)$ is nonzero for all~$i \in (a-1,b)$
but $x_i = 0$ for $i \notin (a-1,b+1)$.
If $\sup_i \dim X_i = 1$ for a chain~$\{X_i\}$, 
then the chain corresponds to an Ising chain or a collection of Ising chains supported on disjoint intervals.
We will simply call a sequence of vectors~$\{x_i \in X_i\,|\, i \in I\}$ on an interval~$I = (a-1,b+1)$
an {\bf Ising sequence} if $X_i \ni x_i \mapsto \phi_{i \to i+1}(x_i) = x_{i+1} \in X_{i+1}$ for all $i \in (a-1,b)$,
\emph{regardless} of whether it forms a direct summand of~$\{X_i\}$.
For example, if $\vec X_{-1} = \braket{\binom{0}{1}} = \cev X_0$
and $\vec X_0 = \braket{\binom{1}{0}} = \cev X_1$,
and all other subspaces~$\cev X_i,\vec X_i$ are zero,
then a sequence $\{ \binom{1}{1}  \in X_0 \}$ is a one-element Ising sequence,
but does not form a direct summand;
it cannot even be extended to form a direct summand.

We first describe conditions under which an Ising sequence defines a direct summand.
Denote by $\hat X_i$ the intersection 
\begin{align}
\hat X_i = \cev X_i \cap \vec X_i.
\end{align}

\begin{lemma}\label{lem:DirectSummandCondition}
An Ising sequence~$\{ x_i \}$ on~$(a-1,b+1)$ 
where $a,b \in [-\infty,\infty]$
is a direct summand of a chain~$\{ X_i \}$
if and only if there are subspaces 
$\cev Y_i \subseteq \cev X_i$ and $\vec Y_i \subseteq \vec X_i$ such that
\begin{align}
\begin{cases}
	\braket{x_i} \oplus \cev Y_i = \cev X_i &\text{for } i\in(a,b+1), \\
	\braket{x_i} \oplus \vec Y_i = \vec X_i &\text{for } i\in(a-1,b), \\
	\cev Y_i \cap \hat X_i = \vec Y_i \cap \hat X_i &\text{for } i\in(a,b), \\
	\phi_{i \to i+1}: \vec Y_i \cong \cev Y_{i+1} &\text{for } i \in (a-1, b),\\
	\hat X_{a} \subseteq \vec Y_{a} &\text{if } a > -\infty, \\
	\hat X_b \subseteq \cev Y_b &\text{if } b < \infty.
\end{cases}\label{eq:sixlines}
\end{align}
\end{lemma}

\begin{proof}
The ``only if'' direction is routine to check;
if $\{X_i\} = \braket{x_i} \oplus \{ Y_i \}$,
then $\cev Y_i$ and $\vec Y_i$ satisfy all those properties.

To show the ``if'' direction, 
we have to construct a direct complement of~$\{\braket{x_i}\}$.
We assume that $X_i = \cev X_i + \vec X_i$ for all~$i \in \ZZ$;
otherwise any direct complement of~$\cev X_i + \vec X_i$ 
within $X_i$ is a direct summand automatically.
Define $\{Z_i = \cev Z_i + \vec Z_i \}$ by
\begin{align} \begin{aligned}
	(\cev Z_i, \vec Z_i) &= (\cev Y_i, \vec Y_i) &&&&&&\text{ for } i \in (a,b),
\\	(\cev Z_a, \vec Z_a) &= (\cev X_a, \vec Y_a) &&&&&&\text{ if } a > -\infty,
\\	(\cev Z_b, \vec Z_b) &= (\cev Y_b, \vec X_b) &&&&&&\text{ if } b < \infty,
\\	(\cev Z_i, \vec Z_i) &= (\cev X_i, \vec X_i) &&&&&&\text{ for } i \notin (a-1,b+1).
\end{aligned} \end{align}

First, we show that $X_i = \braket{x_i} \oplus Z_i$ for all~$i \in \ZZ$.
If $i \in \ZZ \setminus (a-1,b+1)$, then $\braket{x_i} = 0$ and $X_i = Z_i$, 
so $X_i = \braket{x_i} \oplus Z_i$.
If $a \in \ZZ$, then $\braket{x_a} + Z_a = \braket{x_a} + \cev X_a + \vec Y_a $,
which equals $\cev X_a + \vec X_a = X_a$ by the second line of~\eqref{eq:sixlines}.
If $x_a \in Z_a = \cev X_a + \vec Y_a$, then $x_a = \cev x_a + \vec y_a$
for some $\cev x_a \in \cev X_a$ and $\vec y_a \in \vec Y_a$.
It follows that $x_a - \vec y_a = \cev x_a \in \vec X_a \cap \cev X_a = \hat X_a \subseteq \vec Y_a$ by the fifth line of~\eqref{eq:sixlines}.
Then, $0 \neq x_a \in \vec Y_a$, but this is contradictory to the second line of~\eqref{eq:sixlines}.
Hence, $x_a \notin Z_a$ and $\braket{x_a} \cap Z_a = 0$,
so $\braket{x_a} \oplus Z_a = X_a$ if $a > -\infty$.
A symmetric argument shows that $\braket{x_b} \oplus Z_b = X_b$ if $b < \infty$,
using the first and sixth lines of~\eqref{eq:sixlines}.
If $i \in (a,b)$,
then $\braket{x_i} + Z_i = \braket{x_i} + \cev Y_i + \vec Y_i = \braket{x_i} + \cev Y_i + \braket{x_i}+\vec Y_i = \cev X_i + \vec X_i = X_i$
by the first and second lines of~\eqref{eq:sixlines},
so $\braket{x_i} + Z_i = X_i$.
Suppose $x_i = \cev y_i + \vec y_i \in Z_i$ for some $\cev y_i \in \cev Y_i$ and $\vec y_i \in \vec Y_i$.
Since $x_i \in \cev X_i$, 
we have $x_i - \cev y_i = \vec y_i \in \cev X_i \cap \vec Y_i \subseteq \hat X \cap \vec Y_i = \hat X_i \cap \cev Y_i$
where the last equality is by the third line of~\eqref{eq:sixlines},
so $\vec y_i \in \cev Y_i$.
But, then, $0 \neq x_i = \cev y_i + \vec y_i \in \cev Y_i \cap \braket x = 0$.
Hence, $\braket{x_i} \cap Z_i = \braket{x_i} \cap (\cev Y_i + \vec Y_i) = 0$ and $\braket{x_i} \oplus Z_i = X_i$.

Second, we show that the isomorphisms of the parent chain~$\{X_i\}$ restrict to isomorphisms along~$\{Z_i\}$.
The fourth line of~\eqref{eq:sixlines} takes care of the ``interior'' of the interval:
for $i \in (a-1,b)$ 
we have $\phi_{i \to i+1}|_{\vec Z_i} : \vec Z_i = \vec Y_i \xrightarrow{\quad\cong\quad} \cev Y_{i+1} = \cev Z_{i+1}$.
The ``exterior'' of the interval remains intact:
for $i \in \ZZ \setminus (a-1,b)$
we have $\phi_{i \to i+1} |_{\vec Z_i} : \vec Z_i = \vec X_i \xrightarrow{\quad\cong\quad} \cev X_{i+1} = \cev Z_{i+1}$.
\end{proof}

Note that for any linear spaces $E,F,G$ with $G \subseteq F$,
it is easy to show that 
\begin{align}
(G+E)\cap F = G + (E \cap F). \label{eq:D-E-F}
\end{align}

\begin{lemma}\label{lem:InfiniteFullIsingSequences}
For any linear space chain~$\{X_i = \cev X_i + \vec X_i \}$ with $\cev X_0 \neq 0$,
any infinite Ising sequence~$\{ x_i \in X_i \,|\, i \in (a-1,b+1) \}$ forms a direct summand
where $x_a \in \vec X_a \setminus \hat X_a$ if $a > -\infty$
or $x_b \in \cev X_a \setminus \hat X_b$ if $b < \infty$.
\end{lemma}

\begin{proof}
For all $i \in \ZZ$, we take any subspaces $\cev W_i, \vec W_i$ 
such that $\cev X_i = \hat X_i \oplus \cev W_i$,
and $\vec X_i = \hat X_i \oplus \vec W_i$.

($a > -\infty, b = \infty$)
We have $x_a \in \vec X_a \setminus \hat X_a$.
Take any subspace $\vec V_a$ such that $\vec W_a = \braket{x_a} \oplus \vec V_a$,
and define $\vec Y_a = \hat X_a \oplus \vec V_a$.
For $k > a$ we inductively define
\begin{align} \begin{aligned} \label{eq:YtoRight}
	\cev Y_{k} &= \phi_{k-1 \to k}(\vec Y_{k-1}) ,
\\	\hat W_{k} &= \cev Y_{k} \cap \vec X_{k} \,,
\\	\vec Y_{k} &= \hat W_{k} \oplus \vec W_{k} \,.
\end{aligned} \end{align}
We have to check the first three lines of~\eqref{eq:sixlines};
all the rest is either vacuous or obvious by construction.
The third line is immediate: $\vec Y_k \cap \hat X_k = \hat W_k \cap \hat X_k = \cev Y_k \cap \hat X_k$.
The first two lines are straightforward by induction in~$k$. 
The base case is $\braket{x_a} \oplus \vec Y_a = \vec X_a$,
which is clear by construction.
Then for $k > a$ we have $\braket{x_k} \oplus \cev Y_k = \phi_{k-1 \to k}( \braket{x_{k-1}} \oplus \vec Y_{k-1} ) = \cev X_k$.
In turn, using~\eqref{eq:D-E-F} we have
$\braket{x_k} \oplus \vec Y_k 
= \braket{x_k} \oplus (\cev Y_k \cap \vec X_k) \oplus \vec W_k 
= ((\braket{x_k} \oplus \cev Y_k) \cap \vec X_k) \oplus \vec W_k 
= (\cev X_k \cap \vec X_k) \oplus \vec W_k
= \vec X_k$.
This completes the induction step,
proving the lemma if $a > - \infty$.

($a = -\infty, b < \infty$)
A left-right symmetric argument proves the lemma $b < \infty$.
Use the inverses $\phi_{i \ot i+1}$ of $\phi_{i \to i+1}$.

($a = -\infty, b = \infty$)
In this case, $x_i \in \hat X_i$ for all~$i \in \ZZ$.
Take any subspace $\hat W_0$ such that $\hat X_0 = \braket{x_0} \oplus \hat W_0$.
Define $\cev Y_0 = \cev W_0 \oplus \hat W_0$ and $\vec Y_0 = \vec W_0 \oplus \hat W_0$.
Then, $\braket{x_0} \oplus \cev Y_0 = \cev X_0$,  $\braket{x_0} \oplus \vec Y_0 = \vec X_0$,
and $\cev Y_0 \cap \hat X_0 = \hat W_0 = \vec Y_0 \cap \hat X_0$.
We proceed to the right as in the $a>-\infty$ case
and to the left as in the $b< \infty$ case.
\end{proof}

We are going to find a finite Ising sequence that forms a direct summand,
and the argument will use induction in the length of the sequence.
To that end, the following separation of~\eqref{eq:sixlines} will be useful,
which will make conditions on desired subspaces~$\cev Y_i,\vec Y_i$
look similar to those on a one-site shorter interval.
For any~$i \in \ZZ$, we denote by $\phi_{i \ot i+1}$ the inverse of $\phi_{i \to i+1}$.

\begin{lemma}\label{lem:TwoGroupsOfConditions}
Let $\{x_i \,|\, i \in [a,b] \}$ be an Ising sequence a finite interval $[a,b]$.
Define
\begin{align}
\hat{\hat X}_{b-1} &= \hat X_{b-1} \cap \phi_{b-1 \ot b}(\hat X_{b}), &
\hat{\hat X}_{a+1} &= \hat X_{a+1} \cap \phi_{a \to a+1}(\hat X_{a}).
\label{eq:hathat}
\end{align}
Then, for subspaces~$\vec Y_i, \cev Y_i$, \eqref{eq:DecomposedAtRight} $\iff$ \eqref{eq:sixlines} $\iff$ \eqref{eq:DecomposedAtLeft}:
\begin{align}
(i)\begin{cases}
		\braket{x_i}\oplus \cev Y_i=\cev X_i &\text{for } i \in [a+1, b-1],\\
		\braket{x_i}\oplus \vec Y_i=\vec X_i &\text{for } i \in [a, b-2],\\
		\cev Y_i \cap \hat X_i = \vec Y_i\cap \hat X_i &\text{for }i \in [a+1,b-2],\\
		\phi_{i \to i+1}: \vec Y_i \cong \cev Y_{i+1} &\text{for } i\in [a,b-2],\\
		\hat X_{a}\subseteq \vec Y_{a},\\
		\hat{\hat X}_{b-1}\subseteq \cev Y_{b-1},
\end{cases}
~ 	
(ii)\begin{cases}
	\braket{x_b} \oplus \cev Y_b = \cev X_b,\\
	\braket{x_{b-1}} \oplus \vec Y_{b-1} = \vec X_{b-1},\\
	\cev Y_{b-1}\cap \hat X_{b-1}=\vec Y_{b-1}\cap \hat X_{b-1},\\
	\phi_{b-1 \to b}: \vec Y_{b-1} \cong \cev Y_b,\\
	\exists T \subseteq \vec Y_{b-1}~:~ \phi_{b-1 \ot b}(\hat X_b)=\hat{\hat X}_{b-1}\oplus T.
\end{cases}
	\label{eq:DecomposedAtRight}\\
(i)\begin{cases}
		\braket{x_i}\oplus \cev Y_i=\cev X_i &\text{for } i \in [a+2, b],\\
		\braket{x_i}\oplus \vec Y_i=\vec X_i &\text{for } i \in [a+1, b-1],\\
		\cev Y_i \cap \hat X_i = \vec Y_i\cap \hat X_i &\text{for }i \in [a+2,b-1],\\
		\phi_{i \to i+1}: \vec Y_i \cong \cev Y_{i+1} &\text{for } i\in [a+1,b-1],\\
		\hat X_{b}\subseteq \cev Y_{b},\\
		\hat{\hat X}_{a+1}\subseteq \vec Y_{a+1},
\end{cases}
~ 	
(ii)\begin{cases}
	\braket{x_{a+1}} \oplus \cev Y_{a+1} = \cev X_{a+1},\\
	\braket{x_{a}} \oplus \vec Y_{a} = \vec X_{a},\\
	\cev Y_{a+1}\cap \hat X_{a+1} = \vec Y_{a+1}\cap \hat X_{a+1},\\
	\phi_{a \to a+1}: \vec Y_{a} \cong \cev Y_{a+1},\\
	\exists S \subseteq \cev Y_{a+1}~:~ \phi_{a \to a+1}(\hat X_a)=\hat{\hat X}_{a+1}\oplus S.
\end{cases}
	\label{eq:DecomposedAtLeft}
\end{align}
\end{lemma}

Observe the similarity of the group of conditions to those of~\eqref{eq:sixlines}.
The boundary site $b$ or $a$ is replaced by~$b-1$ or~$a+1$.
The boundary condition $\hat X_b \subseteq \cev Y_b$ of~\eqref{eq:sixlines}
is replaced by a similar-looking condition~$\hat{\hat X}_{b-1} \subseteq \cev Y_{b-1}$ in~\eqref{eq:DecomposedAtRight}.
The boundary condition $\hat X_a \subseteq \vec Y_a$ of~\eqref{eq:sixlines}
is replaced by a similar-looking condition~$\hat{\hat X}_{a+1} \subseteq \vec Y_{a+1}$ in~\eqref{eq:DecomposedAtLeft}.

\begin{proof}
We only prove that \eqref{eq:sixlines} $\iff$ \eqref{eq:DecomposedAtRight};
a symmetric argument will prove that \eqref{eq:sixlines} $\iff$ \eqref{eq:DecomposedAtLeft}.

Suppose that there are subspaces~$\vec Y_{b-1} \subseteq \vec X_{b-1}$ 
and~$\cev Y_{b} \subseteq \cev X_b$ such that
$\phi_{b-1 \to b}: \vec Y_{b-1} \cong \cev Y_b$.
This is the fourth line of~\eqref{eq:sixlines} applied to the rightmost site.
Then,
the boundary condition~$\hat X_b \subseteq \cev Y_b$ in the sixth line of~\eqref{eq:sixlines}
is equivalent to $\phi_{b-1 \ot b}(\hat X_b) \subseteq \vec Y_{b-1}$.
The latter condition implies that there exists $T \subseteq \vec Y_{b-1}$ 
such that $\phi_{b-1 \ot b}(\hat X_b) = \hat{\hat X}_{b-1} \oplus T$.
Clearly, $\hat{\hat X}_{b-1} \subseteq \vec Y_{b-1}$.
If we further assume that $\cev Y_{b-1} \cap \hat X_{b-1} = \vec Y_{b-1} \cap \hat X_{b-1}$ 
(the third line of~\eqref{eq:sixlines}),
then $\hat{\hat X}_{b-1} \subseteq \cev Y_{b-1}$.
Conversely, if we have $\phi_{b-1 \ot b}(\hat X_b) = \hat{\hat X}_{b-1} \oplus T$ for some $T \subseteq \vec Y_{b-1}$
and $\hat{\hat X}_{b-1} \subseteq \cev Y_{b-1}$,
then $\hat{\hat X}_{b-1} \subseteq \vec Y_{b-1}$
so $\phi_{b-1 \ot b}(\hat X_b) \subseteq \vec Y_{b-1}$.
Summarizing, if we assume that $\phi_{b-1 \to b}: \vec Y_{b-1} \cong \cev Y_b$
and $\cev Y_{b-1} \cap \hat X_{b-1} = \vec Y_{b-1} \cap \hat X_{b-1}$,
then
\begin{equation}
	\hat X_b \subseteq \cev Y_b
	\iff 
	\phi_{b-1 \ot b}(\hat X_b) \subseteq \vec Y_{b-1}
	\iff 
	\begin{cases}
		\hat{\hat X}_{b-1}\subseteq \cev Y_{b-1},\\
		\exists T \subseteq \vec Y_{b-1} ~:~ \phi_{b-1 \ot b}(\hat X_b)=\hat{\hat X}_{b-1} \oplus T .
	\end{cases} \label{eq:bdryCon}
\end{equation}
Now, the first four conditions of~\eqref{eq:DecomposedAtRight}(i)
are obviously those of~\eqref{eq:sixlines} except for the right end,
and the first four of~\eqref{eq:DecomposedAtRight}(ii) cover the right end.
The left boundary condition~$\hat X_a \subseteq \vec Y_a$ of~\eqref{eq:sixlines} 
appears in~\eqref{eq:DecomposedAtRight}(i).
The combination of the sixth line of~\eqref{eq:DecomposedAtRight}(i)
and the fifth line of~\eqref{eq:DecomposedAtRight}(ii)
is equivalent to the right boundary condition~$\hat X_b \subseteq \cev Y_b$ by~\eqref{eq:bdryCon}.
Therefore, \eqref{eq:sixlines} $\iff$ \eqref{eq:DecomposedAtRight}.
\end{proof}

To enable the induction below, 
we define, for any chain~$\{X_i\}$, two new chains~$\{C_i = \cev C_i + \vec C_i\}$ and~$\{D_i = \cev D_i + \vec D_i\}$ 
by the following table.
The two special subspaces $\hat{\hat X}_a$ and $\hat{\hat X}_b$ are defined in~\eqref{eq:hathat}.
\begin{align}
\begin{array}{c|c|c|c|c|c|c|c}
\cev C_i 	& 0				& 0									& \hat{\hat X}_{a+1}	& \multicolumn{4}{c}{\cev X_i} \\
\vec C_i 	& 0 				& \phi_{a \ot a+1}(\hat{\hat X}_{a+1})	& \vec X_{a+1} 		& \multicolumn{4}{c}{\vec X_i}\\
\hline
i 			& (-\infty,a-1]	&  a									& a+1 				& [a+2,b-2] 	&  b-1				& b 									& [b+1,\infty)\\
\hline
\cev D_i 	& \multicolumn{4}{c|}{\cev X_i} 		& \cev X_{b-1}		& \phi_{b-1 \to b}(\hat{\hat X}_{b-1}) &  0\\
\vec D_i 	& \multicolumn{4}{c|}{\vec X_i} 		& \hat{\hat X}_{b-1} & 0									& 0
\end{array}
\label{eq:CiDi}
\end{align}
The isomorphisms~$\phi^C,\phi^D$ of the new chains 
are inherited from the original chain~$\{X_i\}$ in an obvious manner.

\begin{lemma}\label{lem:induction}
Let $\{x_i \neq 0 \,|\, i \in [a,b] \}$ be a finite Ising sequence of~$\{X_i\}$.
If $\{ x_i \,|\, i \in [a,b-1]\}$ forms a direct summand of~$\{D_i\}$, 
then $\{x_i \,|\, i \in [a,b] \}$ forms a direct summand of~$\{X_i\}$.
By symmetry,
if $\{ x_i \,|\, i \in [a+1,b]\}$ forms a direct summand of~$\{C_i\}$, 
then $\{x_i \,|\, i \in [a,b] \}$ forms a direct summand of~$\{X_i\}$.
\end{lemma}

\begin{proof}
We only prove the first claim; the second follows by the left-right symmetry.

The assumption
gives subspaces~$\vec Y_i$ of~$\vec D_i = \vec X_i$ for $i \in [a,b-2]$ 
and~$\cev Y_i$ of~$\cev D_i = \cev X_i$ for $i \in [a+1,b-1]$,
satisfying~\eqref{eq:DecomposedAtRight}(i).
We have to find~$\vec Y_{b-1} = \phi_{b-1 \ot b}(\cev Y_b)$ that satisfy~\eqref{eq:DecomposedAtRight}(ii).
To this end, we take a subspace~$T$ such that
\begin{align}
	\phi_{b-1 \ot b}(\hat X_b) = \hat{\hat X}_{b-1} \oplus T
\end{align}
from the definition of~$\hat{\hat X}_{b-1}$ in~\eqref{eq:hathat}.
Since $\hat{\hat X}_{b-1}\subseteq \cev Y_{b-1}$ implies, obviously,
$\hat{\hat X}_{b-1}=(\cev Y_{b-1}\cap \hat X_{b-1}) \cap \phi_{b-1 \ot b}(\hat X_b)$,
we see that $(\cev Y_{b-1}\cap \hat X_{b-1}) \cap T = 0$,
so 
\begin{align}
	(\cev Y_{b-1}\cap \hat X_{b-1}) \oplus T \subseteq \vec X_{b-1}.
\end{align}
It follows that
\begin{align}
\big( (\cev Y_{b-1} \cap \hat X_{b-1}) \oplus T  \big) \cap \hat X_{b-1} 
&\subseteq
\big( (\cev Y_{b-1} \cap \hat X_{b-1}) + \phi_{b-1 \ot b}(\hat X_b)  \big) \cap \hat X_{b-1} \label{eq:YXTX}\\
&= (\cev Y_{b-1}\cap \hat X_{b-1}) + \hat{\hat X}_{b-1} & \text{due to }\eqref{eq:D-E-F}\nonumber\\
&= \cev Y_{b-1}\cap \hat X_{b-1} & (\hat{\hat X}_{b-1} \subseteq \hat X_{b-1}) .\nonumber
\end{align}
If $x_{b-1} \in (\cev Y_{b-1} \cap \hat X_{b-1}) \oplus T$,
then, since $x_{b-1} \in \hat X_{b-1}$ as well,
we must have $x_{b-1} \in \cev Y_{b-1} \cap \hat X_{b-1}$,
but the first line of~\eqref{eq:DecomposedAtRight}(i) 
says $\cev X_{b-1} = \braket{x_{b-1}} \oplus \cev Y_{b-1}$,
implying $x_{b-1} \notin \cev Y_{b-1}$, a contradiction.
Therefore,
\begin{align}
	\braket{x_{b-1}} \oplus  (\cev Y_{b-1} \cap \hat X_{b-1}) \oplus T  \subseteq \vec X_{b-1}.
\end{align}
Hence, there exists a subspace $V$ such that
\begin{align}
 \braket{x_{b-1}} \oplus \underbrace{ (\cev Y_{b-1} \cap \hat X_{b-1}) \oplus T  \oplus V}_{\vec Y_{b-1}} = \vec X_{b-1},
\end{align}
where we have defined~$\vec Y_{b-1}$.
Of~\eqref{eq:DecomposedAtRight}(ii) the first, second, fourth, and fifth lines are now obvious.
To show the third line of~\eqref{eq:DecomposedAtRight}(ii),
we claim that $V \cap \hat X_{b-1} = 0$.
Suppose $v \in V \cap \hat X_{b-1}$. 
From $\cev X_{b-1} = \braket{x_{b-1}} \oplus \cev Y_{b-1}$
we have $v = f x_{b-1} + \cev y_{b-1}$ for some~$f \in \FF$ and~$\cev y_{b-1} \in \cev Y_{b-1}$.
Then, $\cev y_{b-1} = v - f x_{b-1} \in \hat X_{b-1}$,
so $\cev y_{b-1} \in \cev Y_{b-1} \cap \hat X_{b-1}$.
But $0 = f x_{b-1} + \cev y_{b-1} - v \in \braket{x_{b-1}} \oplus (\cev Y_{b-1} \cap \hat X_{b-1}) \oplus V$
must be a unique expression, implying that $v = 0$.
Therefore,
\begin{align}
\vec Y_{b-1} \cap \hat X_{b-1} 
&= \big( (\cev Y_{b-1} \cap \hat X_{b-1}) \oplus T  \oplus V \big) \cap \hat X_{b-1} \\
&= (\cev Y_{b-1} \cap \hat X_{b-1}) + (V \cap \hat X_{b-1}) & \text{due to }\eqref{eq:D-E-F}, \eqref{eq:YXTX}\nonumber\\
&= \cev Y_{b-1} \cap \hat X_{b-1} .\nonumber
\end{align}
The proof is complete by~\ref{lem:DirectSummandCondition}.
\end{proof}

Next, we construct an Ising sequence that is maximal in a certain sense.
For any $i < j$ we define $\phi_{i \to j}$ recursively as follows:
\begin{equation}
	\begin{aligned}
		&\dom(\phi_{i \to i})=X_i,~~\phi_{i\to i}=\text{id}|_{X_i},\\
		&\dom(\phi_{i \to j})=(\phi_{i \to i+1})^{-1}(\cev X_{i+1}\cap\dom(\phi_{i+1 \to j})),~~\phi_{i \to j}=\phi_{i+1 \to j}\circ\phi_{i \to i+1}.
	\end{aligned}
\end{equation}
In other words, if $i<j$, the domain of $\phi_{i \to j}$ consists of all $x_i\in X_i$ 
such that a sequence $x_{k+1} = \phi_{k,k+1}(x_k) \in \vec X_{k+1}$ for $k \in [i,j)$ is defined.
The map~$\phi_{i \to j}$ sends a subspace of~$\vec X_i$ isomorphically onto a subspace of~$\cev X_j$.
For $i \leq j$, we denote by $\phi_{i \ot j}$ the inverse of~$\phi_{i \to j}$.
The site index on the left of an arrow will always be equal to or less than that on the right of the arrow.

Suppose we are given a chain~$\{X_i\}$ such that $\cev X_0\neq 0$.
Consider $\dom(\phi_{j \ot 0}) \subseteq \cev X_0 \subseteq X_0$ with $j<0$.
It is clear that its dimension is nonincreasing as $j \searrow -\infty$. 
Therefore, the dimensions must stabilize to a limit
\[
\lim_{j \searrow -\infty}\dim\dom(\phi_{j \ot 0})
\]
after finitely many steps.
If this limit is zero, then there exists a finite $a \in (-\infty,-1]$ 
such that $\im(\phi_{a+1 \ot 0})\cap \hat X_{a+1} \neq 0$ but $\im(\phi_{a \ot 0}) \cap \hat X_{a}=0$.
Then we pick a vector $x_0 \in \dom(\phi_{a \ot 0})$
and apply $\phi$'s leftwards starting from~$x_0$;
by definition of~$\phi_{a \ot 0}$, we obtain $x_a, x_{a+1}, \ldots, x_{-1}, x_0$.
We apply $\phi$'s rightwards until at some $b \in [0,\infty]$ we must stop: 
either $\phi_{0 \to n}(x_0) = x_n$ exists for all $n > 0$, or $b$ is finite so
$x_{b-1} \in \vec X_{b-1}$ but $x_b \in \cev X_b \setminus \hat X_b$.
If the limit is positive ($\lim_{j\searrow -\infty} \dim \dom(\phi_{0,j}) > 0$), in which case $a = -\infty$, 
then we simply pick any $x_0 \in \bigcap_{j<0}\dom(\phi_{0,j})$,
and define a sequence of vectors~$x_i$ for $i\in(-\infty, b+1)$ that are related by $\phi$'s.
We say that the sequence~$\{x_i \,|\, i\in (a-1,b+1) \}$ is {\bf full}.
Although the left boundary~$a$ is determined by the chain~$\{X_i\}$ with $\cev X_i \neq 0$,
the right boundary~$b$ is determined by both the chain and the choice of~$x_0 \in \dom(\phi_{a \ot 0})$.
A full Ising sequence~$\{ x_i \in X_i \,|\, i \in (a-1,b+1) \}$, too, depends on~$x_0$.

\begin{lemma}\label{lem:decomposition}
For any linear space chain~$\{X_i = \cev X_i + \vec X_i \}$ with $\cev X_0 \neq 0$,
any finite full Ising sequence forms a direct summand.
\end{lemma}

\begin{proof}
We are going to use the characterization in~\ref{lem:DirectSummandCondition} and~\ref{lem:TwoGroupsOfConditions}
of a direct summand in terms of subspaces~$\cev Y_i, \vec Y_i \subseteq X_i$.
Let $I = [a,b] = (a-1,b+1)$ be the finite interval over which a finite full Ising sequence is defined,
where $a < 0 \le b$.
We are going to prove the lemma by induction first in~$b \ge 0$ with~$a = -1$ and then in~$a < 0$.

(Induction~i) The base case is where $a = -1$ and $b = 0$, 
so the full Ising sequence consists of two elements~$x_{-1}$ and $x_0$.
By construction, $a=-1$ implies that $\im(\phi_{-1 \ot 0}) \cap \hat X_{-1}=0$.
Since $\hat X_{-1} \subseteq \im( \phi_{-1 \ot 0} ) = \vec X_{-1}\cong \cev X_0 = \dom(\phi_{-1 \ot 0} )$,
we find that $\hat X_{-1}=0$.
Therefore, the condition~\eqref{eq:sixlines} becomes
\begin{align} \begin{aligned}
	\braket{x_0} \oplus \cev Y_0 &= \cev X_0 \,, &&&&& \braket{x_{-1}} \oplus \vec Y_{-1} &= \vec X_{-1} \,,
\\	\phi_{-1\to 0}: \vec Y_{-1} &\cong \cev Y_{0} \,, &&&&& \hat X_0 &\subseteq \cev Y_0 \,.
\end{aligned} \end{align}
Since $b=0$, we must have $x_0 \notin \vec X_0$.
So, $\braket{x_0} \oplus \hat X_0 \subseteq \cev X_0$.
Let $\cev Z_0$ be any direct complement of $\braket{x_0} \oplus \hat X_0$ within $\cev X_0$,
so $\cev X_0 = \braket{x_0} \oplus \hat X_0 \oplus \cev Z_0$.
Define $\cev Y_0 = \hat X_0 \oplus \cev Z_0$ and $\vec Y_{-1} = \phi_{-1 \ot 0}(\cev Y_0)$.
These satisfy all the four conditions.

(Induction~ii) The induction hypothesis is that the lemma is true with any $I = [-1,b']$ where $0 \le b' < b$.
We are going to show the lemma with $I = [-1,b]$.

We claim that $\{d_i = x_i \in D_i \,|\, i \in [a,b-1] \}$ is a full Ising sequence of the new chain~$\{D_i\}$
defined in~\eqref{eq:CiDi}.
Towards the left of the origin,
$\dom(\phi^D_{j \ot 0}) = \dom(\phi_{j \ot 0})$ for all~$j<0$ 
since they are determined by $\phi_{i-1 \ot i} = \phi^D_{i-1 \ot i}$ with $i \le 0$.
Hence, we may construct a full Ising sequence with $d_0 = x_0 \in \cev D_0 = \cev X_0$ in the new chain,
and the sequence elements~$d_j$ to the left $(a \le j \le 0)$ and to the right ($0 < j < b$)
are identical to those of the original sequence.
In the original chain~$\{X_i\}$, by definition of full Ising sequences,
we have $x_b \notin \hat X_b$,
which implies $x_{b-1} \notin \hat{\hat X}_{b-1}$.
Hence, we cannot apply $\phi^D_{b-1 \to b}$ to~$d_{b-1}$.
Therefore, $\{d_i = x_i \,|\, i \in [a,b-1] \}$ is a full Ising sequence of the new chain.

Now, the induction hypothesis applied to the new chain~$\{D_i\}$ with the full Ising sequence $\{x_i\,|\,i\in[a,b-1]\}$
shows that it forms a direct summand of~$\{D_i\}$.
By~\ref{lem:induction}, we complete this induction step.

(Induction~iii) The induction hypothesis is that the lemma is true with $I = [a',b]$ 
where $b \in [0,\infty)$ is arbitrary but $a < a' < 0$.
We are going to show the lemma with $I = [a,b]$.%
\footnote{
	The argument here is mostly the same as in~(Induction~ii),
	but there is a difference 
	as our definition of a full Ising sequence is not left-right symmetric.
}

We claim that $\{c_i = x_i \in C_i \,|\, i \in [a+1,b] \}$ is a full Ising sequence of the new chain~$\{C_i\}$
defined in~\eqref{eq:CiDi}.
First, we have to show that $\phi^C_{a+1 \ot 0} \neq 0$ but $\phi^C_{a \ot 0} = 0$.
From~\eqref{eq:CiDi} it is clear that $\phi^C_{a+1 \ot 0} = \phi_{a+1 \ot 0} \neq 0$.
Since $\{x_i \in X_i \,|\, i \in [a,b]\}$ is full,
we know $\im(\phi_{a \ot 0}) \cap \hat X_a = 0$,
implying $\phi_{a\to a+1}(\im(\phi_{a \ot 0})) \cap \phi_{a\to a+1}(\hat X_a) = 0$.
By definition, $\phi_{a \to a+1}(\im(\phi_{a \ot 0})) = \phi_{a \to a+1}(\hat X_a) \cap \im(\phi_{a+1 \ot 0})$.
So, $\im(\phi_{a+1 \ot 0}) \cap \phi_{a\to a+1}(\hat X_a) = 0$,
implying that $\im(\phi_{a+1 \ot 0}) \cap \hat{\hat X}_a = 0$,
or equivalently $\phi^C_{a \ot 0} = 0$.
Second, we have to show that the sequence $x_{a+1},x_{a+2},\ldots,x_0,\ldots$ extends up to~$x_b$ but no more,
but this is obvious since the new chain is identical to the original on the right of the origin.
Therefore, $\{x_i\,|\, i \in [a+1,b]\}$ is a full Ising sequence of~$\{C_i\}$.

The induction hypothesis applied to the new chain~$\{C_i\}$ with the full Ising sequence $\{x_i\,|\,i\in[a+1,b]\}$
shows that it forms a direct summand of~$\{C_i\}$.
By~\ref{lem:induction}, we complete this induction step, and hence the lemma.
\end{proof}

\begin{lemma}\label{lemma-lookatstab}
Any linear space chain~$\{X_i\}$ with $\dim X_k > 1$ for some~$k$ is a direct sum of two nonzero subchains.
\end{lemma}

\begin{proof}
We may assume without loss of generality that $\cev X_i + \vec X_i=X_i$ for all~$i\in\ZZ$; 
if $\cev X_i + \vec X_i \neq X_i$ for some $i$, 
then any direct complement of $\cev X_i + \vec X_i$ within~$X_i$ 
gives a nonzero direct summand, an Ising chain on single site $i$.
Then, since $X_k = \cev X_k + \vec X_k \neq 0$,
either $\vec X_{k-1} \cong \cev X_k \neq 0$ or $\vec X_k \cong \cev X_{k+1} \neq 0$.
By shifting the origin, we may assume that $\cev X_0 \neq 0$.
Then, there is a nonzero full Ising sequence.
If it is finite, we have~\ref{lem:decomposition},
or if it is infinte, we have~\ref{lem:InfiniteFullIsingSequences},
to find a nontrivial direct summand. 
\end{proof}

Applying~\ref{lemma-lookatstab} recursively, by an application of Zorn's lemma,
we complete the proof of the theorem~\ref{thm:IsingIn1D}.

\section{More on the WPT model}\label{app:WPT}

\tocless\subsection{Horizontal (bottom) boundary.}
\label{sec:WPT_hbot}

The same idea used in the construction for the horizontal top boundary applies here as well.
Let
\begin{align}
	\ISG_{B(y_b)} &= \Big\langle \big\{ P_0(x,y-1) \,, Z_{x,y-1/2} ~\big|~ x\in\ZZ,\, y-y_b\in\ZZ_+ \big\} \cup \big\{ Z_{x,y} ~\big|~ (x,y)\in\Lambda,\, y<y_b \big\} \Big\rangle \,.
\end{align}
The blending circuit consists of 5 steps.
First transition $\ISG_{B(0)} \to \ISG_{B(1)}$ by measuring $\{Z_{x,0} \,|\, x\in\ZZ\}$, then apply the bulk translation circuit for $y>0$ to restore the ISG back to $\ISG_{B(0)}$.
\begin{subequations} \begin{align}
	\ISG^{(B)}_0 &= \ISG_{B(0)} \,,
\\	\ISG^{(B)}_1 &= \ISG_{B(1)} \,,
\\	\ISG^{(B)}_2 &= \Big\langle \big\{ P_1(x,y) \,, X_{x,y-1/2} X_{x,y} ~\big|~ x\in\ZZ,\, y\in\ZZ_+ \big\} \cup \big\{ Z_{x,y} ~\big|~ (x,y)\in\Lambda,\, y\leq0 \big\} \Big\rangle \,,
\\	\ISG^{(B)}_3 &= \ISG_{B(1/2)} \,,
\\	\ISG^{(B)}_4 &= \Big\langle \big\{ P_3(x,y) \,, X_{x,y-1} X_{x,y-1/2} ~\big|~ x\in\ZZ,\, y\in\ZZ_+ \big\} \cup \big\{ Z_{x,y} ~\big|~ (x,y)\in\Lambda,\, y<0 \big\} \Big\rangle \,.
\end{align} \end{subequations}
The local logical operators of $\ISG^{(B)}_0$ are generated by $L_j = X_{j,0} Z_{j+1,0}$.
After a cycle, these operators transforms as $L_j \mapsto L_{j-1}$.

Alternately, we can measure $X$ along the $y=0$ row on the first step, and the logical operators shifts in the opposite direction.
\begin{align} \begin{aligned}
	\ISG^{(B')}_1 &= \Big\langle \big\{ P_0(x,y) \,, Z_{x,y-1/2} ~\big|~ x\in\ZZ,\, y\in\ZZ_+ \big\} \\&\qquad \cup \big\{ X_{(x,0)} ~\big|~ x\in\ZZ \big\} \cup \big\{ Z_{x,y} ~\big|~ (x,y)\in\Lambda,\, y<0 \big\} \Big\rangle \,.
\end{aligned} \end{align}

\tocless\subsection{Reversing boundary flow--vertical boundary revisited}
\label{sec:WPT_vright_rev}

In \S\ref{sec:WPT_vright} we constructed a blend between the WPT ($x \leq 0$) and trivial ($x > 0$) circuits with the boundary logicals moving downwards along the same direction as the bulk circuit.
Here we construct an alternate blend (between the pair of bulk circuits in the same geometry), but with the boundary logicals moving \emph{upwards} instead.

The construction is based on the observation that the blends Eq.~\eqref{eq:WPT_top_blendl} and~\eqref{eq:WPT_top_blendr} have opposite index.
The composition of the circuit \eqref{eq:WPT_top_blendr} with the inverse of \eqref{eq:WPT_top_blendl} is a 1D LR cycle which translates the logicals by two units: $L_j \mapsto L_{j+2}$.
The simplest realization of this is a 3-step circuit involving logicals $L_j$ obeying the algebra~\eqref{eq:anomalous_1D_L}, and a chain of ancila qubits with operators $(Z_k,X_k)$ for $k\in\ZZ$ is:
\begin{align}
\label{eq:anomalous_1D_L_shift2}
\xymatrix @M=2mm {
	\Braket{ Z_k \,|\, k\in\ZZ } \ar[r]^(0.45){1}
&	\Braket{ L_k Z_{k-1} X_k \,|\, k\in\ZZ } \ar[r]^(0.62){2}
&	\Braket{ X_k \,|\, k\in\ZZ } \ar@/^1.5pc/[ll]^(0.76){3}
}
\end{align}
At step~0, the base stabilizer group generated by $Z_k$ has logical generated by $L_k$.
Between adjacent pairs of ISGs, their shared logicals are, for $k\in\ZZ$: 
$L_k Z_{k-1} Z_{k+1}$ (steps 0 \& 1),
$L_{k+2} X_{k} X_{k+2}$ (steps 1 \& 2),
and $L_k$ (steps 2 \& 0).
Under a cycle the logicals transforms $L_k \mapsto L_k Z_{k-1} Z_{k+1} \mapsto L_{k+2} X_{k} X_{k+2} \mapsto L_{k+2}$ from step~0 to step~3.

The vertical blend is a 7-step cycle combining \eqref{eq:WPT_right_blendd} with \eqref{eq:anomalous_1D_L_shift2}:
\begin{subequations} \begin{align}
	\ISG^{(R')}_{7n+t} &= \ISG^{(R)}_{4n+t} \qquad t \in \{0,1,2,3,4\} ,
\\\begin{split}
	\ISG^{(R')}_{7n+5} &= \Big\langle \big\{ P_0(x,y) \,, Z_{x+2,y-1/2} ~\big|~ x\in\ZZ_{-},\, y\in\ZZ \big\}
	\\&\qquad	\cup \big\{ Z_{0,y} X_{0,y+1} Z_{1,y} X_{1,y+1} ~\big|~ y\in\ZZ \big\} \cup \big\{ Z_{x,y} ~\big|~ (x,y)\in\Lambda,\, x\geq2 \big\} \Big\rangle \,,
\end{split}\\\begin{split}
	\ISG^{(R')}_{7n+6} &= \Big\langle P_0(x,y) \,, Z_{x+2,y-1/2} ~\big|~ x\in\ZZ_{-},\, y\in\ZZ \big\}
	\\&\qquad	\cup \big\{ X_{1,y} ~\big|~ y\in\ZZ \big\} \cup \big\{ Z_{x,y} ~\big|~ (x,y)\in\Lambda,\, x\geq2 \big\} \Big\rangle \,.
\end{split}
\end{align} \end{subequations}
This blend has the same index as that of \eqref{eq:WPT_right_blendd} stacked with a pure 1D translation circuit (moving upwards), but without altering the boundary logical algebra.

\bigskip
\section{More on the HH code}\label{app:HH}

\newcommand{\localrev}[3]{$\renewcommand{\arraystretch}{1.1} \begin{array}{@{} r @{\,=\,} l @{}} \mathcal{A} & \big\langle #1 \big\rangle_{\!+t} \\ \calS & \big\langle #2 \big\rangle_{\!+t} \\ \mathcal{B} & \big\langle #3 \big\rangle_{\!+t} \end{array}\mkern-4mu$}
\newcommand{\LOpArr}[1][0ex]{\multirow{2}{*}{\raisebox{#1}{$\downarrow$}}}

\paragraph*{String operators.}
We first briefly discuss logical string operators of the bulk HH code.
(As explained in \S\ref{sec:lrtrans}, these objects are technically not operators, but \emph{locally finite products}.)
We label strings operators by a color label and Pauli label:
	for $C \in \{R,G,B\}$ and $p \in \{X,Y,Z\}$,
	$\mathrm{Str}_C^p$ denote a product of Pauli $p$ operators along $C$-edges which connect a path of $C$-plaquettes.
Recall that a $C$-edge connects between two $C$-plaquettes.
We say two $C$-plaquettes are adjacent if they are connected by a $C$-edge.
For example, on the honeycomb lattice, each $C$-plaquette has exactly 6 adjacent $C$-plaquettes.
A path of $C$-plaquettes is a (infinite or finite periodic) sequence of adjacent $C$-plaquettes.
Note that $\mathrm{Str}_C^p$ (when restricted to any finite region) is generated by $E_C^p$,
	and commutes with $E_{C'}^p$ and $E_C^{p'}$ for any $C'$, $p'$.

With this notation, the logical string operators for each ISG are generated by
\begin{align} \label{eq:HH_stringOp} \begin{aligned}
	\ISG_\stepR &: & \mathrm{Str}_{R}^Y &\sim \mathrm{Str}_{R}^Z \,, & \mathrm{Str}_{G}^X \sim \mathrm{Str}_B^X \,,
\\	\ISG_\stepG &: & \mathrm{Str}_{G}^X &\sim \mathrm{Str}_{G}^Z \,, & \mathrm{Str}_{R}^Y \sim \mathrm{Str}_B^Y \,,
\\	\ISG_\stepB &: & \mathrm{Str}_{B}^X &\sim \mathrm{Str}_{B}^Y \,, & \mathrm{Str}_{R}^Z \sim \mathrm{Str}_G^Z \,.
\end{aligned} \end{align}
We have written a $\sim$ symbol to indicate when two logical representatives are related by a locally finite product built from the ISG.
As the model is locally reversible, each pair of steps has a common pair of string operators.

\tocless\subsection{Blending with WPT model}

We can construct a topological blending between the HH code with that of the WPT (cf.~\S\ref{sec:WPT}), 
that is, with no local logical operators.
First, on one half of the plane, put the HH code with the boundary described in this section, such that the boundary logicals move leftwards per cycle;
	on the other half of the plane, put one of the boundaries of the WPT model that also has a left-moving boundary.
To temporally synchronize the two models, an extra idle step should be added to the HH code such that it has four steps per cycle (e.g.\ two consecutive $\stepB$ steps).
To lattice match the two models, every unit cell of the WPT model should correspond to 3 sites along the boundary shown in Fig.~\ref{fig:HHboundary_zigzag}.
In this set up, the base logical algebra is generated by two copies of the \MCAlg:
	$L^{\stepR\prime}_j$ and $L^{W\!P}_j$ (for $j\in\ZZ$) from the HH code and WPT respectively.
The lattice matching condition would ensure that $L^{\stepR\prime}_j$ and $L^{W\!P}_j$ are spatially near each other.
We glue together these two boundaries in the same spirit as \S\ref{sec:WPT_double} (albeit the two bulks are different).
Let $L^{\stepR\prime}_j(t)$ and $L^{W\!P}_j(t)$ be the evolution of the the logicals at step $t$; we expect $L^{\stepR\prime/{W\!P}}_j(t+4) = L^{\stepR\prime/{W\!P}}_{j-1}(t)$ by constuction.
Next, construct a new glued circuit with the ISGs from the initial blending with added stabilizers $L^{\stepR\prime}_j(t) \, L^{W\!P}_j(t)$ for all $j\in\ZZ$.
This is the desired blending with no local logical operators,
proving that the WPT and HH code belong to the same class under topological blending equivalence.

\bigskip
\tocless\subsection{Planar implementation with chiral edge}
\label{sec:HHplanar}

Here we consider the HH code on the (finite) plane.
If the periodicity of the global circuit is to be maintained at 3 steps per cycle,
	the MQCA index must be nonzero and be constant along the boundary the code.
That is, any period-3 planar realization must have a chiral dynamical flow of information along the edge.

\begin{figure}[htb]
	\centering
	\begin{minipage}{54mm} Step $\stepR$ \\ \includegraphics[width=48mm]{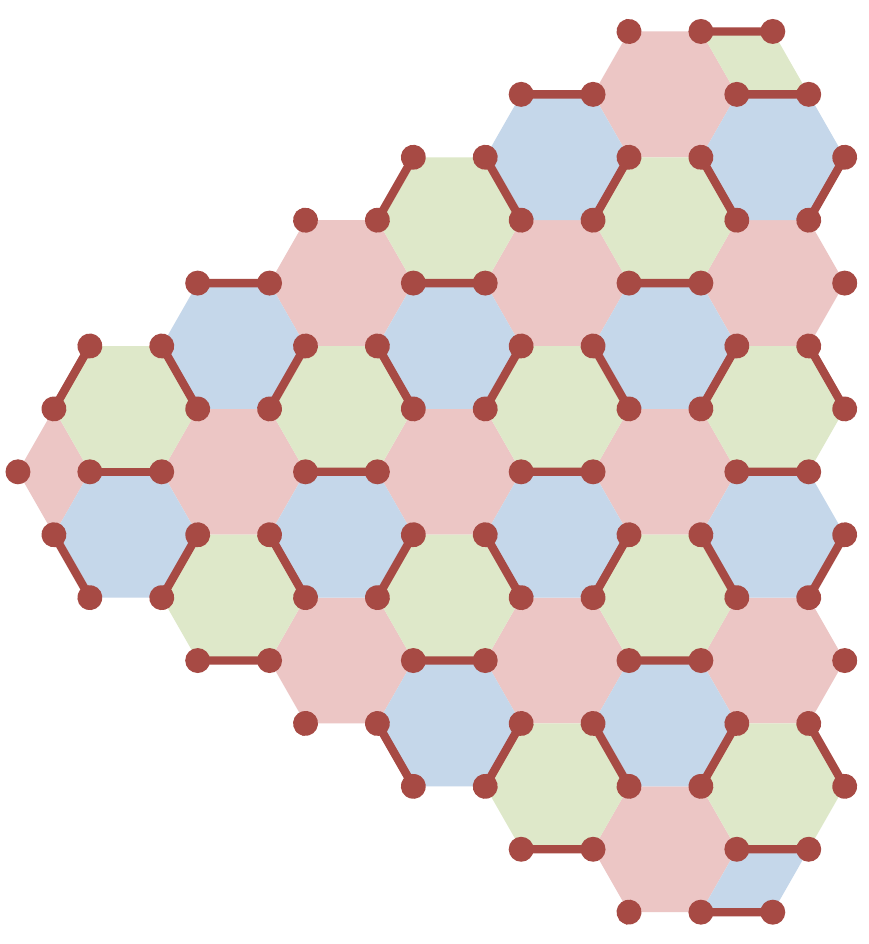} \end{minipage}
	\begin{minipage}{54mm} Step $\stepG$ \\ \includegraphics[width=48mm]{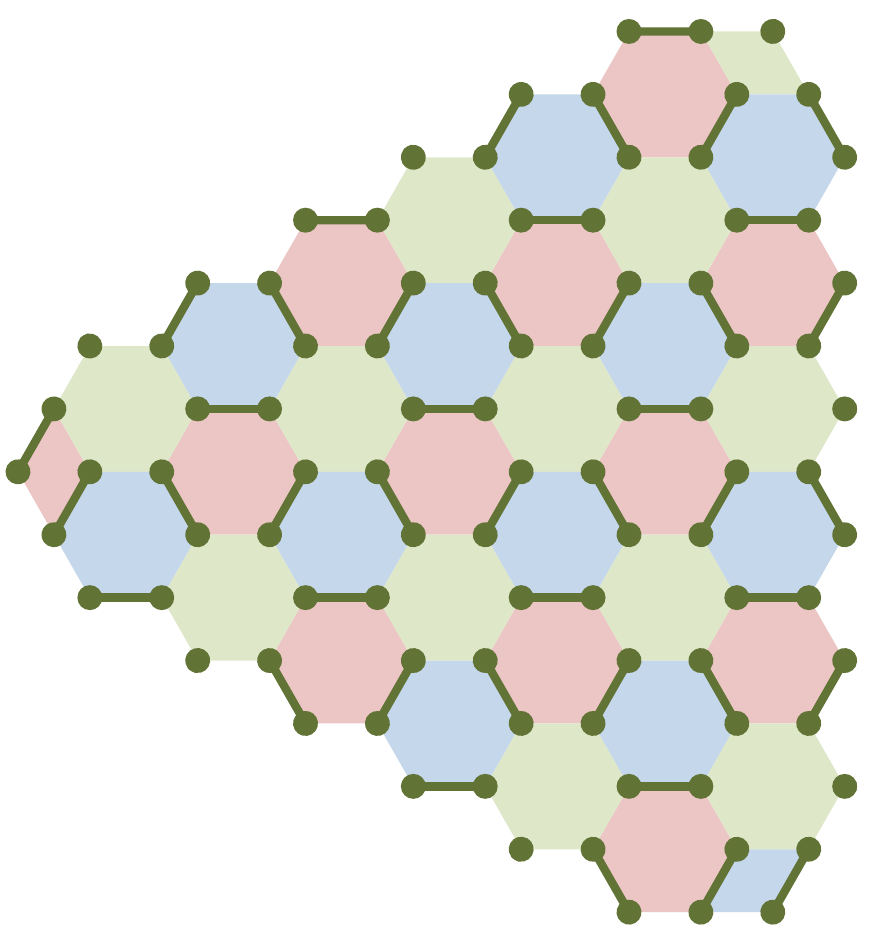} \end{minipage}
	\begin{minipage}{54mm} Step $\stepB$ \\ \includegraphics[width=48mm]{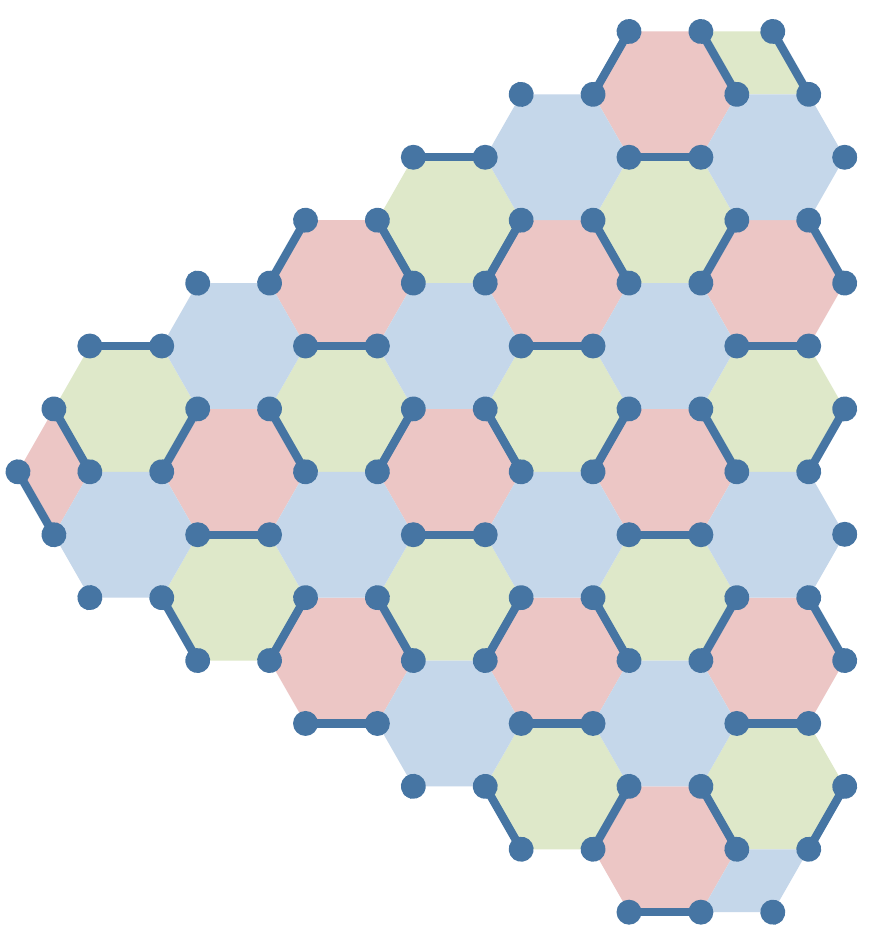} \end{minipage}
	\caption{%
	A planar configuration of the HH code with chiral boundary dynamics.
	The $\stepR$, $\stepG$, $\stepB$ steps involve measurements on single and pairs of sites, with Pauli operators $X$, $Y$, $Z$ respectively.
	There are an extensive number (proportional to the perimeter) of boundary logical operators, obeying the \MCAlg{} for a periodic chain.
	The boundary logical operators evolve in a counterclockwise direction.
	} \label{fig:HHtriangle}
\end{figure}

Figure~\ref{fig:HHtriangle} illustrates a possible realization with a triangle geometry.
Each side of the triangle is a zigzag truncated boundary described in \S\ref{sec:HHboundary1}, the plaquettes at the corners are replaced with rhombi.
The logical algebra is generated by 3- and 4-body operators of the form~\eqref{eq:HHboundary1_logicals} along the boundary.
For a triangle of size $N$ (shown in the figure is $N=7$), the boundary logical operators obey the \MCAlgEq{} for a periodic chain with $2N$ sites.
(The $2N$ operators are not independent; as the products $\prod_\text{even $i$} L_i$ and $\prod_\text{odd $i$} L_i$ belong to the ISG.)
Under a period, the boundary operators shifts by one in the counterclockwise direction.

\begin{figure}[ht]
	\centering
	\includegraphics[width=50mm]{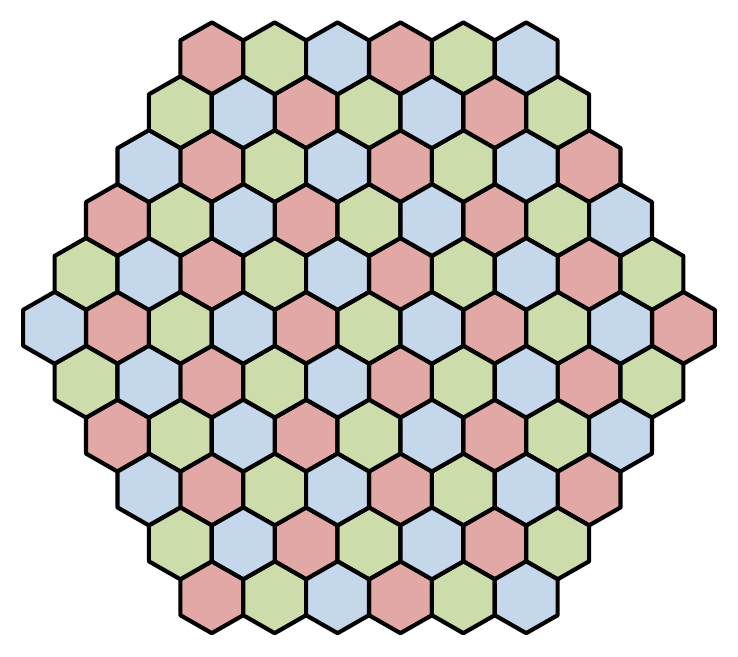}
	\caption{A planar configuration of the HH code with the hexagon geometry.
	} \label{fig:HHhexagon}
\end{figure}
The planar hexagon geometry in Fig.~\ref{fig:HHhexagon},
which is not to be confused with the individual plaquettes,
requires a more nuance construction.
Unlike the triangle geometry, the edges are not all isomorphic,
	the plaquette labeling of red, green, blue runs clockwise along three of six edges, and counterclockwise along the other three.
If we were to attempt implementing the circuit from \S\ref{sec:HHboundary1} along each boundary,
	half of the edges have counterclockwise-propagating dynamics while the other half have clockwise-propagating dynamics;
	neighboring edges of the hexagon have opposing direction for information flow.
Since logical operators cannot accumulate at any vertex,
	it means that ``Majorana Bell pairs'' will be created at certain corners,
	travel in opposite directions,
	and measured out of the system at different corners.
To maintain a fixed chirality along the entirety of the planar hexagon boundary,
	we need to implement boundary circuits alternating between ``foward-moving'' (e.g.~\S\nameref{app:HHzz:for})
	and ``backward-moving'' (e.g.~\S\ref{sec:HHboundary1}) along alternating edges.

\tocless\subsection{Various boundary circuit of the zigzag edge}

We demonstrate the local reversibility for various measurements circuits at the zigzag boundary of the HH code.
We will also track the transformation of logical operators under circuit dynamics.
Throughout, we use the following notation.
\begin{itemize}[itemsep=1pt, parsep=0pt, topsep=2pt]
\item	$\bar{n}$ means $-n$ (for integers $n$).
\item	Sites along the topmost row are numbered, with the `0' at the tip of a red hexagon.
\item	Pauli operators $X_n, Y_n, Z_n$ acts on site $n$ at the top boundary.
\item	Certain boundary circuit utilizes ancilla qubits.
	The ancilla will be (visually) placed above one of the boundary sites.
	Pauli operators $X^n, Y^n, Z^n$ (with superscript $n$) acts on ancilla atop of site $n$.
\item	If a plaquette/polygon is colored (one of red, green, or blue),
	then it means that the ISG contains the plaquette operator which is a product of Paulis on every plaquette/polygon vertex
	(corresponding to $\prod_v X_v$, $\prod_v Y_v$, or $\prod_v Z_v$ respectively).
	\begin{itemize}[itemsep=1pt, parsep=0pt, topsep=2pt]
	\item	It is possible for a colored plaquette to be stablized by additional plaquette operators.  (E.g., a red-colored plaquette may also be stablized by $\prod_v Z_v$.)
	\item	It is also possible for an uncolored plaquette/polygon to be stablized by some plaquette operator; we may not color every possible polygon to avoid visual overload.
	\end{itemize}
\item	Under the \textbf{reversibility} column:
	Each block $(\calA, \calS, \calB)$ describes the pair $(\ISG_t, \ISG_{t+1})$,
	with $\calS = \ISG_t \cap \ISG_{t+1}$,
	and conjugate bases $\calA, \calB$.
\item	The column under \textbf{logical operators} tracks the evolution of logical operators.
	Entries between $\ISG_t$ and $\ISG_{t+1}$ belongs to the shared logical space.
\item	$\calS_0 = \Braket{ P_R^X, P_G^Y, P_B^Z }$.  They are common stabilizers for every step.
\item	$E_{R'}^X$, $E_{G'}^Y$, $E_{B'}^Z$ consist of (the two-third) nonvertically oriented edges that are adjacent to a hexagon plaquette.
	They include operators on the boundary, e.g.~$Y_0Y_1 \in E_{G'}^Y$.
	They will never include single-site Paulis or operators involving ancilla.
\item	Here $\braket{ \{P\} }_{\!+t}$ means the group generated by the set $P$ and its horizontal translates.
	E.g., $\braket{Z_4}_{\!+t} = \braket{Z_{6j+4} ~|~ j\in\ZZ}$.
\end{itemize}

\tocless\subsubsection{Truncated boundary I.}
\insertlabel{app:HHzz:rev}{\thesubsection{}.\arabic{subsubsection}}
This is the period-3 circuit introduced in~\S\ref{sec:HHboundary1}.
There are 2 independent logical generators per unit cell, and the MQCA index is $-\half$.
\begin{center} \renewcommand{\arraystretch}{1.9}
	\begin{tabular}{r @{\;\;} c @{\quad} l r@{\hspace{8pt}} c @{\hspace{7pt}} c }
		& \HHzzSiteNumbering & reversibility && \multicolumn{2}{c}{logical operators}
	\\	\multirow{2}{*}{$\ISG_0$} & \multirow{2}{*}{$\vcenter{\hbox{\includegraphics[width=70mm]{HHzz_isg_R12.pdf}}}$}
		&&& \ooalign{\hfil\LOpArr[-2ex]\hfil \cr $Z_{\bar2}Z_{\bar1}Y_1Y_2$} & \ooalign{\hfil\LOpArr[-2ex]\hfil \cr $X_2X_3X_4$}
	\\\cline{3-3}	&& \multirow{2}{*}{\localrev{ E_{R'}^X, X_0 }{ \calS_0, P_B^{X/Y} }{ E_{G'}^Y, Y_2 }}
		&& \multirow{2}{*}{$Y_{\bar2}Y_{\bar1}Y_1Y_2$} & \multirow{2}{*}{$X_0X_1X_3X_4$}
	\\	\multirow{2}{*}{$\ISG_{1}$} & \multirow{2}{*}{$\vcenter{\hbox{\includegraphics[width=70mm]{HHzz_rev3_isg1.pdf}}}$} &&& \LOpArr&\LOpArr
	\\\cline{3-3}	&& \multirow{2}{*}{\localrev{ E_{G'}^Y, Y_2 }{ \calS_0, P_R^{Y/Z} }{ E_{B'}^Z, Z_{\bar2} }}
		&& \multirow{2}{*}{$Y_{\bar4}Y_{\bar3}Y_{\bar1}Y_0$} & \multirow{2}{*}{$Z_0Z_1Z_3Z_4$}
	\\	\multirow{2}{*}{$\ISG_{2}$} & \multirow{2}{*}{$\vcenter{\hbox{\includegraphics[width=70mm]{HHzz_rev3_isg2.pdf}}}$} &&& \LOpArr&\LOpArr
	\\\cline{3-3}	&& \multirow{2}{*}{\localrev{ E_{B'}^Z, Z_{\bar2} }{ \calS_0, P_G^{Z/X} }{ E_{R'}^X, X_0 }}
		&& \multirow{2}{*}{$X_{\bar4}X_{\bar3}X_{\bar1}X_0$} & \multirow{2}{*}{$Z_{\bar2}Z_{\bar1}Z_1Z_2$}
	\\	\multirow{2}{*}{$\ISG_{0}$} & \multirow{2}{*}{$\vcenter{\hbox{\includegraphics[width=70mm]{HHzz_isg_R12.pdf}}}$} &&& \LOpArr[2ex]&\LOpArr[2ex]
	\\\cline{3-3}	&&&& $X_{\bar4}X_{\bar3}X_{\bar2}$ & $Z_{\bar2}Z_{\bar1}Y_1Y_2$
	\end{tabular}
\end{center}
The boundary logical quotient algebra is that of the \MCAlg.
With the assignment of the logicals of $\ISG_\stepR$ to fermionic operators via
	$Z_{6n-2}Z_{6n-1}Y_{6n+1}Y_{6n+2} \cong \ii \gamma_{2n-1} \gamma_{2n}$,
	and $X_{6n+2}X_{6n+3}X_{6n+4} \cong \ii \gamma_{2n} \gamma_{2n+1}$,
	this MQCA is equivalent to the fermionic transformation $\gamma_k \mapsto \gamma_{k-1}$.
(Since every fermion is translated uniformly, this circuit has the special feature where logical operators remain bounded in size under repeated application of the MQCA.)

\newpage
\tocless\subsubsection{Truncated boundary II.}
\insertlabel{app:HHzz:truncII}{\thesubsection{}.\arabic{subsubsection}}
This is the period-3 circuit introduced in~\S\ref{sec:HHboundary2}.
There are 4 independent logical generators per unit cell, and the MQCA index is $+\half$.
\begin{center} \renewcommand{\arraystretch}{1.9}
	\begin{tabular}{r @{\;\;} c @{\quad} l r@{\hspace{-3.1pt}} c }
		& \HHzzSiteNumbering & reversibility && \multicolumn{1}{c}{logical operators}
	\\	\multirow{2}{*}{$\ISG_0$} & \multirow{2}{*}{$\vcenter{\hbox{\includegraphics[width=70mm]{HHzz_c3t2_isg0.pdf}}}$}
		&&& \ooalign{\hfil\LOpArr[-2ex]\hfil \cr $X_0$}
	\\\cline{3-3}	&& \multirow{2}{*}{\localrev{ E_{R'}^X }{ \calS_0, P_B^{X/Y} }{ E_{G'}^Y }}
		&& \multirow{2}{*}{$X_0X_1X_2$}
	\\	\multirow{2}{*}{$\ISG_{1}$} & \multirow{2}{*}{$\vcenter{\hbox{\includegraphics[width=70mm]{HHzz_c3t2_isg1.pdf}}}$} &&& \LOpArr
	\\\cline{3-3}	&& \multirow{2}{*}{\localrev{ E_{G'}^Y }{ \calS_0, P_R^{Y/Z} }{ E_{B'}^Z }}
		&& \multirow{2}{*}{$Z_0Z_1X_2Y_3Y_4$}
	\\	\multirow{2}{*}{$\ISG_{2}$} & \multirow{2}{*}{$\vcenter{\hbox{\includegraphics[width=70mm]{HHzz_c3t2_isg2.pdf}}}$} &&& \LOpArr
	\\\cline{3-3}	&& \multirow{2}{*}{\localrev{ E_{B'}^Z }{ \calS_0, P_G^{Z/X} }{ E_{R'}^X }}
		&& \multirow{2}{*}{$Z_0Z_1Y_2X_3Y_4Z_5Z_6$}
	\\	\multirow{2}{*}{$\ISG_{0}$} & \multirow{2}{*}{$\vcenter{\hbox{\includegraphics[width=70mm]{HHzz_c3t2_isg0.pdf}}}$} &&& \LOpArr[2ex]
	\\\cline{3-3}	&&&& \footnotesize $(Z_0Z_1Z_2)(X_2X_3X_4)(Y_4Y_5Y_6)X_6$
	\end{tabular}
\end{center}
There are 3 additional logical generators (per unit cell) that are shared between all the ISGs:
\begin{align}
    \big\{ X_{6j+2} X_{6j+3} X_{6j+4} \;,\, Y_{6j-2} Y_{6j-1} Y_{6j} \;,\, Z_{6j} Z_{6j+1} Z_{6j+2} \;\big|\; x\in\ZZ \big\} .
\end{align}
This model (including the boundary) admits an $\ell$-local conjugate bases with $\ell = 1$.
Starting with $L(0) = X_0 \in \ISG_\stepR$,
	its operator size grows linearly as $\operatorname{size} L(t) = \operatorname{size} L(0) + 2t$.
Here, we denote by $\operatorname{size} L(t)$  the minimum size among all equivalent representatives.
This model saturates the bound established in~\ref{thm:LP}.

The boundary logical quotient algebra is also that of the \MCAlg.
However, it is more convenient to make comparison to two Majorana chains, generated by $\beta_k$ and $\gamma_k$ as follows:
	$Z_{6n+0} Z_{6n+1} Z_{6n+2} \cong \ii \gamma_{3n+0} \gamma_{3n+1}$,
	$X_{6n+2} X_{6n+3} X_{6n+4} \cong \ii \gamma_{3n+1} \gamma_{3n+2}$,
	$Y_{6n+4} Y_{6n+5} Y_{6n+6} \cong \ii \gamma_{3n+2} \gamma_{3n+3}$,
	and $X_{6n} \cong \ii \gamma_{3n} \beta_{n}$.
This MQCA is equivalent to the fermionic transformation $\beta_n \mapsto \beta_{n+1}$, $\gamma_n \mapsto \gamma_n$.

\newpage
\tocless\subsubsection{Forward-moving boundary.}
\insertlabel{app:HHzz:for}{\thesubsection{}.\arabic{subsubsection}}

We present a period-3 circuit with logical operators translating uniformly to the right.
This circuit has two logical generators per unit cell, and implements an MQCA with index $+\half$.
The base stabilizer group is the same as that of~\S\ref{sec:HHboundary1}.
Hence, the composition of this circuit with that of \S\ref{sec:HHboundary1} 
gives a period-6 LR cycle whose boundary MQCA is the identity map.

The measurement sequence is
\begin{align} \begin{aligned}
	M_\stepR &= E_R^X \cup \{ X_{6j-2}X_{6j-1} \,, X_{6j+1}X_{6j+2} \,, X_{6j+0} ~|~ j\in\ZZ \} \,,
\\	M_\stepG &= E_G^Y \cup \{ Y_{6j+0}Y_{6j+1} \,, Y_{6j+3}Y_{6j+4} \,, Z_{6j+0}Z_{6j+1}Z_{6j+2} ~|~ j\in\ZZ \} \,,
\\	M_\stepB &= E_B^Z \cup \{ Z_{6j-4}Z_{6j-3} \,, Z_{6j-1}Z_{6j+0} \,, Y_{6j-2}Y_{6j-1}Y_{6j+0} ~|~ j\in\ZZ \} \,.
\end{aligned} \end{align}
Observe that step~$\stepG$ consists of both Pauli~$Y$ and Pauli~$Z$ measurements, as does step~$\stepB$.

This circuit indeed implments an LR cycle, the ISGS are illustrated as follows.
\begin{center} \renewcommand{\arraystretch}{1.9}
	\begin{tabular}{r @{\;\,} c @{\;\;} l r@{\hspace{-1pt}} c @{\;\;} c }
		& \HHzzSiteNumbering & reversibility && \multicolumn{2}{c}{logical operators}
	\\	\multirow{2}{*}{$\ISG_0$} & \multirow{2}{*}{$\vcenter{\hbox{\includegraphics[width=70mm]{HHzz_isg_R12.pdf}}}$}
		&&& \ooalign{\hfil\LOpArr[-2ex]\hfil \cr $Z_{\bar2}Z_{\bar1}Y_1Y_2$} & \ooalign{\hfil\LOpArr[-2ex]\hfil \cr $X_2X_3X_4$}
	\\\cline{3-3}	&& \multirow{2}{*}{\localrev{ E_{R'}^X, X_0 }{ \calS_0, P_B^{X/Y} }{ E_{G'}^Y, Z_0Z_1Z_2 }}
		&& \multirow{2}{*}{$Y_{\bar2}Y_{\bar1}Y_1Y_2$} & \multirow{2}{*}{$X_0X_1X_3X_4$}
	\\	\multirow{2}{*}{$\ISG_{1}$} & \multirow{2}{*}{$\vcenter{\hbox{\includegraphics[width=70mm]{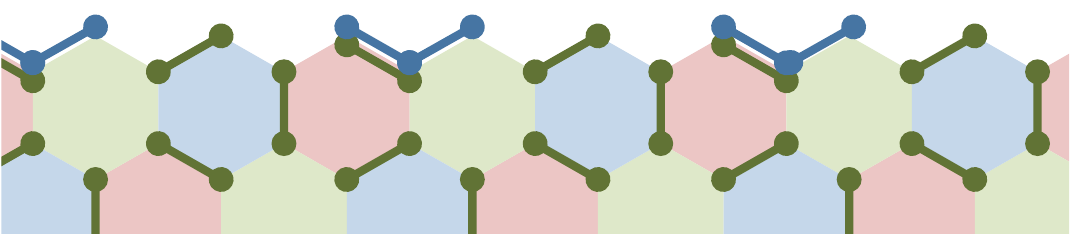}}}$} &&& \LOpArr&\LOpArr
	\\\cline{3-3}	&& \multirow{2}{*}{\localrev{ E_{G'}^Y, Z_0Z_1Z_2 }{ \calS_0, P_R^{Y/Z} }{ E_{B'}^Z, Y_{\bar2}Y_{\bar1}Y_0 }}
		&& \multirow{2}{*}{\parbox{8.8ex}{$Y_{\bar2}Y_{\bar1}Y_0$\\$\times Y_2Y_3Y_4$}} & \multirow{2}{*}{\parbox{9.5ex}{$Z_2Z_3Z_4$\\$\times Z_6Z_7Z_8$}}
	\\	\multirow{2}{*}{$\ISG_{2}$} & \multirow{2}{*}{$\vcenter{\hbox{\includegraphics[width=70mm]{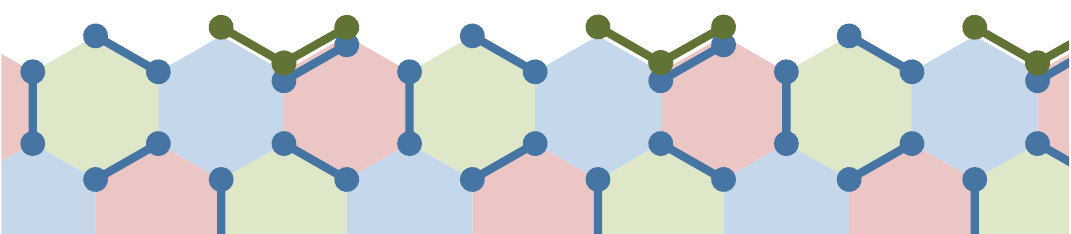}}}$} &&& \LOpArr&\LOpArr
	\\\cline{3-3}	&& \multirow{2}{*}{\localrev{ E_{B'}^Z, Y_{\bar2}Y_{\bar1}Y_0 }{ \calS_0, P_G^{Z/X} }{ E_{R'}^X, X_0 }}
		&& \multirow{2}{*}{$X_2X_3X_5X_6$} & \multirow{2}{*}{$Z_4Z_5Z_7Z_8$}
	\\	\multirow{2}{*}{$\ISG_{0}$} & \multirow{2}{*}{$\vcenter{\hbox{\includegraphics[width=70mm]{HHzz_isg_R12.pdf}}}$} &&& \LOpArr[2ex]&\LOpArr[2ex]
	\\\cline{3-3}	&&&& $X_2X_3X_4$ & $Z_4Z_5Y_7Y_8$
	\end{tabular}
\end{center}
With $\ISG_0$ as the base stabilizer of this circuit,
this circuit implements the inverse MQCA as that of \S\nameref{app:HHzz:rev},
although the intermediate circuits steps are not reversed.

\newpage
\tocless\subsubsection{Gapped boundary with three ancilla.}
\insertlabel{app:HHzz:3anc}{\thesubsection{}.\arabic{subsubsection}}

This is a period-6 circuit from~\S\ref{sec:HHzz_gapped}.
We place an ancilla qubit over every even site--3 ancilla per unit cell.
There are no nontrivial local logical operators.
\begin{center} \renewcommand{\arraystretch}{2.2}
	\begin{tabular}{r @{\;\;} c @{\quad} l }
		& \HHzzSiteNumbering & reversibility
	\\	\multirow{2}{*}{$\ISG_{0}$} & \multirow{2}{*}{$\vcenter{\hbox{\includegraphics[width=70mm]{HHzz_c6t_isg0.pdf}}}$}  &
	\\\cline{3-3}	&& \multirow{2}{*}{\localrev{ E_{R'}^X, X^0, X^2, X^4 }{ \calS_0, P_B^{X/Y}, X_0X_1X_2X^2, X^2X_2X_3X_4 }{ E_{G'}^Y, Y^0, Y_2Y^2, Y^4 }}
	\\	\multirow{2}{*}{$\ISG_{1}$} & \multirow{2}{*}{$\vcenter{\hbox{\includegraphics[width=70mm]{HHzz_c6t_isg1.pdf}}}$}
	\\\cline{3-3}	&& \multirow{2}{*}{\localrev{ E_{G'}^Y, Y^0, Y_2Y^2, Y^4 }{ \calS_0, P_R^{Y/Z}, Z_0Z_1Z_2Z^2, Z^2Z_2Z_3Z_4 }{ E_{B'}^Z, Z^0, Z^2, Z^4 }}
	\\	\multirow{2}{*}{$\ISG_{2}$} & \multirow{2}{*}{$\vcenter{\hbox{\includegraphics[width=70mm]{HHzz_c6t_isg2.pdf}}}$}
	\\\cline{3-3}	&& \multirow{2}{*}{\localrev{ E_{B'}^Z, Z^{\bar2}, Z^0, Z^2 }{ \calS_0, P_G^{Z/X}, Z_{\bar2}Z_{\bar1}Z_0Z^0, Z^0Z_0Z_1Z_2 }{ E_{R'}^X, X^{\bar2}, X_0X^0, X^2 }}
	\\	\multirow{2}{*}{$\ISG_{3}$} & \multirow{2}{*}{$\vcenter{\hbox{\includegraphics[width=70mm]{HHzz_c6t_isg3.pdf}}}$}
	\\\cline{3-3}	&& \multirow{2}{*}{\localrev{ E_{R'}^X, X^{\bar2}, X_0X^0, X^2 }{ \calS_0, P_B^{X/Y}, Y_{\bar2}Y_{\bar1}Y_0Y^0, Y^0Y_0Y_1Y_2 }{ E_{G'}^Y, Y^{\bar2}, Y^0, Y^2 }}
	\\	\multirow{2}{*}{$\ISG_{4}$} & \multirow{2}{*}{$\vcenter{\hbox{\includegraphics[width=70mm]{HHzz_c6t_isg4.pdf}}}$}
	\\\cline{3-3}	&& \multirow{2}{*}{\localrev{ E_{G'}^Y, Y^{\bar4}, Y^{\bar2}, Y^0 }{ \calS_0, P_R^{Y/Z}, Y_{\bar4}Y_{\bar3}Y_{\bar2}Y^{\bar2}, Y^{\bar2}Y_{\bar2}Y_{\bar1}Y_0 }{ E_{B'}^Z, Z^{\bar4}, Z^{\bar2}Z^{\bar2}, Z^0 }}
	\\	\multirow{2}{*}{$\ISG_{5}$} & \multirow{2}{*}{$\vcenter{\hbox{\includegraphics[width=70mm]{HHzz_c6t_isg5.pdf}}}$}
	\\\cline{3-3}	&& \multirow{2}{*}{\localrev{ E_{B'}^Z, Z^{\bar4}, Z^{\bar2}Z^{\bar2}, Z^0 }{ \calS_0, P_G^{Z/X}, X_{\bar4}X_{\bar3}X_{\bar2}X^{\bar2}, X^{\bar2}X_{\bar2}X_{\bar1}X_0 }{ E_{R'}^X, X^{\bar4}, X^{\bar2}, X^0 }}
	\\	\multirow{2}{*}{$\ISG_{0}$} & \multirow{2}{*}{$\vcenter{\hbox{\includegraphics[width=70mm]{HHzz_c6t_isg0.pdf}}}$}
	\\\cline{3-3}	&&
	\end{tabular}
\end{center}
Because the boundary is gapped, only one type of semi-infinite string operator~\eqref{eq:HH_stringOp} may terminate at the edge.
\begin{itemize}[itemsep=1pt, parsep=0pt, topsep=2pt]
\item	For $\ISG_0$,
	$\mathrm{Str}_{G/B}^X$ may terminate at the boundary,
	but $\mathrm{Str}_{R}^{Y/Z}$ cannot.
\item	For $\ISG_1$,
	$\mathrm{Str}_{G}^{Z/X}$ may terminate at the boundary,
	but $\mathrm{Str}_{B/R}^Y$ cannot.
\item	For $\ISG_2$,
	$\mathrm{Str}_{R/G}^Z$ may terminate at the boundary,
	but $\mathrm{Str}_{B}^{X/Y}$ cannot.
\end{itemize}
The termination rules for $\ISG_{t+3}$ are swapped in comparison to $\ISG_t$.

\newpage
\tocless\subsubsection{Gapped boundary with one ancilla.}
\insertlabel{app:HHzz:1anc}{\thesubsection{}.\arabic{subsubsection}}
This is a period-6 variant of circuit from~\S\ref{sec:HHzz_gapped}.
We place an ancilla qubit over every site in $6\ZZ$.
There are no nontrivial local logical operators.
\vspace{-1ex}
\begin{center} \renewcommand{\arraystretch}{2.2}
	\begin{tabular}{r @{\;\;} c @{\quad} l }
	\\	& \HHzzSiteNumbering & reversibility
	\\	\multirow{2}{*}{$\ISG_{0}$} & \multirow{2}{*}{$\vcenter{\hbox{\includegraphics[width=70mm]{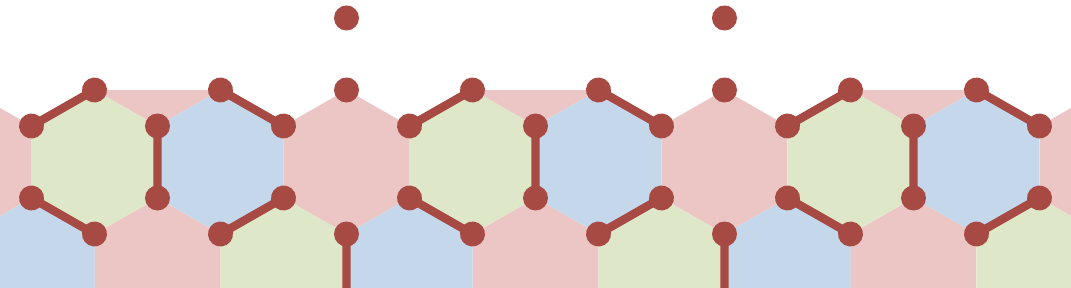}}}$}  &
	\\\cline{3-3}	&& \multirow{2}{*}{\localrev{ E_{R'}^X, X^0 }{ \calS_0, P_B^{X/Y}, X_0X_1X_2X^0, X^0X_2X_3X_4 }{ E_{G'}^Y, Y_2Y^0 }}
	\\	\multirow{2}{*}{$\ISG_{1}$} & \multirow{2}{*}{$\vcenter{\hbox{\includegraphics[width=70mm]{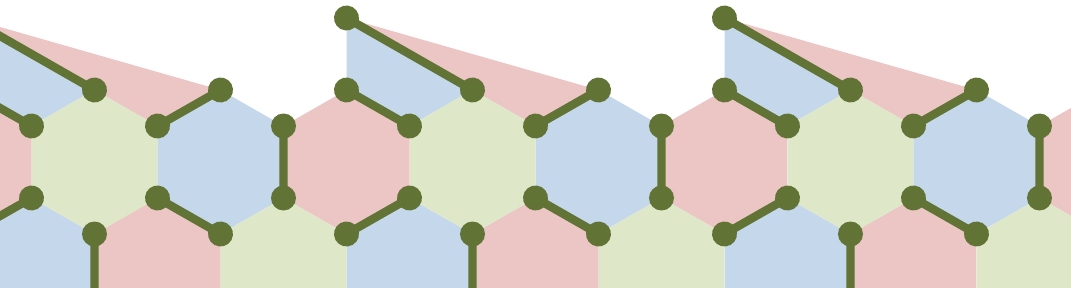}}}$}
	\\\cline{3-3}	&& \multirow{2}{*}{\localrev{ E_{G'}^Y, Y_2Y^0 }{ \calS_0, P_R^{Y/Z}, Z_0Z_1Z_2Z^0, Z^0Z_2Z_3Z_4 }{ E_{B'}^Z, Z^0 }}
	\\	\multirow{2}{*}{$\ISG_{2}$} & \multirow{2}{*}{$\vcenter{\hbox{\includegraphics[width=70mm]{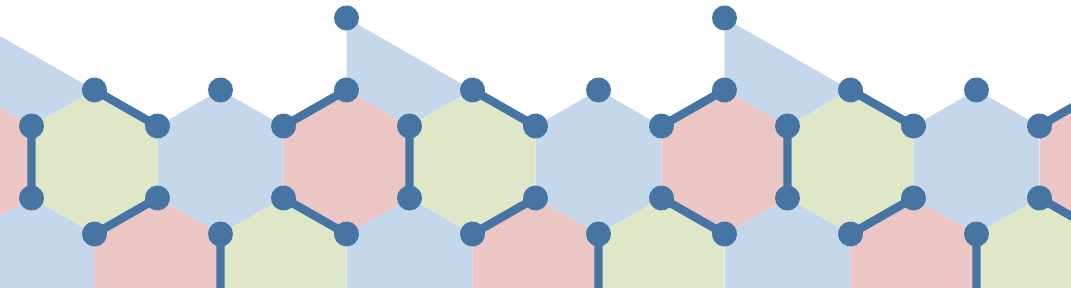}}}$}
	\\\cline{3-3}	&& \multirow{2}{*}{\localrev{ E_{B'}^Z, Z^0 }{ \calS_0, P_G^{Z/X}, Z_{\bar2}Z_{\bar1}Z_0Z^0, Z^0Z_0Z_1Z_2 }{ E_{R'}^X, X_0X^0 }}
	\\	\multirow{2}{*}{$\ISG_{3}$} & \multirow{2}{*}{$\vcenter{\hbox{\includegraphics[width=70mm]{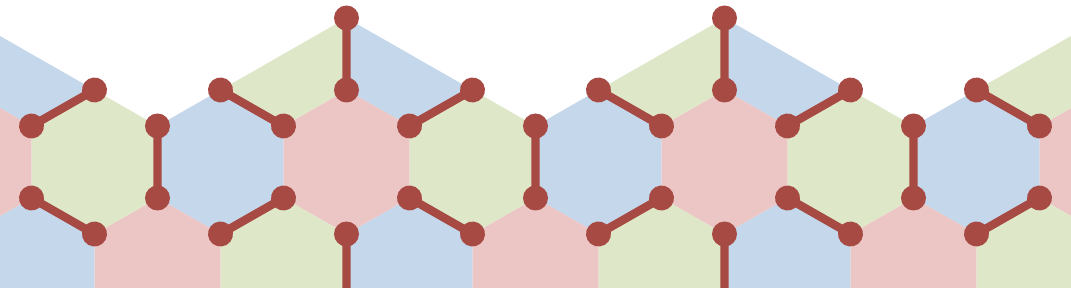}}}$}
	\\\cline{3-3}	&& \multirow{2}{*}{\localrev{ E_{R'}^X, X_0X^0 }{ \calS_0, P_B^{X/Y}, Y_{\bar2}Y_{\bar1}Y_0Y^0, Y^0Y_0Y_1Y_2 }{ E_{G'}^Y, Y^0 }}
	\\	\multirow{2}{*}{$\ISG_{4}$} & \multirow{2}{*}{$\vcenter{\hbox{\includegraphics[width=70mm]{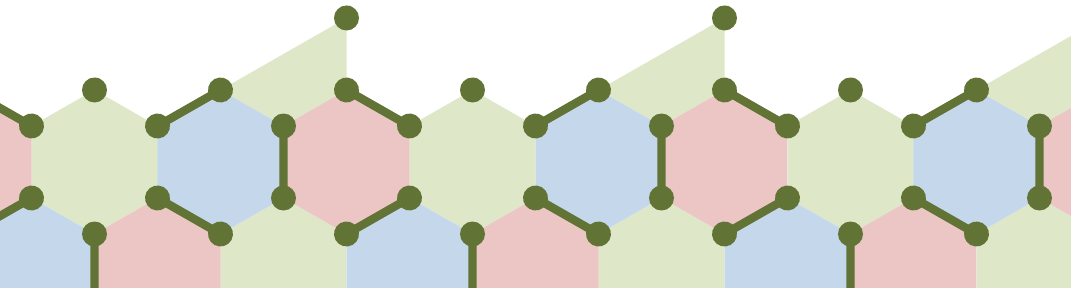}}}$}
	\\\cline{3-3}	&& \multirow{2}{*}{\localrev{ E_{G'}^Y, Y^0 }{ \calS_0, P_R^{Y/Z}, Y_{\bar4}Y_{\bar3}Y_{\bar2}Y^0, Y^0Y_{\bar2}Y_{\bar1}Y_0 }{ E_{B'}^Z, Z_{\bar2}Z^0 }}
	\\	\multirow{2}{*}{$\ISG_{5}$} & \multirow{2}{*}{$\vcenter{\hbox{\includegraphics[width=70mm]{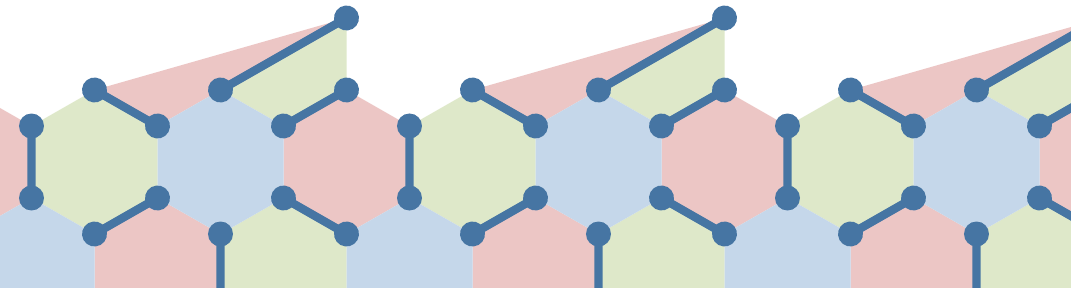}}}$}
	\\\cline{3-3}	&& \multirow{2}{*}{\localrev{ E_{B'}^Z, Z_{\bar2}Z^0 }{ \calS_0, P_G^{Z/X}, X_{\bar4}X_{\bar3}X_{\bar2}X^0, X^0X_{\bar2}X_{\bar1}X_0 }{ E_{R'}^X, X^0 }}
	\\	\multirow{2}{*}{$\ISG_{0}$} & \multirow{2}{*}{$\vcenter{\hbox{\includegraphics[width=70mm]{HHzz_c6s_isg0.pdf}}}$}
	\\\cline{3-3}	&&
	\end{tabular}
\end{center}
This model is nearly identical to previous case.
The termination rules for semi-infinite string operator are identical.

\newpage
\tocless\subsection{Various boundary circuits of the armchair edge}

\tocless\subsubsection{Gapped armchair boundary}
\insertlabel{app:HHac:gapped}{\thesubsection{}.\arabic{subsubsection}}
Here we present a measurement sequence along am armchair edge with no nontrivial logical operators.
The measurement circuit consists of 1- and 2-body measurements only,
and may be relevant to quantum error correction.
We leave the verification of local reversibility up to the reader.
\begin{align*} \renewcommand{\arraystretch}{1.9}
	\begin{array}{l @{\quad} l c c}
		&& \renewcommand{\arraystretch}{1} \begin{tabular}{c} allowed string \\[-0.9ex] termination \end{tabular}
		& \renewcommand{\arraystretch}{1} \begin{tabular}{c} disallowed string \\[-0.9ex] termination \end{tabular}
	\\	\ISG_{0} & \vcenter{\hbox{\includegraphics[width=70mm]{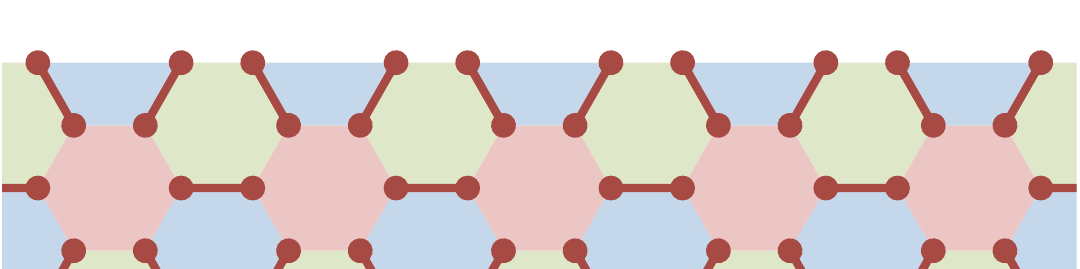}}} & \mathrm{Str}_{R}^{Y/Z} & \mathrm{Str}_{G/B}^X
	\\	\ISG_{1} & \vcenter{\hbox{\includegraphics[width=70mm]{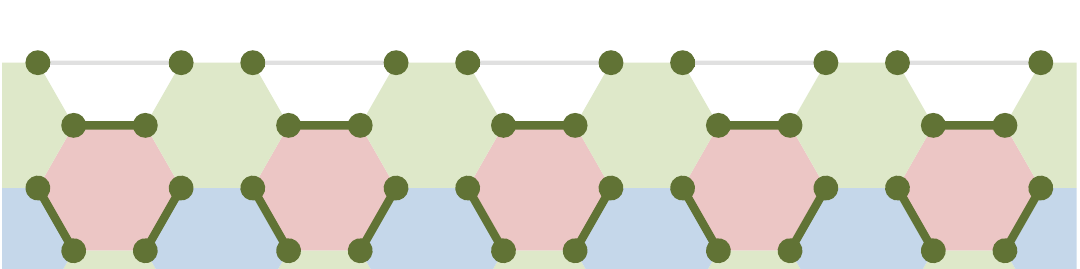}}} & \mathrm{Str}_{B/R}^Y & \mathrm{Str}_{G}^{Z/X}
	\\	\ISG_{2} & \vcenter{\hbox{\includegraphics[width=70mm]{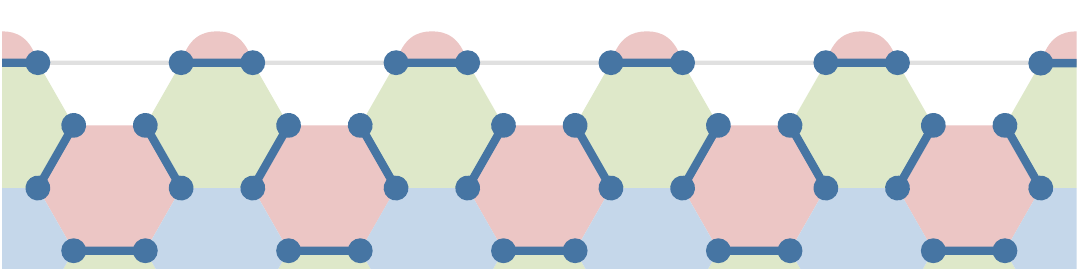}}} & \mathrm{Str}_{B}^{X/Y} & \mathrm{Str}_{R/G}^Z
	\\	\ISG_{3} & \vcenter{\hbox{\includegraphics[width=70mm]{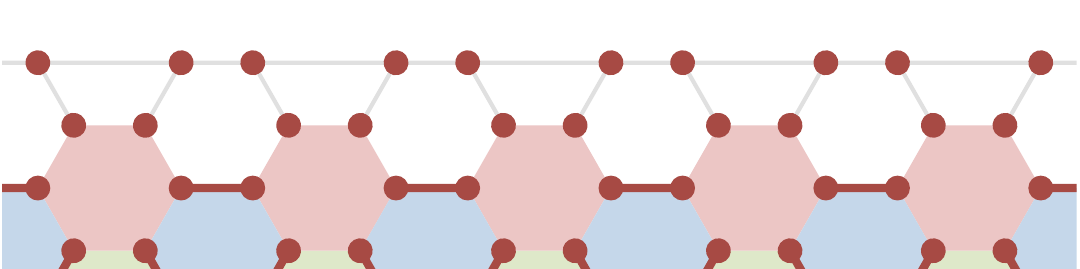}}} & \mathrm{Str}_{G/B}^X & \mathrm{Str}_{R}^{Y/Z}
	\\	\ISG_{4} & \vcenter{\hbox{\includegraphics[width=70mm]{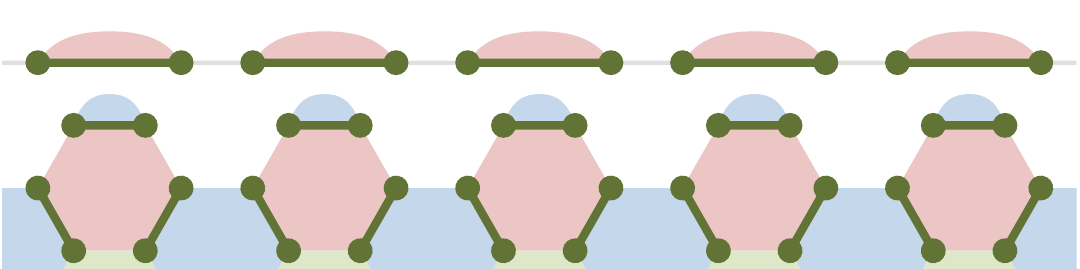}}} & \mathrm{Str}_{G}^{Z/X} & \mathrm{Str}_{B/R}^Y
	\\	\ISG_{5} & \vcenter{\hbox{\includegraphics[width=70mm]{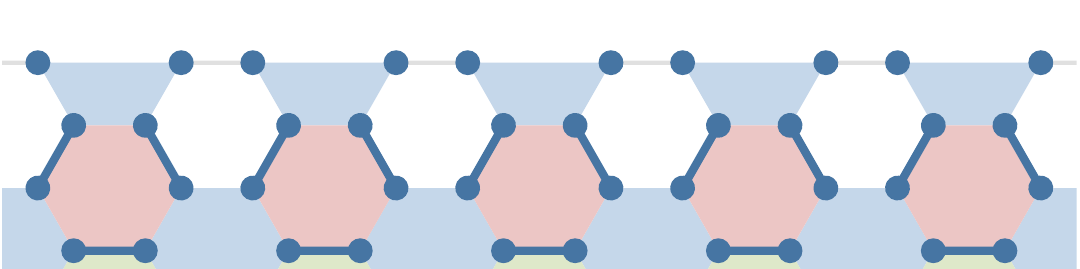}}} & \mathrm{Str}_{R/G}^Z & \mathrm{Str}_{B}^{X/Y}
	\end{array}
\end{align*}

\newpage
\tocless\subsubsection{Armchair boundary II}
\insertlabel{app:HHac:p3r}{\thesubsection{}.\arabic{subsubsection}}
This is a period-3 circuit along the armchair edge.
There are 2 independent logical generators per unit cell, and the MQCA index is $+\half$.
\begin{align*} \renewcommand{\arraystretch}{1.9}
	\begin{array}{l @{\quad} l}
		\ISG_{0} & \vcenter{\hbox{\includegraphics[width=70mm]{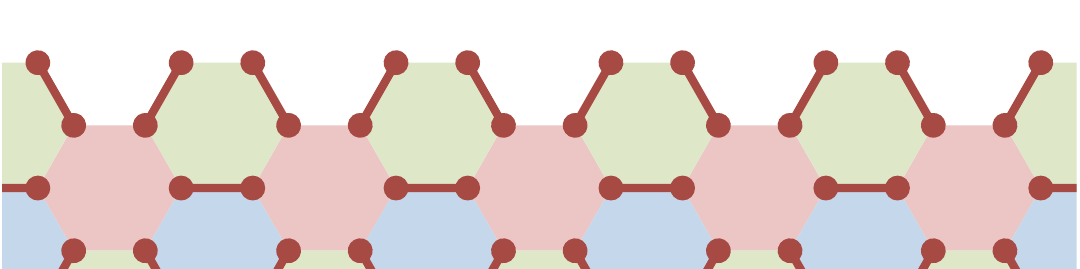}}}
	\\	\ISG_{1} & \vcenter{\hbox{\includegraphics[width=70mm]{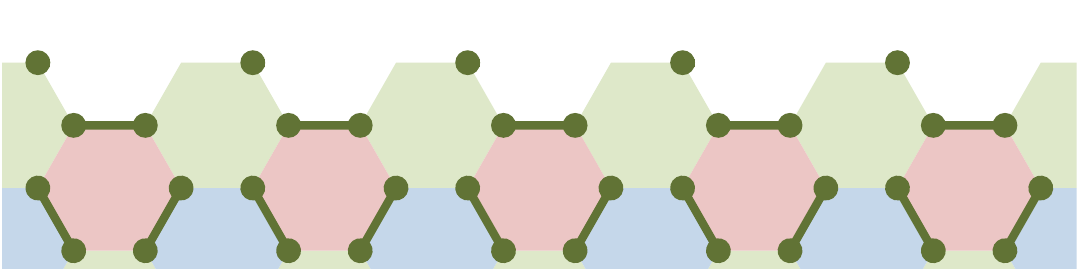}}}
	\\	\ISG_{2} & \vcenter{\hbox{\includegraphics[width=70mm]{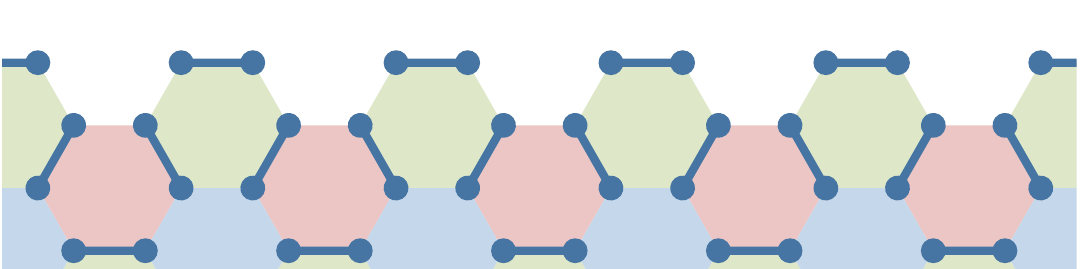}}}
	\end{array}
\end{align*}
The logicals of $\ISG_0$ forms a \MCAlg{} generated by $L_{2n} = Z_{6n-3} Z_{6n-2} Z_{6n-1} Z_{6n+0}$ and $L_{2n+1} = X_{6n+0} X_{6n+1}$.
Under a cycle, the MQCA maps
\begin{align}
	L_{2n} &\mapsto L_{2n+1} \,, & L_{2n-1} &\mapsto L_{2n-1} L_{2n} L_{2n+1} \,,
\end{align}
exhibiting operator size growth.
With the assignment of the logicals to fermionic operators via $L_{2n} = \ii \gamma_{2n-1} \gamma_{2n}$ and $L_{2n+1} = \ii \gamma_{2n} \gamma_{2n+1}$,
	this MQCA is equivalent to the fermionic transformation $\gamma_{2n-1} \mapsto \gamma_{2n+1}$, $\gamma_{2n} \mapsto \gamma_{2n}$.

\end{document}